\documentclass[a4paper,11pt]{book}

\usepackage[T1]{fontenc}
\usepackage[utf8]{inputenc}
\DeclareUnicodeCharacter{2009}{\,} 

\usepackage[square,numbers]{natbib} 
\usepackage[nottoc]{tocbibind}
\usepackage[redeflists]{IEEEtrantools}

\usepackage[french,german,english]{babel}

\usepackage{lmodern} 

\usepackage{fourier} 
\usepackage{setspace} 

\usepackage{graphicx,xcolor}
\graphicspath{{images/}}

\usepackage{subfigure}
\usepackage{booktabs}
\usepackage{lipsum}
\usepackage{microtype}
\usepackage{url}

\usepackage{fancyhdr}

\fancypagestyle{plain}{
	\fancyhf{}

	\fancyfoot[C]{\thepage}}
\fancypagestyle{addpagenumbersforpdfimports}{
	\fancyhead{}
	
	\fancyfoot{}
	\fancyfoot[C]{\thepage}
}

\usepackage{listings}
\usepackage{hyperref}
\hypersetup{pdfborder={0 0 0},
	colorlinks=true,
	linkcolor=black,
	citecolor=black,
	urlcolor=black}
\ifpdf
\usepackage[final]{pdfpages}
\else
\usepackage{calc}
\usepackage{breakurl}
\usepackage[nlwarning=false]{hypdvips}
\usepackage{backref}
\renewcommand*{\backref}[1]{}
\fi
\usepackage{bookmark}

\makeatletter
\renewcommand\@pnumwidth{20pt}
\makeatother

\makeatletter
\def\cleardoublepage{\clearpage\if@twoside \ifodd\c@page\else
    \hbox{}
    \thispagestyle{empty}
    \newpage
    \if@twocolumn\hbox{}\newpage\fi\fi\fi}
\makeatother \clearpage{\pagestyle{plain}\cleardoublepage}

\usepackage{color}
\usepackage{tikz}
\usepackage[explicit]{titlesec}
\newcommand*\chapterlabel{}
\titleformat{\chapter}[display]  
	{\normalfont\bfseries\Huge} 
	{\gdef\chapterlabel{\thechapter\ }}     
 	{0pt} 
 	  {\begin{tikzpicture}[remember picture,overlay]
    \node[yshift=-8cm] at (current page.north west)
      {\begin{tikzpicture}[remember picture, overlay]
        \node[anchor=north east,yshift=-7.2cm,xshift=34mm,minimum height=30mm,inner sep=0mm] at (current page.north west)
        {\parbox[top][30mm][t]{15mm}{\raggedleft \rule{0cm}{0.6cm}\color{black}\chapterlabel}};  
        \node[anchor=north west,yshift=-7.2cm,xshift=37mm,text width=\textwidth,minimum height=30mm,inner sep=0mm] at (current page.north west)
              {\parbox[top][30mm][t]{\textwidth}{\rule{0cm}{0.6cm}\color{black}#1}};
       \end{tikzpicture}
      };
   \end{tikzpicture}
   \gdef\chapterlabel{}
  } 
\titlespacing*{name=\chapter,numberless}{-3.7cm}{83.2pt-\parskip}{-3.2pt+\parskip}
\titlespacing*{\chapter}{-3.7cm}{50pt-\parskip-\parskip}{30pt+\parskip+\parskip}
\titlespacing*{\section}{0pt}{13.2pt}{1em-\parskip}  
\titlespacing*{\subsection}{0pt}{13.2pt}{1em-\parskip}
\titlespacing*{\subsubsection}{0pt}{13.2pt}{1em-\parskip}
\titlespacing*{\paragraph}{0pt}{13.2pt}{1em-\parskip}

\newcounter{myparts}
\newcommand*\partlabel{}
\titleformat{\part}[display]  
	{\normalfont\bfseries\Huge} 
	{\gdef\partlabel{\thepart\ }}     
 	{0pt} 
 	  {\ifpdf\setlength{\unitlength}{20mm}\else\setlength{\unitlength}{0mm}\fi
	  \addtocounter{myparts}{1}
	  \begin{tikzpicture}[remember picture,overlay]
    \node[anchor=north west,xshift=-65mm,yshift=-6.9cm-\value{myparts}*20mm] at (current page.north east) 
      {\begin{tikzpicture}[remember picture, overlay]
        \draw[fill=black] (0,0) rectangle(62mm,20mm);   
        \node[anchor=north west,yshift=-6.1cm-\value{myparts}*\unitlength,xshift=-60.5mm,minimum height=30mm,inner sep=0mm] at (current page.north east)
        {\parbox[top][30mm][t]{55mm}{\raggedright \color{white}Part \partlabel \rule{0cm}{0.6cm}}};  
        \node[anchor=north east,yshift=-6.1cm-\value{myparts}*\unitlength,xshift=-63.5mm,text width=\textwidth,minimum height=30mm,inner sep=0mm] at (current page.north east)
              {\parbox[top][30mm][t]{\textwidth}{\raggedleft \rule{0cm}{0.6cm}\color{black}#1}};
       \end{tikzpicture}
      };
   \end{tikzpicture}
   \gdef\partlabel{}
  } 
\titlespacing*{\part}{11.06cm}{26.4pt-\parskip-\parskip}{0pt}

\usepackage{amsmath}
\usepackage{amsfonts}
\usepackage{amssymb}
\usepackage{mathtools}
\makeatletter
\def\resetMathstrut@{%
  \setbox\z@\hbox{%
    \mathchardef\@tempa\mathcode`\(\relax
      \def\@tempb##1"##2##3{\the\textfont"##3\char"}%
      \expandafter\@tempb\meaning\@tempa \relax
  }%
  \ht\Mathstrutbox@1.2\ht\z@ \dp\Mathstrutbox@1.2\dp\z@
}
\makeatother

\newcommand\blfootnote[1]{%
	\begingroup
	\renewcommand\thefootnote{}\footnote{#1}%
	\addtocounter{footnote}{-1}%
	\endgroup
}

\usepackage{afterpage}
\usepackage{svg}
\usepackage{tabu}   
\usepackage{amsthm}
\usepackage[super]{nth}
\theoremstyle{definition}
\newtheorem{definition}{Definition}
\newtheorem{theorem}{Theorem}
\newcommand{\norm}[1]{\left\lVert#1\right\rVert}
\usepackage{upgreek}
\usepackage[ruled,vlined]{algorithm2e}
\SetKwInput{KwInput}{Inputs}
\SetKwInput{KwOutput}{Output}   
\usepackage{longtable}
\usepackage[acronym]{glossaries}
\usepackage[graphicx]{realboxes}
\usepackage{acronym}
\usepackage{array}
\usepackage{booktabs}
\usepackage{tabularx}
\usepackage{makecell}
\usepackage{adjustbox}
\usepackage{rotating}
\usepackage[flushleft]{threeparttable}
\usepackage{pifont}
\usepackage{amssymb}
\usepackage{pbox}
\usepackage{color, colortbl}
\usepackage{lscape}

\usepackage{colorprofiles}
\usepackage[a-2b,mathxmp]{pdfx}[2018/12/23]
\makeatletter
\AtBeginDocument{\let\mathaccentV\AMS@mathaccentV}
\makeatother
\hypersetup{pdfstartview=}
\begin{document}

\setlength{\parindent}{10pt}
\setlength{\parskip}{0pt} 
\frontmatter
\begin{titlepage}
\begin{center}

\null\vspace{1cm}
{\huge \textbf{Towards Everyday Virtual Reality} \\[0.3cm] \textbf{through Eye Tracking}} \\[2.6cm]

{\large \textbf{Dissertation}} \\[0.3cm]
{\large der Mathematisch-Naturwissenschaftlichen Fakult{\"a}t}\\[0.3cm]
{\large der Eberhard Karls Universit{\"a}t T{\"u}bingen}\\[0.3cm]
{\large zur Erlangung des Grades eines}\\[0.3cm]
{\large Doktors der Naturwissenschaften}\\[0.3cm]
{\large (Dr. rer. nat.)}\\[3cm]

vorgelegt von\\[0.3cm]
{\large M. Sc. Efe BOZKIR}\\[0.3cm]
{\large aus Kar{\c{s}}{\i}yaka, T{\"u}rkei}\\[3cm]

\vfill
{\large T{\"u}bingen}\\[0.3cm]
{\large 2021}\\[0.3cm]

\end{center}
\end{titlepage}

\thispagestyle{empty}

\mbox{}

\vspace{16.8cm}
\noindent Gedruckt mit Genehmigung der Mathematisch-Naturwissenschaftlichen Fakult{\"a}t der \\Eberhard Karls Universit{\"a}t T{\"u}bingen.\\[0.5cm]

\noindent \hspace{-0.2cm}\begin{tabular}{l l}
Tag der m\"undlichen Qualifikation:& 25.02.2022 \vspace{0.2cm}\\
Dekan:& Prof.~Dr.~Thilo~Stehle \vspace{0.2cm}\\
1.~Berichterstatterin:& Prof.~Dr.~Enkelejda Kasneci \vspace{0.2cm}\\
2.~Berichterstatter:& Jun.-Prof.~Dr.~Michael Krone\\
\end{tabular}

\cleardoublepage
\thispagestyle{empty}

\vspace*{5cm}

\begin{raggedleft}
    	\textit{``Our true mentor in life is science.''}\\
     -- Mustafa Kemal Atatürk\\
\end{raggedleft}

\vspace{4.25cm}

\begin{center}
    To my family
\end{center}

\setcounter{page}{0}
\chapter*{Acknowledgments}
\vspace{5pt}
\markboth{Acknowledgments}{Acknowledgments}
\addcontentsline{toc}{chapter}{Acknowledgments}

Firstly, I would like to thank my supervisor, Prof. Dr. Enkelejda Kasneci for accepting me to the Chair for Human-Computer Interaction (former Perception Engineering), guiding me throughout my doctoral studies, and for being one of the most supportive people that I have ever worked with through my entire academic and professional career. Special and sincere thanks to Jun.-Prof. Dr. Michael Krone, Prof. Dr. Oliver Bringmann, and Prof. Dr. Richard Göllner for evaluating my work. Furthermore, I thank Margot Reimold for helping me with all the bureaucratic burdens and making my life easier.\\

\vspace{-5px}

During this work, I was quite lucky to have plenty of pleasant and fruitful collaborations, including interdisciplinary ones. I particularly thank Prof. Dr. Richard Göllner, Dr. Lisa Hasenbein, Philipp Stark, and Joseph Ferdinand for the engaging collaborations on the classroom studies. Furthermore, I thank Prof. Dr. Rafael F. Schaefer, Prof. Dr. Nico Pfeifer, Dr. Onur Günlü, Dr. Mete Akgün, and Ali Burak Ünal for the delightful discussions and work on the privacy-related topics. In addition, I thank our collaborators in Greece in the DAAD project, especially Prof. Dr. Athanassios Skodras.\\

\vspace{-5px}

Without an enjoyable work environment, I would not be successful. Therefore, I thank the past and present members of Chair for Human-Computer Interaction, many of them becoming my collaborators in the meantime, including Dr. Wolfgang Fuhl, Dr. Shahram Eivazi, Dr. Thomas Kübler, Dr. Nora Castner, Dr. David Geisler, Dr. Thiago Santini, Yao Rong, Benedikt Hosp, Daniel Weber, Björn Severitt, and others. Special thanks go to my co-workers in the office, Dr. Ömer Sümer, Hong Gao, and Babette Bühler, for the privilege of beneficial discussions and for sharing nice moments. Furthermore, I thank the members of the Data Science \& Analytics Research Group for the nice ambiance and especially Martin Pawelczyk for fruitful brainstormings on different ideas.\\

\vspace{-5px}

A great appreciation goes to my friends in Tübingen outside the lab for all the fun, silly jokes, and unforgettable memories together. I particularly thank Caner Bağcı, Dr. Murat Seçkin Ayhan, Dr. Mete Akgün, Ali Burak Ünal, Mehmet Direnç Mungan, Emre Barış Karaaslan, Dr. Ezgi Süheyla Karaaslan, and Bilge Sürün. Moreover, I thank my friends in Munich, especially Selin Kenet, Umut Kaya, and Ümit Suat Mayadalı. Even though we were physically away from each other during all that time, I very much enjoyed our online games and chats, especially during the pandemic. You made my life more liveable during weird times.\\

\vspace{-5px}

I have been living away from my family for a long time but thanks to their warm welcome in Tübingen from the very first day, I have again remembered the family feeling with them. My thanks go to my cousin Dr. Acun Papakçı, Dr. Elif Kösedağı, Arinna, and Taru.\\

\vspace{-5px}

Being the only child of my parents led me to have the most incredible friends ever. Even though I live far away from all of you, I feel your sister- and brotherhood all the time. My thanks go to Esra İçer, Dr. Arda Erdut, Utku Barış Özsoy, Kutay Yücetürk, Yuşa Öz, Ömer Kaya, Emin Melih Özsoy, Levent Budakoğlu, Altuğ Bayram, and Mertcan Yığın. Moreover, I sincerely thank my cousin Batuhan Sarıoğlu for being there all the years, even though our paths fell physically apart after studying together in Istanbul. Additionally, I thank uncle Ercüment and Ayhan, for their encouragement and trust even starting from my high school years. I really appreciate it.\\

\vspace{-5px}

Lastly and most importantly, I thank my mother Neşe and my father Fevzi for their unconditional love, support, encouragement, and confidence. These literally mean the world to me.

\bigskip

\hfill Efe~BOZKIR

\cleardoublepage
\chapter*{Abstract}
\vspace{5pt}
\markboth{Abstract}{Abstract}
\addcontentsline{toc}{chapter}{Abstract} 

With developments in computer graphics, hardware technology, perception engineering, and human-computer interaction, virtual reality and virtual environments are becoming more integrated into our daily lives. Head-mounted displays, however, are still not used as frequently as other mobile devices such as smart phones and watches. With increased usage of this technology and the acclimation of humans to virtual application scenarios, it is possible that in the near future an everyday virtual reality paradigm will be realized.\\

\vspace{-5px}

When considering the marriage of everyday virtual reality and head-mounted displays, eye tracking is an emerging technology that helps to assess human behaviors in a real time and non-intrusive way. Still, multiple aspects need to be researched before these technologies become widely available in daily life. Firstly, attention and cognition models in everyday scenarios should be thoroughly understood. Secondly, as eyes are related to visual biometrics, privacy preserving methodologies are necessary. Lastly, instead of studies or applications utilizing limited human participants with relatively homogeneous characteristics, protocols and use-cases for making such technology more accessible should be essential.\\

\vspace{-5px}

In this work, taking the aforementioned points into account, a significant scientific push towards everyday virtual reality has been completed with three main research contributions. Human visual attention and cognition have been researched in virtual reality in two different domains, including education and driving. Research in the education domain has focused on the effects of different classroom manipulations on human visual behaviors, whereas research in the driving domain has targeted safety related issues and gaze-guidance. The user studies in both domains show that eye movements offer significant implications for these everyday setups. The second substantial contribution focuses on privacy preserving eye tracking for the eye movement data that is gathered from head-mounted displays. This includes differential privacy, taking temporal correlations of eye movement signals into account, and privacy preserving gaze estimation task by utilizing a randomized encoding-based framework that uses eye landmarks. The results of both works have indicated that privacy considerations are possible by keeping utility in a reasonable range. Even though few works have focused on this aspect of eye tracking until now, more research is necessary to support everyday virtual reality. As a final significant contribution, a blockchain- and smart contract-based eye tracking data collection protocol for virtual reality is proposed to make virtual reality more accessible. The findings present valuable insights for everyday virtual reality and advance the state-of-the-art in several directions.\\

\begin{otherlanguage}{german}
\cleardoublepage
\chapter*{Zusammenfassung}
\vspace{5pt}
\markboth{Zusammenfassung}{Zusammenfassung}
\addcontentsline{toc}{chapter}{Zusammenfassung}

Durch Entwicklungen in den Bereichen Computergrafik, Hardwaretechnologie, Perception Engineering und Mensch-Computer Interaktion, werden Virtual Reality und virtuelle Umgebungen immer mehr in unser tägliches Leben integriert. Head-Mounted Displays werden jedoch im Vergleich zu anderen mobilen Geräten, wie Smartphones und Smartwatches, noch nicht so häufig genutzt. Mit zunehmender Nutzung dieser Technologie und der Gewöhnung von Menschen an virtuelle Anwendungsszenarien ist es wahrscheinlich, dass in naher Zukunft ein alltägliches Virtual-Reality-Paradigma realisiert wird.\\ 

\vspace{-5px}

Im Hinblick auf die Kombination von alltäglicher Virtual Reality und Head-Mounted-Displays, ist Eye Tracking eine neue Technologie, die es ermöglicht, menschliches Verhalten in Echtzeit und nicht-intrusiv zu messen. Bevor diese Technologien in großem Umfang im Alltag eingesetzt werden können, müssen jedoch noch zahlreiche Aspekte genauer erforscht werden. Zunächst sollten Aufmerksamkeits- und Kognitionsmodelle in Alltagsszenarien genau verstanden werden. Des Weiteren sind Maßnahmen zur Wahrung der Privatsphäre notwendig, da die Augen mit visuellen biometrischen Indikatoren assoziiert sind. Zuletzt sollten anstelle von Studien oder Anwendungen, die sich auf eine begrenzte Anzahl menschlicher Teilnehmer mit relativ homogenen Merkmalen stützen, Protokolle und Anwendungsfälle für eine bessere Zugänglichkeit dieser Technologie von wesentlicher Bedeutung sein.\\

\vspace{-5px}

In dieser Arbeit wurde unter Berücksichtigung der oben genannten Punkte ein bedeutender wissenschaftlicher Vorstoß mit drei zentralen Forschungsbeiträgen in Richtung alltäglicher Virtual Reality unternommen. Menschliche visuelle Aufmerksamkeit und Kognition innerhalb von Virtual Reality wurden in zwei unterschiedlichen Bereichen, Bildung und Autofahren, erforscht. Die Forschung im Bildungsbereich konzentrierte sich auf die Auswirkungen verschiedener Manipulationen im Klassenraum auf das menschliche Sehverhalten, während die Forschung im Bereich des Autofahrens auf sicherheitsrelevante Fragen und Blickführung abzielte. Die Nutzerstudien in beiden Bereichen zeigen, dass Blickbewegungen signifikante Implikationen für diese alltäglichen Situationen haben. Der zweite wesentliche Beitrag fokussiert sich auf Privatsphäre bewahrendes Eye Tracking für Blickbewegungsdaten von Head-Mounted Displays. Dies beinhaltet Differential Privacy, welche zeitliche Korrelationen von Blickbewegungssignalen berücksichtigt und Privatsphäre wahrende Blickschätzung durch Verwendung eines auf randomisiertem Encoding basierenden Frameworks, welches Augenreferenzunkte verwendet. Die Ergebnisse beider Arbeiten zeigen, dass die Wahrung der Privatsphäre möglich ist und gleichzeitig der Nutzen in einem akzeptablen Bereich bleibt. Wenngleich es bisher nur wenig Forschung zu diesem Aspekt von Eye Tracking gibt, ist weitere Forschung notwendig, um den alltäglichen Gebrauch von Virtual Reality zu ermöglichen. Als letzter signifikanter Beitrag, wurde ein Blockchain- und Smart Contract-basiertes Protokoll zur Eye Tracking Datenerhebung für Virtual Reality vorgeschlagen, um Virtual Reality besser zugänglich zu machen. Die Ergebnisse liefern wertvolle Erkenntnisse für alltägliche Nutzung von Virtual Reality und treiben den aktuellen Stand der Forschung in mehrere Richtungen voran.

\end{otherlanguage}

\cleardoublepage
\pdfbookmark{\contentsname}{toc}
\tableofcontents

\cleardoublepage
\phantomsection
\listoffigures

\cleardoublepage
\phantomsection
\listoftables

\cleardoublepage
\phantomsection
\addcontentsline{toc}{chapter}{List of Abbreviations}
\chapter*{List of Abbreviations}
\markboth{List of Abbreviations}{List of Abbreviations}
\begin{acronym}
\vspace{10pt}
\acro{ADHD}{\quad Attention-Deficit/Hyperactivity Disorder}
\acro{AMT}{\quad Amazon Mechanical Turk}
\acro{ANOVA}{\quad Analysis of Variance}
\acro{AR}{\quad Augmented Reality}
\acro{ART}{\quad Aligned Rank Transform}
\acro{CFPA}{\quad Chunk-based Fourier Perturbation Algorithm}
\acro{CPT}{\quad Continuous Performance Task}
\acro{DARE}{\quad Decomposable and Affine Randomized Encoding}
\acro{DCFPA}{\quad Difference- and Chunk-based Fourier Perturbation Algorithm}
\acro{DFT}{\quad Discrete Fourier Transform}
\acro{DP}{\quad Differential Privacy}
\acro{DT}{\quad Decision Tree}
\acro{EEG}{\quad Electroencephalography}
\acro{ETH}{\quad Ether}
\acro{FOV}{\quad Field of View}
\acro{FPA}{\quad Fourier Perturbation Algorithm}
\acro{GAN}{\quad Generative Adversarial Network}
\acro{GDPR}{\quad General Data Protection Regulation}
\acro{HCI}{\quad Human-Computer Interaction} 
\acro{HDD}{\quad Head-down Display} 
\acro{HMAC}{\quad Keyed-Hashing for Message Authentication} 
\acro{HMD}{\quad Head-mounted Display} 
\acro{HUD}{\quad Head-up Display} 
\acro{I-VT}{\quad Velocity-Threshold Identification} 
\acro{IDFT}{\quad Inverse Discrete Fourier Transform} 
\acro{IVR}{\quad Immersive Virtual Reality} 
\acro{k-NN}{\quad k-Nearest Neighbor}
\acro{LPA}{\quad Laplace Perturbation Algorithm}
\acro{MR}{\quad Mixed Reality}
\acro{NMSE}{\quad Normalized Mean Square Error}
\acro{PDF}{\quad Probability Density Function}
\acro{OOI}{\quad Object-of-Interest}
\acro{RBF}{\quad Radial Basis Function}
\acro{RE}{\quad Randomized Encoding}
\acro{RF}{\quad Random Forest}
\acro{SHA}{\quad Secure Hash Algorithm}
\acro{SMC}{\quad Secure Multi-party Computation} 
\acro{SVM}{\quad Support Vector Machine} 
\acro{SVR}{\quad Support Vector Regression} 
\acro{TTC}{\quad Time-to-Collision}
\acro{VR}{\quad Virtual Reality}
\acro{XR}{\quad Extended Reality}
\end{acronym}

\setlength{\parskip}{1em}

\mainmatter
\cleardoublepage
\chapter{List of Publications}
\label{chapter_publications}
\vspace{-1em}
\section*{Accepted Publications Relevant to This Thesis}

\begin{enumerate}
	\item\label{publist_lbl_PLOSONE} \textbf{Efe Bozkir}, Onur Günlü, Wolfgang Fuhl, Rafael F. Schaefer, and Enkelejda Kasneci. Differential privacy for eye tracking with temporal correlations. \emph{PLoS ONE}, 16(8):e0255979, 2021. doi: 10.1371/journal.pone.0255979.

	\item\label{publist_lbl_CHI21} Hong Gao, \textbf{Efe Bozkir}, Lisa Hasenbein, Jens-Uwe Hahn, Richard
	Göllner, and Enkelejda Kasneci. Digital transformations of classrooms in virtual reality. In~\emph{Proceedings of the 2021 CHI Conference on Human Factors in Computing Systems}, New York, NY, USA, 2021. ACM. doi: 10.1145/3411764.3445596.
	
	\item\label{publist_lbl_VR21} \textbf{Efe Bozkir}, Philipp Stark, Hong Gao, Lisa Hasenbein, Jens-Uwe Hahn, Enkelejda Kasneci, and Richard Göllner. Exploiting object-of-interest information to understand attention in VR classrooms. In~\emph{2021 IEEE Virtual Reality and 3D User Interfaces (VR)}, New York, NY, USA, 2021. IEEE. doi: 10.1109/VR50410.2021.00085.
	
	\item\label{publist_lbl_AIVR20} \textbf{Efe Bozkir}, Shahram Eivazi, Mete Akgün, and Enkelejda Kasneci. Eye tracking data collection protocol for VR for remotely located subjects using blockchain and smart contracts. In~\emph{2020 IEEE International Conference on Artificial Intelligence and Virtual Reality (AIVR) Work-in-progress papers}, New York, NY, USA, 2020. IEEE. doi: 10.1109/AIVR50618.2020.00083.
	
	\item\label{publist_lbl_ETRA20} \textbf{Efe Bozkir}, Ali Burak Ünal, Mete Akgün, Enkelejda Kasneci, and Nico Pfeifer. Privacy preserving gaze estimation using synthetic images via a randomized encoding based framework. In~\emph{ACM Symposium on Eye Tracking Research and Applications}, New York, NY, USA, 2020. ACM. doi: 10.1145/3379156.3391364.
	
	\item\label{publist_lbl_SAP19} \textbf{Efe Bozkir}, David Geisler, and Enkelejda Kasneci. Assessment of driver attention during a safety critical situation in VR to generate VR-based training. In~\emph{ACM Symposium on Applied Perception 2019}, New York, NY, USA, 2019. ACM. doi: 10.1145/3343036.3343138.
	
	\item\label{publist_lbl_VRW19} \textbf{Efe Bozkir}, David Geisler, and Enkelejda Kasneci. Person independent, privacy preserving, and real time assessment of cognitive load using eye tracking in a virtual reality setup. In~\emph{2019 IEEE Conference on Virtual Reality and 3D User Interfaces (VR) Workshops}, New York, NY, USA, 2019. IEEE. doi: 10.1109/VR.2019.8797758.
	
\end{enumerate} 

\section*{Other Publications not Relevant to This Thesis}

\begin{enumerate}
\setcounter{enumi}{7}
	\item Wolfgang Fuhl, \textbf{Efe Bozkir}, and Enkelejda Kasneci. Reinforcement learning for the privacy preservation and manipulation of eye tracking data. In~\emph{Artificial Neural Networks and Machine Learning -- ICANN 2021}, Cham, Switzerland, 2021. Springer. doi: 10.1007/978-3-030-86380-7\_48
	
	\item Ömer Sümer, \textbf{Efe Bozkir}, Thomas Kübler, Sven Grüner, Sonja Utz,
	and Enkelejda Kasneci. FakeNewsPerception: An eye movement dataset on the perceived believability of news stories. \emph{Data in Brief}, 35, 2021. doi: 10.1016/j.dib.2021.106909.
	
	\item Wolfgang Fuhl, \textbf{Efe Bozkir}, Benedikt Hosp, Nora Castner, David Geisler, Thiago C. Santini, and Enkelejda Kasneci. ``Encodji: Encoding gaze data into emoji space for an amusing scanpath classification approach''. In~\emph{Proceedings of the 11th ACM Symposium on Eye Tracking Research \& Applications}, New York, NY, USA, 2019. ACM. doi: 10.1145/3314111.3323074.
	
\end{enumerate}

\newpage

\section{Scientific Contribution}
\label{scientific_contribution}
This work advances the state-of-the-art in multiple directions and is considered a significant research step towards the achievement of everyday virtual reality through eye tracking. Contributions include valuable insights on human visual attention and cognition studied in multiple domains, namely education and driving, privacy preserving manipulation of eye movements with differential privacy and a randomized encoding-based framework, and a versatile protocol employing blockchain and smart contracts for eye tracking data collection suitable for subjects that are located remotely in virtual reality.

Chapter~\ref{chapter_introduction} introduces each topic under the umbrella subject of everyday virtual reality. Seven scientific publications at renowned venues from 2019 to 2021 served as motivation for this work and are introduced in Chapter~\ref{chapter_motivation_findings} along with their fundamental methodologies and findings. Chapter~\ref{chapter_discussion} discusses the outcomes and provides an outlook on how to make virtual reality collectively more available in human everyday life.

\cleardoublepage

\chapter{Introduction} 
\label{chapter_introduction}
With recent developments in different fields of computer science, such as computer graphics, sensing technology, and artificial intelligence, and with the decreased cost of smart glasses and head-mounted displays (HMDs)~\cite{FT_oculus_price_cut}, virtual and augmented reality (VR/AR) are fast becoming integrated into our daily lives. It is estimated that the \acs{VR} headset market size will grow at an annual rate of $\approx 28\%$ from $2021$ to $2028$~\cite{VR_marketsize_research}. This indicates that humans will use such devices more regularly in variety of daily routines, including entertainment, education, and training. Currently, different devices with wide range of technical capabilities are available, from cheap, low-end options like Google Cardboard~\cite{google_card_board} to high-end devices like HTC Vive Pro Eye~\cite{htc_pro_eye_online}.

Virtual reality is defined by LaValle~\cite[p.~1]{VR_Book_Lavalle:2020} as ``inducing a targeted behavior in an organism by using sensory stimulation, while the organism has little or no awareness of the interference.''. In his definition, there are four main components including targeted behavior, organism, sensory stimulation, and awareness. Targeted behavior is essentially defined as the experience that the living organism is having. For instance, it could be the experience of a student attending a lecture in an immersive \acs{VR} classroom or a novice driver training in \acs{VR} for safety critical situations that could occur in the real world. According to LaValle~\cite{VR_Book_Lavalle:2020}, the organism could be any life form; however, in the context of this work, the humans are the focus. In addition, while sensory stimulation is carried out by integrating regular or alternative senses for humans, awareness of the experience is related to sense of presence.

Sense of presence is an important issue for virtual reality. It can be discussed through a philosophic perspective by taking Matrix\footnote{\textit{The Matrix Trilogy} is a science fiction trilogy that was written and directed by the Wachowskis~\cite{matrix_triology_link}.}-like utopic scenarios and environments into account. Alternative and virtual realities have been traced back to the \nth{18} century and are thought to originate with the writings of Immanuel Kant~\cite{sommerring_kant:1796} who discussed the realities that occur in one's mind, but differ from the real world~\cite{VR_Book_Lavalle:2020}. Later in \nth{20} century, Antonin Artaud used the term ``la réalité virtuelle'' to describe theatre as being similar to a second reality in his work~\cite{Artaud:1938_fr}. The terms ``artificial reality'' and later ``virtual reality'' in the technical domain were used by Myron Krueger and Jaron Lanier in 1960s and 1980s, respectively~\cite{krueger_artificial_reality_1983,metaphysics_of_vr_heim1993,burbules_2004}. Philosophical discussions of terminology and the utmost possibilities of virtual realities and their effects on humans aside, \acs{VR} researchers and practitioners have been working on different aspects of this technology, such as hardware, software, and human perception, for the last several decades. While we are still far from an extreme sense of presence in \acs{VR} with the current technology due to reasons such as low resolutions in \acs{VR} \acs{HMD}s, limited field-of-view (FOV), or possible cybersickness when \acs{HMD}s are being used, Jason Paul of NVIDIA has estimated that in 2017 we may be only two decades away from generating resolutions that match human eyes~\cite{nvidia_estimate_VR}. Working towards such a goal, research and development of \acs{VR} hardware and software go hand in hand with research of human physiology and perception, thus connecting the perception engineering discipline as a whole~\cite{VR_Book_Lavalle:2020}.

With continuous efforts in perception engineering, it is likely that \acs{VR} technology will be used more frequently in human daily life. Garner et al.~\cite{Garner2018_everydayVR} define everyday \acs{VR} as any kind of activity that people are linked to once a day. The authors provide real world examples as well as use-cases for classrooms, skills training, and workplaces. Apart from the applications, there is also a push in the \acs{VR} research community in this direction. In 2021, the Workshop on Everyday Virtual Reality was organized at the IEEE VR for the \nth{7} year in a row~\cite{ieee_wevr21}.

A number of applications and research questions in the \acs{VR} domain are evaluated by using pre- or post-tests and questionnaires such as presence and realism~\cite{schubert_presence,lombard_presence}, simulator sickness~\cite{SSQ_original_paper}, and mental workload~\cite{NASATLX_20years_later}. While these evaluation methods are well established and psychologically relevant, they do not offer insights on the temporal dynamics of visual attention and cognition during the use of \acs{VR} systems, a useful approach when everyday use-cases are considered. In particular, eye, head, and hand tracking sensors could be used for immersive \acs{VR} environments that are realized with \acs{HMD}s. Furthermore, the data acquired from these sensing modalities could be combined with self-reported measurements to gain more insights from users.

Eyes and their movements are of a particular interest since it is possible to analyze where people look at specific points in time along with the presented stimulus as long as an accurate estimation of eye regions and gaze is carried out. While in-the-wild scenarios for such tasks are more challenging due to illuminations or occlusions~\cite{fuhl_pupil_det_in_the_wild}, \acs{HMD}-based \acs{VR} setups offer more controlled conditions as the eye tracking sensors are located within the \acs{HMD}s. This enables not only eye images that can be recorded closer to eye regions, but also presents the possibility of configurations with controlled illumination that could be convenient for evaluating pupil related measurements for mental workload~\cite{beatty:1982}. With recent developments in sensing technologies, it is possible to have integrated eye trackers in modern high-end \acs{HMD}s (e.g., HTC Vive Pro Eye) or plug eye tracker sensors to such \acs{HMD}s~\cite{Kassner:2014:POS:2638728.2641695}. These have further eased the building of data collection pipelines for eye movements and the understanding of human visual attention using data from immersive virtual environments.

While it is possible to infer valuable information using eye movements in \acs{VR} environments, such as salient regions of the scene~\cite{sitzmann_saliency_tvcg18}, human intentions~\cite{David-John_intent_VR_etra21}, or forecasting eye fixations~\cite{fixationnet_bulling_tvcg21} which can be used for user assistive tasks, eyes include biometric information. For instance, personal authentication via iris textures is a well known approach~\cite{Daugman_1993, KUMAR20101016}. Additionally, eye movements~\cite{Zhang:2018:CAU:3178157.3161410} or the combination of aggregated eye movement features~\cite{Steil:2019:PPH:3314111.3319913} assist with biometric authentication. Zhang et al.~\cite{Zhang:2018:CAU:3178157.3161410} report that eye movement-based authentication schemes could be used with \acs{VR} devices for applications such as in-app purchases~\cite{george:seamlessAndSecureVR:2017}. In addition, Steil et al.~\cite{Steil:2019:PPH:3314111.3319913} found that people agree on sharing eye tracking data if a governmental health agency is available within the process or the data is used for research purposes. Taking authentication and personal identification possibilities into account, eye tracking data should be maintained in such a way that if authentication is not required for a task performed in the virtual environment, personal identities are still protected from adversaries. Still, other tasks which require eye movement data such as gaze guidance or foveated rendering should not be significantly affected due to privacy protection. Recently, in both Europe and in the US, data protection regulations have been legally introduced~\cite{EUdataregulations2018,CCPA_18}. With the increased amount of daily \acs{VR} applications and their usage in the commercial domain, it is foreseeable that the applications and scenarios that enable daily usage of \acs{VR} \acs{HMD}s along with biological signals leading to authentication such as eye movements will require dedicated privacy protection mechanisms.

Apart from the potential of eye movements in \acs{VR} and accompanying privacy concerns, another direction that would collectively increase the accessibility of \acs{VR} environments is the enabling of remote interaction including multiple people as opposed to single player-like application scenarios. This is still a challenge as \acs{VR} \acs{HMD}s are not used as frequently as other personal devices like mobile phones or smart watches. One reason may be that such devices are perceived as being solely for entertainment purposes and, even with their decreased costs, it is not straightforward to use them in the daily life. However, especially with occurrences such as the COVID-19 pandemic, the value of remote work and collaboration has become more appreciated. With the immersion that \acs{VR} setups provide, many domains could be digitally transformed. High-end \acs{HMD}s along with native \acs{VR} applications provide high quality resolutions and immersive and realistic environments. When multi-person setups and corresponding communication between people are considered, however, it is necessary to have communication via web-services and often by introducing third-party entities in order to enable such conditions. This creates opportunity for additional parties to manipulate data if the collected data is shared across parties. In addition, implementing and adapting such applications and systems requires a lot of effort. In terms of development efforts, accessibility, and cost, web-based \acs{VR} is an emerging area that enables people to use \acs{VR} environments at a lower cost with solutions such as Google Cardboard. While easier implementation and integration is inherent in such applications, it is not very straightforward to obtain complex 3D scenes and high quality sensor readings such as eye tracking (i.e., due to unavailability of standardized sensors and low sampling frequencies of such setups.). Overall, in the near or distant future both paradigms will be developed and utilized while \acs{VR} environments and \acs{HMD}s become mainstream in our daily lives. The issues with both approaches will be solved or mitigated depending on the prerequisites of the application domains and trade-offs will be decided. 

Taking all into consideration, the rest of this chapter is organized as follows. As the content of this thesis is research on everyday \acs{VR} through the leveraging of eye movements in \acs{VR}, privacy preserving eye tracking, and the accessibility of \acs{VR} using biological signals such as eye movements, relevant introductions to each topic have been covered. The possibilities of eye tracking signals in \acs{VR} are discussed in Section~\ref{lbl:section_eyetracking_inVR}. Then, privacy considerations with an emphasis on authentication are explored in Section~\ref{lbl:section_privacy_eyetracking}. Later, the accessibility of \acs{VR} for everyday setups with a focus on eye movements is examined in Section~\ref{lbl:sectionAccessiblity_VR}. The final section, Section~\ref{lbl:toward_everydayVR_summary}, introduces how these topics could be combined within the framework of everyday virtual reality.

\section{Eye Tracking in Virtual Reality} 
\label{lbl:section_eyetracking_inVR}
Assessing eye movements could yield a plethora of possibilities for human-computer interaction. Eye movements represent the information gathered when humans look at specific areas of the presented stimulus. Such movements do not happen fully consciously. In the last several decades, researchers have used eye movements to evaluate human behaviors in different tasks or domains such as during reading~\cite{Rayner1998EyeMI}, multimedia learning~\cite{alemdag2018}, web search~\cite{Granka_2004_SIGIR}, driving~\cite{10.1145/1743666.1743701}, in medicine~\cite{Holzman179}, linguistics~\cite{conklin2018eye}, user experience design~\cite{bergstrom2014eye}, psychology~\cite{Mele2012}, education~\cite{Jarodzka_Holmqvist_Gruber_2017}, programming~\cite{10.1145/3145904}, or marketing~\cite{Eye_tracking_for_vis_marketing_2008}. This use in such a variety of applications and domains in fact shows the potential of eye movements for future directions. In the context of VR, eye movements are also used to assess human attention in an offline way or to actively provide gaze related interaction during virtual experiences. In this section, a proportion of works related to eye tracking and virtual reality are analyzed and reviewed.

To apply fine grained eye movement analyses in the context of immersive virtual environments with \acs{HMD}s, first eye regions and gaze vectors are estimated using eye tracking sensors that are located within \acs{HMD}s. Some \acs{HMD}s provide the opportunity to use the integrated eye trackers (e.g., HTC Vive Pro Eye) or there is also the potential of integrating eye tracking plug-ins (e.g., Pupil-Labs Eye Trackers~\cite{Kassner:2014:POS:2638728.2641695}) without significant effort. One could deploy custom and low-cost sensors~\cite{10.1145/2968219.2968274} along with gaze estimation models and track the eyes as well. These approaches include the steps of estimating eye regions such as pupil and iris and detecting gaze directions using computer vision and machine learning techniques. While gaze estimation is especially important for detecting eye movements such as fixations or saccades, semantic segmentation of eye regions is also convenient for privacy reasons. One might, for example, want to obfuscate iris texture~\cite{iris_obfuscation_etra21} if eye videos are saved during the experience. Kansal and Nathan~\cite{eyenet_segmentation_ICCVW19} proposed a convolutional encoder-decoder network for eye region segmentation, whereas Chaudhary et al.~\cite{chaudhary_ritnet_iccvw_19} used combination of U-Net~\cite{UNet_miccai15} and DenseNet~\cite{DenseNet_CVPR17} to carry out real-time semantic segmentation for eyes. Boutros et al.~\cite{Boutros_2019_ICCV} have proposed a baseline multi-scale segmentation network in which the number of parameters are significantly diminished while reducing the overall accuracy marginally. Cycle \acs{GAN}s~\cite{CycleGAN2017} for eye region segmentation~\cite{Fuhl_2019_ICCV}, ellipse segmentation framework for robust gaze tracking~\cite{ellseg_semantic_seg_for_gaze_tracking_TVCG21}, fast and efficient few-shot segmentation for eyes~\cite{eyeseg_ECCVW20}, domain adaptation for eye segmentation~\cite{domain_adaptation_for_eye_segmentation_ECCVW20}, and identity-preserving eye image generation from semantic segmentation~\cite{D-ID-Net_identity_preserving_image_generation_iccvw19} have also been proposed in the literature. These works focused on practical usability, especially in terms of real-time working functionality or privacy-preservation which are optimal for \acs{VR}/\acs{AR} use-cases.

In addition to the segmentation tasks, different approaches have also been introduced for gaze estimation by using pupil detection~\cite{Starbust_2005, Swirski2013, Kassner:2014:POS:2638728.2641695, Excuse_Wolfgang_2015, Javadi_etal_2015, Else_Wolfgang_2016, fuhl2016non, fuhl2018bore, CBF_Wolfgang_2018, santini2018purest, santini2019get} and recently by mainly machine learning~\cite{Sugano_CVPR2014, Zhang_CVPR2015, Wood_etal_ECCV2016, wood2016_etra, fuhl2016pupilnet, fuhl2017pupilnet, Wang_2018_CVPR, Park_ECCV2018, Park2018, NVGaze_chi19, Yu_2020_CVPR}. While some of the aforementioned works are not directly designed for \acs{VR} and \acs{HMD}s, they provide implications for gaze estimation. Affordable and low-cost solutions are also introduced in this context~\cite{Stengel_MM15_Affordable, low_cost_gaze_tracking_google_cardboard_vrst16}. Estimating the pupil and eye gaze opens up the possibility of detecting eye movement events such as fixations and saccades which can be linked to visual attention and other user states. Fixations are periods during which users fix their gaze on a certain area in the stimulus for a significant amount of time. On the other hand, saccades represent a high-speed shift of eye gaze from one fixation to another. Together, they generate the visual scanpath~\cite{kubler2014subsmatch, geisler2020minhash}. Previous literature has found that long fixations are related to individuals engaging more with the content in the stimulus or that they represent an increased amount of cognitive processes~\cite{just1976eye}. Additionally, longer saccade durations indicate inefficient scanning or searching~\cite{goldberg1999computer}, whereas large saccade amplitudes are related to distant attention shift~\cite{goldberg2002eye}. Detection of such events could be done by applying threshold-based algorithms~\cite{salvucci2000identifying,agtzidis2019360} or by using probabilistic approaches~\cite{tafaj2012bayesian, thiago_bayesian_identification_etra16}. These measures might represent different concepts with presented stimulus; however, they are a valuable source of information for \acs{VR} setups and for related design decisions. In addition to fixation and saccade related features, with the detection of eye regions such as pupil area, one could also use pupil diameter for assessing cognitive load~\cite{beatty:1982, appel2021cross}.

In the \acs{VR} context, using such generated features or raw signals provides a variety of online and offline opportunities. While not being limited to the following, a few of these opportunities include foveated/gaze-contingent rendering, saliency prediction, eye-based interaction, user intent analysis and prediction, assessment of cognitive load, and visual attention analysis. Foveated rendering concentrates on rendering the content that users visually focus in high quality while presenting the remaining stimulus in relatively lower quality. From a computational perspective, this might contribute to making \acs{VR} \acs{HMD}s more accessible in daily scenarios. To this end, a lot of effort has been concentrated on researching different aspects~\cite{9005240,towards_foveated_rendering_gazetracked_VR, Arabadzhiyska2017,is_foveated_perceivable_in_VR_MM17, data_augmentation_for_saccade_landing_point_pred_etra20,gazestereo3d_utilizes_fixation_gaze_est_TOG16, Predicting_gaze_dpth_good_for_gazecontingent_ETRA18}. Arabadzhiyska et al.~\cite{Arabadzhiyska2017} studied prediction of saccade ending points in the context of foveated rendering by considering quality mismatch is not noticeable during saccadic supression. Griffith and Komogortsev~\cite{data_augmentation_for_saccade_landing_point_pred_etra20}, on the other hand, proposed a data augmentation strategy to improve saccade landing point prediction. Hsu et al.~\cite{is_foveated_perceivable_in_VR_MM17} analyzed the perceivability of foveated rendering in different settings in which the technical assessments could prove beneficial for the \acs{VR} community. Meng et al.~\cite{9005240} proposed a foveated rendering approach based on eye-dominance and indicated that their approach offers a better rendering performance than conventional approaches. The task of predicting future gaze locations is not only related to foveated rendering, but also to saliency prediction, and has recently been studied for virtual environments as well~\cite{8998375,SGaze_tvcg19,fixationnet_bulling_tvcg21}. While saliency maps and fixation predictions on the images are studied extensively in the literature~\cite{Kummerer16054,Kummerer_2017_ICCV,10.1007/978-3-030-58558-7_25}, humans explore immersive virtual environments with \acs{HMD}s differently. Thus, attention models also differ compared to conventional setups~\cite{sitzmann_saliency_tvcg18}. Several works have focused on saliency related tasks in 3D virtual environments for head movement prediction~\cite{360VideoSaliency_in_HMD}, salient object detection~\cite{salient_object_det_360_and_dataset_TVCG20}, or visualization of 3D heatmaps~\cite{model_based_realtime_vis_3dheatmaps_ETRA16}. In most of the works, eye tracking data provides a variety of benefits both from the research and practical perspective.

While it is likely that estimating saliency for \acs{VR} or improving foveated rendering configurations will lead people to use this technology more comfortably in daily life, understanding human behavior and interaction during the use of such technology is another key factor and more related to the focus of this work. To this end, Hirzle et al.~\cite{foundations_gaze_interaction_chi19} discuss the foundations of gaze interaction in \acs{HMD}s. In terms of eye-based interaction, many different aspects including gaze- and blink-based inputs~\cite{BCI_gaze_textextry_VR_iui18, gaze_typing_in_VR_impact_of_many_factors_ETRA18, blinktype_text_entry_VR_ismar20}, object interactions via gaze~\cite{10.1145/3314111.3319815, gaze-enhanced_menus_in_VR_vrst20}, and navigation in \acs{VR}~\cite{blinks_and_redirected_walking_algos_vrst18, virtual_navig_automatic_man_pos_before_teleport_vrst20} have been studied. In these works, there are several findings that concern the use of eye features. For instance, Lu et al.~\cite{blinktype_text_entry_VR_ismar20} have reported that according to users' subjective feedback, for hands-free interaction in \acs{VR} blinking is preferred as opposed to dwelling over content for a specific amount of time, which is considered more common. Nguyen and Kunz~\cite{blinks_and_redirected_walking_algos_vrst18} have found that users do not detect the scene rotations up to a certain degree during blinks due to visual supression, which is helpful for redirected walking algorithms in \acs{VR}. Furthermore, Sidenmark and Lundström~\cite{10.1145/3314111.3319815} have indicated that interactions with stationary objects during hand interactions in \acs{VR} might be favorable in terms of attaining fixations.

User intention analysis, efforts to understand human visual attention, and cognitive load assessment via eyes during immersive \acs{VR} experiences are of a particular interest for many researchers. This is due to the possibility of supporting and assisting users during the \acs{VR} experience when such information is known. Additionally, while it is not possible, at present, to create identical configurations with the real world due to technological limitations, one can create hypothetical or utopic scenarios by using human behaviors obtained from different situations. Taking these into account, this direction differs from others such as foveated rendering, saliency estimation, or interactions within virtual spaces. Additionally, understanding human visual attention and perception can also improve the interaction experience during the virtual experience. In the human-aware interaction direction, prediction of touch intentions using gaze~\cite{sparse_haptic_proxy_chi17}, prediction of interaction intentions using gaze information~\cite{David-John_intent_VR_etra21}, prediction of tasks using eye movements~\cite{predicting_task_basedon_eye_movements_VR_etra20} are possible. From a visual attention perspective, a data-driven and eye tracking based approach for locating elements in 3D virtual spaces~\cite{vis_att_optimizing_obj_placement_via_gaze}, immersion preserving attention guidance~\cite{attention_guidance_chi20_HiveFive}, and improvement of driving habits with the help of visual attention analyses~\cite{8448290} have been studied. As mental and cognitive load can be also assessed using eyes, particularly via pupil sizes~\cite{beatty:1982, kasneci2021your}, such information may be helpful for context sensitive assistance and support for users. Recently, Luong et al.~\cite{mental_workload_physiological_performance_features_ismar20} have studied the real time recognition of mental workload using physiological features, including the features related to pupillary activities in a \acs{VR} flight simulator, and obtained up to $65\%$ accuracy. Another aviation use-case was demonstrated by Wilson et al.~\cite{aviation_cognitive_load_IMWUT21} with deep learning including eye features such as pupillary responses and blinks. The authors have shown that it is possible to classify two level cognitive load over $81\%$ accuracy. Similarly, Kübler et al.~\cite{Kubler_Kasneci_Vintila_2017} have studied the pupillary responses to hazard perception in a 360-degree virtual reality setup and found that pupil dilations are helpful for perception.

Apart from the aforementioned directions and use-cases, eye tracking and gaze measurements have been used in \acs{VR} research in different settings such as in medicine~\cite{7829437}, expert-novice analysis based training~\cite{hosp_goalkeeper_expertise_21}, or education~\cite{ahuja_etal_chi21_classroom_twins}. The majority of the studies provide implications through eye movements or pupil related activities, which shows the overall usefulness of eye tracking for \acs{VR} and its future potential.

\section{Privacy and its Considerations for Eye Tracking}
\label{lbl:section_privacy_eyetracking}
Despite the advantages of using sensor data and eye tracking, privacy risks exist. Some of the privacy risks of extended reality (XR) that also mention eye tracking data were discussed by Mhaidli and Schaub~\cite{some_privacy_risks_of_xr_chi21}. Silva et al. focused on eye tracking support for visual analytics and identified that ensuring privacy~\cite[p.~8]{eye_track_support_for_Vis_Analytics_inc_privacy_ETRA19} is a major theme in terms of oppotunities and challenges for eye movements. Katsini et al.~\cite{katsini_secpriv_survey_chi20} discussed the aspects of eye gaze in terms of security and privacy by focusing more on authentication schemes. Similarly, Liebling and Preibusch~\cite{Liebling2014} have argued the need for privacy mechanisms in pervasive eye tracking. Recently, there has already been a push for privacy in both the eye tracking and \acs{VR} communities. With legal regulations like General Data Protection Regulation (GDPR)~\cite{EUdataregulations2018} or similar, it is possible that there will be more emphasis in this direction especially with \acs{VR} devices being used more frequently in everyday life.

One of the most straightforward use-cases of eye data in the biometrics domain is iris authentication. Iris textures can be used reliably for biometric authentication~\cite{Daugman_1993,KUMAR20101016} and iris recognition systems are already used widely, for instance at airports (e.g., in the UK, The Netherlands, and in Canada)~\cite[pp.~2-3]{Daugman2009}. Security protocols such as multi-party computations or homomorphic encryption schemes have recently been applied for such purposes as well~\cite{Barni_SEMBA_19,XSong2020}. While not providing the same level of authentication accuracy as iris recognition schemes, multiple studies have shown that eye movements can also be used for biometric authentication~\cite{Kinnunen:2010:TTP:1743666.1743712,6712725,Komogortsev2010}. Due to the lower success rates, Komogortsev et al.~\cite[p.~4]{Komogortsev2010} have argued that such methods could be integrated with biometric authentication systems that use face or iris recognition as an additional security layer. Eberz et al.~\cite{Eberz:2016:LLE:2957761.2904018} have suggested that eye movement-based authentication could be applicable with settings available in consumer level equipment. Similarly, Zhang et al.~\cite{Zhang:2018:CAU:3178157.3161410} have argued for possible use-cases for \acs{VR} setups with eye movement-based authentication. Furthermore, Zhu et al.~\cite{blinkey_VR_authentication_IMWUT20} proposed a two-factor user authentication method for \acs{VR} \acs{HMD}s based on blinking behaviors and pupil sizes. Such schemes are considerably helpful for authentication-requiring use-cases such as in-app purchases or login scenarios in \acs{VR} and are resilient to shoulder surfing attacks especially in \acs{HMD} consisting situations. However, apart from the authentication-requiring scenarios, collecting and aggregating eye movement data without privacy protection could introduce a privacy breach given the wider use of \acs{VR} devices and the works that map eye movements to user authentication and identification.

Computational privacy and related algorithms could be developed for and applied to any kind of time-series data which is similar to eye movement observations obtained from \acs{VR} displays. However, if the intent is to make privacy preserving mechanisms work practically, one should consider the real time working capability of such mechanisms. From the privacy perspective, multiple aspects should be covered. These include differential privacy (DP)~\cite{dwork2006,dwork_DP_only,dwork2014}, secure multi-party computation~\cite{Yao_1982,Yao_1986} (SMC)-based solutions, and more practical use-cases such as data masking. Differential privacy is an overall scheme for sharing data without compromising the information by which individuals participate in a corresponding dataset by introducing noise on the data~\cite{dwork2006}. The main issue is to find a proper utility-privacy trade-off. In the \acs{SMC}-based works, the main idea is to compute an output without compromising the raw data of the input parties, usually by sharing secrets. Yao~\cite[p.~1]{Yao_1982} gives the example of two millionaires who want to know who is richer without providing information about their own wealth. Other solutions such as randomized encoding~\cite{applebaum2006cryptography,applebaum2017garbled} could be also applied; however, computation complexity or communication costs play an important role in terms of the practical usability of the solutions. From an eye tracking in \acs{VR} perspective, multiple input parties could calculate intentions, activities, or estimate gaze without compromising the raw eye movement information that is obtained from their eye trackers.

In the literature, privacy aspects of eye tracking data have not been researched in-depth yet. Recently, Steil et al.~\cite{steil_diff_privacy} and Liu et al.~\cite{Liu2019} applied standard differential privacy mechanisms on eye movement features (e.g., rate of fixations, mean saccade amplitudes) and on heatmaps, respectively. While these works are the first in the literature to introduce differential privacy mechanisms on the eye tracking data, they do not address the correlations and privacy loss due to them~\cite{Zhangetal_2020} in the differential privacy context. Li et al.~\cite{263891} proposed a formal approach for area-of-interests that works in real time. There are several approaches that focus on iris obfuscation~\cite{iris_obfuscation_etra21} by degrading iris authentication~\cite{John:2019:EDI:3314111.3319816,John_pixel_noise_etra20} or removing iris biometric on the eye images~\cite{John_sec_util_iris_auth_TVCG20,10.1145/3379156.3391375}. Adversarial attacks have also been performed on the classifiers based on eye tracking~\cite{adversarial_hagestedt_etra20}. Furthermore, while not being directly related to \acs{VR}, Steil et al.~\cite{Steil:2019:PPH:3314111.3319913} have studied an automated shutter mechanism on the field camera of an eye tracker based on the scene privacy which could be applied to \acs{AR} setups. In this work, the authors found that scene privacy, namely privacy of the stimulus, could be detected to some extent solely by using eye movement features.

\section{Accessibility of Virtual Reality for Everyday Scenarios}
\label{lbl:sectionAccessiblity_VR}
It is likely that \acs{VR} applications are going to be used by wider communities going forward, taking the market size estimations~\cite{VR_marketsize_research} and current research into account. Even today, users can access \acs{VR} applications through application stores (e.g., Steam~\cite{steam_store}, Oculus Go Store~\cite{oculus_app_lifecycle}), \acs{VR} supported platforms (e.g., Mozilla Hubs~\cite{mozilla_hubs_link}), or even through YouTube~\cite{youtube_online_link} with 360-degree videos. The access to \acs{VR} content through web browsers and services is indeed a valuable contribution since more people can potentially access \acs{VR}. However, native applications designed for a specific system allow for higher quality graphical stimulus and usually a better human-computer interaction experience due to already available physics engines and potential easier use of sophisticated sensors.

Considering either native or web-based \acs{VR} applications, a lot of research and development in computer hardware, software, and perception engineering has been carried out to make different kinds of \acs{VR} systems and applications more available in daily life. On the research side, particularly in perception engineering, human attention and behaviors are important. Developed systems are usually evaluated with human data for many reasons. For instance, if the purpose of a system is entertainment, practitioners may focus on eliminating cybersickness during the virtual experience. In another example, for a system that is designed for training, practitioners may target for scenarios that users might encounter in their real lives. Overall, these approaches require collection of human data during confrontation with developed simulations. The knowledge that is obtained through research on human data is used to provide more advanced virtual experiences. With the many iterations of this loop, everyday usage of \acs{VR} and \acs{HMD}s will be more feasible in the future.

On the research side, collection of human data in \acs{VR} experiments is mostly done in laboratory settings, limiting number of participants in the experiments. This can also lead to homogeneous characteristics of the participants. When the variety of users that can actually own \acs{HMD}s and experience \acs{VR} applications is considered, there is a vicious cycle that can be broken through the development of remote data collection routines and protocols. Such data covers both questionnaires, which is usually easier to collect by using web services, and other sensor data such as eye tracking. To tackle this issue, Ma et al.~\cite{Ma_etal_WebbasedVR_power_by_crowd_2018} have proposed a solution using Amazon Mechanical Turk (AMT)~\cite{amt_online_link} and a web-based \acs{VR} application. In their solution, users have to validate the existence of their VR devices by taking photos including their worker IDs. While they did not collect any eye movements, the effort is a valid solution for the problem of enabling the crowd for \acs{VR} experiments. Very recently, Rivu et al.~\cite{rivu2021remote} addressed a very similar issue and reported four approaches to conducting remote \acs{VR} studies including providing a standalone application through direct download or an application store (e.g., Steam), uploading a standalone application to a \acs{VR} platform (e.g., Mozilla Hubs), or directly setting up the \acs{VR} application on a \acs{VR} platform. The authors have indicated that providing the standalone \acs{VR} application through a direct download option is preferable for the most advanced functionalities and data collection options; however, recruitment would need to be done through social media or forums. In addition, they have also emphasized the importance of clear instructions for the experiment and possible ways for the user to connect with the experimenter if needed. In general, the options Rivu et al.~\cite{rivu2021remote} have proposed should go hand-in-hand and used by practitioners depending on the requirements. For instance, if \acs{VR} platforms do not have an option to collect a specific type of data (e.g., eye movements, hand tracking), then it is more reasonable to provide the applications via a direct link. On the other hand, if it is too much effort to advertise the application with the use of direct download, one could pursue recruitment through application stores by choosing the second option proposed by Rivu et al.~\cite{rivu2021remote}.

In summary, researchers should strive not only to propose novel solutions in perception engineering and \acs{VR}, but to make these solutions accessible for wider communities during research cycles and to evaluate the solutions with a greater number of participants. Until now, this aspect, in particular, is researched by few and remains an open research direction.

\section{Towards Everyday Virtual Reality}
\label{lbl:toward_everydayVR_summary}
Developments in eye tracking methodology, usage of eye tracking data in \acs{VR} studies and applications, privacy preserving considerations for such data, and efforts to make \acs{VR} more accessible will help integrate virtual experiences in everyday life. To this end, in line with the everyday \acs{VR} definition provided by Garner et al.~\cite{Garner2018_everydayVR}, this work proposes novel solutions to a variety of problems facing everyday \acs{VR} and presents an overall framework. 

Multi-modal gaze, i.e., head pose and eye gaze, in \acs{VR} and pupillary measurements could help both for design considerations of virtual environments and for online user assistive tasks during virtual experiences. Design considerations for some of the everyday tasks and setups are more important than the others. For instance, one could design direct replicas of real world configurations based on scientific findings in real world settings for some of the everyday setups in \acs{VR}. However, as attention models could differ in real and virtual worlds, a dedicated attention assessment and related cognitive processes should be carried out. A classroom environment for learning falls directly in the category of everyday \acs{VR} and relevant knowledge transfer can be done from real world studies as comparison. With the increasing popularity of e-learning platforms and, very recently, the necessity for online learning due to the COVID-19 pandemic, learning in remote setups has emerged as standard. However, most of the setups lack immersion and provide limited interaction. Learning in \acs{VR} might solve such disadvantages, but, at the same time, human behaviors should be analyzed before offering \acs{VR}-based learning to make it highly accessible.

Another potential, but unconventional direction for \acs{VR} environments that take place more frequently in everyday life is training. This may be considered unconventional in the umbrella term of everyday \acs{VR} since once required expertise is gained after regular use of a dedicated virtual environment one may not necessarily need to use it moving forward. In terms of training, there are multiple ways to proceed. A straightforward application is providing novices in some specific domain with an immersive environment in which to gain new information. An alternative is replicating unexpected and unusual occasions that can happen in everyday life. Training in the real world for such occasions might be dangerous or even impossible depending on the domain. For instance, safety related situations that happen in maritime applications~\cite{safetystuff_maritime_VR_2020} could benefit from virtual training. More related to everyday \acs{VR} is the example of driver training. Skill training apart, according to Bialkova and Ettema~\cite[p.~2]{8809586}, driving and transportation scenarios are considered within the everyday VR context. While actual driving instruction is performed with real vehicles, there are some safety critical situations during the training of novice drivers that are ethically impossible to create in the real world. For example, the scenario of a pedestrian crossing the street unexpectedly. Additionally, with the growing number of semi-autonomous vehicles and the assumption that fully autonomous vehicles will be available in the future, human-machine interaction is likely to be crucial in traffic scenarios. While not as unique as the maritime example, the number of different interaction scenarios with semi-autonomous vehicles that one may encounter is limited in real world driving learning. With \acs{VR} and pre-programmed routines, novices can train for such situations in \acs{VR} in large part thanks to gaze-based assistance and behavior analysis. However, before all of these, it should be researched whether such gaze-based assistance for safety critical situations has a positive effect on drivers.

The growing importance of privacy-related topics, legally and technically, will require eye tracking-enabled \acs{VR} experiences and configurations to take privacy issues into account. Data conventionally collected before and after the experiments with questionnaires can be anonymized without too much effort. Users at least are usually more accustomed to such data collection protocols. Furthermore, the usage of questionnaires or similar methods does not make much sense when users' everyday experiences with virtual environments (e.g., at their homes or personal spaces) are taken into account. As suggested by Steil et al.~\cite[p.~8]{steil_diff_privacy}, users might not be extensively aware of the kind of inference generated using eye tracking in the \acs{VR} and \acs{AR} context. Differential privacy~\cite{dwork2006,dwork_DP_only} is especially common in the database applications area, and could be used effectively in eye tracking signals by calibrating the required amount of noise to add to signals while providing data for further inference. Well-established formal methods in the differentially private eye tracking field will further help usage of individual protected eye tracking data collected in everyday scenarios and setups. Additionally, \acs{VR} \acs{HMD}s are on the way to becoming personal devices like mobile phones and smart watches. For providing assistance during experience, machine learning models are trained with huge amounts of data, collected from various people. Such models can also be trained in cloud for scalability and better processing power. In a very primitive and naive setup, users can upload their collected eye tracking data directly to the cloud. However, one may not share the raw data due to aforementioned inference possibilities especially by the commercial applications. In such cases, use of secure multi-party computation~\cite{Yao_1982, Yao_1986}, randomized encoding~\cite{applebaum2006cryptography,applebaum2006computationally}, or similar procedures are reasonable as long as real time capability of such solutions in the test time is empirically evidenced. The solution can be anything from privacy preserving gaze estimation to gaze prediction in \acs{VR}; however, such setups depend on the classifiers that are used and real time working capability can be affected by the effectiveness of the communication with the cloud instances. Overall, the use of cloud-based predictive models by preserving the privacy computationally is actually very relevant for everyday \acs{VR} setups when the contemporary cloud infrastructure is utilized~\cite{cloud_infrastucture_spent_2020}.

To enable everyday usage, the inclusion of more people in \acs{VR} systems, from a data perspective, is necessary. While the decreased cost of \acs{HMD}s, marketing, and other relevant attempts could affect this, one important factor is the availability of large scale data collection possibilities and related protocols and applications. For behavioral data other than sensors such as eye tracking, \acs{AMT} has been used widely for crowdsourcing purposes. However, there are several issues with such paradigms if they are applied to \acs{VR}/\acs{AR} sensor data collection, particularly eye tracking. Either forwarding from \acs{AMT} (or similar) to web-based custom \acs{VR} services~\cite{Ma_etal_WebbasedVR_power_by_crowd_2018} or the supplying of stand-alone applications~\cite{rivu2021remote} should be done. As web-based \acs{VR} currently faces major challenges such as lower resolutions or lack of integrated high quality eye tracker devices, providing standalone applications via different channels such as forums or social media is more reasonable. However, in this approach, one should solve the issue of participant compensation and data sharing optimally without any centralized third-parties due to the behavioral biometrics that are inherent in eye movements. If third-parties are involved in the process, an extra layer of protection e.g., by using cryptography, can be applied, but this would increase the overall complexity. Furthermore, validation of data authenticity should be performed to check for adversarial participant behavior.
\cleardoublepage

\chapter{Motivation and Main Findings} 
\label{chapter_motivation_findings}
In this chapter, the papers I have authored in the direction of everyday \acs{VR} during my doctoral studies have been summarized in terms of motivation, methodology, and findings according to the open directions that are laid out in Chapter~\ref{chapter_introduction}. The published papers are available in Chapters~\ref{appendix_A},~\ref{appendix_B}, and~\ref{appendix_C}. 

Eye tracking in \acs{VR} for visual attention and cognitive processes has been researched and studied in two different domains, namely education and driving. In the education-related studies, the research focus is relating conventional eye movements and pupillometry and virtual objects of interests with cognitive processes and visual attention, respectively. In the driving studies, the effect of gaze-aware assistance on human behavior and the feasibility of cognitive load estimation during a safety critical pedestrian crossing have been researched. These are introduced in Section~\ref{section_main_visattention}. 

Privacy preserving eye tracking has been researched in two directions as well. First, differential privacy mechanisms for eye movement signals that have been collected from \acs{VR} and \acs{AR}-related setups taking temporal correlations into account are discussed. Next, the privacy preserving gaze estimation task for \acs{VR} using two input and one function parties is introduced in Section~\ref{section_main_privacy}.

Lastly, for the accessibility of \acs{VR} along with eye tracking, a blockchain- and smart contract-based eye tracking data collection protocol for remote participants is introduced in Section~\ref{section_main_accessibility}. 

Overall, three sub-directions that are researched in this thesis introduce novel use-cases while advancing the state-of-the-art in multiple directions, and propose together a framework that is depicted in Figure~\ref{fig_thesis_overall_framework} which further pushes for everyday \acs{VR}.

\begin{figure}[ht]
\centering
\includegraphics[clip, trim = 2cm 2cm 5cm 1cm, width=\linewidth]{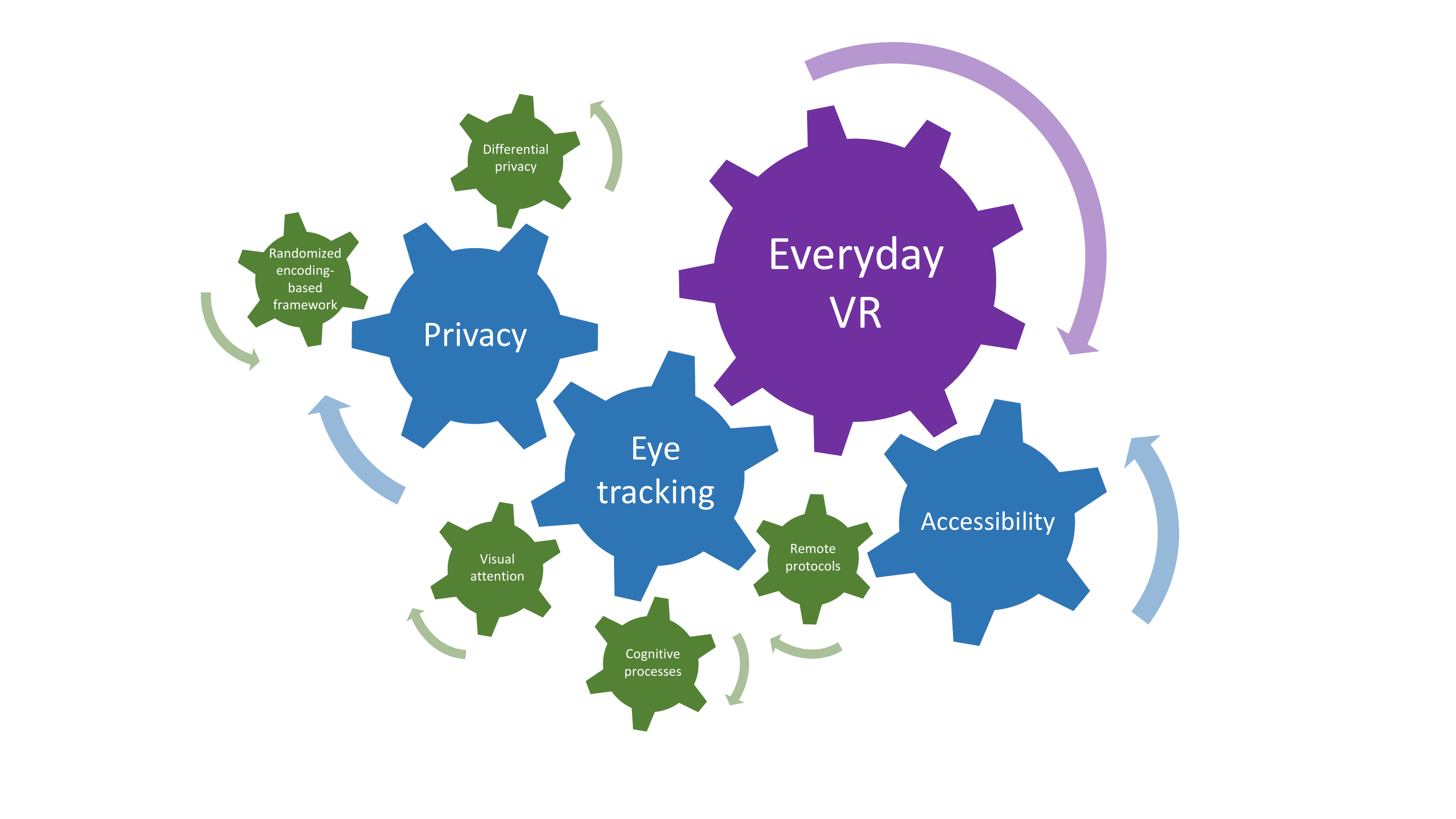}
\caption{Overall framework of contributions in the thesis.}
\label{fig_thesis_overall_framework}%
\end{figure}

\newpage

\section{Visual Attention and Cognitive Processes in Virtual Reality}
\label{section_main_visattention}
The first two subsections (Sections~\ref{subsection_main_CHI},~\ref{subsection_main_VR}) focus on the education content in \acs{VR} for the design considerations of virtual spaces in everyday educational VR, whereas the latter two (Sections~\ref{subsection_main_SAP},~\ref{subsection_main_VRW}) focus on the driving domain with the purpose of assessing the feasibility of \acs{VR} for safety critical scenario training.

\subsection{Eye Movements in Virtual Classrooms}
\label{subsection_main_CHI}

This subsection is based on the paper~\ref{publist_lbl_CHI21} in Chapter~\ref{chapter_publications}, \emph{Digital transformations of classrooms in virtual reality} at \emph{2021 CHI Conference on Human Factors in Computing Systems}.

\subsubsection*{Motivation and Main Methodology}
How to design and visualize interaction related content is crucial when everyday virtual learning environments are taken into consideration, as these might affect motivation, engagement, performance, and eventually learning outcomes of students in the long term. While there are many studies which concentrate on a variety of issues for the real world classroom~\cite{goldberg2019attentive,fauth2020don}, few conclusions are drawn for \acs{VR} environments. In the real world studies, human behaviors are usually extracted and analysed based on head and body movements. In \acs{VR}, it is possible to extract eye movements and pupillary information, which may be linked to cognitive processes, given the fact that modern high-end \acs{HMD}s integrate eye trackers. In fact, virtual environments, especially immersive ones, could have a large range of benefits for humans when environmental and physical inaccessibilities are considered. For instance, during the COVID-19 pandemic, many institutions temporarily switched to remote teaching and learning. While this switch was particularly motivated by the pandemic, remote teaching and learning can also be useful for handicapped people or when in-person meetings are not feasible due to the restrictions of physical distance. For example, while during the pandemic, many of the scientific conferences and meetings were organized remotely or as hybrid events. After the pandemic, such a hybrid model could remain to accommodate difficulties due to remoteness. Virtual environments have already been realized using relatively trivial tools such as Mozilla Hubs~\cite{mozilla_hubs_link}. Before more sophisticated tools are employed for such teaching, learning, and interaction purposes, the effects of different configurations on humans should be researched in several virtual environment types, including halls, auditoriums, and classrooms. The major motivation of this work is to understand the effects of different virtual classroom configurations on students including different avatar representation styles of virtual characters, namely realistic and cartoon, different participant locations in the classroom, including back and front, and different hand-raising behaviors of virtual peer-learners, particularly 20\%, 35\%, 65\%, and 80\%. An example view from the used \acs{VR} classroom is depicted in Figure~\ref{figmain_VR_Classroom}.

\begin{figure}[!h]
   \includegraphics[width=1\linewidth]{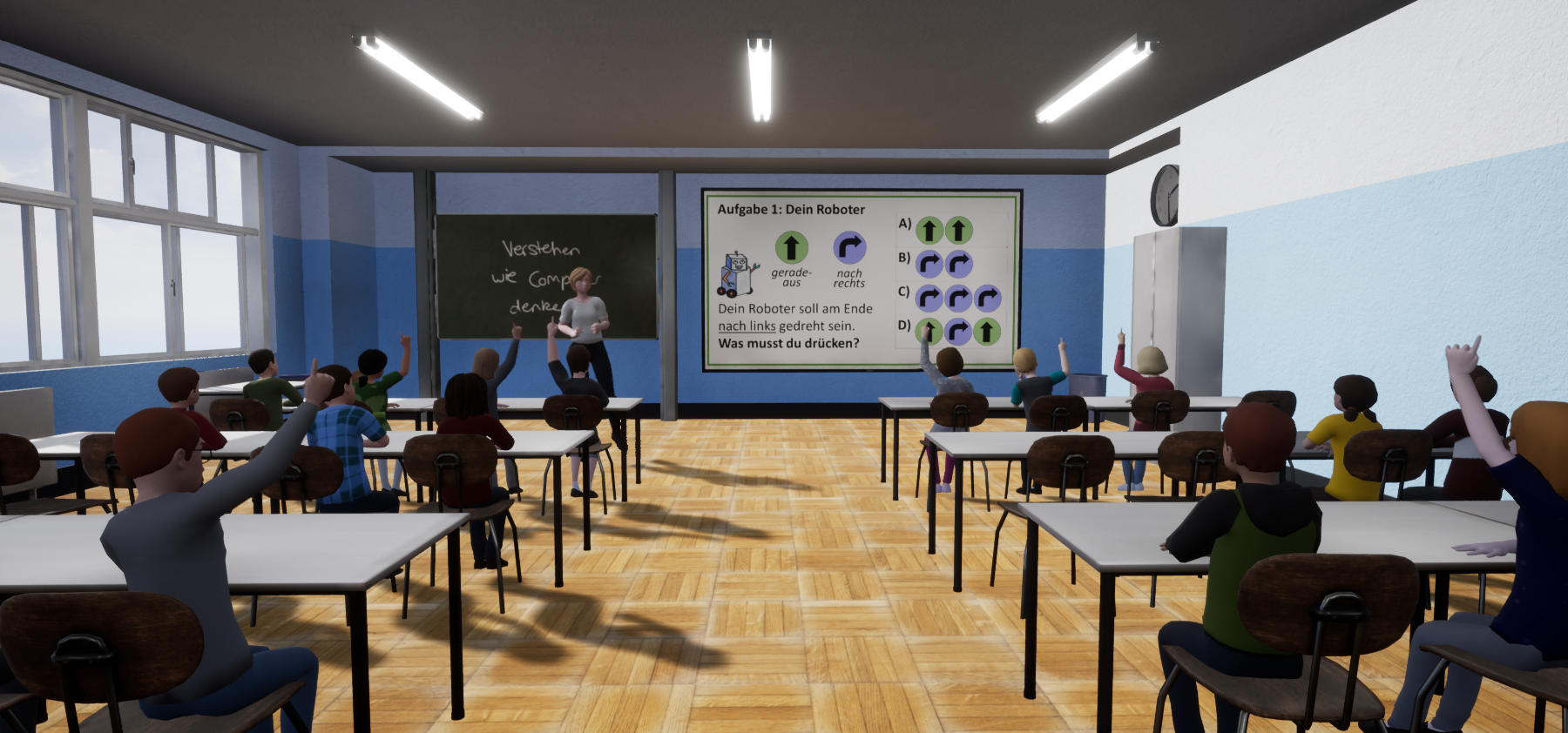}
   \caption{An example view from the used VR classroom.}
  \label{figmain_VR_Classroom}%
\end{figure}

Using raw eye and head tracking data that was collected from 381 sixth-grade students in a between-subjects design during a computational thinking lecture that took approximately $15$ minutes, eye fixation, saccade, and pupil diameter related features were extracted. The raw data only provides angular information that was reported by the eye tracker within the \acs{HMD}. Therefore, a processing pipeline is needed to extract features such as durations of fixations and saccades, amplitudes of saccades, number of fixations, and pupil diameter values. Apart from the pupil diameter values, fixation and saccade extraction algorithms should be tailored especially for the \acs{VR}, as in \acs{VR} one should take head movements into account compared to conventional eye tracking experiments that include chin-rests. For \acs{VR} setups, Agtzidis et al.~\cite{agtzidis2019360} have proposed a Velocity-Threshold Identification (I-VT) similar thresholding approach for 360-degree videos, and, in this work, a similar approach was followed by setting velocity and duration thresholds for fixations and saccades. For the pupil diameters, a standard processing pipeline including smoothing~\cite{savitzky64} and baseline correction~\cite{Mathot2018} components was employed. After pre-processing and feature extraction phases, full factorial \acs{ANOVA}s ($\alpha = 0.05$) were applied for different features to find out whether different classroom manipulations have significant effects on human visual behaviors. For multiple comparisons and non-parametric versions of the analyses, Tukey-Kramer and Aligned Rank Transform (ART)~\cite{10.1145/1978942.1978963} were applied, respectively.

\subsubsection*{Main Findings}
Different classroom manipulations have statistically significant effects on human visual behaviors, including durations of the fixations and saccades, number of fixations, saccade amplitudes, and pupil diameters. More specifically, in terms of participant locations in the virtual classroom, participants sitting in the back of the classroom had significantly longer fixations than those sitting in the front. Furthermore, the front sitting participants had longer saccade durations and larger saccade amplitudes than those sitting in the back. 

In terms of avatar representations, the analyses have yielded beneficial, but less significant results. The participants that experienced cartoon-styled avatars in the virtual classroom had longer fixation durations compared to the participants who observed realistic-styled avatars. In addition, participants that experienced the \acs{VR} classroom with realistic-styled avatars had longer saccadic durations and larger pupil diameters during the virtual lecture than the participants that observed cartoon-styled avatars.

Lastly, in terms of attention towards different hand-raising behaviors, mixed results were found. In particular, participants had significantly larger pupil diameters with 80\% hand-raising peer-learners compared to 35\% of the virtual peers raising their hands. Additionally, in the 65\% hand-raising condition, participants had significantly more fixations than in the 80\% hand-raising configuration. These results have important implications for the design of everyday virtual and interactive classrooms, especially in terms of cognitive processes.

\subsection{Visual Attention on Virtual Objects in Virtual Classrooms}
\label{subsection_main_VR}

This subsection is based on the paper~\ref{publist_lbl_VR21} in Chapter~\ref{chapter_publications}, \emph{Exploiting object-of-interest information to understand attention in VR classrooms} at \emph{2021 IEEE Virtual Reality and 3D User Interfaces}.

\subsubsection*{Motivation and Main Methodology}
The findings that are introduced in Section~\ref{subsection_main_CHI} provide implications primarily for cognitive processes during the virtual classroom experience. Another aspect that is important in virtual environments in general is the assessment of interactions with the virtual objects and the provided 3D stimulus. Virtual objects are essentially gathered together in the virtual environments to create the overall 3D stimulus. Thanks to modern game engines that are mostly used to design such environments, it is also possible to associate physics-related components with 3D virtual objects. For instance, each peer-learner, instructor, and any other static or dynamic content can be represented as objects in the classroom environment and volumetric details can be fetched for further analysis. In terms of human-computer interaction, several interaction models with objects during the experience are possible, such as with hand-held controllers, solely by hand tracking, audio-based input, or gaze-based methods. However, before providing such interaction models, preferably online and real time, one should study visual attention on different objects under different manipulations. Therefore, the main focus of this work is studying the gaze-based attention mainly on the most important objects and object groups in the \acs{VR} classroom that is studied in Section~\ref{subsection_main_CHI}. The most important objects are considered virtual instructor, virtual peer-learners by aggregating attention from each peer-learner, and the lecture screen where the lecture content is visualized. Such analysis and features differ from conventional eye tracking features such as fixations or saccades because they are more related to cognitive processing. Studying the focus on 3D objects in the environment is more related to visual attention and interactions in particular parts of the \acs{VR} classroom.

To do such analysis, the head and eye gaze information should be mapped to the 3D environment. To this end, using the head pose reported by the \acs{HMD} and the gaze vector provided by the eye tracker, an invisible ray is traced to the classroom environment. As virtual objects have invisible geometric colliders around them, the 3D hit point of the traced ray~\cite{ROTH1982109} is calculated. Participants can attend some objects very shortly in an insignificant amount of time and this does not indicate much about the visual attention. To overcome this issue, a duration threshold of 200 milliseconds was used for minimum attention duration. This scheme was applied for each frame and attention times were obtained for each object. The set threshold is greater than fixation detection thresholds in the previous literature. This was done intentionally since participants could have any kind of eye movements on the attended virtual objects. Afterwards, to analyze the implications of each classroom manipulation, full factorial \acs{ANOVA}s ($\alpha = 0.05$) for attention time on each relevant object, namely virtual peer-learners, instructor, and screen were applied. Similar to the method in Section~\ref{subsection_main_CHI}, Tukey-Kramer post-hoc test and \acs{ART}~\cite{10.1145/1978942.1978963} were used for multiple comparisons and non-parametric analyses, respectively.

\subsubsection*{Main Findings}
Analyses showed that the participants that were located in the back of the classroom had significantly longer attention time on the virtual peer-learners than the participants who were located in the front. On the contrary, the front sitting participants had significantly longer attention time on the virtual instructor and lecture screen than the participants who sat in the back.

Cartoon-styled avatars attracted more attention time on the virtual peers than the realistic-styled avatars, while in the realistic-styled avatar configuration, the virtual instructor drew more attention time than in the cartoon-styled avatar configuration. In terms of attention time on the screen, realistic- or cartoon-styled avatar configurations did not differ significantly.

Mixed effects were obtained in the analyses for different hand-raising configurations of the virtual peers as is the case for the conventional eye movement features discussed in Section~\ref{subsection_main_CHI}. The attention time on the peer-learners was the most in the extreme hand-raising configurations, namely 20\% and 80\%. More particularly, the attention time on the peers was significantly greater in the 80\% hand-raising configuration than in the 65\%. Furthermore, 20\% and 65\% configurations differed statistically significantly in terms of focus on the peer-learners, with the 20\% configuration being longer. In the analyses for the attention time on the virtual instructor, the only significant difference was found between 65\% and 80\%, with focus time in the 65\% hand-raising configuration being longer. Furthermore, focus on the lecture screen was significantly longer with 65\% of the virtual peers raising their hands than with 80\% of the peers raising their hands. The focus time on the screen in the 65\% hand-raising configuration was also significantly higher than the 35\% hand-raising configuration; however, the effect was smaller compared to the effect in the 80\% hand-raising configuration.

\subsection{Driver Attention Analysis during a Safety Critical Situation in Virtual Reality}
\label{subsection_main_SAP}
This subsection is based on the paper~\ref{publist_lbl_SAP19} in Chapter~\ref{chapter_publications}, \emph{Assessment of driver attention during a safety critical situation in VR to generate VR-based training} at \emph{2019 ACM Symposium on Applied Perception}.

\subsubsection*{Motivation and Main Methodology}
Whether \acs{VR} setups could help training for a variety of scenarios that humans may encounter in their real lives and cannot train for due to safety reasons or a low likelihood of occurrence is an open research direction in the context of everyday \acs{VR}. In the driving domain, safety critical scenarios are of a particular interest since by using \acs{VR} novice drivers can train at their homes to get accustomed to these scenarios and improve their skills. However, before generating the training packages for these purposes, it is important to understand whether \acs{VR} configurations with gaze-based assistance help and draw the attention of drivers in safety critical situations. To this end, an unexpected pedestrian crossing scenario in a \acs{VR} setup has been designed, with two conditions. The conditions include an experimental condition where participants are informed about the criticality of the pedestrian by gaze-aware red cues around the figure, and a control condition where participants are not informed. A view from the driving vehicle cockpit is depicted in Figure~\ref{figmail_start_point_road}~\cite[p.~2]{bozkir2019person}.

\begin{figure}[!h]
   \includegraphics[width=1\linewidth]{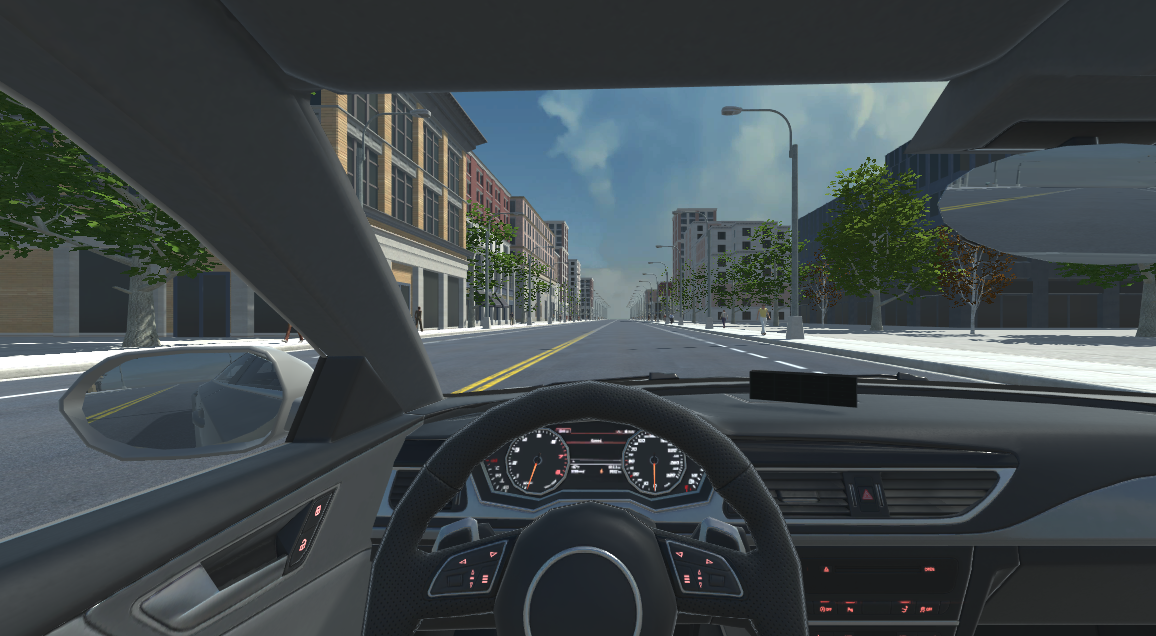}
   \caption{An example view of the road that includes critically crossing pedestrian from the driving vehicle's cockpit. \textcopyright~2019 IEEE.}
  \label{figmail_start_point_road}%
\end{figure}

The study involved 16 participants with between-subjects design. The critical pedestrian started crossing the street when driving vehicle was closer than approximately 45 meters as this distance simulates an expected time-to-collision between approximately 1.8-5 seconds. In such conditions, Rasouli et al.~\cite{DBLP:journals/corr/RasouliKT17} report the high occurrence likelihood of joint attention between pedestrians and drivers. However, in the opposite case, the outcome might be fatal. Since such a scenario cannot be generated and validated in the real world, \acs{VR} setups stand alone for applying such scenarios with low costs. The gaze-aware pedestrian warning was activated when the distance to the crossing pedestrian was approximately 77 meters and deactivated if the participant's gaze was within 5 meters of the crossing pedestrians for at least 0.85 seconds. Intersection between participant gaze and the pedestrian 3D object was carried out with the help of ray-casting~\cite{ROTH1982109}, similar to the study in Section~\ref{subsection_main_VR}. To assess the usefulness of the setup along with gaze-aware cues, the closest distance between vehicles and crossing pedestrians, driver inputs on accelerator and brake, and participant pupil diameters were evaluated. Pupil diameters were smoothed~\cite{savitzky64} and baseline corrected~\cite{Mathot2018} in a similar manner to the study in Section~\ref{subsection_main_CHI}. The experimental and control conditions were compared with two sample T-test in terms of closest distances to crossing pedestrians. Furthermore, within each condition, paired T-tests were applied to driver accelerator inputs and pupil diameters for baseline driving and risky driving timeframes. Baseline driving corresponds to driving without any intervention such as gaze-aware warnings or pedestrian crossing, and is calculated using the observations just before the gaze-aware warnings or start of the pedestrian crossing for experimental and control conditions, respectively. Risky driving timeframe corresponds to the time after the start of the critical pedestrian crossing. 

\subsubsection*{Main Findings}
Analysis on the closest distances to the crossing pedestrian showed that the participants who received the gaze-aware critical pedestrian warnings before the crossing passed statistically significantly distant to the pedestrian than the participants that did not receive the warnings. This indicates that the warnings helped the participants drive safer.

According to the analysis of inputs on accelerator and brake pedals, significant differences on accelerator inputs between baseline and risky driving timeframes started earlier in the participants who received gaze-aware warning cues, indicating that they realized the criticality earlier. In addition, five of the participants who did not receive pedestrian cues performed full braking whereas no participant receiving the cues performed a full brake. In terms of pupil diameters, a trend similar like in the accelerator inputs was observed. Overall results indicated safe and smooth driving experiences for the participants who were provided with gaze-aware risky pedestrian warnings.

\subsection{Cognitive Load Estimation during Virtual Driving}
\label{subsection_main_VRW}
This subsection is based on the paper~\ref{publist_lbl_VRW19} in Chapter~\ref{chapter_publications}, \emph{Person independent, privacy preserving, and real time assessment of cognitive load using eye tracking in a virtual reality setup} at \emph{2019 IEEE Conference on Virtual Reality and 3D User Interfaces Workshops}.

\subsubsection*{Motivation and Main Methodology}
\acs{VR} setups have potential to provide timely benefits to users if their states can be estimated in real time during virtual experiences, which could be helpful in the context of everyday \acs{VR}. As aforementioned, one of the non-intrusive ways to do this is through the eyes of the users. Considering the controlled illumination \acs{VR} provides unlike real world configurations, user cognitive load estimation based on pupillometry is one direction that could be investigated.

Based on the \acs{VR} driving experiment discussed in Section~\ref{subsection_main_SAP}, data annotations were carried out using pupil diameter measurements for low and high cognitive load, based on the time points that pedestrian warnings and pedestrian crossing started for experimental and control groups, respectively. An empirical timeshift of 0.8 seconds was introduced as well, considering that pupils did not dilate directly on time as the manipulations were introduced due to biological factors. Once the data of each participant was annotated for binary classification purposes, classifiers including Support Vector Machine (SVM), Decision Tree (DT), k-Nearest Neighbors (k-NN), and Random Forest (RF) were trained in a leave-one-participant-out cross-validation configuration and validated accordingly. The feature vectors included driver performance measurements incorporating accelerator and brake inputs as well as pupil diameter values. In the validation phase, accuracy, precision, recall, and F1-scores were calculated to assess performance of the classification tasks. In addition, to evaluate the real time working capability of each classifier, the cognitive load prediction time span of each test sample was calculated and averaged using all samples. 

\subsubsection*{Main Findings}
The assessment of the accuracy, precision, recall, and F1-scores revealed that cognitive load estimation for \acs{VR} setups has potential and could be applied in a person-independent way successfully. Particularly, \acs{SVM} performed the best with over 80\% accuracy, whereas \acs{DT} and \acs{RF} worked comparably with accuracies between 70-80\%. 1-NN, 5-NN, and 10-NN were evaluated for the \acs{k-NN} with 10-NN working the best with approximately 79\% accuracy.

Since these estimations should be applied in real time for user assistive tasks, real time working capabilities were assessed with a \acs{VR}-capable computer. On average, \acs{SVM}s and \acs{DT}s worked the fastest with approximately 0.3 milliseconds for estimating cognitive load per sample. Estimation times for \acs{k-NN}s with different k values took relatively longer with approximately 0.74 and 0.76 milliseconds, whereas \acs{RF}-based estimation took the most time per sample with approximately 5.4 milliseconds.  

\section{Privacy Preserving Eye Tracking for Virtual Reality}
\label{section_main_privacy}
In this section, differential privacy for eye tracking and privacy preserving gaze estimation based on eye landmark data are introduced in Section~\ref{subsection_main_diff_privacy} and Section~\ref{subsec_main_ppge}, respectively.

\subsection{Differential Privacy for Eye Tracking}
\label{subsection_main_diff_privacy}
This subsection is based on the paper~\ref{publist_lbl_PLOSONE} in Chapter~\ref{chapter_publications}, \emph{Differential privacy for eye tracking with temporal correlations} at \emph{PLoS ONE} in 2021.

\subsubsection*{Motivation and Main Methodology}
Differential privacy is a rigorous framework for protecting the information about whether an individual participated in a dataset or not~\cite{dwork2006,dwork_DP_only}. Considering a database that includes incomes of individuals, when data is queried for the mean value of the total number of individuals, the mean income for N individuals is obtained. If the same database is queried for N-1 individuals, using two mean values, an adversary could automatically infer the remaining individual's income. In differential privacy, privacy protection is achieved by adding randomly generated noise to the query outcomes based on a privacy parameter $\epsilon$ using, for instance, Laplace or Gaussian distributions, so that the query answers do not significantly change based on the participation of individuals. Provided privacy is increased when the $\epsilon$ value is decreased. On the one hand, the privacy of individuals is preserved thanks to differentially private mechanisms. One the other hand, due to the added noise data quality and utility are decreased. Therefore, it is crucial to find reasonable trade-offs. Formally, $\epsilon$-Differential Privacy ($\epsilon$-DP) is defined as follows.

\begin{definition}{$\epsilon$-Differential Privacy (\emph{$\epsilon$-DP~\cite{dwork2006,dwork_DP_only}}).}
A randomized mechanism $M$ satisfies $\epsilon$ differential privacy for all databases $D_1$ and $D_2$ which differ at most in one element for every $S \subseteq Range(M)$ with the following.
\begin{equation} 
\Pr[M(D_1) \in S] \leq e^{\epsilon} \Pr[M(D_2) \in S].
\end{equation}
\end{definition}

The amount of noise added depends on query sensitivities which are defined in the following.

\begin{definition}{\emph{Query sensitivity~\cite{dwork2006}.}}
For any random query of $X^n$ and $w \in \{1,2\}$, the query sensitivity ($\Updelta_{w}$) of $X^n$ is defined as the smallest value for all databases $D_1$ and $D_2$ that differ maximum in one element with
\begin{equation} 
||\, X^n(D_1) - X^n(D_2) ||_{w} \leq \Updelta_{w}(X^n)
\end{equation}
\end{definition}

The standard Laplace mechanism of differential privacy (i.e., Laplace perturbation algorithm (LPA)) achieves differential privacy for $\lambda = \Updelta_{1}(X^n)/{\epsilon}$~\cite{dwork2006} with noisy observations generated according to $\widetilde{X}^n = X^n(D) + Lap^{n}(\lambda)$, where $Lap^{n}(\lambda)$ consists of $n$ independent random variables from a Laplace distribution with a zero mean and variance $2\lambda^{2}$.

While this framework is widely used in database applications, it is not very straightforward to apply it to time-series data. Firstly, in time-series data, observations from different time points are temporally correlated. When standard differential privacy mechanisms are applied, an adversary can use this background information and make inferences on the data with the assumption of independent noise realizations on each time point. According to Zhao et al.~\cite{8269219} and Cao et al.~\cite{8333800}, there exist privacy leaks with the correlations and actual $\epsilon$ value for a pre-defined $\epsilon$ increases with such effect. Secondly, the required amount of noise to make time-series data differentially private significantly increases with long signals. As eye movements and the features that are extracted from raw eye movements are temporally correlated and can have long durations depending on the stimuli or individuals, standard differential mechanisms might perform poorly. In the previous work, there are privacy frameworks for correlated or sensor data such as Pufferfish~\cite{pufferfish_transaction_paper} or Olympus~\cite{olympus_framework_privacy}; however, these require different assumptions such as the necessity of a domain expert or modeling the privacy as adversarial networks. The Fourier perturbation algorithm (FPA) proposed by Rastogi and Nath~\cite{Rastogi:2010:DPA:1807167.1807247} deals with temporally correlated time-series data, and noisy observations are generated according to Algorithm~\ref{algomain_FPA}~\cite{Rastogi:2010:DPA:1807167.1807247,bozkir2020differential}, where DFT, IDFT, and PAD correspond to discrete Fourier transform, inverse discrete Fourier transform, and zero padding, respectively. Unlike the value claimed by Rastogi and Nath~\cite{Rastogi:2010:DPA:1807167.1807247}, the \acs{FPA} achieves $\epsilon$-differential privacy for $\lambda = \frac{\sqrt{n}\sqrt{k}\Updelta_{2}(X^n)}{\epsilon}$, with its proof in the Section~\ref{appendix:B1}.

\begin{algorithm}
  \KwInput{$X^n$, $\lambda$, $k$}
  \KwOutput{$\widetilde{X}^n$}
  $F^{k} = DFT^{k}(X^n)$. \\
  $\widetilde{F}^k = LPA(F^{k}, \lambda)$. \\
  $\widetilde{X}^n = IDFT(PAD^{n}(\widetilde{F}^k))$.
\caption{Fourier Perturbation Algorithm (FPA).}
\label{algomain_FPA}
\end{algorithm}

To solve these issues and make the eye movement features obtained from \acs{VR}/\acs{AR} setups differentially private, \acs{LPA}~\cite{dwork2006} and corrected \acs{FPA}~\cite{Rastogi:2010:DPA:1807167.1807247} are evaluated. Furthermore, two additional mechanisms, particularly chunk-based FPA and difference- and chunk-based FPA have been proposed and named as CFPA and DCFPA, respectively. The major purposes of these extensions are decreasing temporal correlations, the sensitivities required to achieve differential privacy, and ideally computational complexities as the chunk sizes are selected power of 2~\cite{Onurbio3}, namely 32, 64, and 128. The \acs{CFPA} applies the \acs{FPA} mechanism to each chunk, whereas the \acs{DCFPA} applies the \acs{FPA} to the consecutive difference signals within the chunks. Difference signals are observed to be significantly decorrelated, which implies that in this privacy mechanism, the privacy reduction due to the temporal correlations is less than the others. The \acs{CFPA} and \acs{DCFPA} processes are summarized in Figure~\ref{figmain_CFPA-DCFPA}~\cite[p.~9]{bozkir2020differential}.

\begin{figure}[!h]
   \includegraphics[width=1\linewidth]{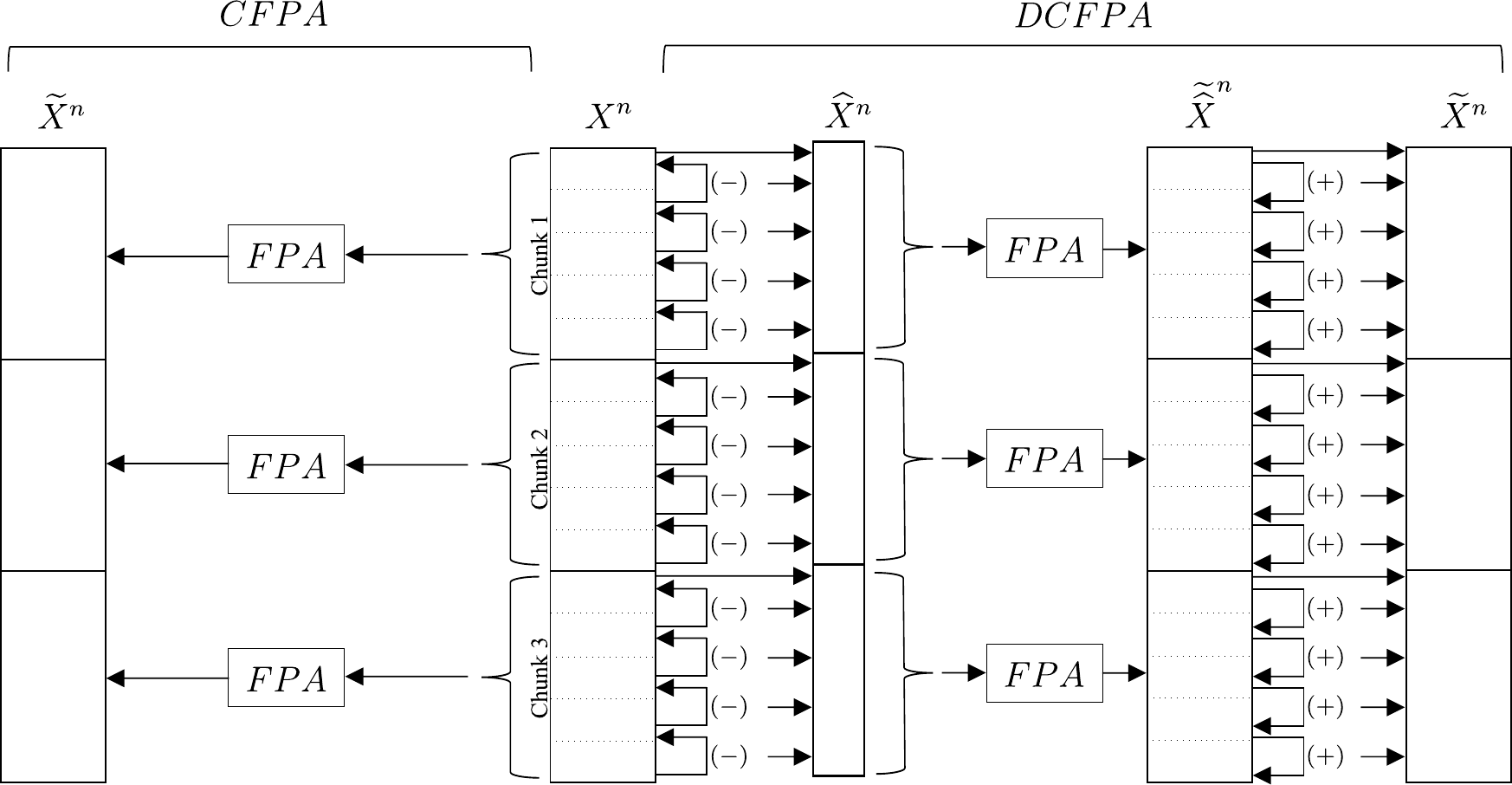}
   \caption{Workflow of the CFPA and DCFPA.}  
  \label{figmain_CFPA-DCFPA}%
\end{figure}

To evaluate proposed mechanisms, MPIIDPEye~\cite{steil_diff_privacy} and MPIIPrivacEye~\cite{Steil:2019:PPH:3314111.3319913} datasets have been used with five different privacy levels, namely $\epsilon = [0.48, 2.4, 4.8, 24, 48]$. The former dataset includes 52 eye movement features extracted from a reading task of three different document types (i.e., comics, news, and textbook) in \acs{VR}, whereas the latter dataset includes the same eye movement features from an \acs{AR} similar setup. As a significant amount of noise is introduced to the extracted features to achieve differential privacy, it is important to validate the usefulness of the private data. To do this, absolute normalized mean square error (NMSE) was used especially for comparison of different privacy mechanisms namely, \acs{LPA}, \acs{FPA}, \acs{CFPA}, and \acs{DCFPA}. However, while this metric shows the divergence trend of noisy data from the original data, it does not directly show how usable the private data is for different tasks. To find this out, for the MPIIDPEye, gender and document type classification tasks were employed. For the MPIIPrivacEye, privacy sensitivity detection of the viewed scene was carried out. In addition, person identification tasks were applied to evaluate whether it is possible to recognize the individuals using machine learning classifiers for both MPIIDPEye and MPIIPrivacEye. For these purposes, Support Vector Machines (SVMs), k-Nearest Neighbors (k-NNs), Decision Trees (DTs), and Random Forests (RFs) were used. Except for person identification tasks, all classifiers were trained and evaluated in a person-independent manner to ensure generic outcomes.

\subsubsection*{Main Findings}
As the analyses are split into two groups, including utility based on the \acs{NMSE} and classification accuracies, the findings are reported separately as well. However, before the usability of the data, data correlations are also analyzed. The difference signals used by the \acs{DCFPA} are observed to be less correlated than the original observations for both datasets. This means that the \acs{DCFPA} method is less vulnerable to temporal correlations in terms of privacy.

Utility evaluations based on the absolute \acs{NMSE} showed that the \acs{CFPA} and \acs{DCFPA} outperform the standard Laplace mechanism of differential privacy. All \acs{CFPA} evaluations outperform the \acs{FPA} as theoretically assumed due to reduced sensitivities. Different chunk sizes with \acs{CFPA} perform very similarly, therefore it is reasonable to use larger chunks as they better reduce correlations. The \acs{DCFPA}, especially with smaller chunks, surpasses other methods in the most private settings (i.e., $\epsilon = 0.48$). In the \acs{DCFPA}, smaller chunks (e.g., 32) perform significantly better than others in terms of absolute \acs{NMSE}.

For the MPIIDPEye dataset, three classification tasks, namely document type, gender classifications, and person identification, were applied. Ideally, individuals and their genders should not be recognized through private data. At the same time, since document type classification is treated as a utility task, the accuracy of this task should be as high as possible. Analyses have shown that when the \acs{FPA} is applied along with majority voting for person identification, very high accuracies (i.e., 100\%) are obtained. This is an indication that even with the addition of noise, it is still possible to recognize individuals. With the \acs{DCFPA}, it is not possible to identify individuals accurately either in high (e.g., $\epsilon = 0.48$) or low privacy ($\epsilon = 48$) regions. When the \acs{CFPA} is applied, it is possible to identify individuals accurately starting from $\epsilon = 24$. However, when high privacy regions are considered, classifiers fail to identify individuals as is the case for the \acs{DCFPA}. In terms of gender classification, it is only possible to detect gender to some extent (e.g., up to accuracy of $0.68$) with the \acs{CFPA} in the lowest privacy regions. Hence, almost all the methods are able to hide gender information unlike person identification. Document types are identified accurately with the \acs{FPA} in all privacy regions with accuracy over $0.85$ with random forests. The \acs{CFPA} works with an accuracy over $0.7$ in the low privacy regions, whereas the \acs{DCFPA} performs better with accuracies of $0.64$ and $0.69$ for high ($\epsilon = 0.48$) and middle ($\epsilon = 4.8$) privacy regions, respectively. Even though the \acs{FPA} works better than other methods for document type classification since it is possible to identify individuals very accurately when this method is used, one should consider either the \acs{CFPA} or \acs{DCFPA} when differential privacy mechanisms for eye movements are considered.

For the MPIIPrivacEye, privacy sensitivity classification works very similarly for the \acs{CFPA} and \acs{DCFPA} with accuracies in the vicinity of $0.6$ in all privacy regions. Similar to the MPIIDPEye dataset, when the \acs{FPA} is applied, it is possible to identify individuals very accurately with majority voting. The \acs{CFPA} and \acs{DCFPA} are able to hide personal identifiers successfully as in all privacy regions the person identification accuracies are close to random guessing probability. The accuracies without majority voting follow a similar trend with either higher or lower accuracies depending on the actual values not only for the MPIIPrivacEye dataset, but also for the MPIIDPEye dataset.

\subsection{Privacy Preserving Gaze Estimation based on Eye Landmarks}
\label{subsec_main_ppge}

This subsection is based on the paper~\ref{publist_lbl_ETRA20} in Chapter~\ref{chapter_publications}, \emph{Privacy preserving gaze estimation using synthetic images via a randomized encoding based framework} at \emph{2020 ACM Symposium on Eye Tracking Research and Applications}.

\subsubsection*{Motivation and Main Methodology}
Eye movements and the features extracted from raw data provide a lot of insight into individuals as discussed in the previous section. However, the focus of privacy mechanisms might not be hiding individual participation in a dataset in all setups like differential privacy. Instead, the focus may be protecting the data as whole, especially if the data is related to health status, medicine and etc. As eye movements in \acs{VR} can even be related to diseases~\cite{7829437}, one should consider such methods for privacy as well. Still, as it is beneficial to use eye movements for user assistive systems applied in \acs{VR} environments, for such cases machine learning models should be trained and tested without providing raw data with setups that include multiple parties. This may be due to a lack of processing power for each individual leading to classifier training in the cloud or a lack of data for specific tasks. 

The main purpose of this work is to show the applicability of eye tracking data analytics collected from VR setups with several parties. While the main focus is \acs{VR} \acs{HMD}s, data collected from other equipment such as smart glasses or optical see through displays can be used as well. One of the main prerequisites is the real time working capability of such frameworks. For validation, a gaze estimation task using a baseline model based on Support Vector Regression (SVR), with three parties including two input and one function party, is employed. Input parties are considered data providers for the function party to create and train gaze estimation models. The function party could be thought of as a cloud instance that processes data. For input data, UnityEyes~\cite{wood2016_etra} has been used to generate 20k synthetic eye images in total, similar to those obtained via eye tracker sensors inside \acs{HMD}s and two samples are visualized in Figure~\ref{figmain:generatedImgs}.

\begin{figure}[ht]
  \centering
   \subfigure[A synthetic eye gazing up.]{{\includegraphics[height = 3.75cm]{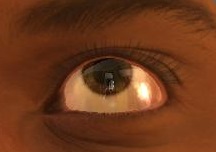}}}
   \qquad
   \quad
   \subfigure[A synthetic eye gazing down.]{{\includegraphics[height = 3.75cm]{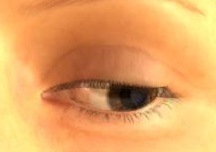} }}%
  \caption{Generated sample eye images with UnityEyes.}
  \label{figmain:generatedImgs}
\end{figure}  

From the generated images, 36 eye landmark-based features~\cite{Park2018} were used. After each input party extracts eye landmarks locally, communication between each party and the cloud instance starts. First input party, named ``Alice'', generates two random vectors and a value and sends these to the second input party, named ``Bob''. Both parties mask their extracted raw eye landmark features with these randomly generated values, and send them to the function party with the gram matrix of their samples. After obtaining the inputs from Alice and Bob, the function party computes the dot product of Alice's and Bob's samples and completes the remaining part of the gram matrix. Later, using the gram matrix, the function party trains the \acs{SVR} to estimate the gaze. The security framework employed in this work is inspired by Ünal et al.~\cite{unal2019framework}, and its security analysis is available in Section~\ref{appendix:B2}. As input parties primarily send the masked data and gram matrix of their samples, and as the function party does not know the generated values from Alice for masking, apart from training the \acs{SVR}, it cannot infer the raw data of Alice or Bob. Similarly, as input parties do not collude, it is not possible for them to make inferences about each other's raw data.

\subsubsection*{Main Findings}
Evaluations on different amounts of data, namely 5k, 10k, and 20k samples, show that when the number of samples are increased, mean angular error slightly decreases even with synthetic data. In particular, mean angular error decreased from 0.21 to 0.18 in the test time when number of samples increased from 5k to 20k. 

As this work is considered as a proof-of-concept for the applicability of a randomized encoding~\cite{applebaum2006cryptography,applebaum2006computationally}-based framework for the \acs{VR} and eye tracking domains, one of the most important issues is evaluating the execution time during testing. In the experiment with 20k samples, it took approximately 4.5 seconds to predict gaze direction of 4k test samples with a standard computer, which corresponds to approximately 1.1 milliseconds per sample.

\section{Accessibility of Virtual Reality}
\label{section_main_accessibility}
This section introduces the eye tracking data collection protocol suitable for remotely located participants in the context of accessible \acs{VR}.

\subsection{Remote Eye Tracking Data Collection Protocol for Virtual Reality}
This subsection is based on the paper~\ref{publist_lbl_AIVR20} in Chapter~\ref{chapter_publications}, \emph{Eye tracking data collection protocol for VR for remotely located subjects using blockchain and smart contracts} at \emph{2020 IEEE International Conference on Artificial Intelligence and Virtual Reality Work-in-progress papers}.

\subsubsection*{Motivation and Main Methodology}
Analyses on visual attention, cognitive processes using eye movements, and privacy preserving manipulation of eye movement signals have great potential to make \acs{VR} more accessible and available in everyday life. However, to make \acs{VR} collectively available, more people need access to \acs{VR} \acs{HMD}s and there should be potential to access more people's data, particularly eye tracking data as in the context of this work. With the recent COVID-19 pandemic, this issue has become more prominent than ever before. To do this, remote data collection protocols are needed for \acs{VR} setups.

Collecting such data remotely is not trivial. Firstly, it is important to keep the data quality high and comparable with laboratory studies. Secondly, data collectors should utilize mechanisms which guarantee that collected data is not altered by malicious users. In the laboratory studies, this is straightforward and trivial since subjects that participate in the experiments do not have access to applications and data collection pipelines. However, when remote collection is considered, subjects carry out tasks on their own computers and thus have opportunity to analyze any kind of application that maybe of interest. Thirdly, in such user studies, subjects are provided with small amount of money or gifts as compensation. In studies of human behavior that are carried out online such as crowdsourcing, services like \acs{AMT}~\cite{amt_online_link} are used for subject compensation. However, \acs{AMT} or similar services do not have direct support for \acs{VR} experiments and eye tracking data collection. Furthermore, such services introduce an extra layer between the experimenter and subjects, which is undesirable from a privacy perspective.

To enable remote eye tracking data collection, a protocol using white-box cryptography~\cite{wyseur_whitebox}, blockchain~\cite{bitcoin_whitepaper}, and smart contracts~\cite{eth_white_paper} has been designed. The overall workflow is as follows. Firstly, subjects obtain the \acs{VR} application from the data collector. After gaining access to the application, the experiment is carried through following instructions. The application informs subjects about the data quality and whether it is good enough for reporting it to the data collector. After this confirmation, subjects initiate data collection smart contract staking double the amount of compensation that they will acquire in the end. Then, the data collector approves the data collection process, staking double the amount of compensation that will be given. At this point, four units of compensation are locked in the blockchain and, as long as both parties do not fulfill their obligations, the staked amounts get stuck in the blockchain. Afterwards, the subject stores the hash value reported by the VR application in the blockchain then sends the collected data to the data collector along with the transaction ID of storing the hash value in the blockchain. The hash value reported to the subject is calculated using the data collector's secret key and white-box cryptography~\cite{whitebox_crypto_alex_biryukov,wyseur_whitebox} so that even though subjects infer the hash function, they do not have the ability to construct a fake hash value for altered data. Therefore, they must behave honestly which means they are not expected to alter the collected data. After obtaining the data and transaction ID of the hash storage on the blockchain, the data collector checks whether the hash value stored by the subject in the blockchain and the local hash value calculated by using the selected secret key overlap. When they match it means that the subject is honest, hence the data collector can confirm the data collection in the smart contract. At this point, the smart contract unlocks the compensation and distributes three units of compensation to the subject and one unit to the data collector. At the end of the process, collected data is transferred to the data collector and subjects receive their compensation without any centralized third party service processing their data or intervening in the compensation management. The aforementioned protocol was realized on the Ropsten Testnet of Ethereum platform~\cite{eth_white_paper} with dummy synthetic data, and is available as follows: \url{https://ropsten.etherscan.io/address/0x0e937a4a4618dd8d5a12ec4a9f8fd61d6bfd13e4}.

\subsubsection*{Main Findings}
Unlike the other studies, this study has shown that remote \acs{VR} data collection for eye tracking is possible with just a little more effort than laboratory studies. For this purpose, \acs{VR} applications should calculate the quality of eye tracking data (e.g., tracking ratios) at the end of each experimental session and report it for further processes. In addition, cryptographic implementations are needed to guarantee that subjects do not alter the data. Lastly, use of blockchain and smart contracts shows that, with a little more effort, compensation distribution along with data transfers can be realized without a centralized service and are a credit to the proposed protocol.

\cleardoublepage

\chapter{Discussion} 
\label{chapter_discussion}
In this chapter, the papers that are summarized in Chapter~\ref{chapter_motivation_findings} and presented in Chapters~\ref{appendix_A},~\ref{appendix_B}, and~\ref{appendix_C} are discussed within the umbrella term of ``Everyday Virtual Reality''. Findings on visual attention and cognitive processes based on eye movements in education and driving domains are discussed in Section~\ref{sec:discussion_visAttention}, based on the papers~\ref{publist_lbl_CHI21},~\ref{publist_lbl_VR21} and~\ref{publist_lbl_SAP19},~\ref{publist_lbl_VRW19} in Chapter~\ref{chapter_publications}, respectively. Implications of privacy preserving eye tracking for \acs{VR} focusing on differential privacy and a randomized encoding-based framework are explained in Section~\ref{sec:discussion_privacy}, based on the papers~\ref{publist_lbl_PLOSONE} and~\ref{publist_lbl_ETRA20}, respectively. Then, how these mechanisms could be combined with the everyday \acs{VR} setups in terms of remote data collection is discussed in Section~\ref{sec:discussion_accessibility}, standing on the paper~\ref{publist_lbl_AIVR20} in Chapter~\ref{chapter_publications}. Finally, the outlook in the realm of all covered topics is drawn in Section~\ref{sec_discussion_outlook}.

\section{Visual Attention and Cognitive Processes}
\label{sec:discussion_visAttention} 
Visual attention and cognitive process related research questions were investigated through two experiments in the education and driving domains. Therefore, the discussion is split into two subsections.

\subsection{Virtual Reality in the Classroom Context}
\label{subsec_discuss_VR_education}
\acs{VR}-based classrooms not only offer the possibility to make online and remote learning more immersive and interactive, they also support the studying different classroom manipulations that are difficult to generate in the real world. While some manipulations are directly related to \acs{VR} environments rather than real world scenarios, such as the visualization styles of avatars (e.g., cartoon or realistic), it is also possible to control, for instance, the number of hand-raising peer-learners during the experience, which could affect student self-concept~\cite{self_concept_1976} in the long run. It is also relatively straightforward to create conditions which typically happen in the real world such as locating students attending to a virtual lecture in the back or front of the virtual classroom.

In the aforementioned studies, these three classroom manipulations have mainly been studied by taking human head and eye movements into account. Locating students in the back or front of the virtual classroom has yielded different implications. The participants sitting in the back had longer fixation durations during the lecture, which implies that they spent more time processing information. This may be related to the relationship between mean fixation durations and task difficulty~\cite{task_difficulty_fixation_durations}. While the task is essentially the same for all participants, participants sitting in the back view the lecture content through a smaller field of view and may have difficulty extracting information. On the contrary, the back sitting participants had shorter saccades and smaller saccade amplitudes, meaning that they shifted their attention less than those sitting in the front which could also be related to content size in their field of view. Furthermore, the back sitting participants engaged significantly more with the virtual peer-learners whereas the front sitting participants engaged more with the virtual instructor and lecture screen. Considering that the virtual instructor and screen are more related to learning lecture content and considering the increased fixation time in the back sitting condition, while designing virtual learning spaces, one might locate students in the frontal regions of the classrooms if visually attending the lecture content is important. If engagement with virtual peers is more important, one might favor either locating the students in the back of the classroom or even organizing the desks in a U-, V-, or O-shape so students can have a better view of their virtual peers. In summary, according to visual attention and cognition findings, the location of the students in the classroom should be determined based on the goal of the virtual lecture.

The avatar representation styles are directly related to \acs{VR} environments because it is not possible to have such configurations in real classrooms. According to the findings of this manipulation, students engaged more with the environment in general with longer fixation and shorter saccade durations when cartoon-styled avatars were presented. While the opposite trend between fixation and saccade durations fits with theoretical expectations, the students that encountered cartoon-styled avatars also engaged more with the peer-learners than the students encountered with realistic-styled avatars. Considering that the students who attended the virtual lecture were small children, engaging more with the cartoon-styled peer-learners is a reasonable and explainable outcome. On the contrary, the students that encountered realistic avatars had larger pupil diameters indicating a higher cognitive load in general. Taking the reasonably controlled illumination of the students' sitting positions into account, the findings indicate that the lecture with realistic characters increased focus and concentration in the learning space. Similar to sitting positions, when such environments are designed, together with the target groups' demographic information, visualization styles should be tailored to the dedicated lecture type, such as interacting with either engaging peers or realistic instructors.

The hand-raising behaviors of the peer-learners can theoretically be manipulated in virtual and real classrooms. However, it is very challenging to have a controlled experiment for such a condition in the real world. Secondly, if these behaviors are manipulated in the real classroom students are already biased by knowing the typical performance of their classmates when reacting to different topics. The behaviors related to such ``artificial'' manipulation might not provide naturalistic conclusions. Therefore, it is more feasible to study these manipulations virtually. However, according to the findings on visual behavior, much research is needed. The results indicate that the extreme hand-raising behaviors of peer-learners yielded higher cognitive load in students. It is likely that when a moderate number of peers raise their hands for questions during the lecture, students find these actions to be straightforward and their level of focus is less than in the extreme levels of hand-raising behaviors. In terms of visual distraction that could occur with many peer-learners raising their hands, the 80\% hand-raising condition gathered the most attention on the peer-learners. This indicates that the efforts of virtual peer-learners to participate in the lecture were noticed by the students. As theoretically expected, this yielded the least attention on the virtual instructor and on the virtual lecture screen. However, the effects of this manipulation are relatively mixed and should be further investigated. In the case of continuous \acs{VR} class attendance, hand-raising should be calibrated carefully due to the possible impact on student self-concepts. Furthermore, one might consider an adaptive manipulation depending on the individual, subject, or topic in the context of everyday \acs{VR}.

In summary, the findings indicate that human visual attention and cognition differ significantly when such educational manipulations are introduced, contributing to the state-of-the-art. While the optimal configurations may depend on multiple factors like the target group characteristics (e.g., ages, \acs{VR} experience, and etc.), lecture content, and the mainstream drawbacks of virtual environments with \acs{HMD}s, such as limited field of view and possible cybersickness when confronted with long durations, such environments have great potential in the digital era. Considering the switch to digital teaching during the recent COVID-19 pandemic, the increasing number of online classes, and even the efforts to generate classroom twins for \acs{VR}~\cite{ahuja_etal_chi21_classroom_twins}, it seems that new teaching and learning paradigms are on the way in our daily life.

\subsection{Virtual Reality in Driving}
\label{subsec_discuss_VR_driving}
The driving studies conducted have shown that multiple implications can be drawn using eye movements not only for driving configurations, but also for time dependent and critical tasks that can be carried out in the everyday \acs{VR} context.

While driving in the real world, many modern vehicles provide driver assistance features such as collision warning and lane keeping. However, in real life, due to human safety, it is not possible to train people for safety critical situations. Humans can train for interactions with pedestrians, driving with automated vehicles, or overtaking scenarios with low-cost \acs{VR} \acs{HMD}s easily at their homes without any significant consequences. To this end, a first step towards interactions with critically crossing pedestrians has shown that even the gaze-aware and minimalistic warning cues for critical pedestrians, which are located around the periphery, help drivers to drive safer and smoother according to assessment of pupil diameters, driver pedal inputs, and distances to the critically crossing pedestrians. Furthermore, during these scenarios cognitive load estimation can be carried out accurately and in real time. While cognitive load assessment for real world situations may be viewed skeptically due to effect of illumination on human pupil sizes, thanks to the relatively controlled illumination that can be provided in \acs{VR}, it is possible to make use of pupil dilations in such setups. Considering that even the Formula 1 drivers practice with simulators~\cite{f1_drivers_simulator} (e.g., Lando Norris practicing with a simulator~\cite{lando_norris_practice}), and their visual attention through eye movements is considered almost superhuman~\cite{F1_tobii_nico_hulkenberg}, and as such similar setups may be used in daily lives for different purposes such as entertainment, it can be argued that ordinary people and particularly novice drivers can train for many different traffic scenarios with the help of relatively low-cost \acs{VR} simulators and eye movements.

As the shift of visual attention in a time dependent context is not necessarily related to driving, there are implications for other everyday \acs{VR} scenarios. Currently, even though \acs{VR} and \acs{HMD}s have several disadvantages such as low resolutions, vergence-accomodation conflict~\cite{kramida_VAC_in_HMDs}, and significant weight of the \acs{HMD}s, with small cues it is possible to shift attention quickly towards important regions of the presented 3D stimulus. In the driving context, this is validated with pedestrians. Other scenarios could include an attention shift for a student attending a class in \acs{VR} to support the overall learning process, for a video gamer to notify important milestones during the game, or for novices in more or less any domain.

\section{Privacy Preserving Eye Tracking}
\label{sec:discussion_privacy} 
Differential privacy mechanisms applied to eye movement features and privacy preserving gaze estimation based on a randomized encoding-based framework are discussed in the following subsections.

\subsection{Differential Privacy}
\label{subsec_discuss_diff_privacy}
Findings on the application of differential privacy mechanisms such as \acs{FPA}, \acs{CFPA}, or \acs{DCFPA} indicate that it is not very trivial to have high utility while preserving privacy due to high amount of noise required by the differential privacy mechanisms applied to correlated time-series data. These effects are discussed in the following based on two evaluation metrics, namely the utility metric based on the \acs{NMSE} and classification accuracies of different tasks.

Firstly, from a higher utility aiming perspective application of chunking in the \acs{CFPA} and \acs{DCFPA} significantly decreased the amount of noise needed to make the eye movement query outcomes differentially private. This can been seen when the \acs{CFPA} is compared with \acs{FPA} and the different chunk sizes (i.e., 32, 64, and 128) within the \acs{DCFPA} are evaluated using the \acs{NMSE}-based utility metric. While chunking could be considered a method from the signal processing domain, according to the Parallel Composition Theorem~\cite{McSherry:2009:PIQ:1559845.1559850}, since the chunks are not overlapping, the differential privacy is preserved. Using larger chunks decorrelates the data more effectively therefore, it should be preferable when the utilities are similar between various chunk sizes. However, when the \acs{DCFPA} is considered, apart from the chunking mechanism since differences between consecutive observations are used, and when the noisy values are propagated within each chunk to obtain final noisy observations, the Sequential Composition Theorem~\cite{McSherry:2009:PIQ:1559845.1559850} is applied for overall privacy and $\epsilon$ calculations. Therefore, when the chunk sizes are larger more noise is needed to achieve differential privacy in this method. While this is a disadvantage due to the divergence of overall eye movement signals from the original signals, since it has been empirically determined that the eye movement difference signals are less correlated compared to the original signals, the privacy leak for this method due to data correlations is less than in others methods. Overall, based on the \acs{NMSE}-based utility metric, since the divergence of the differentially private eye movements is less for the \acs{CFPA} and \acs{DCFPA} compared to the \acs{FPA} and the standard Laplace mechanism of the differential privacy, it is better to use these methods for eye movements. However, there are many trade-offs such as different chunk sizes and $\epsilon$ values, namely different privacy regions. The methods should be tailored according to the eye movement feature generation pipelines, and possibly further tasks that are applied in the everyday VR context.

According to the discussed utility perspective, it is optimal that differentially private eye movement signals are less diverged from original eye movement signals. When it comes to classification tasks, however, the overall goal is more complex than the \acs{NMSE}-based utility metric because there are different tasks such as the classification of document types or person identification. For instance, while it is desirable to have high accuracies in document type or privacy sensitivity classification (e.g., as applied for MPIIDPEye~\cite{steil_diff_privacy} and MPIIPrivacEye~\cite{Steil:2019:PPH:3314111.3319913} datasets), low accuracies in gender prediction and person identification are preferred when the privacy of individuals is considered. Comparing differential privacy mechanisms that are appropriate for time-series data, namely \acs{FPA}, \acs{CFPA}, and \acs{DCFPA}, it is possible to have almost perfect accuracies for person identification tasks with the \acs{FPA} in both datasets. Even if other classification accuracies are obtained in a preferable success, it is likely that the \acs{FPA} is vulnerable to person identification attacks. On the contrary, both the \acs{CFPA} and \acs{DCFPA} significantly decrease person identification accuracies towards guessing probabilities. All the mechanisms successfully hide gender information in a person-independent cross-validation setup in the high privacy regions, which is expected since, even with the clean data, the accuracies are in the vicinity of 70\% also according to the previous work~\cite{steil_diff_privacy}. The \acs{FPA} works over 85\% accuracy in the document type classification task which is comparably higher than the \acs{DCFPA} with 64\% accuracy in the most private regions (i.e., $\epsilon = 0.48$); however, due to its lack of resistance to attacks on person identification, one should determine on a trade-off when using the \acs{FPA}. As the human reading behaviors consist of ``Z''-type (or similar) patterns and the \acs{CFPA} and the \acs{DCFPA} perturbs the eye movement data with chunks, it is suspected that such patterns are removed easily with these mechanisms unlike the \acs{FPA}. Therefore, this is a task-specific outcome of the evaluated methods. In the privacy sensitivity detection solely based on the differentially private eye movements, both the \acs{CFPA} and \acs{DCFPA} outperform the \acs{FPA}, which performs state-of-the-art in terms of differential privacy perspective in the eye tracking domain. Overall, there are multiple trade-offs to consider before applying differential privacy mechanisms on the eye movements. These include further tasks, the stimulus information from the original eye movements that were collected, data correlations, and the amount of background information that an adversary may have on the data. Especially when more practical use-cases in everyday life are considered, practitioners that design privacy mechanisms should take the latter issue into account and propose privacy solutions accordingly.

\subsection{Randomized Encoding}
\label{subsec_discuss_smc}
Unlike differential privacy, if the complete raw data needs to be private cryptographic approaches should be employed. The randomized encoding (RE)-based framework that is utilized falls into this category because the raw data should not be available, for example, to function parties (e.g., third party cloud or server instances). In principle, since eye movements and eye tracking data obtained from \acs{VR} \acs{HMD}s represents visual biometrics, any type of additional task working with encrypted raw eye tracking data could be employed. As aforementioned, when such a system is built for training and testing machine learning models, and a real time interaction mechanism is needed, test times should fit this expectation. For this reason, instead of evaluating more sophisticated tasks such as cognitive load detection, gaze prediction, foveated rendering or similar tasks for everyday \acs{VR} setups, a fundamental gaze estimation task was chosen. This is because gaze estimation is the starting point for all other eye movement related tasks in the \acs{VR} setups.

The use-case of privacy preserving gaze estimation via the randomized encoding based framework has shown that it is possible to estimate gaze using the baseline \acs{SVR} model identically to the non-private version (See the proof in Section~\ref{appendix:B2}). While the decreasing trend of mean angular error is anticipated with a higher amount of data, the most important discussion point is its real time working capability. A prediction time of approximately $1.1$ milliseconds falls in the range of a real time capable system in the \acs{VR} domain. However, since this is only the prediction time on the function party instance, a possible communication latency is introduced during an everyday application scenario. It is assumed that as long as efficient communication between input and function parties is available, the proposed work will function in real time. In addition, from a machine learning perspective the input feature size is 36 for this work and is dedicated to the baseline gaze estimation task. Feature set size might also be important depending on the task and configuration. In the eye tracking literature, there are generic features based on fixations, saccades, pupil diameters, and blinks. For instance, in the works of Steil et al.~\cite{steil_diff_privacy,Steil:2019:PPH:3314111.3319913}, 52 eye movement features are used for estimating stimulus type, gender, or scene privacy accurately based on \acs{SVM}s. As the used feature vector sizes and machine learning models overlap for different tasks, our results imply that all these tasks could be done with multiple input parties as well while preserving the privacy. Such approaches not only protect data privacy, but also help to increase the training data as the total training data consists of data from multiple parties. Therefore, in the case of a lack of data for training models, these frameworks can be utilized as well. The proposed scenario and use-case stands as the only work in the current literature on the intersection of cryptography, randomized encoding, and eye tracking in everyday \acs{VR}.

\section{Remote Eye Tracking Data Collection Possibilities for Virtual Reality}
\label{sec:discussion_accessibility} 
Obtaining eye movement behaviors of human subjects with heterogeneous backgrounds is indeed an important challenge for data driven \acs{VR} systems such as user-assistive or gaze-guidance based on machine learning. Usually these types of studies are conducted within a small group of human subjects of similar ages and backgrounds who might not be representative of different populations. The protocol proposed in Section~\ref{section_main_accessibility} offers a first step to solve this issue and make crowdsourcing possible with high-end \acs{VR} \acs{HMD}s.

As the proposed protocol assumes that the data collectors provide subjects with the \acs{VR} application via a secure and direct method, the advertisement of the experiment should be completed externally via mailing groups, online forums, and etc. While these actions require extra effort to attract \acs{HMD} users for the experiments, Rivu et al.~\cite[pp.~20-21]{rivu2021remote} have reported that this method has the potential for independent studies that require different technical functionalities. However, one should be aware that attracting \acs{HMD} users in such a way will likely create a sample set of experienced \acs{VR} and \acs{HMD} users which may introduce a confounder on user studies that evaluate eye movements.

Another important discussion point is the increased effort required by researchers and developers. For instance, the white-box cryptography~\cite{wyseur_whitebox} paradigm should be implemented within the \acs{VR} application to prevent data altering attacks. This requires the selection of secret keys, a hashing algorithm, and the actual implementation of a cryptographic approach. However, in the future it is likely that software packages will be available for such purposes that every \acs{VR} application can benefit from. In addition, in order to increase the data quality coming from remote participants, practitioners should implement processing pipelines for quality checking and filtering of eye movement data. For instance, one might filter the data by setting thresholds using tracking ratios or eye openness as reported by eye tracking sensors. Although doing these actions within the \acs{VR} application might be seen as overhead by some in general, they are necessary to guarantee the validity of remote data collection.

Lastly, our protocol makes use of blockchain and smart contracts. In particular, we have used Ethereum platform~\cite{eth_white_paper} and its smart contract functionality for our use-case. Blockchain is needed for its immutability for recording the hash value of the collected eye movement data. Smart contracts are utilized for compensation in experiments without a centralized authority between data collector and subjects. While any blockchain-based platform that supports smart contract functionality could be used for this protocol, compensations are distributed in cryptocurrencies in any case. Due to the fact that even well-established cryptocurrencies such as Bitcoin~\cite{bitcoin_whitepaper} and Ether~\cite{eth_white_paper} are highly volatile in value, one might be skeptical about using them for such purposes. In addition, they have not yet been embraced for daily usage by large communities. However, going forward it is possible that such currencies will be used in daily life more frequently (e.g., recently El Salvador have accepted Bitcoin as legal tender along with American dollar~\cite{btc_legal_tender_news}), that communities will soon be more familiar with them. Additionally, such technologies disable the extra layer of centralized institutions between parties, particularly for financial transactions, and decrease the cuts applied by these centralized entities. Another advantage is that, in the case of public blockchain usage, the compensation distribution process is directly transparent via the web. Therefore, while some may consider the use of cryptocurrencies in everyday \acs{VR} experiments utopic or futuristic, there is great potential.

\section{Outlook}
\label{sec_discussion_outlook}
Virtual reality and eye tracking research requires interdisciplinary work. Both research areas consist of components related to computer hardware, computer graphics, human-computer interaction, cognitive science, artificial intelligence, and psychology. On top of these, fields including cryptography and security are involved in the privacy preservation of eye movement data collected from \acs{HMD}s. This work is one of the first that combines these multiple aspects to create an overall framework towards everyday virtual reality. In terms of visual attention and cognition, conventional measures related to fixations, saccades, pupil diameters, object-of-interests similar to area of interests, etc., have been customized and used in the evaluation of \acs{VR} scenarios in education and driving domains along with machine learning and other sensing modalities depending on the context. Privacy preserving paradigms cover two areas: differential privacy and the utilization of a randomized encoding-based framework. In differential privacy, the privatized features are aggregated statistics of fixations, saccades, pupil diameters, or blinks similar to those used in education and driving studies, apart from the feature extraction timespans. The main reason for this is that in the attention and cognition studies, the focus was human visual behaviors throughout the complete experiments, whereas in the differential privacy context, the main goal was time-series representations of such statistics. Contrary to the features used in the differential privacy work, gaze estimation utilizing an \acs{RE}-based framework focused on eye landmarks for a straightforward reason. Estimation of gaze is the initial point of extracting features related to fixations and saccades. At the same time, if the eye movement data is completely encrypted, it is reasonable to encrypt the initial step rather than moderate or final steps of the data processing pipeline. After privacy preserving gaze estimation, one can extract features locally and use them for further analyses with privacy guarantees. Such works on preserving the privacy of individuals that intend to use \acs{VR} related technologies, along with eye tracking, will likely enable usage from wider communities and help \acs{VR} become more accessible in everyday life. Lastly, the remote data collection protocol proposed for eye tracking and \acs{VR} spans all the works as it is possible to collect either raw eye tracking or feature-related data with such protocols. It is even possible to collect data outside of eye tracking, as long as the data quality can be controlled within \acs{VR} applications. As the presented work has multiple aspects, the outlook for each component is discussed separately. Each of these aspects also contribute technically to recent discussions on Metaverse~\cite{snowcrash1992, metaverse_acmcompsurv_2013, all_one_needs_to_know_about_metaverse_21} and they should be taken into consideration.

In the education works, the observed virtual classroom was pre-scripted and did not include a real interaction experiencing students' perspectives. In addition, the peer-learners in the classroom were simple pre-scripted bots. While the results are very important and the first for such setups from a visual attention and cognition perspective, more sophisticated setups could be employed. For instance, pre-scripted peer-learners could be replaced with smart agents, namely computer games like intelligent bots to study attention towards such agents. Furthermore, while synchronization may not be trivial using these custom environments, to generate a more realistic setup each peer-learner and the instructor could be connected to a real person using web-based technologies to replicate a complete virtual and remote classroom environment with actual classmates and teacher. In this case, one should employ other sensing modalities such as hand tracking and audio to ensure interactive virtual environments. Such configurations will also help make \acs{VR} viable in everyday life considering the immersion and real person interactions it provides similar to conventional classrooms in the remote teaching and learning scenario.

The driving research that was conducted is related to setups that one might encounter in real daily driving. The end goal is to enable humans to train for similar time critical situations by making use of eye movements not only in driving, but also in other domains. In terms of driving, future works could focus on interactions within other critical situations such as take-over scenarios, complex traffic situations that include multiple critical pedestrians or vehicles, and even interactions with semi- and fully-automated vehicles. There are works in the direction of automated vehicles, in particular, with more sophisticated and expensive driving simulators (e.g.,~\cite{8082802,7313360}); however, it is an open question as to how visual behaviors and cognition would appear in the case of everyday \acs{VR} with \acs{HMD} scenarios. At the same time, with the growing market size for autonomous vehicles~\cite{Autonomous_vehicles_marketsize_research}, interactions between the human driver and semi-autonomous vehicles will be key, and there is a research gap in this direction. With what \acs{VR} provides, it is likely that such interactions could be studied to ensure a deeper understanding of these relationships in daily life. Additionally, the effects of gaze-aware and minimalistic human visual support, as featured in the aforementioned driving work, could be studied for domains other than driving. Lastly, \acs{VR}-based training packages for critical scenarios could be utilized and evaluated over time.

Differential privacy has been studied in many domains in addition to human-computer interaction. What is proposed in this work is mainly for temporally correlated eye movement feature signals and for reducing the effects of temporal correlations on individual privacy. However, considering substantial improvements on various tasks with deep neural networks, one could employ differential privacy along with deep learning~\cite{goodfellow2016deep}, similar to Abadi et al.'s work~\cite{10.1145/2976749.2978318} while using the eye movements and tasks for \acs{VR} setups. Taking into account the recent open source libraries for privacy such as Opacus~\cite{opacus_online_link}, which enables model training with differential privacy, this direction may be more relevant and focused in the near future. In addition, local differential privacy~\cite{Joseph_Roth_Ullman_Waggoner_2020} scenarios could be employed for not only eye tracking, but also for any type of biometric data collected from the \acs{VR} \acs{HMD}s. As it is required to noise the data locally in local differential privacy, utility and noise dynamics are different and should be studied explicitly.

Privacy preserving gaze estimation work has focused on two input parties and a function party which could be thought of as a cloud instance. However, in real world scenarios the number of input parties may be more than two. Considering many \acs{HMD} users wish to train models while preserving their privacy, this scenario has similarities to N number of hospitals that share their data privately for further processes as proposed by Chen et al.~\cite{10.1145/3375708.3380316}. Such directions are possible for eye tracking data collected from \acs{VR} displays as well. Furthermore, instead of focusing on \acs{SVM}s, different classifiers such as decision trees, random forests, and artificial neural networks could be considered. While the proposed gaze estimation method based on a randomized encoding framework is the first for the eye tracking domain, privacy preserving eye tracking and human-computer interaction is indeed a green field in terms of these research directions in the context of everyday \acs{VR}.

The blockchain-based work for collecting eye movement data from remote participants is a proof-of-concept protocol. Therefore, there are multiple ways to extend this work. The first step is to actually implement the complete workflow within a \acs{VR} application and test its usability by employing self-reported measures provided by participants. While interactions with blockchains and smart contracts are trivial in principle, they are considerably new technologies. The blockchain concept was introduced with Bitcoin~\cite{bitcoin_whitepaper} and smart contracts for blockchains with Ethereum~\cite{eth_white_paper} in the last two decades. Therefore, the usability of such technologies in \acs{VR} applications is a relevant direction to explore. Even though there are some works that explore the user experience of cryptocurrency wallets with augmented reality~\cite{crypto_ar_wallet} and general use-cases for interaction design~\cite{foth_blockchain_for_interaction_design}, not much research has been conducted for the aforementioned. GazeCoin~\cite{gaze_coin_white_paper}, a cryptocurrency for payments based on the eye gaze for \acs{VR}/\acs{AR} was recently introduced as a token. This also shows the potential of such technologies for \acs{VR} and everyday use-cases. For the proposed protocol, while the Ethereum platform and its blockchain were used, newer platforms such as Avalanche~\cite{avax_white_paper} and Polkadot~\cite{polkadot_whitepaper} could be employed along with their extended functionalities. Apart from these, one might strive for more intuitive ways to guarantee data integrity and a decentralized method of experiment compensation in future works.

\section{Conclusion}
\label{sec_discussion_conclusion}
A significant contribution to the scientific research of everyday \acs{VR} using eye tracking was carried out in the context of this thesis. This encompasses research on human attention and cognition understanding based on eye movements and features in multiple domains, including education and driving, privacy preserving manipulations of eye movement features and signals with their algorithmic foundations, and a versatile protocol that may help make \acs{VR} more accessible to a wide range of human participants with different socio-demographic backgrounds.

\addtocontents{toc}{\vspace{\normalbaselineskip}}
\cleardoublepage
\bookmarksetup{startatroot}

\appendix
\chapter{Visual Attention and Cognition in VR through Eye Tracking}
\label{appendix_A}

This chapter includes the following publications:
\vspace{1cm}
\begin{enumerate}
	\item\label{appendix_CHIpaper_label} Hong Gao*, \textbf{Efe Bozkir*}, Lisa Hasenbein, Jens-Uwe Hahn, Richard
	Göllner, and Enkelejda Kasneci. Digital transformations of classrooms in virtual reality. In~\emph{Proceedings of the 2021 CHI Conference on Human Factors in Computing Systems (CHI)}, New York, NY, USA, 2021. ACM. doi: 10.1145/3411764.3445596.
	
	\item\label{appendix_VRpaper_label} \textbf{Efe Bozkir*}, Philipp Stark*, Hong Gao, Lisa Hasenbein, Jens-Uwe Hahn, Enkelejda Kasneci, and Richard Göllner. Exploiting object-of-interest information to understand attention in VR classrooms. In~\emph{2021 IEEE Virtual Reality and 3D User Interfaces (VR)}, New York, NY, USA, 2021. IEEE. doi: 10.1109/VR50410.2021.00085.
	
	\item\label{appendix_SAPpaper_label} \textbf{Efe Bozkir}, David Geisler, and Enkelejda Kasneci. Assessment of driver attention during a safety critical situation in VR to generate VR-based training. In~\emph{ACM Symposium on Applied Perception (SAP)}, New York, NY, USA, 2019. ACM. doi: 10.1145/3343036.3343138.
	
	\item\label{appendix_VRWpaper_label} \textbf{Efe Bozkir}, David Geisler, and Enkelejda Kasneci. Person independent, privacy preserving, and real time assessment of cognitive load using eye tracking in a virtual reality setup. In~\emph{2019 IEEE Conference on Virtual Reality and 3D User Interfaces (VR) Workshops}, New York, NY, USA, 2019. IEEE. doi: 10.1109/VR.2019.8797758.

\end{enumerate}

\blfootnote{
\hspace{-14pt}{\scriptsize * indicates equal contribution.\\}
{\scriptsize Publications are included with minor templating modifications. Definitive versions are available via digital object identifiers at the relevant venues. Publications \ref{appendix_CHIpaper_label} and \ref{appendix_SAPpaper_label} are \textcopyright~2021 ACM and \textcopyright~2019 ACM, respectively, and included with relevant permission. Publications \ref{appendix_VRpaper_label} and \ref{appendix_VRWpaper_label} are \textcopyright~2021 IEEE and \textcopyright~2019 IEEE, respectively, and reprinted, with permission, from \ref{appendix_VRpaper_label} and \ref{appendix_VRWpaper_label}. In reference to IEEE copyrighted material which is used with permission in this thesis, the IEEE does not endorse any of University of Tübingen’s products or services. Internal or personal use of this material is permitted. If interested in reprinting/republishing IEEE copyrighted material for advertising or promotional purposes or for creating new collective works for resale or redistribution, please go to \url{http://www.ieee.org/publications_standards/publications/rights/rights_link.html} to learn how to obtain a License from RightsLink. If applicable, University Microfilms and/or ProQuest Library, or the Archives of Canada may supply single copies of the dissertation.}
}

\newpage

\section[Digital Transformations of Classrooms in Virtual Reality]{Digital Transformations of Classrooms in Virtual Reality}
\label{appendix:A1}

\subsection{Abstract}
With rapid developments in consumer-level head-mounted displays and computer graphics, immersive \acs{VR} has the potential to take online and remote learning closer to real-world settings. However, the effects of such digital transformations on learners, particularly for \acs{VR}, have not been evaluated in depth. This work investigates the interaction-related effects of sitting positions of learners, visualization styles of peer-learners and teachers, and hand-raising behaviors of virtual peer-learners on learners in an immersive \acs{VR} classroom, using eye tracking data. Our results indicate that learners sitting in the back of the virtual classroom may have difficulties extracting information. Additionally, we find indications that learners engage with lectures more efficiently if virtual avatars are visualized with realistic styles. Lastly, we find different eye movement behaviors towards different performance levels of virtual peer-learners, which should be investigated further. Our findings present an important baseline for design decisions for \acs{VR} classrooms.

\begin{figure}
  \includegraphics[width=\textwidth]{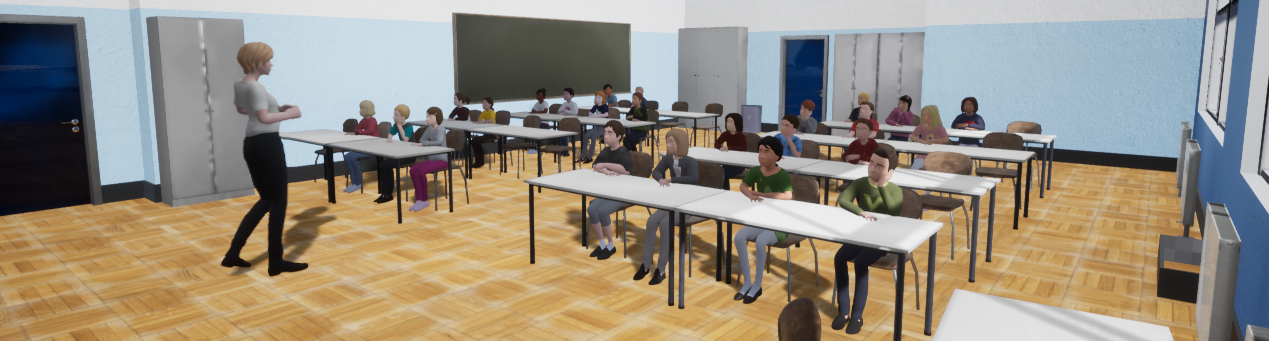}
  \caption{Immersive virtual reality classroom.}
  \label{fig:teaser_CHI21}
\end{figure}

\subsection{Introduction}
Recently, many universities and schools have switched to online teaching due to the COVID-19 pandemic. Online and remote learning may become more prevalent in the near future. However, one of the disadvantages of teaching and learning in such ways compared to conventional classroom-based settings is the limited social interaction with teachers and peer-learners. As this may demotivate learners in the long term, better social engagement providing solutions such as immersive virtual reality (IVR) can be used for teaching and learning. Next-generation \acs{VR} platforms such as Engage\footnote{https://engagevr.io/} or Mozilla Hubs\footnote{https://hubs.mozilla.com/} may offer better social engagement for learners in the virtual environments; however, the effects of such environments on learners have to be better investigated. In addition to the opportunity to provide more efficient social engagement configurations, \acs{VR} also enables building and evaluating situations that are difficult to set up in real life (e.g., due to the privacy-related concerns or current availability).

While \acs{VR} technology has a long history in the education domain~\cite{vreducationdomain, vrandeducation}, the current availability of consumer-grade head-mounted displays (HMDs) allows for the creation of immersive experiences at a reasonable cost, making it possible to employ immersive personalized \acs{VR} experiences in classrooms in the near future~\cite{riftart}. However, the digital transformations of classrooms reflect an important and critical step when developing \acs{VR} environments for learning purposes and require further research. A unique opportunity to understand the gaze-based behavior, and consequently, attention distribution of learners in such \acs{VR} settings is provided through the analysis of the eye movement of learners~\cite{vr_eyetracking_review}. Since some of the high-end \acs{HMD}s already consist of integrated eye trackers, it does not require extensive effort to extract eye movement patterns during simulations in \acs{VR}. A thorough analysis of the eye movements allows to infer information on the users going beyond the gaze position, for example stress~\cite{stress_vr}, cognitive load~\cite{bozkir2019person}, visual attention~\cite{bozkir_vr_attention_et}, evaluation and diagnosis of diseases~\cite{7829437}, future gaze locations~\cite{8998375}, or training evaluation~\cite{8448290}. In the virtual classroom, this rich source of information could even be combined with the virtual teachers' attention, similar to real-world classrooms~\cite{Sumer_2018_CVPR_Workshops,goldberg2019attentive}, to design more responsive and engaging learning environments.

In this study, we design an immersive \acs{VR} classroom that is similar to a real classroom, enabling students to perceive an immersive virtual classroom experience. We focus on exploring the impact of the digital transformation from the classroom to immersive \acs{VR} on learners by analyzing their eye movements. For this purpose, three design factors are studied, including sitting positions of the participating students, different visualization styles of the virtual peer-learners and teachers, and different performance levels of virtual peer-learners with different hand-raising behaviors. Figure~\ref{fig:teaser_CHI21} shows the overall design of the virtual classroom. Consequently, our main contributions are as follows.

\begin{itemize}
\item We design an immersive \acs{VR} classroom and conduct a user study to enable students to virtually perceive ``interactive'' learning.
\item We analyze the effect of different sitting positions on learners, including sitting in the front and back. We find significantly different effects in fixation and saccade durations, and saccade amplitudes in relation to the sitting position.
\item We evaluate the effect of different visualization styles of virtual avatars on learners including cartoon and realistic styles and find significantly different effects in fixation and saccade durations, and pupil diameters.
\item We assess the effect of different performance levels of virtual peer-learners on learners by evaluating various hand-raising percentages, and find significant effects particularly in pupil diameters and number of eye fixations.
\end{itemize}
\vspace{-15px}

\subsection{Related Work}
As head-mounted displays (HMDs) and related hardware become more accessible and affordable, \acs{VR} technology may become an important factor in the educational domain, particularly given its provided immersion and potential for teaching~\cite{review_ivr_education,youngblut1998educational}. Various recent works on \acs{VR} and education indicate that \acs{VR} may offer significant advantages for learning and teaching. For instance, based on the post-session knowledge tests, both augmented and virtual reality (AR/VR) are found to promote intrinsic benefits such as increasing learners' immersion and engagement when used for learning structural anatomy~\cite{doi:10.1002/ase.1696}. In~\cite{alhalabi2016virtual}, the impact of \acs{VR} systems on student achievements in engineering colleges was investigated by evaluating the results of post-quizzes and the results show that \acs{VR} conditions present significant advantages when compared to no-\acs{VR} conditions since students improve their performance, which indicates that \acs{VR} can successfully support teaching engineering classes. Additionally, \acs{VR} was also evaluated to help teachers develop specific skills that can be helpful in their teaching processes~\cite{lambvirtual}. In addition to teaching and learning processes, another aspect under evaluation concerns the types of virtual environment configurations that are used not only for learning, but also for exploring immersion, motivation, and interaction. To this end, different types of \acs{VR} setups have been studied.~\cite{riftart} introduced an immersive \acs{VR} tool to support teaching and studying art history, which indicates, when used for high-school students, an increased motivation towards art history.~\cite{smartphone_vr} explored the possibility of using low-cost \acs{VR} setups to improve daily classroom teaching by using a smartphone-based \acs{VR} system. According to the evaluations using pre- and post tests, the proposed \acs{VR} setup helps students perform better compared to traditional teaching using whiteboard and slides. Furthermore, \acs{HMD}-based \acs{VR} environment was studied in an elementary classroom for teachers to guide their students in exploring learning elements in immersive virtual field trips~\cite{207ed}. It has been concluded that students' motivation was enhanced after the virtual field trips. Overall, such works imply that while increasing motivation and engagement, different types of \acs{VR} environments provide plenty of benefits and can be used to assist learning and teaching processes by providing users with immersive experiences.

One disadvantage of such \acs{VR} and online learning tools is that learners' motivation and performance may be affected by lack of social interaction~\cite{mooc_social_interaction}, peer accompaniment~\cite{doi:10.1177/1052562904271199}, or immersion~\cite{motivation_immersion}. Furthermore, realism in immersive environments can have various implications~\cite{realism}, related to both learning and interaction. To address these issues, several works have focused on how to provide more realistic and immersive environments. For example,~\cite{vrclassroomconstructivist} discusses the design of the \acs{VR} environments for classrooms by replicating real learning conditions and enhancing learning through real-time interaction between learners and instructors. Furthermore,~\cite{8797708} constructed virtual classmates by synthesizing previous learners' time-anchored comments and indicates that when students are accompanied by a small number of virtual peer-learners built with prior learners' comments, their learning outcomes are improved. In addition to virtual peer-learners, the presence of virtual instructors may also have an impact on learning in \acs{VR}.~\cite{livehumanrole} investigated this and reports that learners engaged more with the environment and progressed further with the interaction prompts when a virtual instructor was provided. These works and findings indicate that the styles and types of virtual agents in the virtual environments may have several effects on students' attention and perception during immersion and should be taken into account. The evaluation of real-time visual attention towards similar configurations, which could be carried out using sensors such as eye trackers, may not only help to understand learning processes but also provide empirical insights about interactions during virtual classes for digital transformations of classrooms in \acs{VR}.

From immersion and interaction point of view, video teleconferencing systems share similar goals with the \acs{VR} classrooms as such systems enable people to experience highly immersive and interactive environments~\cite{teleconference_VR} and have been studied in the \acs{VR} context as well. For example,~\cite{teleconference_immersion} proposed a video teleconference experience using a \acs{VR} headset and found that the sense of immersion and feeling of presence of a remote person increases with \acs{VR}. Furthermore, different mixed reality (MR)-based 3D collaborative mediums were studied in terms of teleconference backgrounds and user visualization styles~\cite{teleconference_twoaspects}. The real background scene and realistically constructed avatars promote a higher sense of co-presence. Low-cost setups were investigated also for real-time \acs{VR} teleconferencing~\cite{teleconference_realtime}, as it was done for \acs{VR} learning environments and it is found that it is possible to improve image quality using headsets in these setups. The possibility of having low-cost setups may become an important factor in the future when accessibility and extensive usage of everyday \acs{VR} environments for learning~\cite{alhalabi2016virtual} and interaction~\cite{vrclassroomconstructivist} are considered.

In general, while the visualization styles and rendering are considered to affect learners' perception and attention, in virtual learning environments particularly in \acs{IVR} classrooms, other design factors are also important for attention-related tasks. For instance,~\cite{bailenson_et_al_2008} has studied the effect of being closer to the teacher, being in the teacher's field of view (FOV), and the availability of virtual co-learners in virtual classrooms. In particular, the authors found that students learn more if they are closer to the teacher and by being in the center of the teacher's \acs{FOV}. In addition, when no co-learners or co-learners who have positive attitudes towards the lecture (e.g., looking at the teacher or taking notes) are available, students learn more information about the lecture instead of the virtual room. Gazing time was approximated according to the time students kept the virtual teacher in their \acs{FOV}s; however, real-time gaze information was missing during the experiments. Exact gazing patterns and different eye movement events during learning are particularly needed for understanding moment-to-moment visual behaviors of students. In another work,~\cite{blume_et_al_18} studied the effect of the sitting position on attention-deficit/hyperactivity disorder (ADHD) experiencing students in such classrooms and found indications that front-seated students are affected positively by this configuration in terms of learning. However, similar to~\cite{bailenson_et_al_2008}, the authors did not have gaze information available but identified that the evaluation of eye movements may provide additional insights during learning, particularly in terms of real-time visual interaction, when learning and cognitive processes are taken into consideration. In addition, eye movements are also considered as choice of measurements to study visual perception during learning~\cite{holmqvist_book_eye_tracking,Jarodzka_Holmqvist_Gruber_2017}.~\cite{DazOrueta2014AULAVR} and~\cite{Seo2019JointAV} have studied attention measures and social interaction in similar setups using continuous performance tests and head movements, respectively. The latter work has used head movements as a proxy for visual attention and found that head movements shift between target and interaction partner. This finding partly supports the finding of~\cite{livehumanrole} that the learners' engagement increases when a virtual instructor is presented. However, both works lack eye movement measurements. As also reported by~\cite{Seo2019JointAV}, eye movements should be examined along with head movements to understand attention and interaction more in-depth, since eyes can move differently. In addition,~\cite{Nolin2016ClinicaVRCA} studied the relationship between performance, sense of presence, and cybersickness, whereas~\cite{mangalmurti_2020} examined attention, more particularly \acs{ADHD} with continuous performance task in a virtual classroom. However, both works are more in the clinical domain, which are relatively different from an everyday classroom setup.~\cite{rizzo_bowerly_buckwalter_klimchuk_mitura_parsons_2009} provides a general overview more from clinical perspective. Lastly, although has not been studied extensively in \acs{VR} yet, peer-learners' engagement expressed by hand-raising behavior~\cite{hand_raising_classroom_learning} may also affect the attention and visual behaviors of learners in the \acs{VR} classrooms, which could be further studied.

In summary, while showing that \acs{VR} could be a useful technology to support education, the aforementioned works primarily focused on the importance of used mediums and configurations, visualization styles, participant locations for visual attention, engagement, motivation, and learning of participants in \acs{VR} classrooms. Yet, real-time and moment-to-moment interactions with the environment and visual behaviors of students in an everyday \acs{VR} classroom setup were not studied in depth. Although obtaining such information in real-time is challenging, analyzing eye-gaze and eye movement features can provide valuable understanding into visual attention and interaction in a non-intrusive way, especially for designing such classroom configurations. For instance, long fixations can be related to the increased amount of cognitive process~\cite{just1976eye}, whereas long saccadic behaviors are related to inefficient search behavior~\cite{goldberg1999computer}. Furthermore, pupillometry is highly related to cognitive workload~\cite{Appel:2018,appel2019predicting}. Such information is also argued for consideration in \acs{IVR} environments~\cite{bailenson_2002,bailenson_2004}. In fact, when designing immersive \acs{VR} environments for digital transformations of classrooms in virtual worlds, such features can be key to understand visual attention, cognitive processes, and visual interactions towards different classroom manipulations, which may also affect learning and teaching processes. To address this research gap, we study three configurations in an everyday \acs{VR} classroom setup including different visualization styles of virtual avatars, sitting positions of participants, and hand-raising based performance levels of peer-learners by using eye movement features.

\subsection{Methodology}
The main purpose of our study is to investigate the effects of digital transformations of the classrooms to \acs{VR} settings on learners. Therefore, we designed a user-study to study these effects. In this section, we discuss the participant information, apparatus, experimental design, experiment procedure, measurements, data pre-processing steps, and our hypotheses. Our study and data collection were approved by the institutional ethics committee at the University of T{\"u}bingen (date of approval: 25/11/2019, file number: A2.5.4-106\_aa) as well as the regional council responsible for educational affairs at the district of T{\"u}bingen.

\subsubsection{Participants}
Participants were recruited from local academic track schools via e-mails and invitation letters. After obtaining written informed consent from both students and their parents or legal guardians, all students who indicated interest were admitted to the study. $381$ volunteer sixth-grade students ($179$ female, $202$ male), whose ages range from $10$ to $13$ ($M=11.51$, $SD=0.56$), were recruited to participate in the experiment. Due to hardware problems or incorrect calibration, data from $32$ participants were removed. In addition, data from $61$ participants were also removed due to eye tracker related issues including low eye tracking ratio (lower than $90\%$). Therefore, data from $288$ participants ($137$ female, $151$ male), whose ages range from $10$ to $13$ ($M = 11.47$, $SD = 0.51$), were used for evaluations. We had $16$ different conditions in the experiment and the average number of participants for each condition was $18$ ($SD = 5.3$). In addition to the actual study and data collection, we successfully piloted both our technical setup and the experimental workflow with $55$ similar aged ($M = 11.35$, $SD = 0.52$) sixth-grade students ($20$ female, $35$ male).

\subsubsection{Apparatus}
In our experiments we employed HTC Vive Pro Eye devices with a refresh rate of $90$ Hz and a field of view of $110^{\circ}$. The \acs{VR} environment was designed and rendered using the Unreal Game Engine\footnote{https://www.unrealengine.com/} v$4$.$23$.$1$. The screen resolution for each eye was set to $1440 \times 1600$. To collect eye movement data, we used the integrated Tobii eye tracker with a $120$ Hz sampling rate and a default calibration with $0.5^{\circ}-1.1^{\circ}$ accuracy.

\subsubsection{Experimental Design}
The virtual classroom designed in our study has $4$ rows and $2$ columns of desks along with chairs, as well as other objects which typically exist in the conventional classrooms such as a board and display. In total, there are $24$ virtual peer-learners sitting on the chairs. A virtual teacher standing in front of the classroom teaches a $\approx 15$-minute virtual lecture to the students about computational thinking~\cite{Weintrop2016DefiningCT}. During the lecture, the virtual teacher walks around the podium. The virtual peer-learners and participants sit on the chairs throughout the lecture. The lecture has four phases including \textbf{(a) topic introduction} ($\approx3$ minutes), \textbf{(b) knowledge input} ($\approx4.5$ minutes), \textbf{(c) exercises} ($\approx5.5$ minutes), and \textbf{(d) summary} ($\approx1.5$ minutes). There are distracting behaviors from virtual peer-learners (e.g., raising hands, turning around) in the first, second, and third phases of the lecture.

In the beginning of the first phase, the teacher enters the classroom, stays in the classroom for a while, and then leaves for $\approx20$ seconds, giving participants the opportunity to look around and adjust to the virtual environment. The topic of the lecture is displayed on the board as \emph{“Understanding how computers think”}. During the first phase, the teacher asks five simple questions to interact with the students. Some of the peer-learners raise their hands and answer the questions. In the second phase, the teacher explains two terms to the students, namely, the terms \emph{``loop''} and \emph{``sequence''}. These terms are also shown on the display. Then, the teacher asks four questions about each term and the peer-learners raise their hands to answer the questions. In the third phase, the teacher gives the students two exercises to evaluate whether or not they understand the terms correctly. For each exercise, the students have some time to think. Then, the teacher provides the answers for each exercise, and the peer-learners vote for the correct answer by raising their hands. In the last phase, the teacher stands in the middle of the classroom to summarize the lecture. No questions are asked in this phase; therefore, none of the peer-learners raise their hands.

Our study is in between-subjects design. The participants are located either in the front or back region of the virtual classroom. The participants that sit in the front of the virtual classroom have one row in front of them, whereas the participants that sit in the back have three rows in front of them. The visualization styles of the avatars have two levels as well, in particular cartoon and realistic. Lastly, the hand-raising percentages, which are intended to show the performance levels of the virtual peer-learners, have four different levels, including $20\%$, $35\%$, $65\%$, and $80\%$. Combining all, we have a $2 \times 2 \times 4$ factorial design that forms $16$ different conditions in total. Participants' views from back and front sitting positions, cartoon- and realistic-styled avatars are depicted in Figures~\ref{fig:screenshots_CHI21} (a), (b), (c), and (d), respectively.

\begin{figure*}[ht]
\centering
\subfigure[Back sitting participant experiencing the VR classroom.]{
    {\includegraphics[width=0.4\linewidth,keepaspectratio]{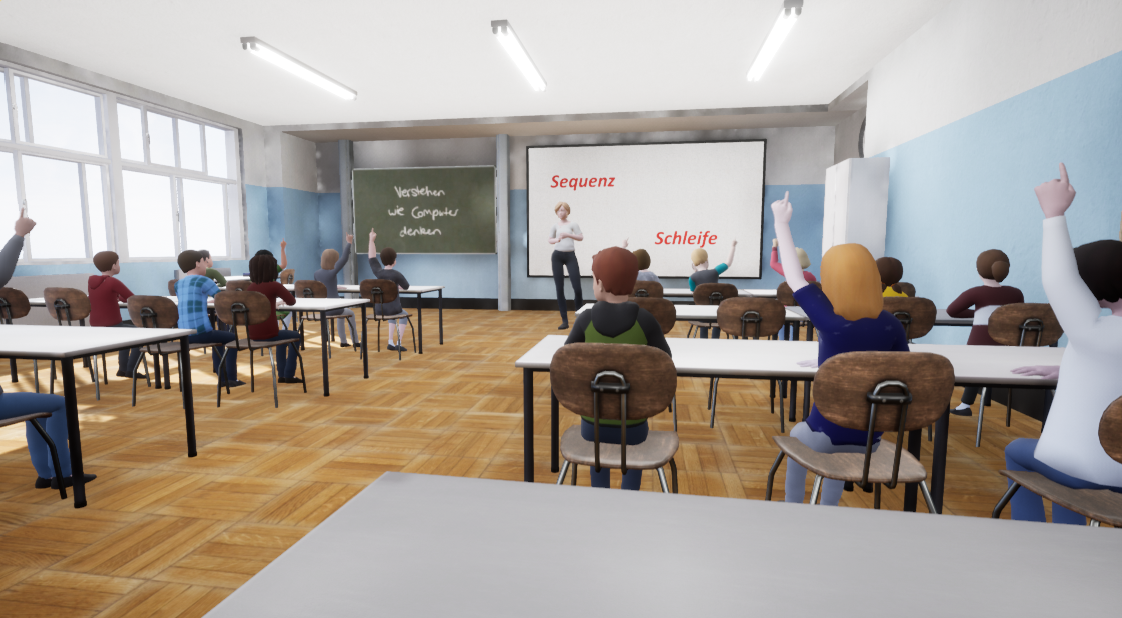}}
}
\qquad
\qquad
\subfigure[Front sitting participant experiencing the VR classroom.]{
    {\includegraphics[width=0.4\linewidth,keepaspectratio]{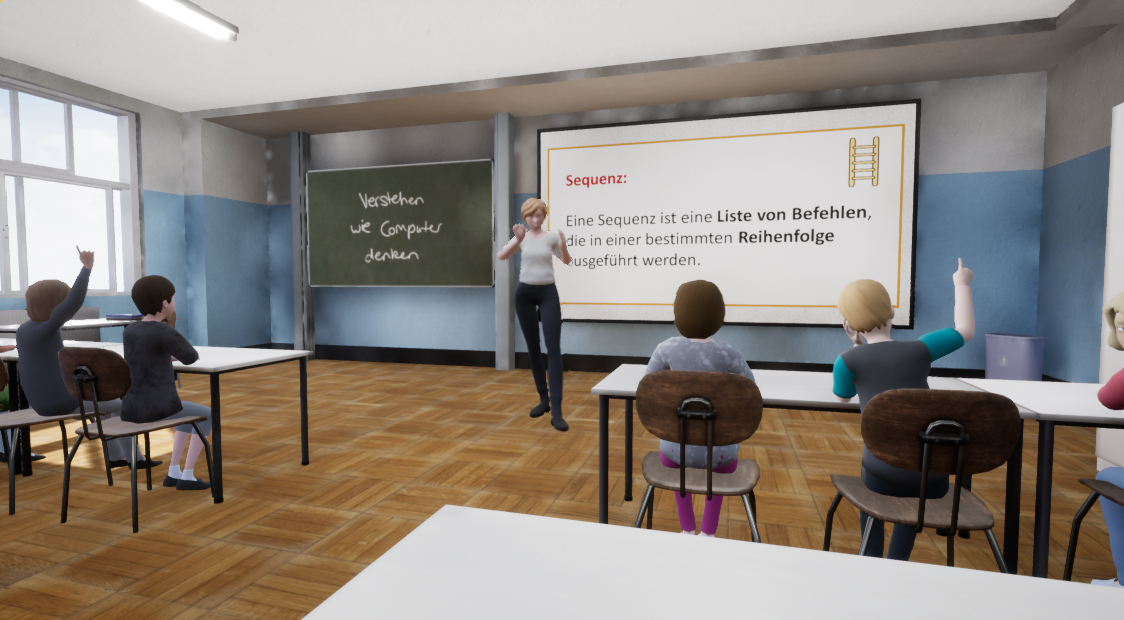}}
}
\qquad
\qquad
\subfigure[Cartoon-styled avatars.]{
    {\includegraphics[width=0.4\linewidth,keepaspectratio]{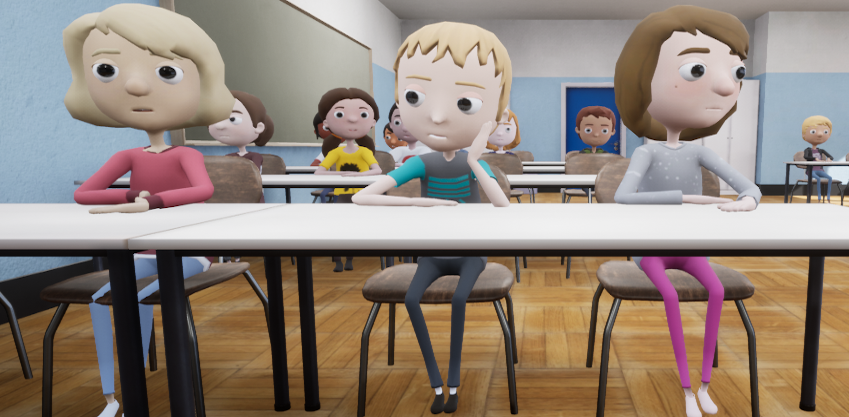}}
}
\qquad
\qquad
\subfigure[Realistic-styled avatars.]{
    {\includegraphics[width=0.4\linewidth,keepaspectratio]{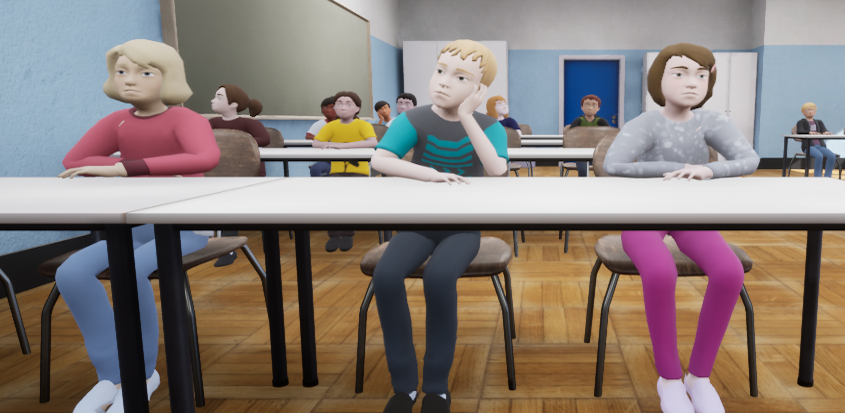} }
}
\caption{Views from the immersive virtual reality classroom.}
\label{fig:screenshots_CHI21}
\end{figure*}

\subsubsection{Procedure}
Each experimental session took $\approx45$ minutes including preparation time. We conducted the experiments in groups of ten participants by assigning each participant randomly to one of the sixteen conditions. Before the data assessment took place at the participating schools, students were informed that they could drop out of the study at any time without consequences. After a brief introduction to the experiment and the data collection process, participants had the opportunity to acclimate with the hardware and the \acs{VR} environment.

The experiment started with the eye tracker calibration. After calibration success, the experimenters pressed the ``Enter'' button to start the actual experiment and data collection process, wherein participants experienced the immersive virtual environment and the lecture. The experiments were supposed to be carried out in one session without breaks, mimicking thus a real classroom teaching session, lasting about 15 minutes. At the end of the experiment, the \acs{VR} application displayed a message telling the participants to take off their \acs{HMD}s. Lastly, participants filled out questionnaires about their experienced presence and perceived realism.

\subsubsection{Measurements}
For this work, our main focus was eye-gaze, head-pose, and pupil related activities of the participants as these are considered to be rich information sources, especially in \acs{VR}. Fixations are the periods during which eyes are stationary within the head while fixated on an area of interest. Saccades, on the other hand, are the high-speed ballistic eye movements that shift eye-gaze from one fixation to another.

Using fixations, saccades, and pupil diameters, plenty of eye movement features are extracted. In this study, we extracted the number of fixations, fixation durations, saccade durations, saccade amplitudes, and normalized pupil diameters to analyze different conditions of the experiment. In the eye tracking literature, longer fixation durations correspond to engaging more with the object or increased cognitive process~\cite{just1976eye}. Fixation durations are mainly related to cognition and attention; however, it is argued that they are affected by the procedures that lead to learning and it is reported that fixation durations can be used to understand learning processes as well~\cite{Negi_Mitra_2020}. For instance,~\cite{CHIEN2015191} has studied fixation patterns during learning in simulation- and microcomputer-based laboratory and found that simulation group had longer fixation duration, which means more attention and deeper cognitive processing. In addition to the fixations, longer saccade durations correspond to less efficient scanning or searching~\cite{goldberg1999computer}, whereas longer saccade amplitudes mean that attention is drawn from a distance~\cite{goldberg2002eye}. Furthermore, a larger pupil diameter is related to higher cognitive load~\cite{beatty:1982}. In addition, while being task dependent,~\cite{doi:10.1177/1541931213601689} has indicated that pupil diameter measurements in high task load correlate with individual's performance. However, as pupil diameter values are also affected by the illumination, a controlled environment is needed to assess it. In our \acs{VR} setup, the illumination is controlled across different conditions. Besides, a general overview of considering eye tracking as a tool to enhance learning with graphics is provided in~\cite{MAYER2010167}.

Additionally, the self-reported presence and realism were assessed by questionnaires. The items in the questionnaires were based on the conceptualizations of~\cite{schubert_presence} and~\cite{lombard_presence} which were developed particularly to assess students' perception of the \acs{VR} classroom situation. The experienced presence and perceived realism were assessed via using a 4-point Likert scales ranging from 1 (``do not agree at all'') to 4 (``completely agree'') with nine (e.g., ``I felt like I was sitting in the virtual classroom.'' or ``I felt like the teacher in the virtual classroom really addressed me.'') and six items (e.g., ``What I experienced in the virtual classroom, could also happen in a real classroom.'' or ``The students in the virtual classroom behaved similarly to real classmates.''), respectively.

\subsubsection{Data Pre-processing}
As the raw eye tracking data collected from the \acs{VR} device does not include fixations, saccades or similar eye movements, we first pre-processed the data to identify these events. Detecting different eye movements in the \acs{VR} setup is a challenging task and different from the traditional eye tracking experiments that include equipment such as chin-rests, as participants have opportunity to move their heads freely in \acs{VR}. In the eye tracking literature, Velocity-Threshold Identification (I-VT) method is used to classify fixations based on velocities~\cite{salvucci2000identifying}. In the \acs{VR} context, ~\cite{agtzidis2019360} applied a similar method to detect eye movement events. We opted for a similar approach.

Before applying the \acs{I-VT}, we first applied linear interpolation for the missing gaze vectors. After the interpolation, we identified the fixations when the \acs{HMD} was stationary. However, the identification of saccades was not restricted by the \acs{HMD} movement. The used velocity and duration thresholds for the \acs{HMD} movement states, fixations, and saccades are depicted in Table~\ref{tab:events_CHI21}, where the velocities and durations are given as $v$ and $\Delta$, respectively. Unlike the fixations and saccades, the pupil diameter values are reported by the eye tracker. As raw pupil diameter values are affected by blinks and noisy sensor readings, we smoothed and normalized the pupil diameter readings using Savitzky-Golay filter~\cite{savitzky64} and divisive baseline correction using a baseline duration of $\approx 1$ seconds~\cite{Mathot2018}, respectively.

\begin{table}[ht]
  \begin{center}
  \caption{Head and eye movement event identification thresholds.}
  \label{tab:events_CHI21}
  \begin{tabular}{ccc}
    \toprule
    Event & Conditions for velocity ($v$) & Conditions for duration ($\Delta$)\\
    \midrule
    Stationary HMD & $v_{head} < 7^{\circ}/s$ & -\\
    Fixation & $v_{head} < 7^{\circ}/s$ and $v_{gaze} < 30^{\circ}/s$ & $100ms < \Delta_{fixation} < 500ms$ \\
    Saccade & $v_{gaze} > 60^{\circ}/s$ & $30ms < \Delta_{saccade} < 80ms$\\
  \bottomrule
\end{tabular}
\end{center}
\end{table}

\subsubsection{Hypotheses}
We developed three hypotheses, each corresponds to one design factor.

\begin{itemize}
\item \textbf{Hypothesis-1 (H1)}: We hypothesize that the different sitting positions of the participants yield different effects on the eye movements. As the participants that sit in the front are closer to the board, displays, and the teacher, we assume that they can attend the virtual lecture more efficiently than participants in the back and have less difficulty extracting information about the lecture. However, as they have a narrower field of view, particularly towards the frontal part of the classroom, they need to shift their attention more than the participants sitting in the back.

\item \textbf{Hypothesis-2 (H2)}: We hypothesize that different visualization styles of virtual avatars affect student visual behaviors differently. More particularly, as students are familiar with realistic styles in the conventional classrooms, we claim that compared to cartoon-styled visualization condition, they attend the scene shorter during fixations in the realistic-styled visualization setting as cartoon-styled avatars are more attractive to the students. Therefore, students engage with the environment more in the cartoon-styled visualization condition than in the realistic-styled condition.

\item \textbf{Hypothesis-3 (H3)}: We hypothesize that different hand-raising percentages of virtual peer-learners can distinctively affect the behaviors of participants. Specifically, we anticipate that when relatively higher percentages of hand-raising levels are provided, such as $65\%$ or $80\%$, the participant's cognitive load will be higher due to the fact that many of the peer-learners attend the lecture with a high focus. Similarly, participants have more fixations in the classroom in the higher hand-raising percentage conditions as a higher number of hand-raising percentage creates an opportunity for various attention and distraction points.
\end{itemize}

\subsection{Results}
As we have three factors that form $16$ different conditions, we applied $3$-way full-factorial analysis of variance (ANOVA) by setting the level of significance to $\alpha = 0.05$ with Tukey-Kramer post-hoc test. For the non-parametric factorial analysis, we used the Aligned Rank Transform (ART)~\cite{10.1145/1978942.1978963} before applying \acs{ANOVA} procedures.

\subsubsection{Analysis on Different Sitting Positions}
Different sitting positions have an impact on the mean fixation and saccade durations, and mean saccade amplitudes. The mean fixation durations of the front and back sitting participants are illustrated in Figure~\ref{fig:sitting_pos_results_CHI21} (a). The participants that sit in the back have significantly longer mean fixation durations ($M = 222.6ms, SD = 14.57ms$) than the participants that sit in the front ($M = 218.75ms, SD = 13.11ms$), with $F(1,272) = 6.7$, $p = .01$.

Both saccade durations and amplitudes are influenced by the sitting positions and are depicted in Figures~\ref{fig:sitting_pos_results_CHI21} (b) and (c), respectively. The results reveal significantly longer saccade durations in the front condition ($M = 50.23ms, SD = 1.7ms$) than in the back condition ($M = 47.9ms, SD = 2.62ms$), with $F(1,272) = 73.76$, $p<.001$. Similarly, the mean saccade amplitude is significantly larger in the front condition ($M = 10.93^{\circ}, SD = 1.54^{\circ}$) than in the back condition ($M = 10.05^{\circ}, SD = 1.38^{\circ}$), with $F(1,272) = 22.6$, $p < .001$.

\begin{figure*}[ht]
  \centering
   \subfigure[Mean fixation durations.]{{\includegraphics[width=0.315\linewidth,keepaspectratio]{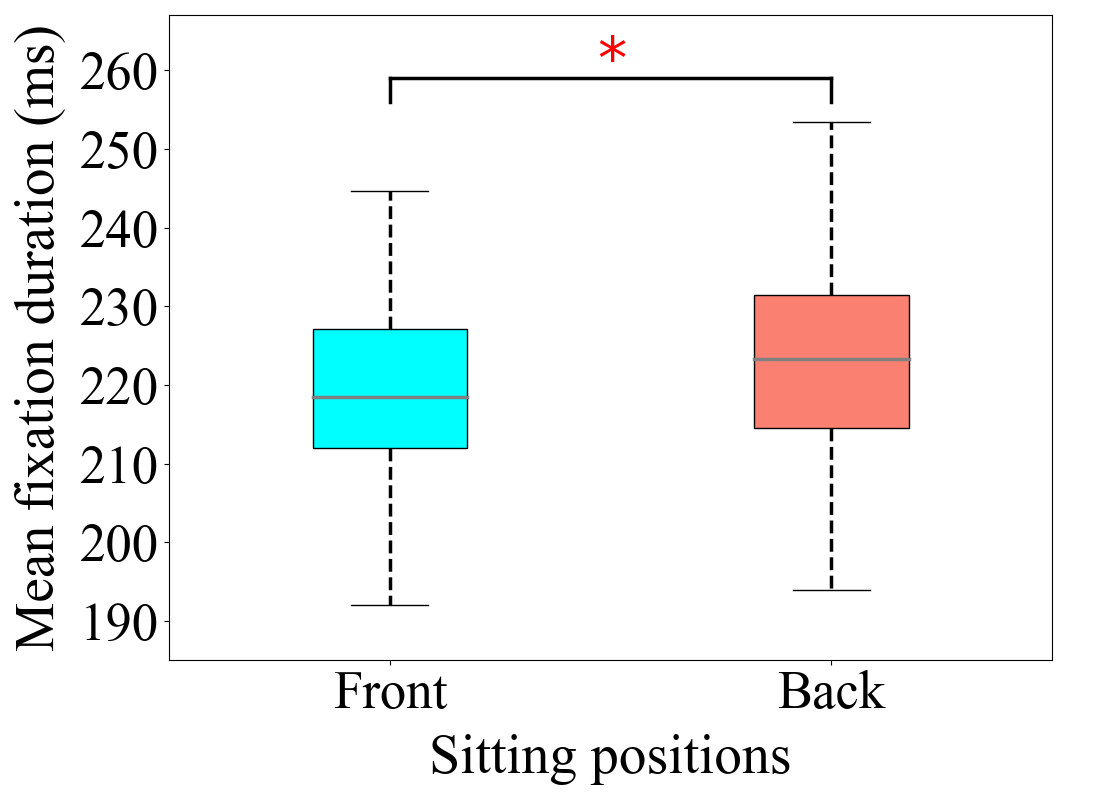}}}%
   \quad
   \subfigure[Mean saccade durations.]{{\includegraphics[width=0.315\linewidth,keepaspectratio]{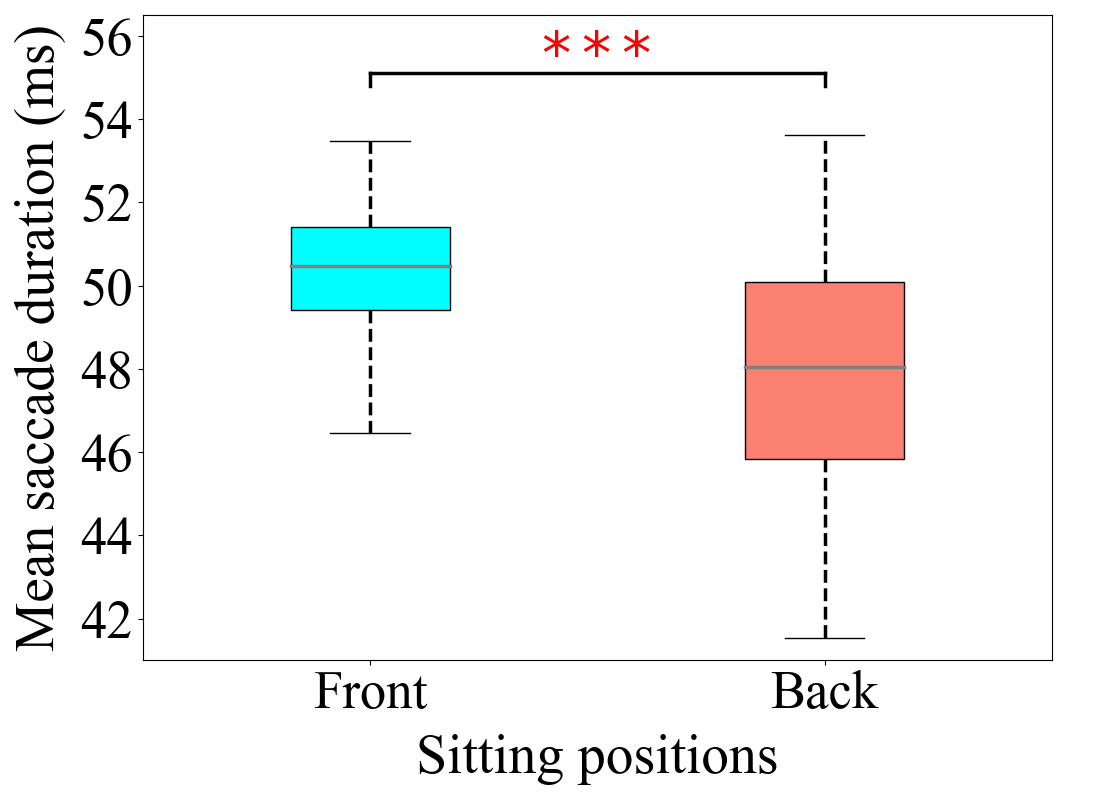} }}%
   \quad
   \subfigure[Mean saccade amplitudes.]{{\includegraphics[width=0.315\linewidth,keepaspectratio]{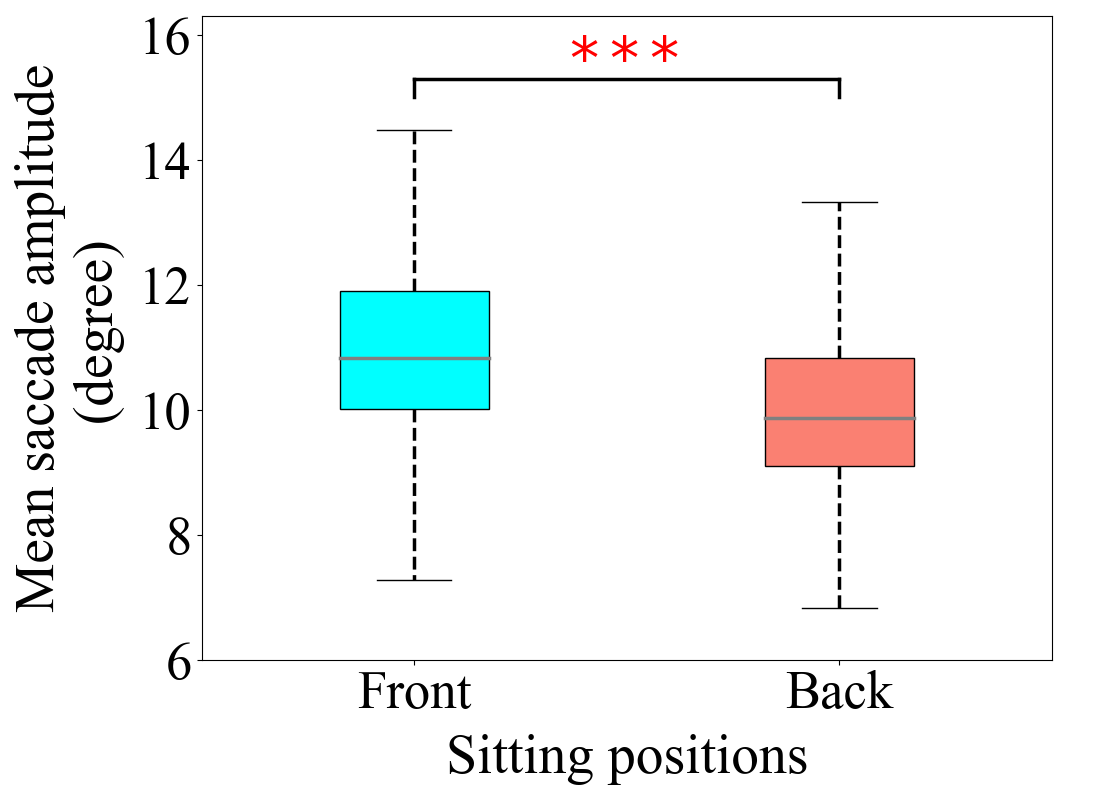}}}
  \caption{Results for different sitting positions. Significant differences are highlighted with * and *** for $p <.05$ and $p <.001$, respectively.}
  \label{fig:sitting_pos_results_CHI21}%
\end{figure*}

\subsubsection{Analysis on Different Avatar Styles}
Different avatar visualization styles affect the mean fixation and saccade durations, and pupil diameters. The results are depicted in Figures~\ref{fig:different_avatar_results_CHI21} (a), (b), and (c), respectively. The mean fixation durations are significantly longer in the cartoon-styled avatar condition ($M = 222.88ms, SD = 14.06ms$) than in the realistic-styled avatar condition ($M = 218.6ms, SD = 13.76ms$), with $F(1,272) = 5.27$, $p = .022$. By contrast, the mean saccade durations are significantly shorter in the cartoon-styled avatar condition ($M = 48.58ms, SD = 2.66ms$) than in the realistic-styled condition ($M = 49.3ms, SD = 2.35ms$), with $F(1,272) = 6.22$, $p = .013$.

\begin{figure*}[ht]
  \centering
   \subfigure[Mean fixation durations.]{{\includegraphics[width=0.315\linewidth,keepaspectratio]{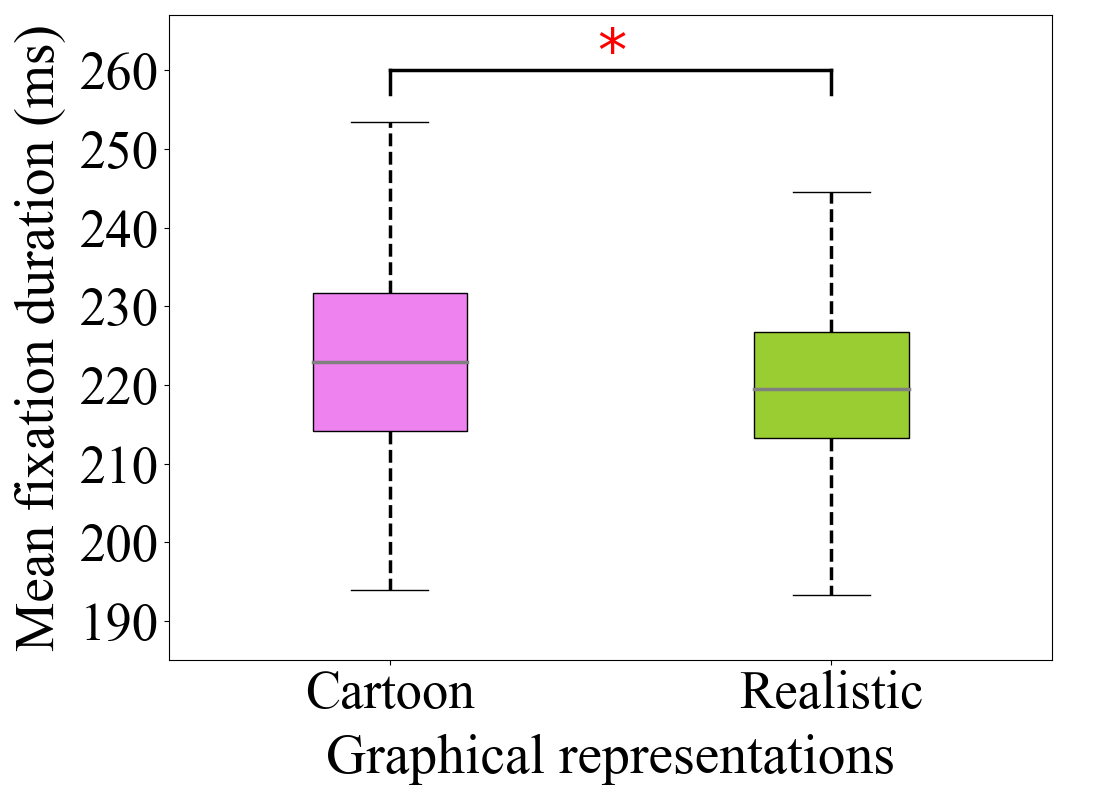}}}%
   \quad
   \subfigure[Mean saccade durations.]{{\includegraphics[width=0.315\linewidth,keepaspectratio]{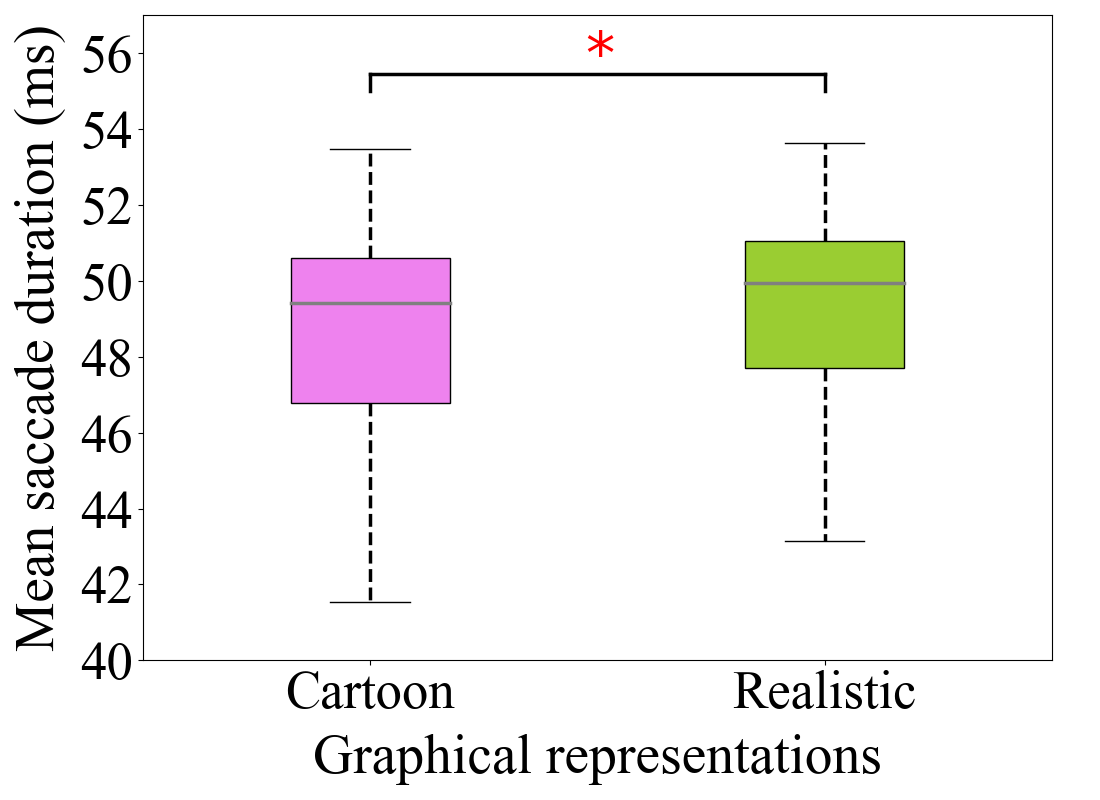} }}%
    \quad
   \subfigure[Pupil diameters.]{{\includegraphics[width=0.315\linewidth,keepaspectratio]{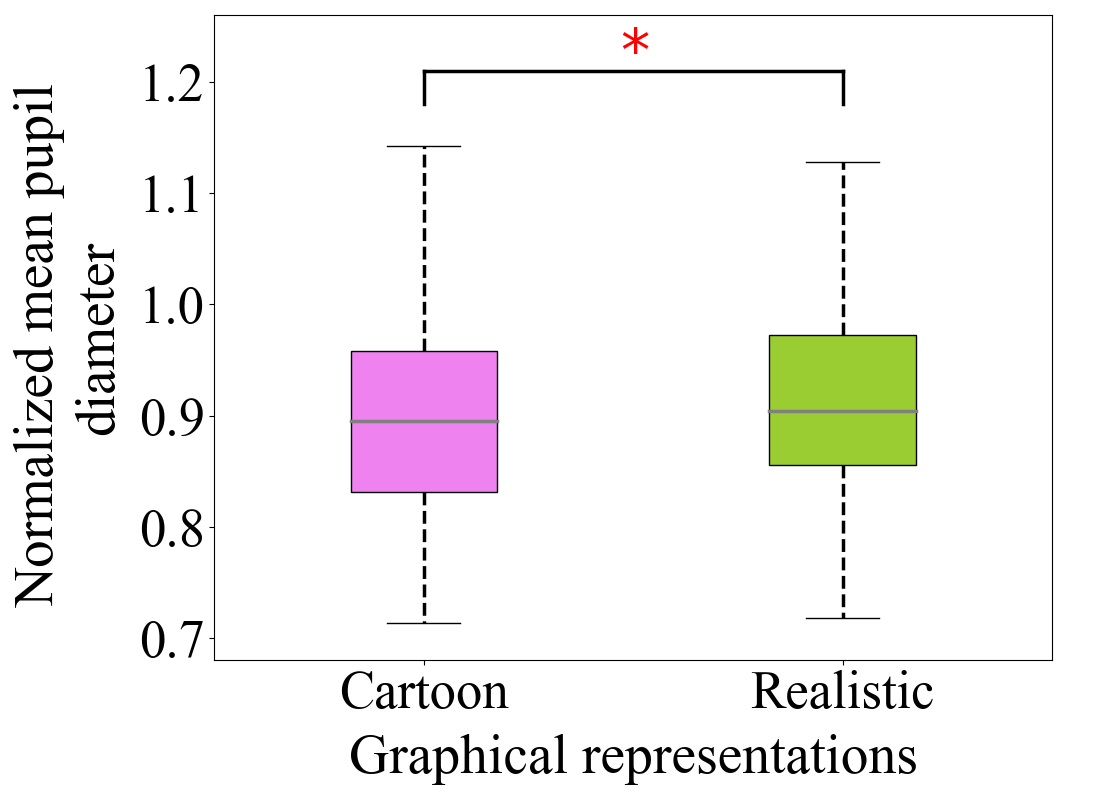}}}
  \caption{Results for different avatar visualization styles. Significant differences are highlighted with * for $p <.05$.}
  \label{fig:different_avatar_results_CHI21}
\end{figure*}

The normalized mean pupil diameter, which reflects the cognitive load, is significantly larger in the realistic-styled avatar condition ($M = 0.94, SD = 0.16$) than in the cartoon-styled avatar condition ($M = 0.91, SD = 0.13$), with $F(1,272) = 3.94$, $p = .048$.

\subsubsection{Analysis on Different Hand-raising Behaviors}
The hand-raising behaviors of virtual peer-learners have significant impacts on the pupil diameters and number of fixations as depicted in Figures~\ref{fig:different_handraising_results_CHI21} (a) and (b), respectively. We found significant effects on normalized mean pupil diameter values with $F(3,272) = 4.78$, $p = .003$. Particularly, mean pupil diameter in the $80\%$ hand-raising condition ($M = 0.96, SD = 0.16$) is significantly larger than in the $35\%$ hand-raising condition ($M = 0.9, SD = 0.12$), with $F(3,272) = 4.78$, $p < .001$. In addition, we found significant effects on number of fixations with $F(3,272) = 3.01$, $p = .03$. More specifically, there are notably more fixations in the $65\%$ hand-raising condition ($M = 1112.92, SD = 245.07$) than in the $80\%$ hand-raising condition ($M = 995.49, SD = 211.98$), with $F(3,272) = 3.01$, $p = .028$.

\begin{figure*}[ht]
  \centering
   \subfigure[Pupil diameters.]{{\includegraphics[width=0.485\linewidth,keepaspectratio]{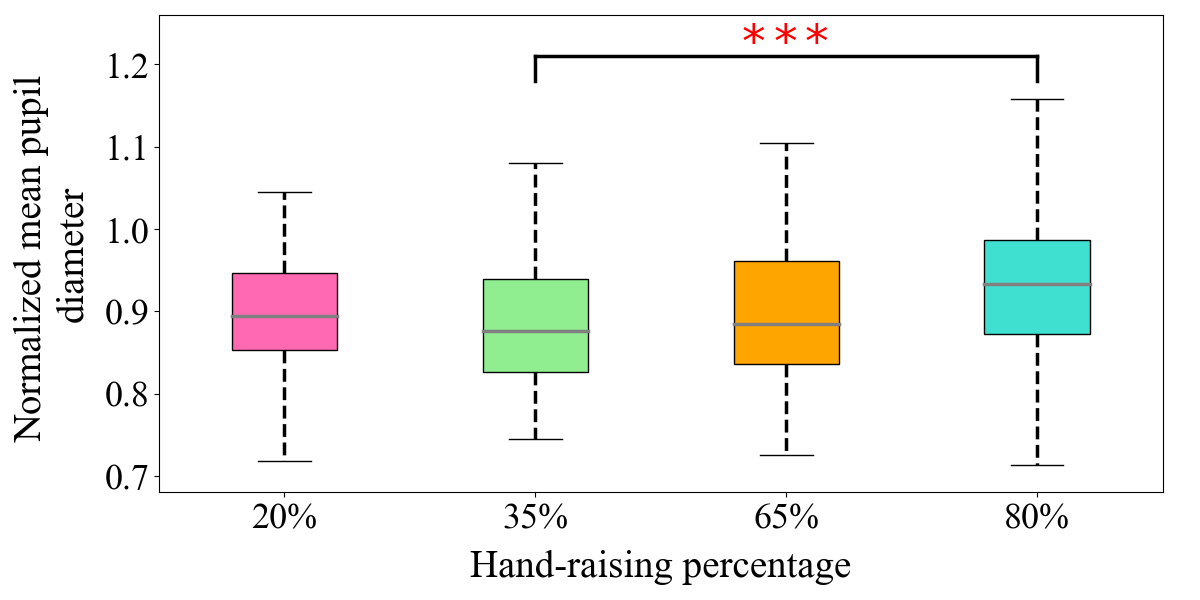} }}%
   \quad
   \subfigure[Number of fixations.]{{\includegraphics[width=0.485\linewidth,keepaspectratio]{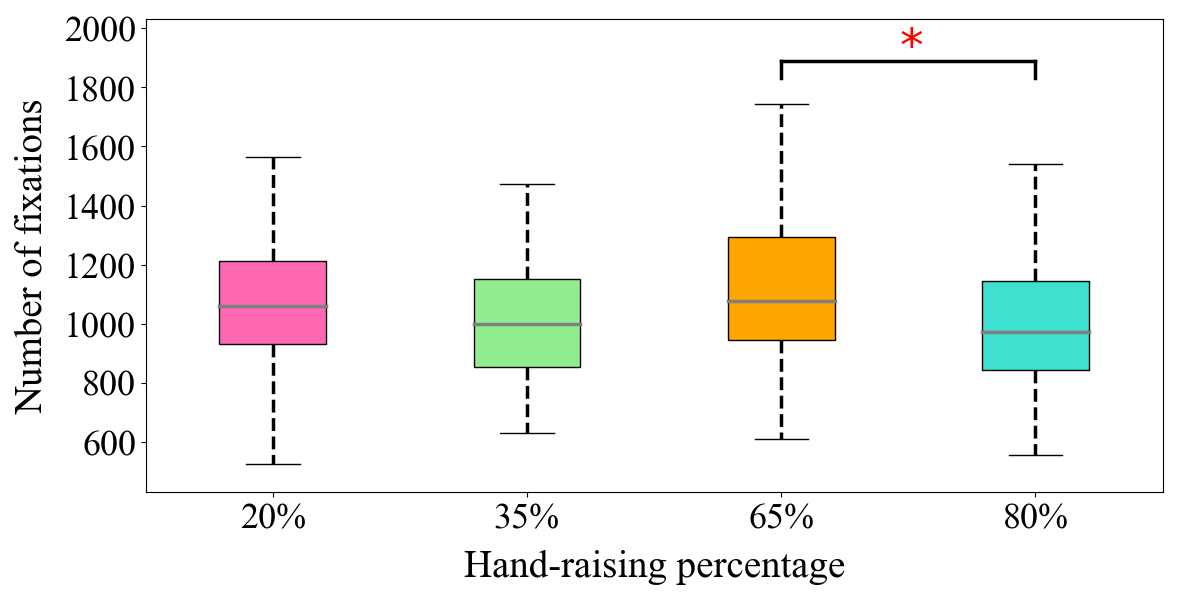}}}%
  \caption{Results for different hand-raising percentages. Significant differences are highlighted with * and *** for $p <.05$ and $p <.001$, respectively.}
  \label{fig:different_handraising_results_CHI21}%
\end{figure*}

\subsubsection{Analysis on Experienced Presence and Perceived Realism}
We did not find significant effects of different experimental conditions on the self-reported experienced presence and perceived realism. Overall, the self-reported experienced presence and perceived realism values are in the vicinity of highest values with ($M = 2.91$, $SD = 0.55$) and ($M = 2.91$, $SD = 0.57$), respectively. These mean that even though we did not obtain statistically significant differences between conditions, the participants experienced high levels of presence and realism in the \acs{IVR} classroom environment.

\subsection{Discussion}
The results show that there are significant differences in the eye movement features between front and back sitting position conditions. Firstly, participants had longer fixations in the back sitting condition. This indicates that they had more processing time than the participants sitting in the front, which can be related to difficulty extracting information, similar to the relationship between task difficulty and mean fixation duration~\cite{task_difficulty_fixation_durations}. Secondly, the participants that sit in the front had longer saccade durations and amplitudes, which suggests that they needed to shift their attention more during the virtual lecture. While being located closer to the lecture content, longer saccade durations indicate that the participants sitting in the front had less efficient scanning behavior~\cite{goldberg1999computer} during the lecture. We assume that this was due to the narrower field of view. These results support our \textbf{H1}. When designing virtual classes, these results should be taken into account, particularly when determining where students should be located in the classroom, depending on the context.

Our results show consequential effects in the eye movement features in different avatar style conditions. As mean fixation durations are longer in the cartoon-styled visualization condition, we assume participants found the cartoon-styled avatars more attractive and attention-grabbing. Therefore, their fixation behaviors were longer during the virtual lecture. On the contrary, the mean saccade durations are longer in realistic-styled conditions as the fixation durations are shorter, which is theoretically expected. Furthermore, the pupil diameters of the participants in the realistic-styled condition are larger, indicating that the cognitive load of these participants was significantly higher during the lecture, which is suggested by the previous work~\cite{beatty:1982}. This is an indication that participants may have taken the lecture more seriously and in a more focused manner when the visualization was realistic. These findings support our \textbf{H2}. Rendering realistic-styled avatars may be computationally expensive depending on the configuration. Therefore, an optimal trade-off should be decided, taking the behavioral results into account while designing the virtual classrooms.

Furthermore, we observe significant effects in attention towards different hand-raising based performance levels of the peer-learners. Particularly, the pupil diameters of the participants in the $80\%$ condition are significantly larger than the pupil diameters of the participants in the $35\%$ condition. We interpret this to mean that when the performance and attendance level of peer-learners was relatively higher, the participants' cognitive load became higher, indicating that they might pay more attention to the lecture content. This partially supports our \textbf{H3}. In addition, a greater number of fixations are observed in the $65\%$ condition than in the $80\%$ condition. We claim that when almost all of the peer-learners participated in hand-raising behaviors during the lecture, participants acknowledged this information without significantly shifting their gaze. However, this claim requires further investigation. Manipulation of different hand-raising conditions may affect student self-concept~\cite{self_concept_1976}, which should be further studied as well.

In our study, the interaction and perception in the immersive \acs{VR} classroom were assessed mainly by using eye-gaze and head-pose information. However, while the virtual teacher and peer-learners talk in the simulations, no response or interaction by means of audio or gestures was expected from the participants. Combining visual perceptions and interactions with such data may provide additional insights particularly for better interaction design in \acs{VR} classrooms. A future iteration can also evolve into an everyday virtual classroom platform where each virtual agent is actually connected to a real person, similar to in platforms such as Mozilla Hubs. To this end, further design settings such as optimal seating arrangement (e.g., U-shape, circle shape) in addition to the sitting positions should be investigated. Evaluation of similar configurations in online learning platforms such as Coursera\footnote{https://www.coursera.org/}, Udemy\footnote{https://www.udemy.com/}, or MOOCs\footnote{https://www.mooc.org/} could provide additional implications for interaction modeling. Furthermore, gaze-based attention guidance can be considered for more interactive \acs{VR} classroom experience and it can be achieved by fine-grained eye movement analysis focusing on short time windows instead of complete experiments. While being out of the scope of this paper, assessing learning outcomes and combining them with visual interaction and scanpath behaviors from immersive \acs{VR} classroom could also offer insights for optimal \acs{VR} classroom design.

\subsection{Conclusion}
In this work, we evaluated three major design factors of immersive \acs{VR} classrooms, namely different participant locations in the virtual classroom, different visualization styles of virtual peer-learners and teachers, including cartoon and realistic, and different hand-raising behaviors of peer-learners, particularly through the analysis of eye tracking data. Our results indicate that participants located in the back of the virtual classroom may have difficulty extracting information during the lecture. In addition, if the avatars in the classroom are visualized in realistic styles, participants may attend the lecture in a more focused manner instead of being distracted by the visualization styles of the avatars. These findings offer valuable insights about design decisions in the \acs{VR} classroom environment. Few indicators were obtained from the evaluation of the different hand-raising behaviors of peer-learners, providing a general understanding of attention towards peer-learner performance. However, these indicators should be further investigated and remain a focus of future work. 

\subsection*{Acknowledgments}
This research was partly supported by a grant to Richard G{\"o}llner funded by the Ministry of Science, Research and the Arts of the state of Baden-W{\"u}rttemberg and the University of T{\"u}bingen as part of the Promotion Program of Junior Researchers. Lisa Hasenbein is a doctoral candidate and supported by the LEAD Graduate School \& Research Network, which is funded by the Ministry of Science, Research and the Arts of the state of Baden-W{\"u}rttemberg within the framework of the sustainability funding for the projects of the Excellence Initiative II. Authors thank Stephan Soller, Sandra Hahn, and Sophie Fink from the Hochschule der Medien Stuttgart for their work and support related to the immersive virtual reality classroom used in this study.

\newpage

\section[Exploiting Object-of-Interest Information to Understand Attention in VR Classrooms]{Exploiting Object-of-Interest Information to Understand Attention in VR Classrooms} 
\label{appendix:A2}

\subsection{Abstract}
Recent developments in computer graphics and hardware technology enable easy access to virtual reality headsets along with integrated eye trackers, leading to mass usage of such devices. The immersive experience provided by virtual reality and the possibility to control environmental factors in virtual setups may soon help to create realistic digital alternatives to conventional classrooms. The importance of such settings has become especially evident during the COVID-19 pandemic, forcing many schools and universities to provide the digital teaching. Researchers foresee that such transformations will continue in the future with virtual worlds becoming an integral part of education. Until now, however, students' behaviors in immersive virtual environments have not been investigated in depth. In this work, we study students' attention by exploiting object-of-interests using eye tracking in different classroom manipulations. More specifically, we varied sitting positions of students, visualization styles of virtual avatars, and hand-raising percentages of peer-learners. Our empirical evidence shows that such manipulations play an important role in students' attention towards virtual peer-learners, instructors, and lecture material. This research may contribute to understanding of how visual attention relates to social dynamics in the virtual classroom, including significant considerations for the design of virtual learning spaces.

\subsection{Introduction}
Everyday use of head-mounted displays (HMDs) is increasing as virtual reality (VR) technology and virtual environments are already being used in various domains such as gaming and entertainment. In addition, some of the consumer-grade \acs{HMD}s are coming to market with integrated eye trackers that may help to assess human attention during immersion and allow for more interactive virtual environments. It is likely that, in the near future, such tools will become widely used mobile devices similar to today's mobile phones or smart watches. To this end, not only should researchers strive to improve the capabilities of these devices, but scrutiny should also be given to understanding human behavior and attention while using such technology. 

Measures of eye movements obtained through eye-tracking are effective indicators of human states and visual behavior to some extent; however, they are dependent on application or task~\cite{10.1145/3352763}. Analyzing and modeling human attention using this data in a specific domain may not be transferable to other domains. Thus, when assessing human attention in digital environments, or more particularly in \acs{VR} for the application in educational technology, specific domain knowledge and configurations should be considered. There is already some history of training and teaching in digital or virtual setups~\cite{KavaLuxtWuen2017ir,review_ivr_education}. Today, due to the COVID-19 pandemic, virtual or digital education has become more popular and even a necessity in many cases. Currently, many schools and universities are carrying out their teaching responsibilities remotely via platforms such as Zoom\footnote{https://www.zoom.us/} or Webex\footnote{https://www.webex.com/}. Such platforms lack the possibility of instructor-student interaction beyond audio and video features and encounter privacy concerns if videos are recorded and stored during classes. \acs{VR} setups offer the immersion, interaction, and privacy preservation that current remote learning platforms lack. In addition, as \acs{VR} allows users to easily control the environmental settings, it is possible to evaluate different classroom manipulations and subsequent effects on human behavior, a step that is exponentially more difficult in real world classrooms.

In this work, we exploit object-of-interest information by using eye-gaze and three main sets of objects in immersive \acs{VR}. We focus on virtual peer-learners, virtual instructor, and screen to understand visual attention through the design of a virtual classroom and a lecture about computational thinking. We choose these objects-of-interests since they are of particular interest with regard to attention towards social dynamics and learning. Our study has three different design factors: Different sitting positions of participating students, different visualization styles of virtual avatars including an instructor and peer-learners, and different hand-raising behaviors of virtual peer-learners. Different sitting positions include seating participating students in the front or back of the virtual classroom. In addition, different visualization styles of avatars consists of two conditions that are cartoon- and realistic-styled avatars. Lastly, different hand-raising behaviors include $20\%$, $35\%$, $65\%$, and $80\%$ of the peer-learners raising their hands to answer questions during the lecture. 
To the best of our knowledge, this is the first work that assesses students' attention by using object-of-interest information in an immersive \acs{VR} classroom through the manipulation of sitting positions of students, visualization styles of peer-learners and instructor, and hand-raising behaviors of peer-learners collectively. Such manipulations may be important indicators of students' visual attention towards lecture contents and social dynamics in the classroom and should be taken into consideration when designing \acs{VR} classrooms. 

\subsection{Related Work}
Since our work benefits from \acs{VR} in education and in eye tracking research, we discuss the state-of-the-art along these two lines. Various studies using \acs{VR} in education settings assess the mechanisms of attention or social dynamics by using pre- or post-tests or by relying on head movement behavior as a proxy for gaze. Using eye tracking in addition to such information presents the possibility of a deeper understanding of visual and situational attention during immersive experiences.

\subsubsection{Virtual Reality in Education and Classrooms}
\acs{VR} offers great promise for supporting teaching and learning procedures, especially when digital learning, physical inabilities, ethical concerns, and situational limitations are considered. An extensive review of immersive \acs{VR} in education and its pedagogical foundations are discussed in~\cite{review_ivr_education} and~\cite{explore_ped_foundations_VR}, respectively. We focus on research on \acs{VR} in education and immersive \acs{VR} classrooms in this section.

The effectiveness of learning in virtual and augmented reality (VR/AR) compared to tablet-based applications and the impact of \acs{VR}-based systems on students' achievements are studied in~\cite{doi:10.1002/ase.1696} and~\cite{alhalabi2016virtual}, respectively, and these works indicate several advantages of \acs{VR}-based conditions. In addition, it has been found that students' motivation increases when \acs{VR} is used as a teaching tool in art history~\cite{riftart} and social studies~\cite{207ed}. \acs{VR} not only supports the effectiveness of learning, but also can improve instructor teaching skills~\cite{lambvirtual}. 

Apart from \acs{VR} applications in teaching and learning, the design and degree of realism in \acs{VR} classrooms have also been studied. Presence of a virtual instructor was found to increase the engagement and progress of users~\cite{livehumanrole}. Furthermore, the processes of synthesizing virtual peer-learners by using previous learner comments~\cite{8797708} and designing \acs{VR} classrooms by replicating real conditions~\cite{vrclassroomconstructivist} which may affect learning are considered.

Several works focused on understanding visual attention and behavior in immersive \acs{VR} classrooms. Bailenson et al.~\cite{bailenson_et_al_2008} and Blume et al.~\cite{blume_et_al_18} studied learning outcomes according to sitting positions and offer compelling evidence that students seated in the front have better learning outcomes. Few studies, however, took head movements into consideration~\cite{doi:10.1089/10949310050078940,DazOrueta2014AULAVR,Nolin2016ClinicaVRCA,Seo2019JointAV} in such setups. In~\cite{DazOrueta2014AULAVR}, the immersive \acs{VR} classroom was used as a tool to study attention measures for attention deficit/hyperactivity disorder (ADHD), whereas in~\cite{Nolin2016ClinicaVRCA} reliability of virtual reality and attention was studied with continuous performance task (CPT) for clinical research. Social interaction using head movements was studied in~\cite{Seo2019JointAV} with users' head movements found to shift between the interaction partner and target. Some studies argued for eye tracking measurements, especially in clinical research for diagnosis or attention related tasks~\cite{rizzo_bowerly_buckwalter_klimchuk_mitura_parsons_2009,mangalmurti_2020}. However, none of the previous works have focused on social interactions and dynamics in the immersive \acs{VR} classroom in an everyday setting by using object-of-interest information and eye movements. 

\subsubsection{Eye Tracking in Virtual Reality}
Eye tracking and gaze estimation are considered challenging tasks in a real world setting because it is difficult to control factors such as occlusions or illumination changes~\cite{fuhl_pupil_det_in_the_wild,eth_xgaze_eccv20}. However, in most of the \acs{VR} setups, eye trackers are located inside of \acs{HMD}s. This creates not only a more controlled and reliable environment for eye tracking, but also provides a unique opportunity to analyze and process human visual behavior during the \acs{VR} experience.

Eye tracking has been used in many applications and shown to be helpful for various tasks in \acs{VR} such as guiding attention in panoramic videos using central and peripheral cues~\cite{directing_attr_attention_20}, predicting motion sickness by using 3D Convolutional Neural Networks~\cite{8642906}, synthesizing personalized training programs to improve skills~\cite{8448290}, foveated rendering using saccadic eye movements and eye-dominance~\cite{Arabadzhiyska2017,9005240}, evaluation and diagnoses of diseases such as Parkinson's disease~\cite{7829437}, re-directed walking using blinking behavior~\cite{redirected_walking_steinicke}, or continuous authentication using eye movements~\cite{Zhang:2018:CAU:3178157.3161410}. While these works have used either the eye tracking or gaze data to derive more meaningful information for related tasks, assessing visual attention via eyes and gaze-based interaction is more relevant for classroom setups in particular. Bozkir et al.~\cite{bozkir_vr_attention_et} assessed visual attention using gaze guidance and pupil dilations in a time-critical situation, whereas Khamis et al.~\cite{VRPursuits_interaction} discussed gaze-based interaction using smooth pursuit eye movements in \acs{VR}. In addition, Sidenmark and Lundstr\"{o}m~\cite{10.1145/3314111.3319815} analyzed eye fixations on interacted objects during hand interaction in \acs{VR} and found that interaction with stationary objects may be favorable. Aforementioned works indicate that eye movements can be used reliably in \acs{VR} setups. Moreover, considering that the majority of objects in a classroom are stationary or have limited spatial movement, visual attention extracted from such data may provide valuable insight into human behavior. While exploiting objects-of-interests could be considered as a primitive task, it forms the foundation of more complex tasks necessary to understand visual attention.

\subsection{Methodology}
The main focus of this work is to investigate object-of-interest information in different manipulations of an immersive \acs{VR} classroom. We focus on three objects that may be considered as the most important objects in the current setup, namely peer-learners, instructor, and screen.

\subsubsection{Participants}
$381$ volunteer sixth-grade students ($179$ female and $202$ male) between $10$ to $13$ years old ($M=11.5$, $SD=0.6$) were recruited for the experiment. In this age group, students are able to use an \acs{HMD}, but do not have much experience with \acs{VR}. They also had no background knowledge about the lecture content. Data from $101$ participants were removed due to hardware related problems, incorrect calibration, low eye tracking ratio (lower than $90\%$), and synchronization issues. The average number of participants per condition was $17.5$ ($SD=5.2$). Finally, we used the data of $280$ participants ($140$ female and $140$ male) with the aforementioned average age and standard deviation. For each condition group separately, participants' gender was also equally distributed ($M = 0.58, SD = 0.08$). The study was approved by the ethics committee of the University of T{\"u}bingen prior to the experiments. Participants and their parents or legal guardians provided written informed consent in advance.

\subsubsection{Apparatus}
For the experiments, HTC Vive Pro Eye devices with integrated Tobii eye trackers were used. The HTC Vive Pro Eye has a refresh rate of $90$ Hz and field of view of $110^{\circ}$. The integrated eye tracker has $120$ Hz sampling rate. The screen resolution per eye was set to $1440 \times 1600$. Unreal Game Engine v$4$.$23$.$1$\footnote{https://www.unrealengine.com/} was used to render the virtual classroom.

\subsubsection{Experimental Design}
The virtual classroom consists of $4$ rows of desks organized in $2$ columns. Next to each desk, chairs are located to let virtual peer-learners sit. There are $24$ virtual peer-learners in the environment and all of them sit on chairs during the entirety of the lecture. Some of the chairs are kept empty so as not to overcrowd the virtual classroom. In addition, the virtual classroom includes other objects, which exist in real classrooms such as board, screen, cupboard, clock, and windows. The lecture content is visualized on the white screen. Additionally, the virtual instructor walks around the podium, replicating behavior similar to that of a real instructor. Figures~\ref{fig:classroom_ss_VR21} (a), (b), (c), and (d) show the overall design, hand-raising peer-learners, realistic-styled peer-learners, and cartoon-styled peer-learners, respectively.

\begin{figure*}[ht]
  \centering
   \subfigure[Overall virtual classroom design.]{{\includegraphics[height = 3.75cm]{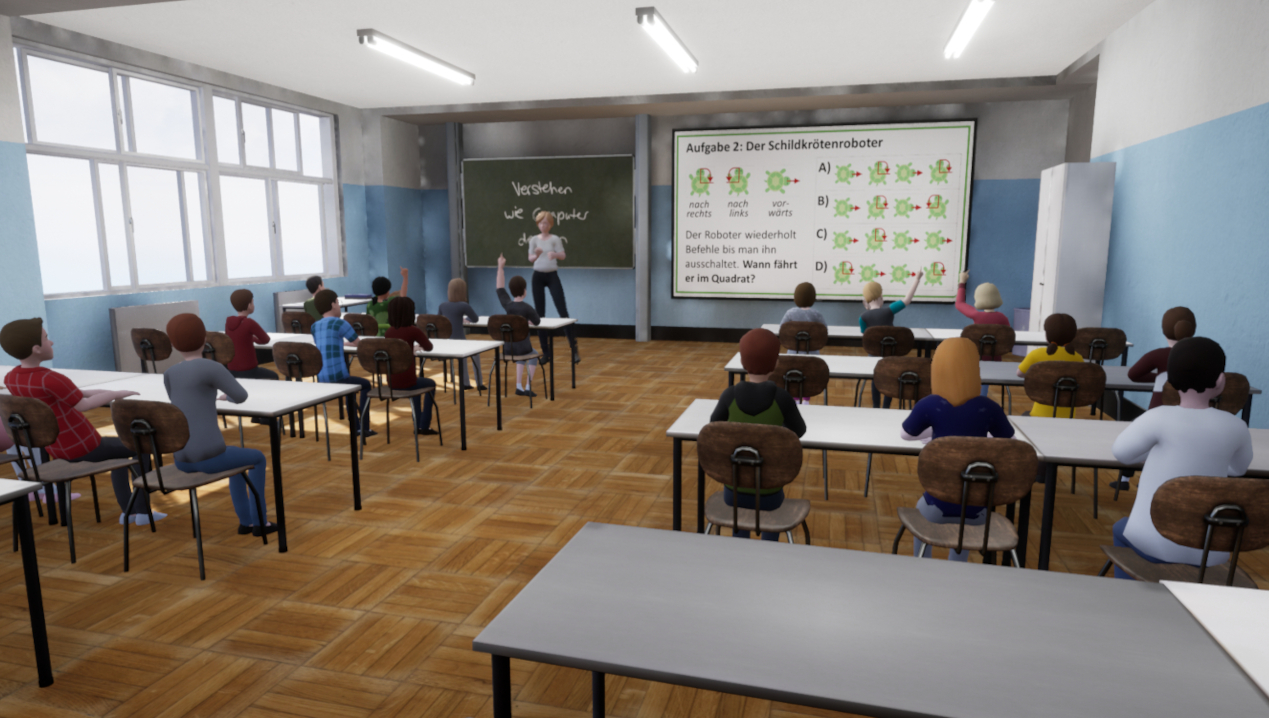}}}
   \qquad
   \subfigure[Hand-raising cartoon-styled peer-learners from back.]{{\includegraphics[height = 3.75cm]{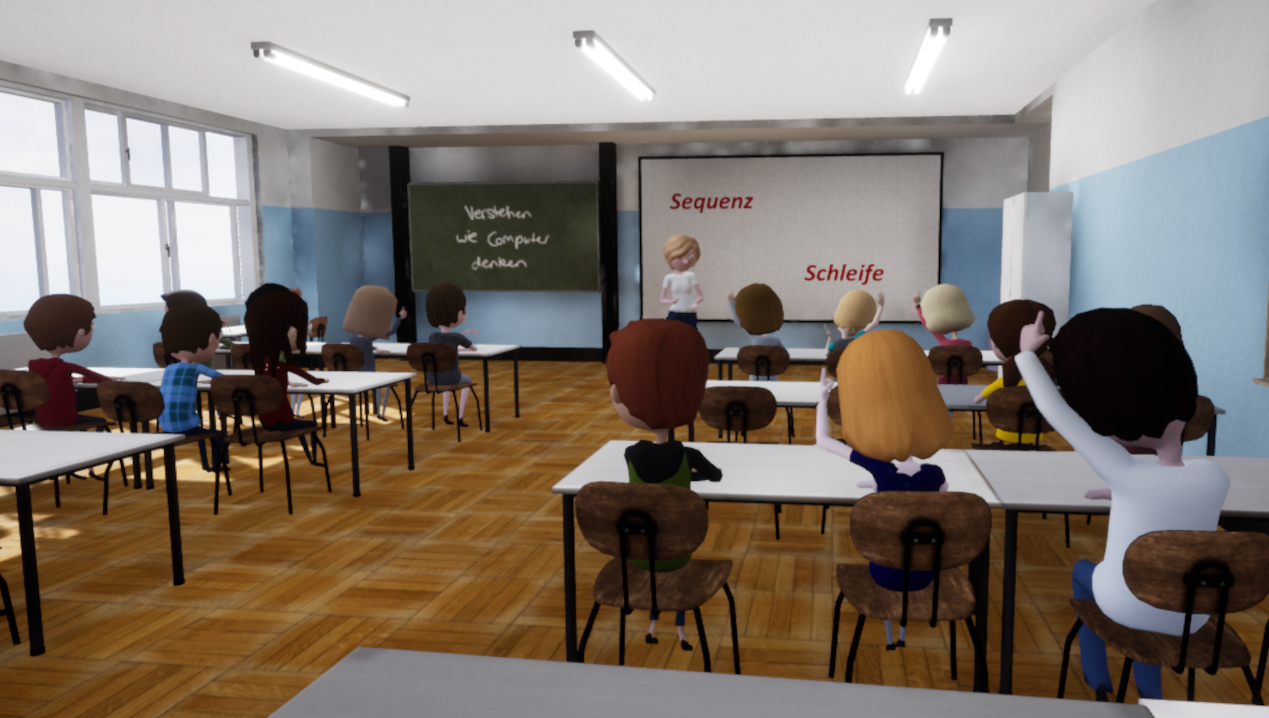} }}%
   \qquad
   \subfigure[Realistic-styled peer-learners.]{{\includegraphics[height = 3.75cm]{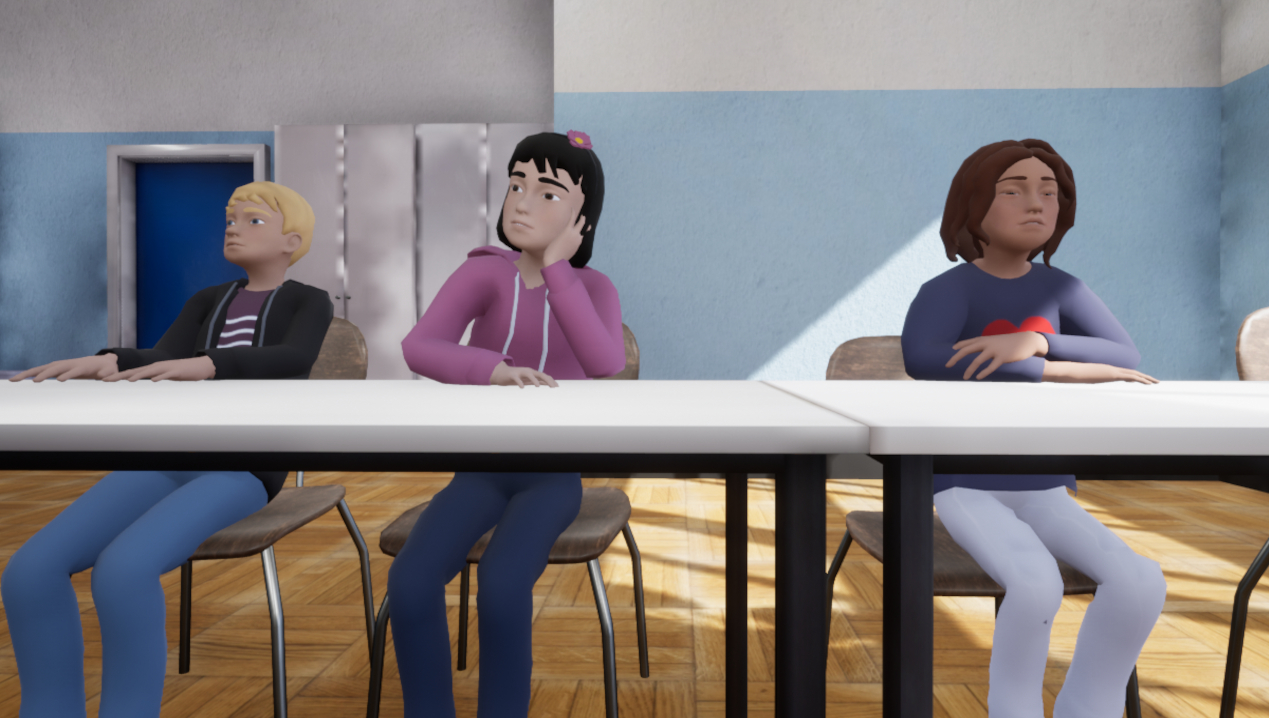} }}
   \qquad
   \subfigure[Hand-raising cartoon-styled peer-learners.]{{\includegraphics[height = 3.75cm]{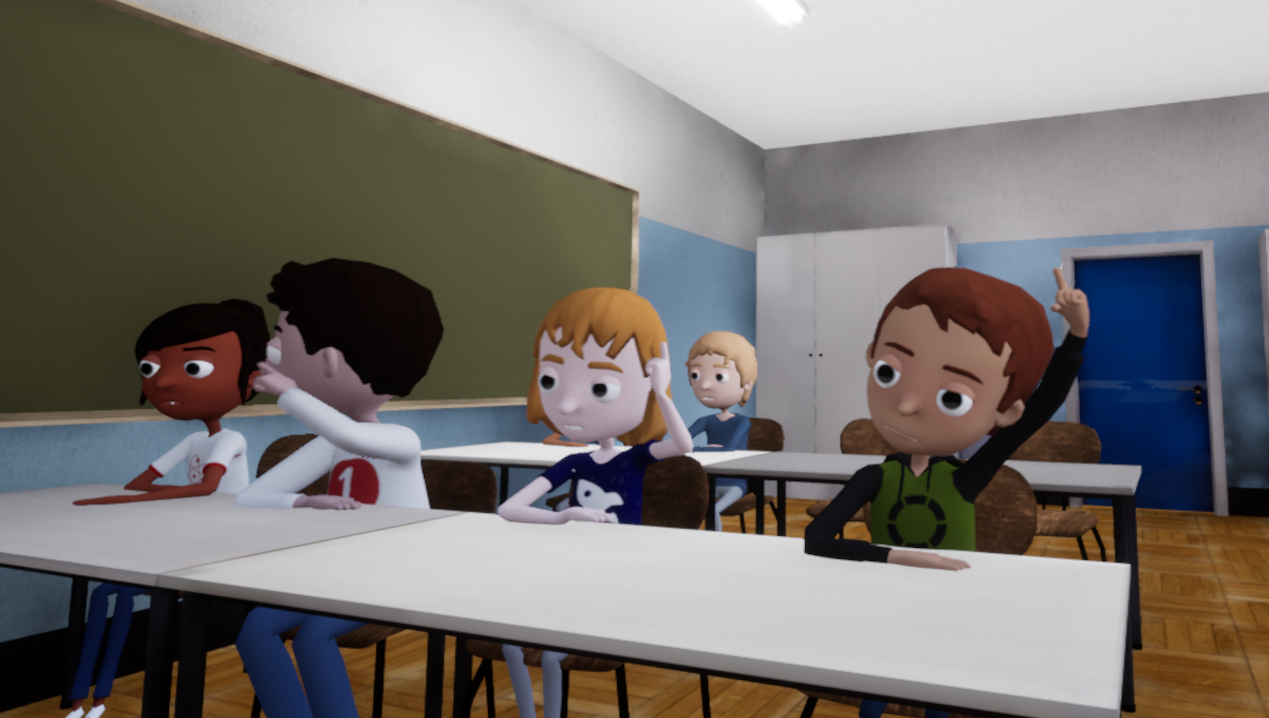}}}
  \caption{Views from the virtual classroom.}
  \label{fig:classroom_ss_VR21}%
\end{figure*}

The content of the virtual lecture is about computational thinking~\cite{Weintrop2016DefiningCT} and the lecture takes $\approx 15$ minutes in total, including $4$ phases. These four phases are grouped as ``Introduction to the topic'', ``Knowledge input'', ``Exercises'', and ``Summary'' and take $\approx3$, $\approx4.5$, $\approx5.5$, and $\approx1.5$ minutes, respectively. The topic of the virtual lecture is visible on the board as ``Understanding how computers think''. The first phase starts with the virtual instructor entering the classroom. After staying for a while, the instructor leaves the classroom for about $20$ seconds. During this time, participants have the opportunity to explore the classroom, look around, and acclimate themselves with the virtual environment. During the initial phase of the lecture, the instructor asks five questions, and some of the virtual peer-learners raise their hands to interact. In the second phase, the instructor describes two terms, ``sequence'' and ``loop'', and shows these terms on the white screen. After the descriptions, the instructor asks four questions about each term and some of the peer-learners raise their hands to answer them. In the third phase, the instructor assigns two exercises and allows students some time to think about them. Later, choices for each exercise are provided by the instructor and, this time, peer-learners raise their hands to vote on the correct answer out of the presented options. In the fourth phase, the instructor summarizes the lecture without asking any questions, which means that peer-learners do not raise their hands. In addition, no hand-raise is expected from the participants as hand poses are not measured during the experiments. 

Our study is conceptualized in a between-subjects design. We evaluated three design factors, namely sitting positions of the participants, visualization styles of virtual avatars, and hand-raising percentages of virtual peer-learners. Participants were seated either in the front or back rows, which means that the participants seated in the front had one row in front of them, whereas participants seated in the back had three rows between them and the screen. Both conditions were aligned in the aisle side of the desks that were on the right side of the classroom. This manipulation can give insights about students' attention during a lecture, when they have either the overview over whole class and see most of their virtual peer-learners or when they are positioned closer to instructor and screen the lecture is presented on. Participants encountered either cartoon- or realistic-styled virtual avatars in the environment, including the virtual instructor and peer-learners. The cartoon-styled avatars have larger heads and tinier arms and legs as compared to the realistic-styled avatars. Since the animation and design of more realistic looking avatars is time and cost expensive, it should be interesting to investigate the impact of such manipulation. In addition, various hand-raising percentages of virtual peer-learners consist of four levels, namely $20\%$, $35\%$, $65\%$, and $80\%$. This means that when a question is asked during the lecture by the virtual instructor, a corresponding percentage of virtual peer-learners raise their hands to answer the question. The last two manipulations are of particular interest, regarding the question how social avatars should be designed in a virtual classroom and how they are perceived by students. Under which condition do students use social information and how does visualization and certain behaviour influence students attention. This helps to simulate and evaluate social dynamics and engagement during the virtual lecture using visual attention. In total, our $2$ (factor $1$) $\times 2$ (factor $2$) $\times 4$ (factor $3$) between-subjects design leads to $16$ treatment groups.

\subsubsection{Procedure}
In the beginning of the experiment, the assistants introduced the experiment and its process to the participants. Participants had the opportunity to familiarize themselves with the hardware and the \acs{VR} environment. Afterwards, the actual experiment and data collection began. Firstly, the eye tracker was calibrated. Then, the experiment was started with assistants pressing a start button. At the end of the virtual lecture, the participants were told to take the \acs{HMD} off by a message which was displayed in the virtual environment. Virtual lectures were carried out without any breaks. After watching the virtual lecture, participants filled out questionnaires about their perceived realism and experienced presence which were conceptualized for the \acs{VR} classroom according to~\cite{lombard_presence,schubert_presence}.

Each session took $\approx45$ minutes in total. The experiments were carried out in groups of ten participants who were randomly allocated to one of the $16$ treatment groups by using a random number generator to ensure the random distribution of conditions within groups. To maintain natural behavior, participants selected the physical seat in the experiment room freely without being informed about experimental conditions. Although research assistants helped with technical issues regarding the use of the \acs{HMD}, participants were blinded to the true purpose and design of the study, as it was solely introduced as a learning experience.

\subsubsection{Data Processing and Measurements}
During the experiments, head location and pose, gaze, and eye related data along with experimental condition were collected. Head movements are particularly helpful for mapping eye-gaze in the virtual environment. These were saved in data sheets for each participant using anonymous identifiers which ensured the privacy of the participants.

As gaze data reported by the eye tracker can be affected negatively by blinks or noisy sensor measurements, we applied a linear interpolation on the gaze vectors to clean the data. Afterwards, using head pose and interpolated gaze data, we applied ray-casting~\cite{ROTH1982109} to map the gaze into the 3D virtual environment. The objects in the 3D environment are surrounded by dedicated colliders; therefore, we were able to calculate 3D gaze points and gazed objects using the procedure visualized in Figure~\ref{fig:ray_cast_VR21}.

\begin{figure}[ht]
  \centering
   \includegraphics[width=\linewidth, keepaspectratio]{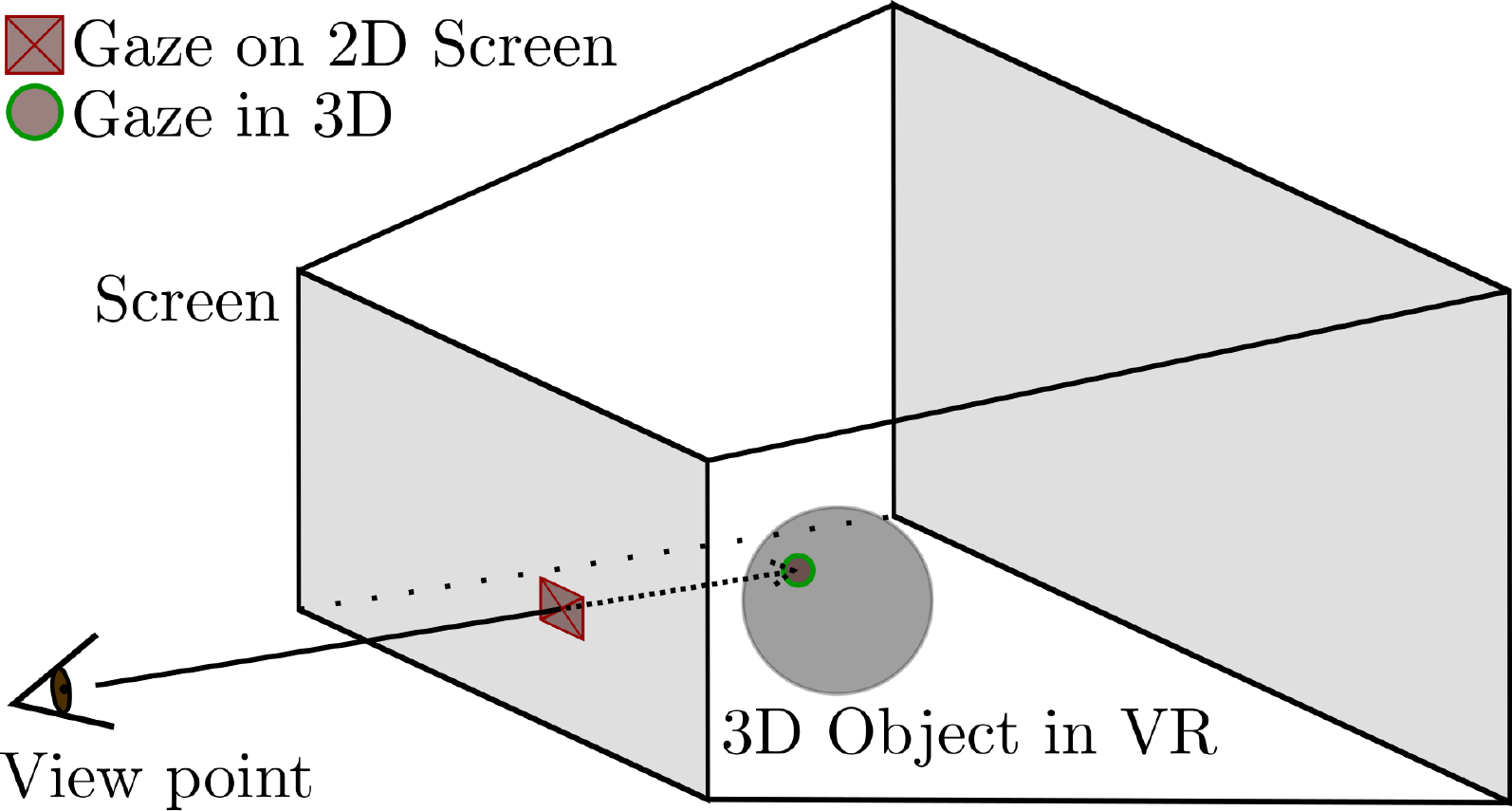}
  \caption{Ray-casting procedure to obtain 3D gazed object.}
  \label{fig:ray_cast_VR21}%
\end{figure}

However, gazed objects may not directly represent visual attention as participants can gaze on some objects unconsciously for a very short time. To overcome this issue, we set an attention threshold of $200$ ms, meaning that we count the objects as object-of-interest if participants stay with their gaze on the objects for at least the amount of the attention threshold. As we assume that both fixations and saccades can occur during attending one object, the selected threshold is larger than classical fixation thresholds applied in eye tracking literature for both conventional~\cite{salvucci2000identifying} or \acs{VR} eye tracking~\cite{agtzidis2019360} setups. While we also experimented with various threshold values, our results show similar trends across different thresholds.

In addition to the data related to visual attention, self-reported perceived realism and experienced presence were obtained at the end of the experiments with $4$-point Likert scales ranging from $1$ (``completely disagree'') to $4$ (``completely agree'') with $6$ (e.g., ``I felt like the teacher and the classmates could be real people'') and $9$ (e.g., ``During the virtual lecture, I almost forgot that I was wearing the \acs{VR} glasses'') items, respectively.

In this study, we focused on three main objects in the virtual classroom, namely peer-learners, virtual instructor, and screen, when we extracted object-of-interest information. We decided that these objects may have a significant impact on social dynamics in the classrooms and for overall course of lecture. In our analyses, the attention time on each peer-learner is aggregated and the object of ``peer-learners'' represents the aggregated object and related attention. In addition, in our classroom setup there is one board and one white screen behind the instructor as depicted in Figure~\ref{fig:classroom_ss_VR21} (a). The lecture content is provided on the white screen only; therefore, in our analysis we refer to the white screen when mentioning screen object.

\subsubsection{Research Hypotheses}
Our hypotheses correspond to the experimental factors of sitting positions, avatar visualization styles, and various hand-raise percentages of virtual peer-learners, respectively. Furthermore, since we analyze behaviors towards three different objects in the virtual classroom, namely peer-learners, instructor, and screen, for simplicity we call attention to attending these objects-of-interests for the rest of the paper.

\paragraph{Visual Attention in Different Sitting Positions (H1) \\}
We expect that participants seated in the front condition have less attention on peer-learners, naturally because they do not have as many peer-learners sitting in front of them as opposed to the participants sitting in the back. In addition, the participants that are located in the front are closer to the virtual instructor and the screen that visualizes lecture content. Due to the proximity and having fewer moving and occluding objects in their field of view (FOV), we hypothesize that these participants have more attention time on both virtual instructor and screen than the participants sit in the back.

\paragraph{Visual Attention in Different Visualization Styles of Virtual Avatars (H2) \\} 
We hypothesize that attention time on peer-learners in the cartoon-styled visualization is longer than in the realistic-styled visualization as cartoon-styled peer-learners are more exciting for participants when ages of our interest group are taken into consideration. In addition, we assume that participants look at the realistic-styled instructor for longer than at cartoon-styled instructor as participants may consider the realistically rendered instructor more credible in a learning environment. Lastly, we do not expect any differences in terms of attention towards virtual screen that lecture content is visualized, as the visualization style of the screen does not change. 

\paragraph{Visual Attention in Different Hand-raising Behaviors of Peer-learners (H3) \\} 
We hypothesize that attention time on peer-learners increases with a higher number of virtual peer-learners raising their hands when questions are asked, as this would create a visually more dynamic classroom. Additionally, we expect that if fewer virtual peer-learners raise their hands, this will lead participants to keep their attention either on the instructor or the lecture screen due to having less amount of visual distractors when questions are provided by the virtual instructor.

\subsection{Results}
In this section, we analyze the total amount of time spent on each object-of-interest (OOI), which we call visual attention, between different conditions. For each \acs{OOI}, we applied a $3$-way full factorial \acs{ANOVA} for statistical comparison using alpha level of $0.05$. For non-parametric analysis, we transformed the data using the aligned rank transform (ART)~\cite{10.1145/1978942.1978963} before applying \acs{ANOVA}s. For the pairwise comparisons, we used Tukey-Kramer post-hoc test as the sample sizes were not equal. While the main focus of this work is to assess visual attention using \acs{OOI} information, here we report experienced presence and perceived realism questionnaires to support our main results. We obtained mean values of $2.91$ for experienced presence and perceived realism with $SD = 0.55$ and $SD = 0.57$, respectively, without any significant differences between conditions.

\subsubsection{Visual Attention on Peer-learners}
Total time spent on peer-learners for different sitting positions, avatar visualization styles, and various hand-raising behaviors are depicted in Figures~\ref{fig:attention_peer-learners_VR21} (a), (b), and (c), respectively. Total time spent on peer-learners is significantly longer in the back seated condition ($M = 115.07$ sec, $SD = 85.28$ sec) than it is in the front seated condition ($M = 33.59$ sec, $SD = 32.45$ sec) with ($F(1,264) = 156.23$, $p < .0001$, $\eta^2 = .36$). 

\begin{figure*}[ht]
  \centering
   \subfigure[Comparison between sitting positions.]{{\includegraphics[height = 5.64cm]{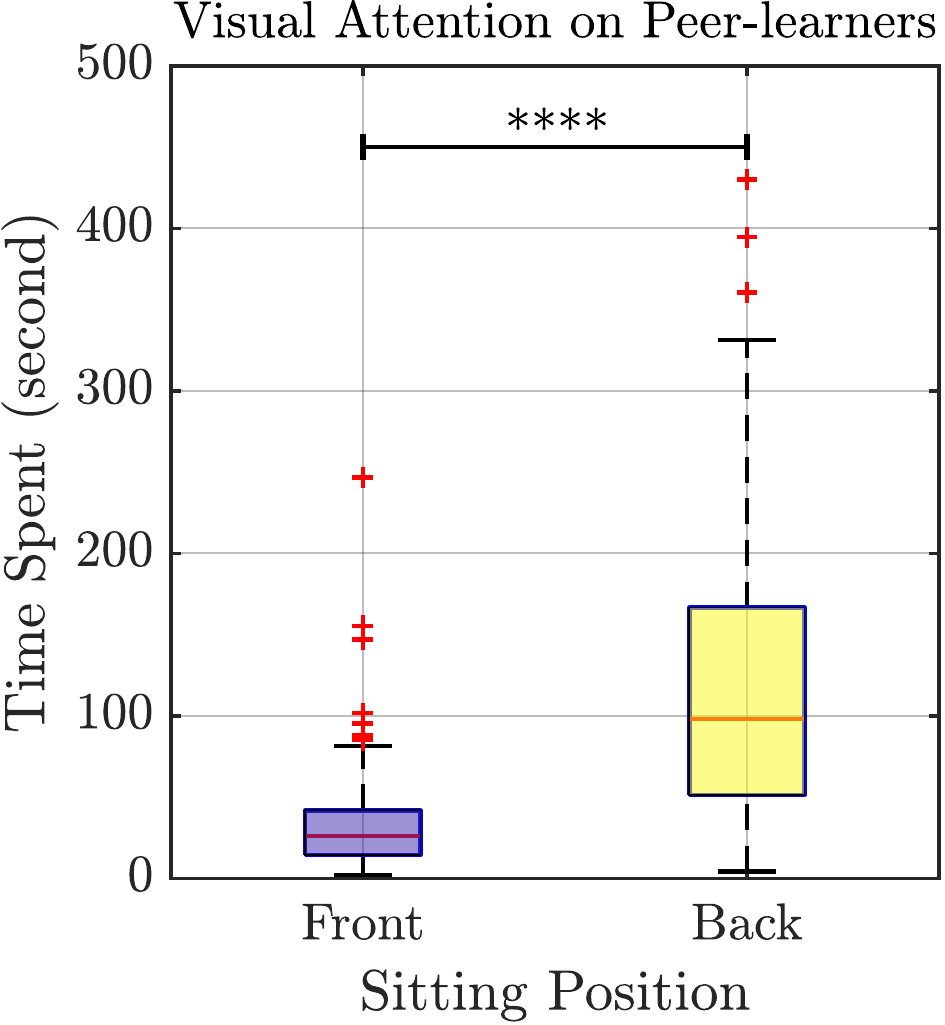} }}%
   \quad
   \quad
   \quad
   \quad
   \subfigure[Comparison between visualization types.]{{\includegraphics[height = 5.64cm]{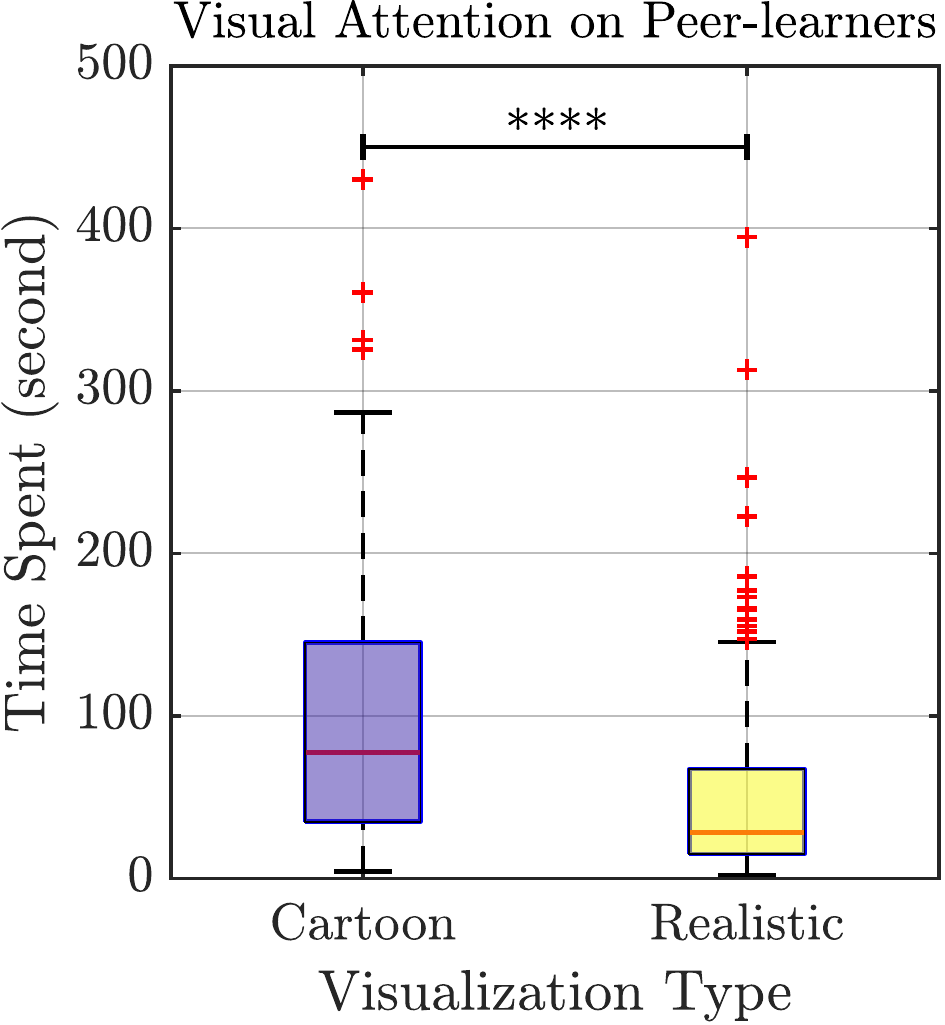} }}
   \quad
    \subfigure[Comparison between hand-raising behaviors.]{{\includegraphics[height = 5.64cm]{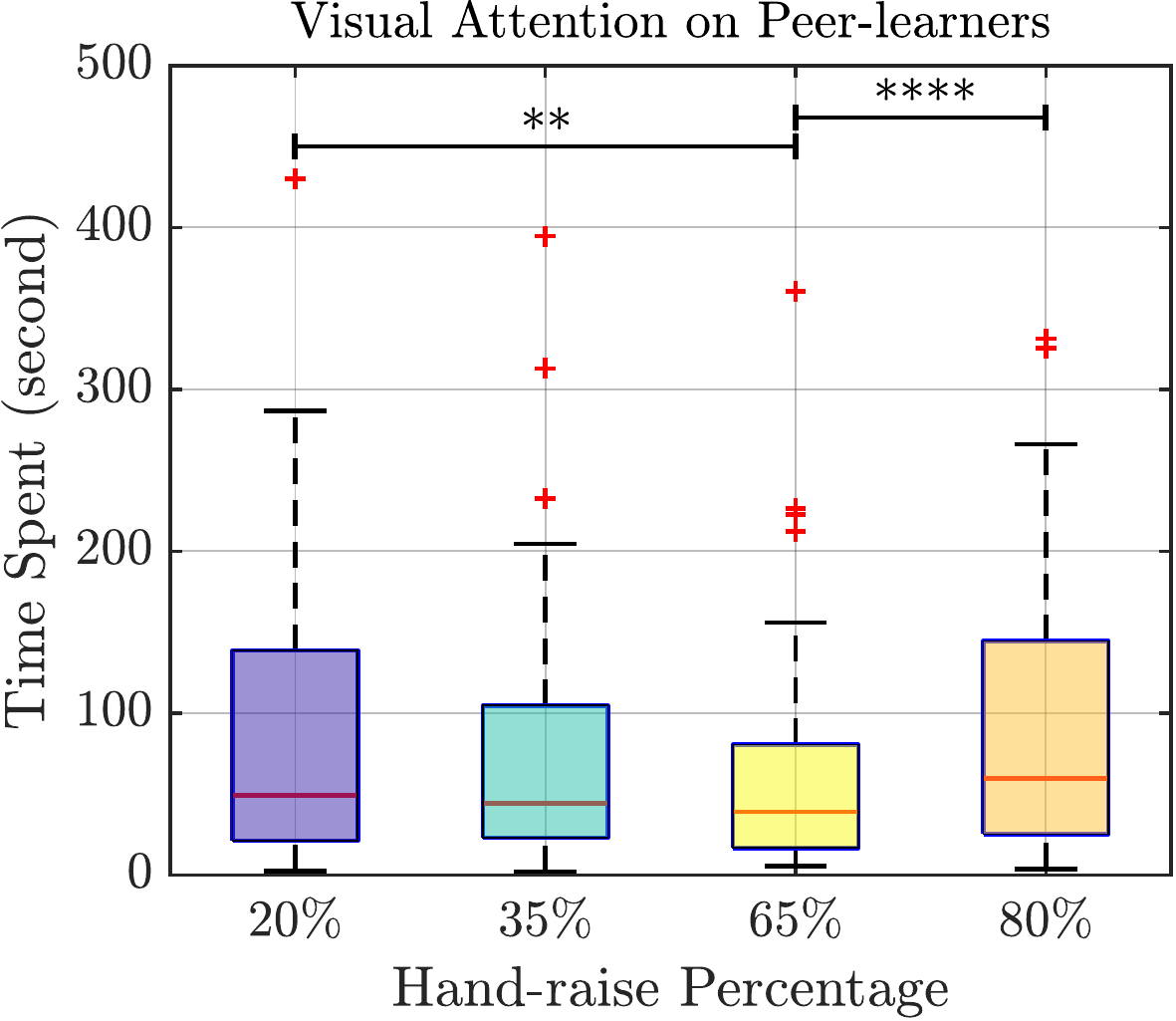}}}
  \caption{Attention towards virtual peer-learners for different classroom manipulation configurations. $**$ and $****$ correspond to the significance levels of $p<.01$ and $p<.0001$, respectively.}
  \label{fig:attention_peer-learners_VR21}%
\end{figure*}

Attention towards peer-learners as different visualization styled avatars differs significantly. Cartoon-styled peer-learners ($M = 98.67$ sec, $SD = 82.79$ sec) drew significantly more attention than the realistic-styled peer-learners ($M = 55.28$ sec, $SD = 65.65$ sec) with ($F(1,264) = 54.13$, $p < .0001$, $\eta^2 = .17$).

Furthermore, for different hand-raising manipulations, attention time on the peer-learners differs significantly with ($F(3,264) = 6.93$, $p < .001$, $\eta^2 = .07$). Particularly, the total time spent on peer-learners in the $80\%$ condition ($M = 88.95$ sec, $SD = 78.15$ sec) is significantly longer than in the $65\%$ condition ($M = 59.23$ sec, $SD = 65.19$ sec) with ($F(3,264) = 6.93$, $p < .0001$, $\eta^2 = .07$). In addition, the total time spent in the $20\%$ condition ($M = 88.62$ sec, $SD = 87.53$ sec) is significantly longer than in the $65\%$ condition ($M = 59.23$ sec, $SD = 65.19$ sec) with ($F(3,264) = 6.93$, $p = .005$). In summary, attention time towards extreme levels of hand-raising percentages are longer than for intermediate levels.

Additionally, we found some significant interaction effects regarding the attention time on the peer-learners. The time on peer-learners in the hand-raising condition depends on the sitting position of the students with ($F(3,264) = 3.88$, $p = .0097$, $\eta^2 = .041 $), as well as the attention time on peer-learners in the avatar visualization styles condition depends on the sitting position with ($F(1,264) = 11.37$, $p < .001$, $\eta^2 = .039$) and vice versa. A small interaction effect was found between the hand-raising condition and the avatar visualization styles with ($F(3,264) = 3.36$, $p = .02$, $\eta^2 = .036$).

\subsubsection{Visual Attention on Instructor}
Total time spent on instructor for different sitting positions, avatar visualization styles, and various hand-raising behaviors are depicted in Figures~\ref{fig:attention_instructor_VR21} (a), (b), and (c), respectively. The participants that are seated in the front ($M = 190.07$ sec, $SD = 93.13$ sec) attended to the virtual instructor significantly more than the participants seated in the back ($M = 80.37$ sec, $SD = 60.78$ sec) with ($F(1,264) = 144.16$ $p < .0001$, $\eta^2 = .34$).

\begin{figure*}[ht]
  \centering
   \subfigure[Comparison between sitting positions.]{{\includegraphics[height = 5.64cm]{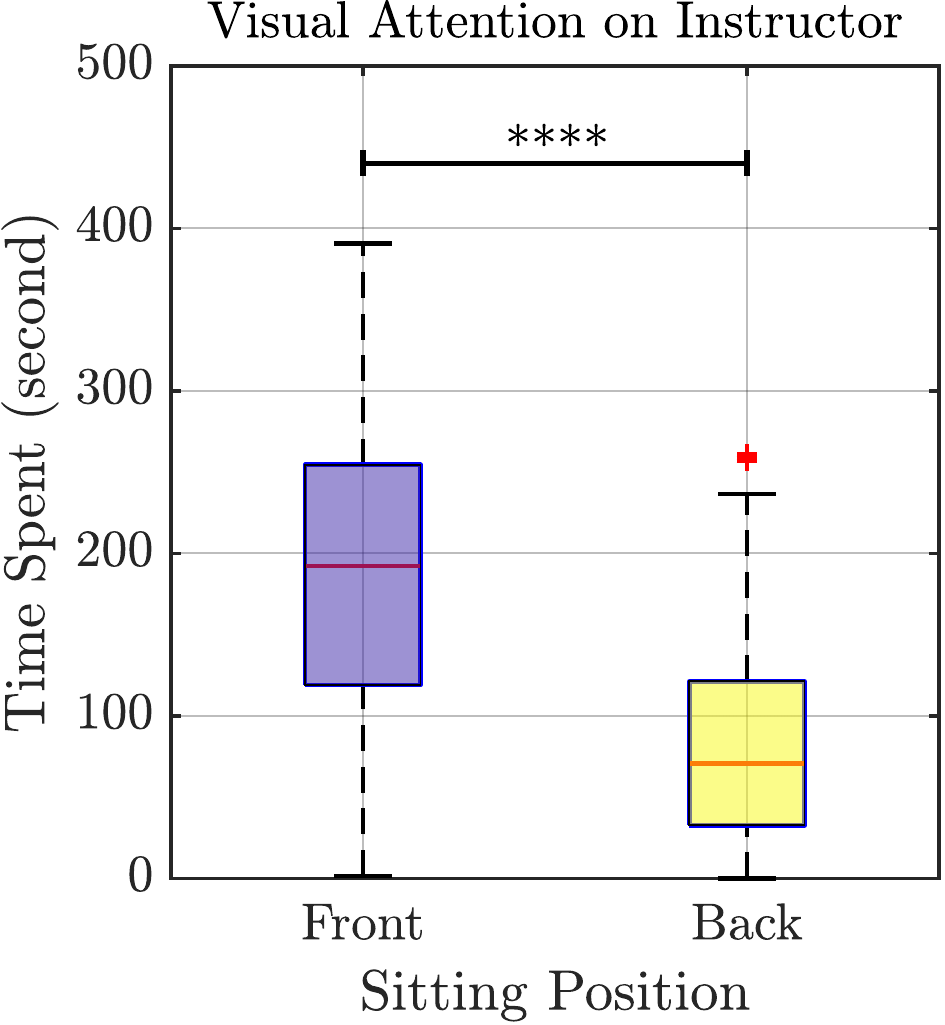} }}%
   \quad
   \quad
   \quad
   \quad
   \subfigure[Comparison between visualization types.]{{\includegraphics[height = 5.64cm]{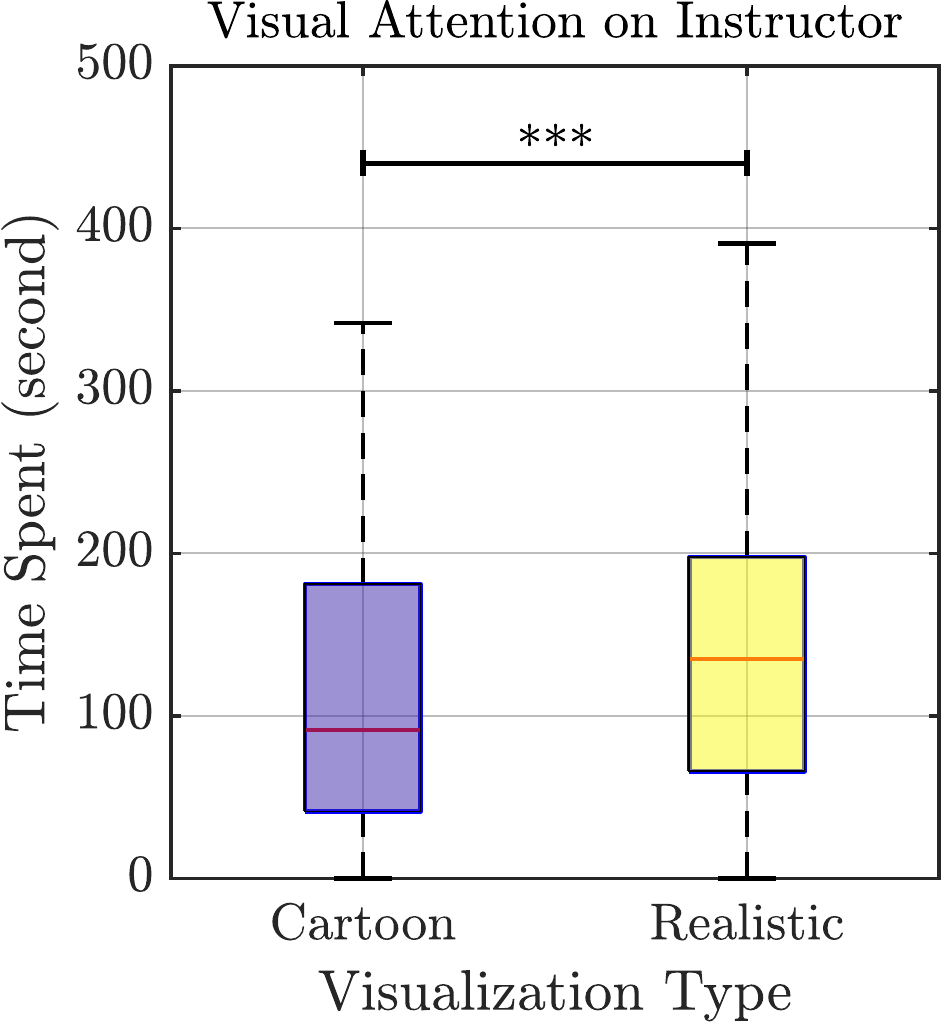} }}
   \quad
   \subfigure[Comparison between hand-raising behaviors.]{{\includegraphics[height = 5.64cm]{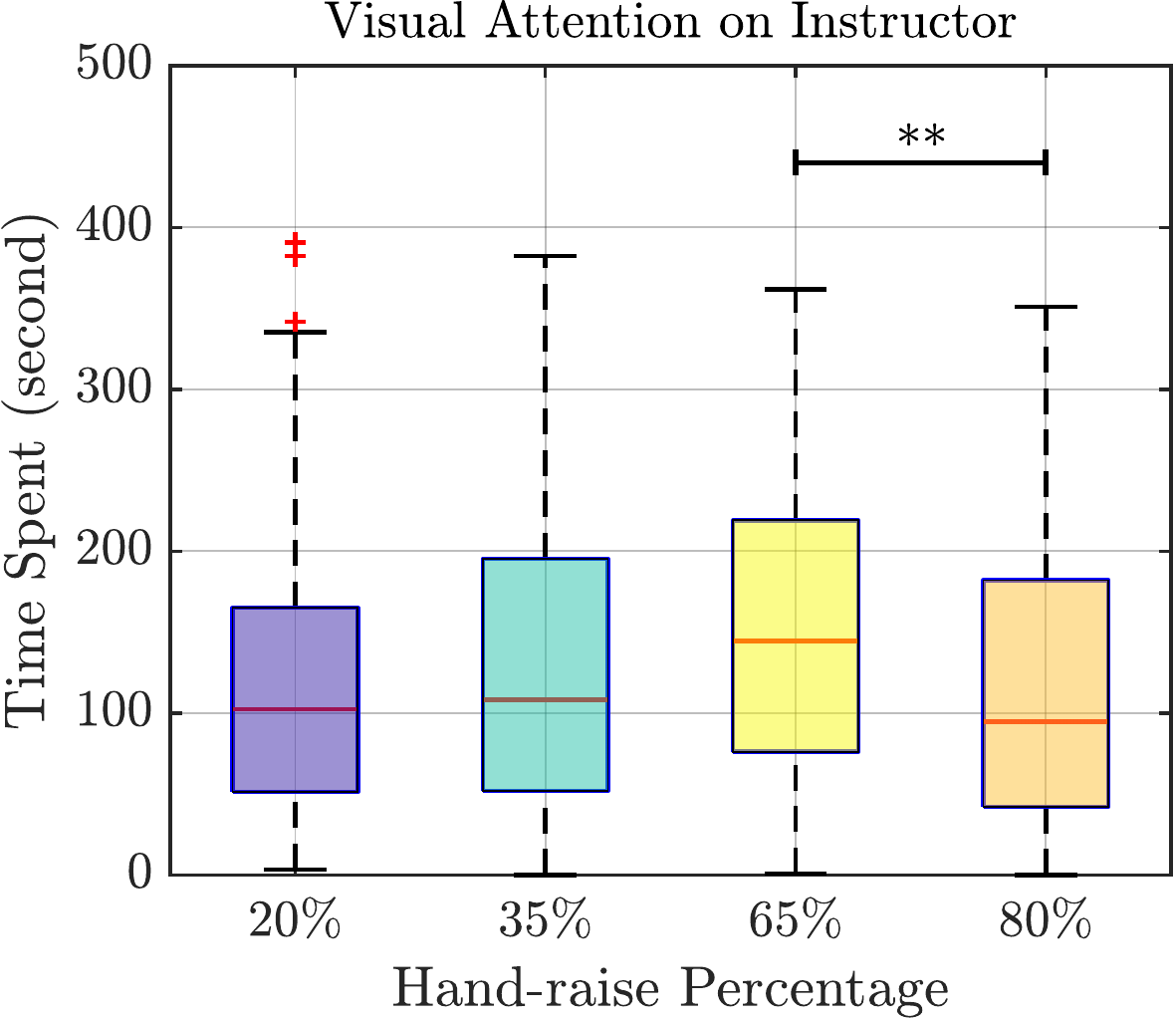}}}
  \caption{Attention towards virtual instructor for different classroom manipulation configurations. $**$, $***$, and $****$ correspond to the significance levels of $p < .01$, $p < .001$, and $p < .0001$, respectively.}
  \label{fig:attention_instructor_VR21}%
\end{figure*}

The virtual instructor drew significantly more attention in the realistic-styled avatar condition ($M = 145.98$ sec, $SD = 96.63$ sec) than in the cartoon-styled avatar condition ($M = 114.82$ sec, $SD = 89.83$ sec) with ($F(1,264) = 11.81$, $p < .001$, $\eta^2 = .04$).

Furthermore, attention time on the instructor is found to differ significantly between different hand-raising behaviors of the peer-learners with ($F(3,264) = 3.54$, $p = .015$, $\eta^2 = .04$). In particular, the total time spent on virtual instructor in the $65\%$ condition ($M = 152.46$ sec, $SD = 91.48$ sec) is significantly longer than the $80\%$ condition ($M = 117.39$ sec, $SD = 91.12$ sec) with ($F(3,264) = 3.54$, $p = .009$, $\eta^2 = .04$). Overall, more attention is drawn by the virtual instructor in the intermediate levels of hand-raising than the extreme levels. There were no interaction effects found for attention time on instructor.

\subsubsection{Visual Attention on Screen}
Total time spent on the screen, where the lecture content visualized for different sitting positions, avatar visualization styles, and various hand-raising behaviors are depicted in Figures~\ref{fig:attention_screen_VR21} (a), (b), and (c), respectively. The participants that are seated in the front ($M = 218.65$ sec, $SD = 78.70$ sec) attended to the lecture screen for a significantly longer period of time than the back seated participants ($M = 154.21$ sec, $SD = 96.88$ sec) with ($F(1,264) = 42.5$, $p < .0001$, $\eta^2 = .14$).  

\begin{figure*}[ht]
  \centering
   \subfigure[Comparison between sitting positions.]{{\includegraphics[height = 5.64cm]{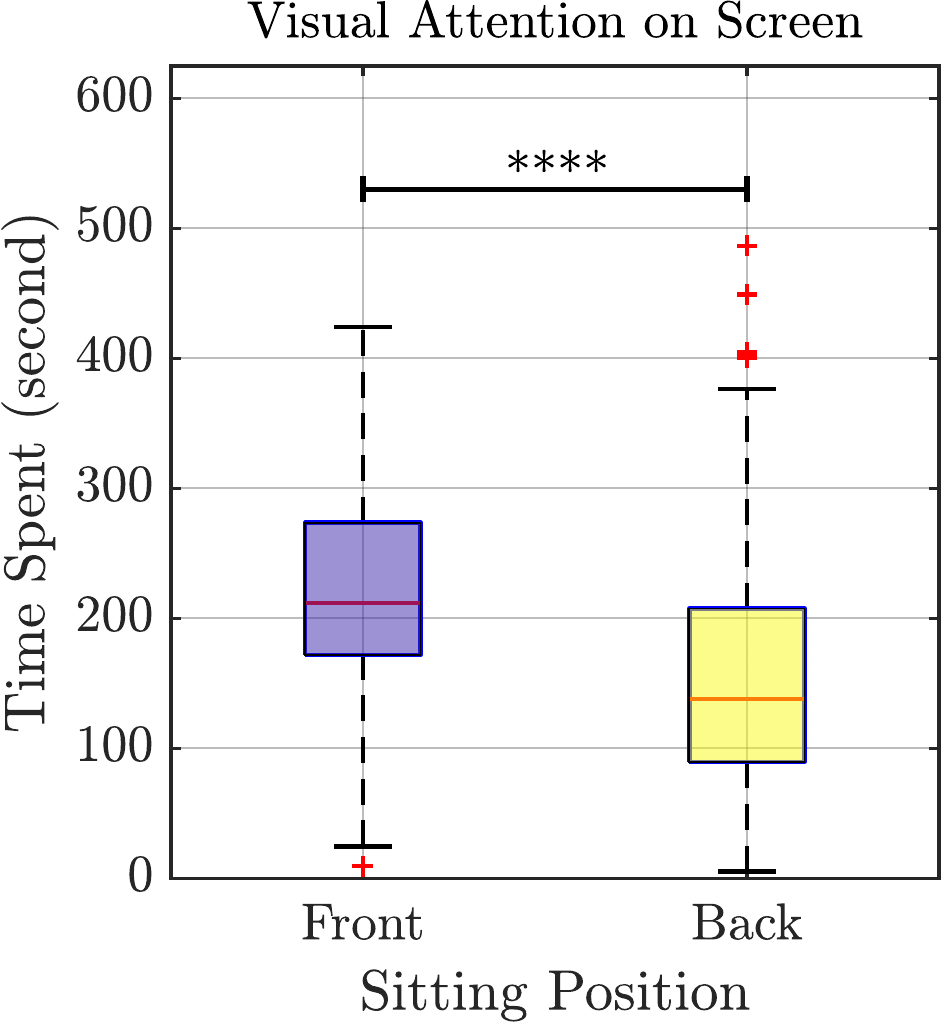} }}%
   \quad
   \quad
   \quad
   \quad
   \subfigure[Comparison between visualization types.]{{\includegraphics[height = 5.64cm]{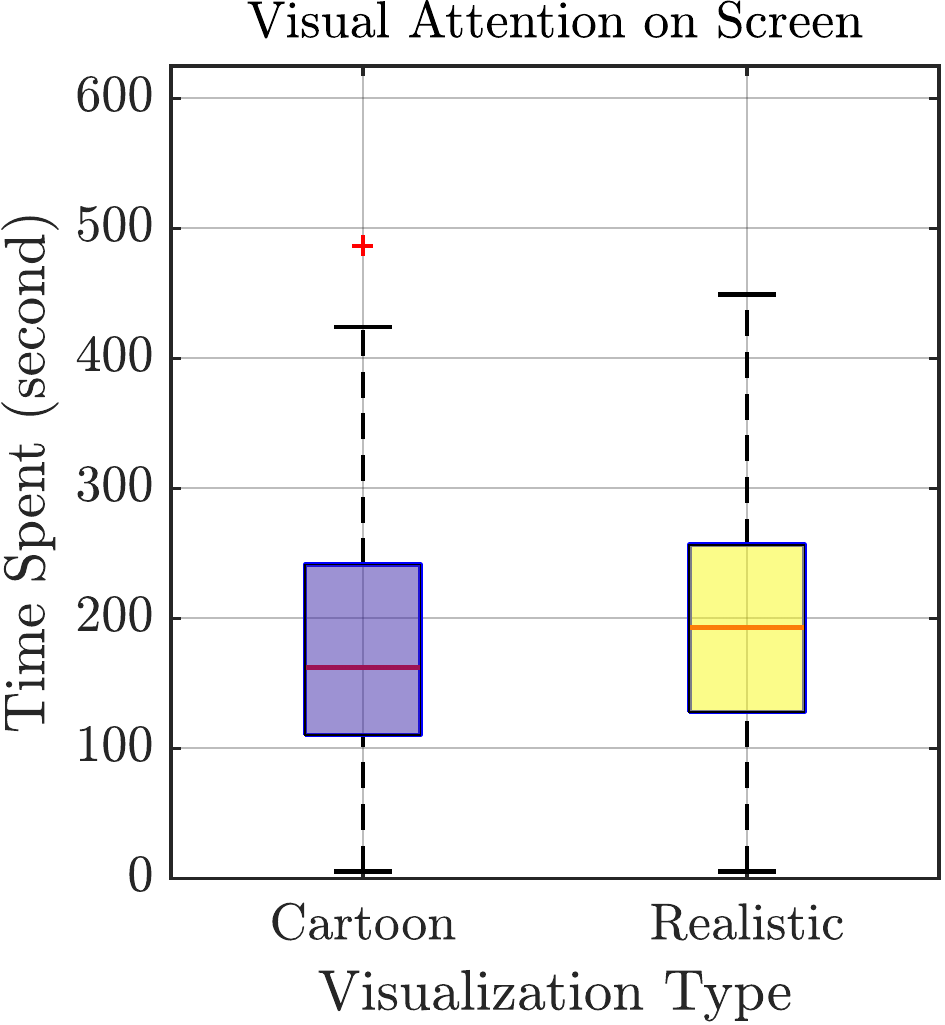} }}
   \quad
   \subfigure[Comparison between hand-raising behaviors.]{{\includegraphics[height = 5.64cm]{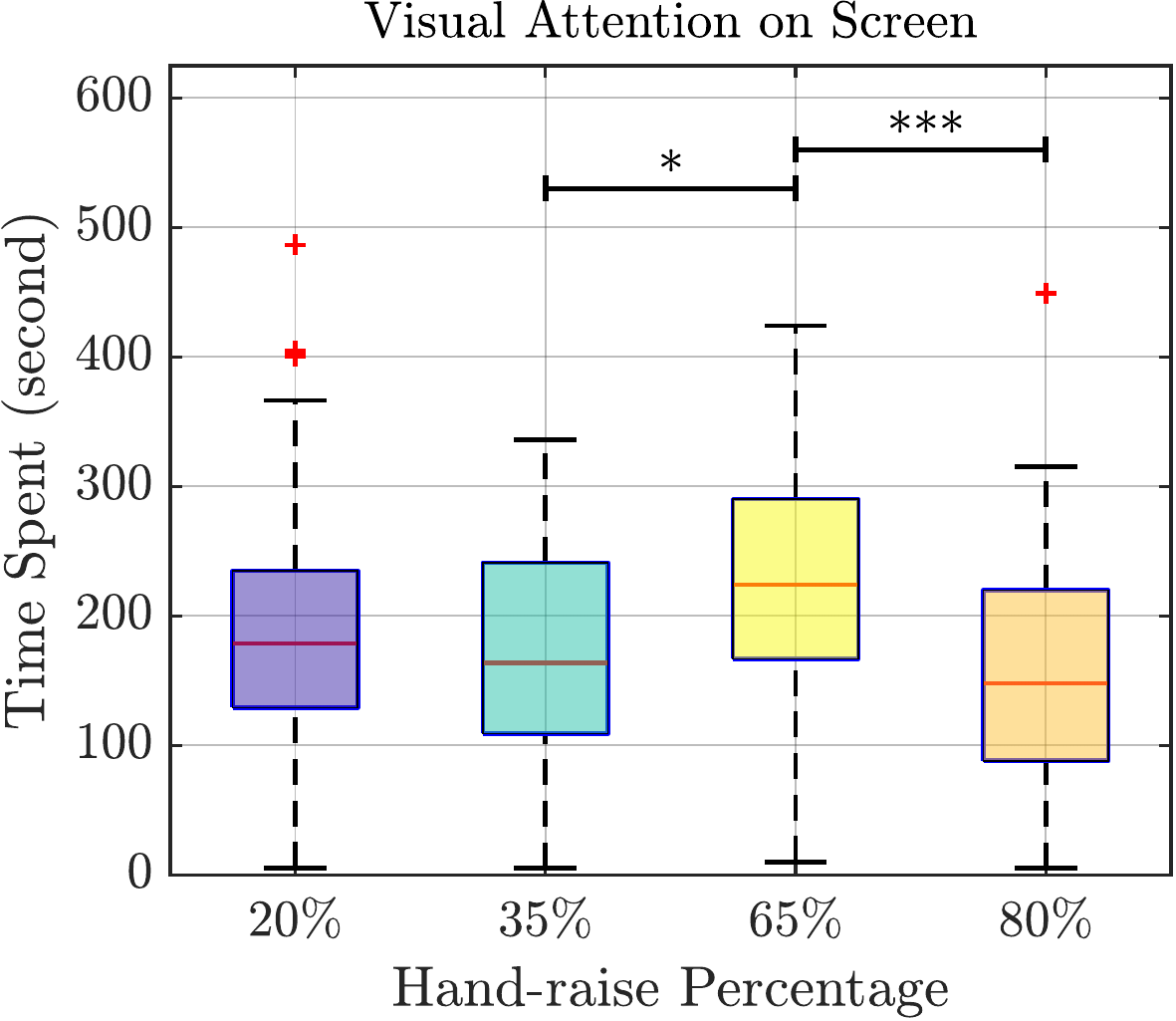}}}
  \caption{Attention towards screen for different classroom manipulation configurations. $*$, $***$, and $****$ correspond to the significance levels of $p < .05$, $p < .001$, and $p < .0001$, respectively.}
  \label{fig:attention_screen_VR21}%
\end{figure*}

We did not find significant effects on screen attention between cartoon- and realistic-styled avatar conditions ($F(1,264) = 1.9$, $p = .17$, $\eta^2 < .01$); however, attention time in realistic style ($M = 193.35$ sec, $SD = 92.30$ sec) was slightly longer than cartoon style ($M = 173.95$ sec, $SD = 96.11$ sec).

In addition, the total attention time on the screen is found to differ significantly between different hand-raising conditions with ($F(3,264) = 5.74$, $p < .001$, $\eta^2 = .06$). In particular, attention time on screen is longer in the $65\%$ hand-raising condition ($M = 222.03$ sec, $SD = 94.90$ sec) than in the $80\%$ condition ($M = 156.06$ sec, $SD = 88.25$ sec) with ($F(3,264) = 5.74$, $p < .001$, $\eta^2 = .06$). In addition, attention time in the $65\%$ condition is also significantly longer than in the $35\%$ hand-raising condition ($M = 174.87$ sec, $SD = 81.28$ sec) with ($F(3,264) = 5.74$, $p = .025$). The overall trend of attention on the lecture screen is similar to virtual instructor with the intermediate conditions being higher than the extreme conditions. There were no interaction effects found for attention time on screen.

\subsection{Discussion}
We discuss experimental results particularly for social interaction and dynamics in \acs{VR} classrooms, usability of eye tracking data, and the advantages of such classrooms along with their limitations.

\subsubsection{Social Dynamics in VR Classroom}
We discuss our findings about social dynamics in the \acs{VR} classroom in three parts, particularly based on \textbf{H1}, \textbf{H2}, and \textbf{H3} which are related to different sitting positions, different avatar visualization styles, and different hand-raise behaviors of peer-learners, respectively.

In our analyses, we found that the participants seated in the front of the classroom attended less on the peer-learners than the participants in the back, which was expected because they had fewer peers in their \acs{FOV}, unless they turn back of the classroom. Assuming that during the course of the lecture, participants are supposed to listen and pay attention to the topics told by the instructor, the visual attention we observed is normal. Briefly, this is an indication that participants focus on the lecture content or instructor instead of visually interacting with their peers when seated in the front. Further, as a supporting evidence to aforementioned result, front seated participants had spent significantly more time visually attending the instructor and the screen than the participants seated in the back. We assume that these results are due to being closer to them and having fewer occluding objects in the frontal participants' \acs{FOV}. These findings confirm our \textbf{H1}. Additionally, the results from the interaction effects support this hypothesis. The differences in visual attention on their virtual peer-learners for the avatar visualization style and hand-raising depend on the sitting position. Participants located in the back of the classroom have more peer-learners in their line of sight and therefore recognize the behaviour of the virtual peer-learners more, than participants seated in the front.  

Our results indicate that students visually attended for longer on the peer-learners when avatars in the classroom were presented in cartoon styles. Considering the number of peer-learners in the environment and the ages of our participants being between $10$-$13$, we argue that participants may have felt like engaging more with their peer-learners due to the emotional reasons as cartoon-styled peers are more appropriate to their ages. Realistic-styled peer-learners may be too ordinary for student engagement with peers in our setup, which led to less amount of attention. On the contrary, participants visually spent more time on the instructor when realistic-styled avatars were used. We conceive that if the avatar styles are ordinary, then the visual attention shifts to the instructor instead of interacting with the peer-learners. Lastly, as we did not find any statistical difference in attention time on the screen between different avatar visualization styles, we conclude that visual attention on the screen is not affected by such avatar visualization styles. Realism that is provided by the avatar styles may introduce additional computational complexity as such visualizations can be computationally expensive or can require additional effort to implement in advance. If the interaction with peer-learners is the main focus of the lecture, then practitioners can opt for cartoon-styled avatars. This also decreases the effort of generating the avatars. Overall, these findings confirm our \textbf{H2}.

In the analysis on different hand-raising behaviors of the peer-learners, we found mixed effects. In the attention time towards peer-learners, we found a clear evidence that attention time in the extreme hand-raising conditions, namely when $80\%$ or $20\%$ of the virtual peer-learners raise their hands after the questions were asked by the virtual instructor is longer than in the intermediate conditions ($35\%$ and $65\%$). The extreme conditions may represent either more or less capable groups of peer-learners in the learning environment and participants may have a higher self-concept when surrounded by a less capable group and the other way around, which is related to the Big-fish-little-pond effect~\cite{Marsh1984DeterminantsOS}. Having reasonably higher attention on peer-learners on these conditions also indicates that \acs{VR} can present an opportunity to create digital environments to further study students' self-concept. On the other hand, intermediate hand-raising conditions may help students to focus more on learning related objects in the classroom instead of peer-learners such as lecture content or instructor as experimentally indicated. However, we expected an approximately linear increase in terms of attention time towards higher hand-raising conditions in the attention time on peer-learners. While we obtained an expected result between the $65\%$ and $80\%$ hand-raising conditions, the results regarding the $20\%$ hand-raising condition do not support our hypothesis \textbf{H3}. This might be due to a moment of surprise when only a handful of peer-learners raises their hands indicating that few number of peer-learners know the answers of the questions. Furthermore, we found that attention time on the instructor tended to be longer in the intermediate levels of hand-raising than in the extreme conditions. Statistically significant results are only found for the difference between the $65\%$ and $80\%$ condition. While a decreasing linear trend towards the higher hand-raising percentages exists between the $65\%$ and $80\%$ for attention on the instructor, the overall trend is against our hypothesis, even though they are aligned with the attention time on peer-learners. Lastly, the experimental results on attention time on the screen is similar as compared to the attention time on the instructor. However, the $35\%$ hand-raising condition drew significantly less attention than the $65\%$ condition, which does not support our hypothesis. Overall, while some of our expectations are verified, \textbf{H3} is not confirmed. Still, the resulting behaviors should be further investigated with regard to effects on students' self-concepts during \acs{VR} learning and considered when creating a classroom students are habituated to. 

In summary, the three different manipulations that we studied have important effects on students' visual behavior in immersive \acs{VR} classrooms in terms of social dynamics. For instance, in practice, students' self-concept can be affected by consistent hand-raising behaviors of virtual avatars over the time. While this may be less problematic in real classrooms as peer students may have different capabilities in different themes, it should carefully considered in the virtual setting, because we could present always the same behavior of the peer-learners. An adaptive strategy for hand-raising behaviors of the virtual peer-learners may be considered in practice. In addition, seating the students in the front along with realistic-styled avatars may help to increase visual attention on the lecture content. However, if a more interactive classroom environment is focused on visual interaction, practitioners can either seat students in locations where they can see their peer-learners clearly or design \acs{VR} classrooms differently in terms of seating plans.

\subsubsection{Usability of Eye Tracking Data}
As eye tracking data is considered a noisy data source, we discuss our insights into the usability of this data, for particularly the immersive \acs{VR} classroom setups. As aforementioned, we defined the visual attention on the different objects by using an attention threshold, which was $200$ ms. In the end, in almost all conditions, the total amount of time that was spent on only the three types of objects was in the vicinity of half of the complete experiment duration despite having a relatively higher attention threshold value compared to fixation detection algorithms in the eye tracking literature. Such amount of total attention time on these three objects empirically validates our assumption of independence between them as well. We removed a significant number of samples from eye movement data due to sensory issues (e.g., lower eye tracking ratio) in order to obtain high-quality data and accurate attention mapping on the objects in the virtual classroom. While this may not be necessary for larger objects such as virtual screen in the classroom, it might cause mapping the attention wrongly for the smaller objects such as virtual avatars if the data quality is low. Considering that the participants were children in our experiments and they did not have experience with virtual reality and eye tracking, number of data removals due to such issues would be more than the experiments that are carried out with adults. In addition, unlike pre- or post-tests, eye tracking allows researchers to analyze time-dependent and temporal visual behavior changes, which can help assess students' states during virtual lectures and adapt to the environment accordingly. Therefore, despite the drawbacks, we suggest using eye movement data in such classrooms as long as an accurate calibration is applied in advance. A further iteration could take relationship of eye movement-based visual attention into consideration or analyze perceived relevance of lecture content along with eye-gaze behaviors such as in~\cite{10.3389/feduc.2020.572546} and~\cite{10.1145/3379155.3391329}, respectively. 

\subsubsection{Advantages and Limitations}
One of the advantages of immersive \acs{VR} classroom setups is the opportunity of simulating different classroom manipulations in remote settings, which are difficult to do in real world, and evaluate students' behaviors and learning under such manipulations. Another advantage of such setups is the possibility of preserving the privacy of students since the videos that include faces are not recorded in such settings. In real world classrooms, it is troublesome to record and store videos of the class while lecturing, even though there are some efforts supporting the automated anonymization~\cite{sumer2020automated} of such data. In contrast, data collected from virtual classrooms can be pseudo-anonymized. However, one should be aware of the amount of personal information that can be extracted from eye movement data and how to manipulate it~\cite{bozkir2020differential,fuhl2020reinforcement,bozkir_ppge}. Furthermore, one should take the relationship between iris texture and biometrics into account and how to preserve privacy in case eye videos are recorded and stored~\cite{10.1145/3379156.3391375}. In addition, we observed during experiments that some of the students intended to raise their hands when seeing the hand-raising behaviors of the virtual peer-learners. While we did not record hand tracking data in our study, it is possible to accurately assess the intentions of students towards questions asked by the virtual instructor by using a hand tracker device on the \acs{HMD}, which is another advantage of \acs{VR} setups compared to real classrooms. Although, hand-raising is a good indicator of children's participation during a lecture, we do not know if students interpret this behaviour of their virtual peers as a sign of competence, engagement, or motivation.

Despite the advantages, there are other technical limitations regarding the use of \acs{VR} classrooms. Long periods of exposure to \acs{VR} lectures can lead to immense levels of cybersickness. In addition, a vast amount of \acs{HMD} movement on the head may cause a drift in eye tracker calibration, leading to incorrect sensor readings. This can affect interaction experience if gaze-aware features are included in virtual environments. These should be taken into consideration when designing a virtual classroom and lecture. Particularly, the duration of the lecture should be chosen carefully to minimize these effects.

\subsection{Conclusion}
To understand the visual attention in \acs{VR} classrooms in different manipulations, we analyzed object-of-interest information based on eye-gaze. We found that participants seated in the front attended more time to the virtual instructor and the screen displaying lecture content. In addition, participants focused on the cartoon-styled peer-learners more than realistic-styled ones, whereas in the realistic-styled avatar manipulation the virtual instructor drew more visual attention. The extreme conditions of hand-raising behaviors drew more attention towards virtual peer-learners, whereas in the intermediate conditions visual attention was focused more on the instructor and screen. These findings are based on the eye movements of the participants and correspond to the social dynamics of \acs{VR} classrooms such as students' self-concept or peer-learner interaction; however, such manipulations may also affect learning outcomes. While our results provide primitive but fundamental cues about how to design immersive \acs{VR} classrooms by taking students' visual behaviors into account for different goals in digital teaching, effects of such manipulations on the learning outcome should be further investigated.

As future work, we plan to specifically investigate the relationship between different manipulations with temporal gaze dynamics as an immediate response to asked questions and related students' performances.

\subsection*{Acknowledgments}
This research was partly supported by a grant to Richard G{\"o}llner funded by the Ministry of Science, Research and the Arts of the state of Baden-W{\"u}rttemberg and the University of T{\"u}bingen as part of the Promotion Program of Junior Researchers. Lisa Hasenbein and Philipp Stark are doctoral candidates and supported by the LEAD Graduate School \& Research Network, which is funded by the Ministry of Science, Research and the Arts of the state of Baden-W{\"u}rttemberg within the framework of the sustainability funding for the projects of the Excellence Initiative II. Authors thank Stephan Soller, Sandra Hahn, and Sophie Fink from the Hochschule der Medien Stuttgart for their work and support related to the immersive virtual reality classroom used in this study.

\newpage

\section[Assessment of Driver Attention during a Safety Critical Situation in VR to Generate VR-based Training]{Assessment of Driver Attention during a Safety Critical Situation in VR to Generate VR-based Training}
\label{appendix:A3}

\subsection{Abstract}
Crashes involving pedestrians on urban roads can be fatal. In order to prevent such crashes and provide safer driving experience, adaptive pedestrian warning cues can help to detect risky pedestrians. However, it is difficult to test such systems in the wild, and train drivers using these systems in safety critical situations. This work investigates whether low-cost virtual reality (VR) setups, along with gaze-aware warning cues, could be used for driver training by analyzing driver attention during an unexpected pedestrian crossing on an urban road. Our analyses show significant differences in distances to crossing pedestrians, pupil diameters, and driver accelerator inputs when the warning cues were provided. Overall, there is a strong indication that \acs{VR} and Head-Mounted-Displays (HMDs) could be used for generating attention increasing driver training packages for safety critical situations.

\subsection{Introduction}
Having safe driving experiences and decreasing the number of crashes are two of the most important issues when it comes to driving safety. Every year, many fatal crashes occur on roads all over the world. According to the Road Safety Annual Report in International Transport Forum 2018, most of the fatal crashes occurred on rural roads; however, the number of fatal crashes in urban roads has been increasing in more than half of the countries since 2000 \cite{IRTAD2018}.

Apart from road or weather conditions, distracted driving can cause fatal crashes. While a total prevention is almost impossible, many crashes can be prevented by training drivers using driver assistant systems. With recent developments in the field of augmented reality (AR) and head-up display (HUD) technology, new means have become available to overlay different warnings to the driver, such as pedestrian warnings or road signs. In fact, many modern cars already employ this technology to a certain degree. The majority of studies that concentrated on driver training and the interaction between these technologies and drivers in safety critical situations used driving simulators. With the recent developments in \acs{VR} and \acs{HMD}s, it is possible to apply these scenarios and trainings in \acs{VR} with lower cost. However, it is an open question whether \acs{VR} and \acs{HMD}s can be used in studying driver training and interaction for safety critical situations. 

In order to assess whether \acs{VR}, \acs{HMD}s, and gaze-aware cues can be useful and driver attention can be increased properly in this context, we focused on an unexpected pedestrian crossing behavior at non-designated crosswalks on urban roads when the Time-to-Collision (TTC) between the driving vehicle and crossing pedestrian is very short ($\approx 1.8$-$5$ seconds). \cite{DBLP:journals/corr/RasouliKT17} mentioned that in this range of \acs{TTC}, there is a high likelihood that joint attention between crossing pedestrian and driver happens. However, in case it does not happen, due to distracted pedestrian or driver, it is more likely that a crash will happen. In our experiments, control group did not receive any critical pedestrian warning cues, whereas the experimental group had the gaze-aware critical pedestrian warning cues. By analyzing closest distances between driving vehicles and crossing pedestrians, pupil diameter changes of drivers between baseline and risky driving timeframes, and driver performance measurements, we found that there is a strong indication that gaze-aware visual warnings for critical pedestrians help increasing the driver attention earlier in \acs{VR}. Therefore, low-cost \acs{VR} setups along with realistic and gaze-aware warnings can be introduced to train drivers for safety critical scenarios. Major contributions of our work are as follows: \textbf{(a)} Demonstrating a very critical scenario in terms of collision risk between driver and pedestrian with and without risky pedestrian warning cues in \acs{VR} and \textbf{(b)} Evaluation of gaze-aware critical pedestrian warning cues in \acs{VR} whether they increase driver attention earlier so that attention increasing \acs{VR}-based training packages can be proposed and further evaluated. Since the dedicated driving scenario is highly dynamic and time-critical, the outcome of the current study can be taken as a basis for any study that includes time-dependent and safety-critical scenarios in \acs{VR}.

\subsection{Related Work}
Driving simulation studies have been conducted in various domains. Two of the most common issues addressed were safety and driver assistance. \cite{Charissis2010} introduced a novel interface for \acs{HUD} over head-down display (HDD). \acs{HUD}s and \acs{AR} cues have been used for various purposes. \cite{doi:10.1016/j.aap.2017.01.019} discussed that specificity of visual warnings provided advantages in gaze, brake reaction times, passing speeds, and collision rates. \cite{Tran:2013:LDA:2516540.2516581} showed the benefits of \acs{HUD}s while turning left, whereas \cite{doi:10.1016/j.aap.2014.05.020} presented positive effects of \acs{AR} cues in terms of time-to-contact and gap response variation to assist elderly drivers during left-turns. In addition, \cite{doi:10.13140/2.1.3582.0801} showed the navigational \acs{AR} aid for recognizing turn locations earlier via 3D volumetric \acs{HUD}. \cite{doi:10.1177/0018720811430502} discussed that adaptive support in \acs{HUD} for lane keeping helped drivers drive more centrally and with less lateral variation. The effect of in-car \acs{AR} system for reducing collisions caused by other vehicles' movements was presented by \cite{doi:10.1109/ISMAR.2013.6671764}. Additionally, increase in situational awareness using \acs{AR} in automated driving for take-over scenarios was studied by \cite{doi:10.1177/1541931214581351} and \cite{doi:10.1109/ITSC.2016.7795767}, whereas classification of drivers' take-over readiness was studied by \cite{8082802}.

While numerous studies can be counted in the context of driver assistance, the studies include pedestrian safety, hazard anticipation, and driver training are more relevant to our work. \cite{doi:10.1016/j.trf.2012.08.007} showed that \acs{AR} cueing increased the response rate for pedestrian and warning sign detection in directing driving attention to roadside hazards. The study of \cite{Pomarjanschi:2012:GGR:2070719.2070721} in a driving simulator with a maximum speed of about $30$km/h showed that gaze guidance reduced number of pedestrian collisions. \cite{7795724} studied three driver awareness levels of a pedestrian in a driving simulator: Perception, vigilance, and anticipation. They showed that \acs{AR} cues were capable of enhancing the driver awareness in all levels. The outdoor study conducted by \cite{8302393} showed that \acs{AR} pedestrian warnings provided positive results on measures such as braking, distances to pedestrians, and gaze-on pedestrian travel distances. The study of \cite{Pradhan_et_al} on eye movements showed that the scanning patterns of novice drivers reflected their failure to recognize potential risks. Driving simulator studies have been used in driver training and \acs{VR} as well. \cite{doi:10.1518/hfes.45.2.218.27241} found out that drivers who were trained in a simulator improved their driving skills in turning into correct lane and proper signal use. Furthermore, \cite{Fisher_et_al} evaluated hazard anticipation and found that trained drivers recognized the risks more often. \cite{8448290} showed the effect of improvement of bad driving habits via synthesizing personalized training programs in \acs{VR}. \cite{Mangalore_et_al} assessed drivers' hazard anticipation across \acs{VR} and driving simulators to evaluate the usage of \acs{VR} headsets and justified that \acs{VR} headsets could be used for measuring driving performance. \cite{YouOrMe} studied personality traits on sacrifice decisions including pedestrians during \acs{VR}-based driving. While the studies which include driving simulators and hazardous situations showed great potential for driver training, it is an open question whether visual cues for critical situations in \acs{VR} can increase driver attention properly, so that \acs{VR}-based training packages can be proposed and synthesized for safety critical situations.

\subsection{Experiment}
We focused on driver behavior in a very critical scenario when pedestrians tried to cross the road with \acs{TTC} was between $\approx 1.8$-$5$ seconds in \acs{VR}. In this range of \acs{TTC}, there is a high likelihood that pedestrian or joint attention occurs \cite{DBLP:journals/corr/RasouliKT17}. However, if it does not occur, the outcome can be fatal. Our experiment included a control group that did not receive any cues, and an experimental group that received gaze-aware critical pedestrian cues. Our major hypothesis is that if the gaze-aware warning cues can successfully increase the driver attention earlier in the safety critical situation in \acs{VR}, similar low-cost \acs{VR} setups along with adaptive warnings could be proposed for driver training for these situations.

\subsubsection{Participants}
$16$ volunteer participants ($4$ female, $12$ male) whose ages range from $25$ to $50$ ($M\approx 31$) and driving experiences range from $5$ to $30$ years ($M=12$) participated in the experiment. Participants were separated into two groups. A control group, receiving no critical pedestrian warning cues, and an experimental group, receiving the warning cues.

\subsubsection{Apparatus}
HTC-Vive along with Pupil-Labs Binocular Add-on \cite{Kassner:2014:POS:2638728.2641695}, which has binocular $120$hz eye tracking cameras and clip-on rings, Logitech G27 Steering Wheel and Pedals, and Phillips headphones were used to create driving setup. Eye tracking was measured using the open-source hmd-eyes of Pupil-Labs with Pupil Service version 1.7. Virtual city was created using Unity3D game engine. For the environment, vehicles, and pedestrians, we purchased and used models from Urban City Pack, City Park Exterior Props, Traffic Sign Set, Modern People, Traffic Cars, Realistic Car HD 02, Realistic Car Controller, Simple Waypoint System, and Playmaker asset packages. We designed the main roads long and straight so that the drivers would have opportunity to speed up as they want and drive naturally. Example scenes from our virtual environment are shown in Figure~\ref{fig:exampleScenes_SAP19}.

\begin{figure}[ht]
  \centering
   \subfigure[Cockpit of driving car.]{{\includegraphics[height = 3.75cm]{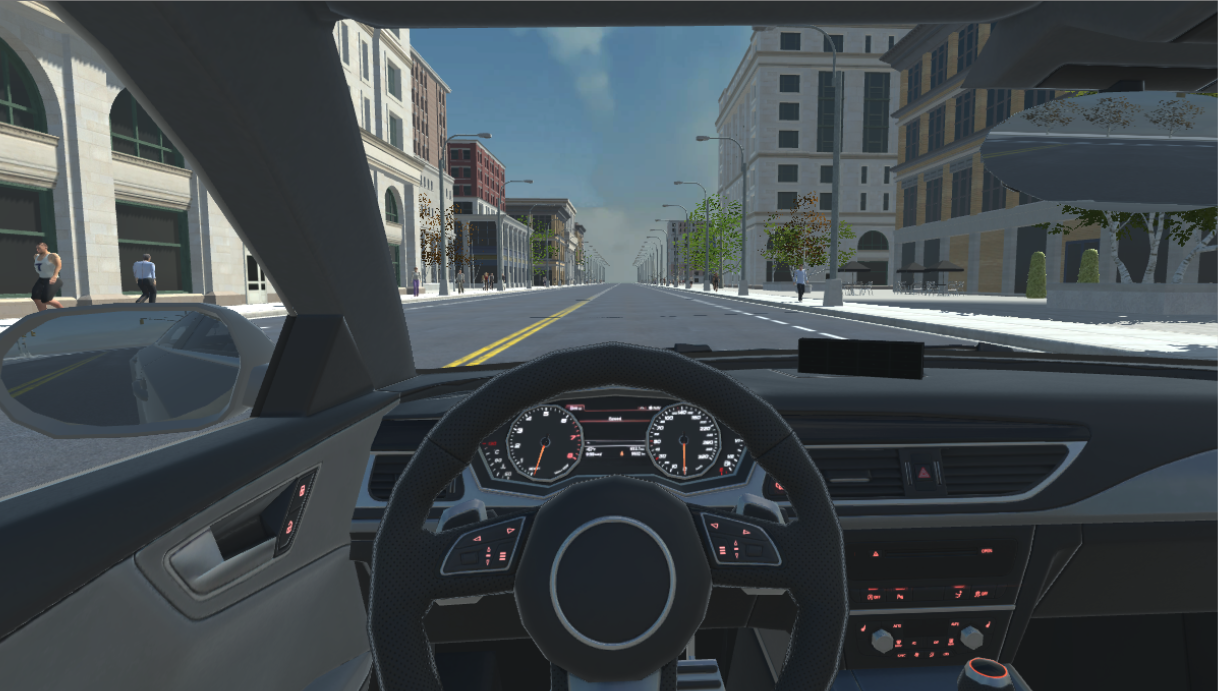}}}
   \qquad
   \subfigure[Main road.]{{\includegraphics[height = 3.75cm]{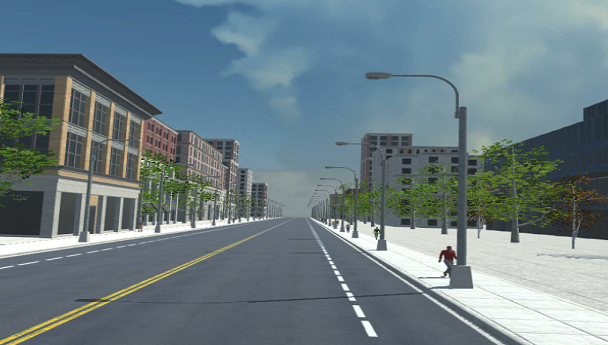} }}%
  \caption{Example scenes from VR environment.}
  \label{fig:exampleScenes_SAP19}%
\end{figure}

The dedicated setups were run on a PC equipped with an NVIDIA Titan X graphics card with 12GB memory, a 3.4GHz Intel i7-6700 processor, and 16GB of RAM. 

Since the visual warning cues for experimental group are very important in our setup, Figure~\ref{fig:PedestrianWithWithoutCues_SAP19} shows a pedestrian model with and without warning cues.

\begin{figure}[ht]
  \centering
   \subfigure[With cue.]{{\includegraphics[height = 3.75cm]{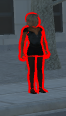}}}
   \qquad
   \qquad
   \qquad
   \subfigure[Without cue.]{{\includegraphics[height = 3.75cm]{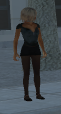} }}%
  \caption{Pedestrian with and without warning cues.}
  \label{fig:PedestrianWithWithoutCues_SAP19}%
\end{figure}

\subsubsection{Procedure}
In the beginning of the experiment, participants were informed about the purpose and scope of the experiment orally. They had the opportunity to stop and cancel the experiment anytime. At the end of the experiment, participants filled a small questionnaire about demographic and qualitative information. The experiment consisted of two phases. For both phases, participants were given written instructions before starting. In the first phase, participants acclimated the setup. This phase did not include any pedestrians or dynamic objects apart from the driver's car; no data were collected during this phase. Generally, this phase lasted in 5-10 minutes, although if participants had not felt comfortable, they could have continued driving. Once they felt comfortable with the setup, they continued to the second phase.

In the beginning of the second phase, 2D calibration with $16$ points using hmd-eyes of Pupil-Labs was performed. After calibration success, participants started the experiment. The starting location of the driving vehicle was in the beginning of the main road, where a critical pedestrian crossing happened. Since there was no intersection until the end of this road, all of the participants were required to drive until the end. At the end of the road, they could have turned left or right and continued driving, however our data analyses did not concentrate on the data acquired after the turn, since they could have encountered with different scenarios. The speed limit of the driven road was $90$km/h, and participants were supposed to realize this by traffic signs. The driving vehicle was also equipped with maximum speed warning. 

The critically crossing pedestrian scenario was as follows. At the beginning of each run, two occurrences of a critical pedestrian were generated along with other non-critical pedestrians on the side walks. The critically crossing pedestrian was determined at random, as active and proceeded to dangerously cross the street before the driving vehicle. Both of these occurrences had dedicated gaze-aware warning cues. Pedestrians were not located in the beginning of the road, so that the drivers had the opportunity to speed-up or slow-down until the crossing. Pedestrian warnings were activated for the experimental group when the distance between front of the driving vehicle and critical pedestrians became $\approx 77m$. The crossing pedestrian started crossing the road from the right side when the distance between vehicle and pedestrian was $d_{critical} \approx 45m$. We assumed that drivers would obey the speed limit ($90$km/h) and also drive faster than $30$km/h. This way, parameter of $d_{critical}$ helps to map expected \acs{TTC} to $\approx 1.8s \leq TTC \leq 5s$ interval. Ray-casting \cite{ROTH1982109} method was used to map gaze signal, which was obtained from Pupil-Labs software, from 2D canvas to 3D environment by the help of Unity3D colliders \cite{Unity3D_Colliders} that were attached to virtual objects. Once the drivers' gaze signal in 3D environment was closer than $5$ meters to the pedestrians for $\approx 0.85$ seconds, the cues were deactivated. Therefore, the cues became gaze-aware. Since the control group did not receive cues, the timeframes consisted of different milestones for each group. $t_w$ and $t_m$ correspond to start of the critical pedestrian warning and start of the pedestrian movement respectively. For the control group, baseline driving corresponds to $[t_{m}-\delta{t}, t_{m}]$, whereas for the experimental group, it is  $[t_{w}-\delta{t},t_{w}]$. $[t_{m}, t_{m}+\delta{t}]$ is the risky driving timeframe for both groups. Setting different values of $\delta{t}$ means changing the durations of the timeframes.

\subsubsection{Measurements} 
The metrics analyzed were the closest distances between the crossing pedestrians and the driving vehicles, driver performance measurements including inputs on accelerator and brake pedals, and pupil diameter changes between baseline and risky driving timeframes. Particularly, since the critically crossing pedestrian is only safety critical for the driving vehicle inside of the driven lane, we took the closest distance in this lane. Driver inputs on pedals are also indicators of perception and reflect the smoothness of the driving experience as well. Lastly, pupil enlargement corresponds to increase in cognitive load \cite{Appel:2018}. Pupil diameter values were fetched from Pupil-Labs software in pixel units. For smoothing and normalization, we applied Savitzky-Golay filter \cite{savitzky64}, and divisive baseline correction using baseline duration of $0.5$ seconds and median~\cite{Mathot2018}.

\subsubsection{Hypotheses}
Our hypotheses are based on the driver attention and actions. Since the experimental group was provided with the risky pedestrian cues, we expected that the closest distances between the crossing pedestrians and the driving vehicles for the experimental group would be more than the control group. In addition, when the visual cues were provided to the drivers, we expected that they would understand the criticality earlier, and their cognitive load would increase earlier. Pupil dilation is one of the indicators of the cognitive load increase, therefore we expected that pupil dilation of the experimental group would happen earlier. Furthermore, cues would affect driver inputs on accelerator and brake pedals, hence it was expected that experimental group drivers would take their foot off the accelerator earlier and perform smoother braking behavior than the control group drivers. In all, we expected that the experimental group would perform safer and smoother driving experience than the control group. 

\subsection{Results}
Analyses for the distances, driver performance measurements and pupil diameters during baseline driving and risky driving timeframes were calculated using MATLAB and are as follows.

\subsubsection{Closest Distance to Crossing Pedestrian}
We measured the closest distances between the crossing pedestrians and driving vehicles until pedestrians completed half of their trajectories, since during the second half, the pedestrians were not safety critical to the driving vehicles anymore. Figure~\ref{fig:Distance-Group-Relationship_SAP19} shows the results for this metric. We applied two sample T-test with alpha level of $0.05$ and found significant difference between two groups with $p \approx 0.00059$ (Cohen's $d \approx 2.21$). One of the participants in the control group hit the pedestrian, and the experiment was terminated at that moment. In addition, since the velocities of the vehicles were not fixed, the difference in distances could vary. However, the deceleration trend in the experimental group started from $t_{w}$, which is a strong indication that they acknowledged the critical situation earlier than the control group and behaved accordingly. Overall, it is clear that the experimental group participants drove safer than the control group participants.

\begin{figure}[ht]
  \centering
  \includegraphics[height = 4.5cm]{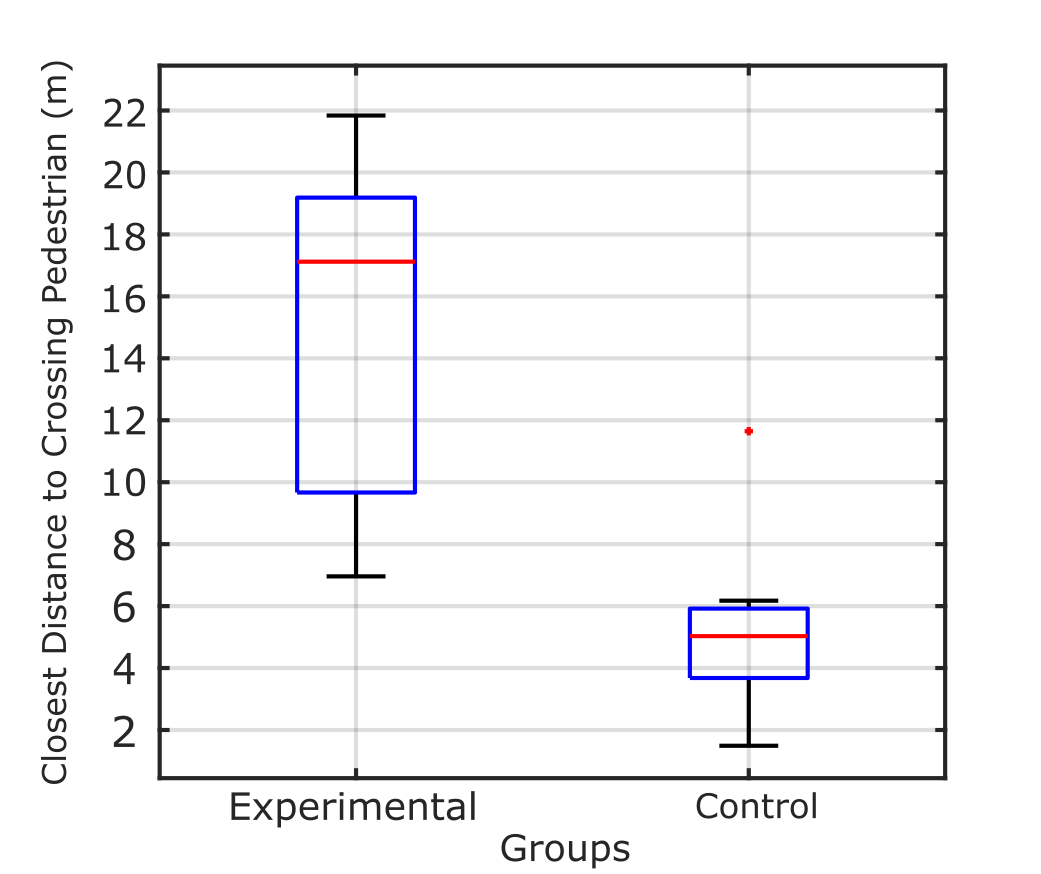}
  \caption{Closest distance to crossing pedestrian - Experiment group relationship.}
  \label{fig:Distance-Group-Relationship_SAP19}%
\end{figure}

\subsubsection{Driver Performance Measurements}
Driver inputs on accelerator and brake pedals are the two main indicators of safe and smooth driving. Therefore, we analyzed the normalized driver inputs on accelerator and brake during different durations of baseline and risky driving timeframes. First, we applied paired T-test with alpha level of $0.05$ between baseline and risky driving timeframes using normalized mean accelerator inputs. As expected, significant differences for experimental group even for very short durations (e.g. $\delta{t} \approx 50ms$, $p=0.0158$, Cohen's $d \approx 1.12$) were found. However, significant differences were found for the control group starting from $\delta{t} \approx 1.4s$ ($p = 0.0495$, Cohen's $d \approx 0.84$). Figure~\ref{fig:acceleratorInputs_SAP19} shows the dedicated analyses. Finding significant differences in shorter $\delta{t}$ values means that the drivers acknowledged the critical situation earlier. Therefore, it is a significant indicator that visual pedestrian cues helped drivers drive safely even during a very dangerous situation. Furthermore, we analyzed braking behaviors by analyzing whether participants performed full brake, since the braking happens in very short time. In total, five of the participants in the control group performed full brake, whereas none of the participants in the experimental group did this. This indicates that visual cues also helped to have smoother driving experience.

\begin{figure}[ht]
  \centering
   \subfigure[Accelerator input - Baseline \& Risky driving relationship for the experimental group.]{{\includegraphics[height=4.5cm]{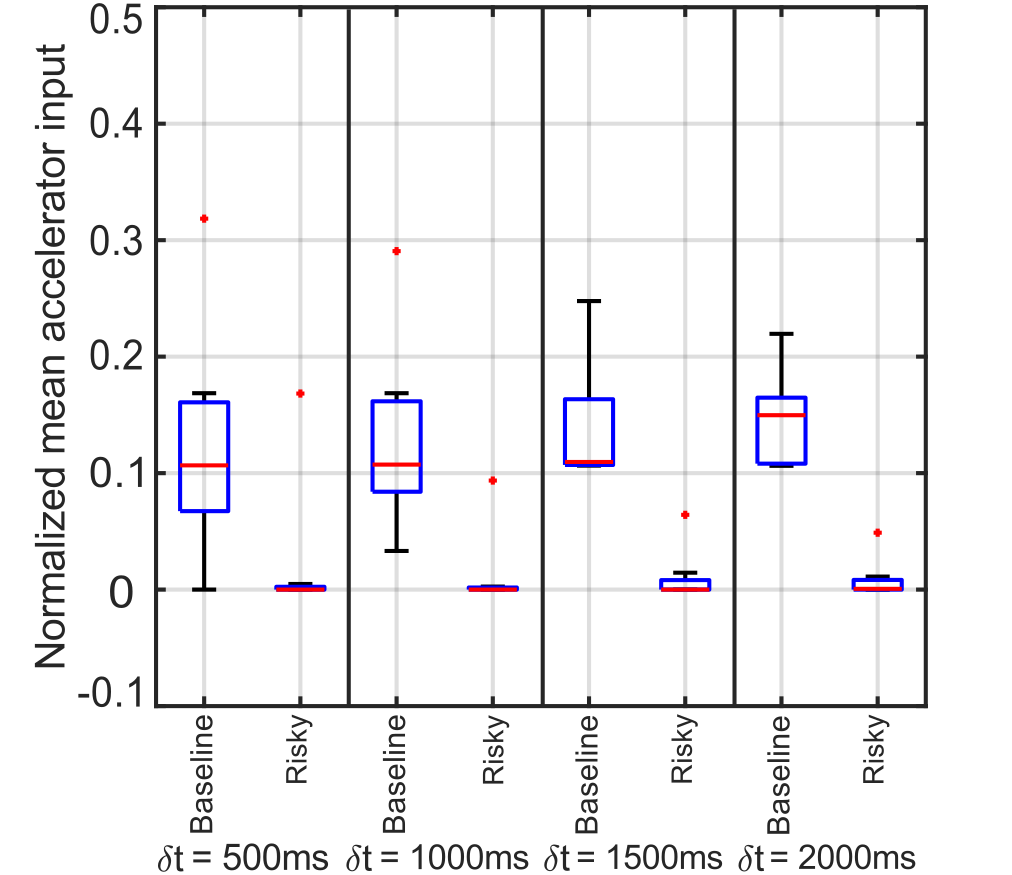}}}
   \qquad
   \qquad
   \subfigure[Accelerator input - Baseline \& Risky driving relationship for the control group.]{{\includegraphics[height=4.5cm]{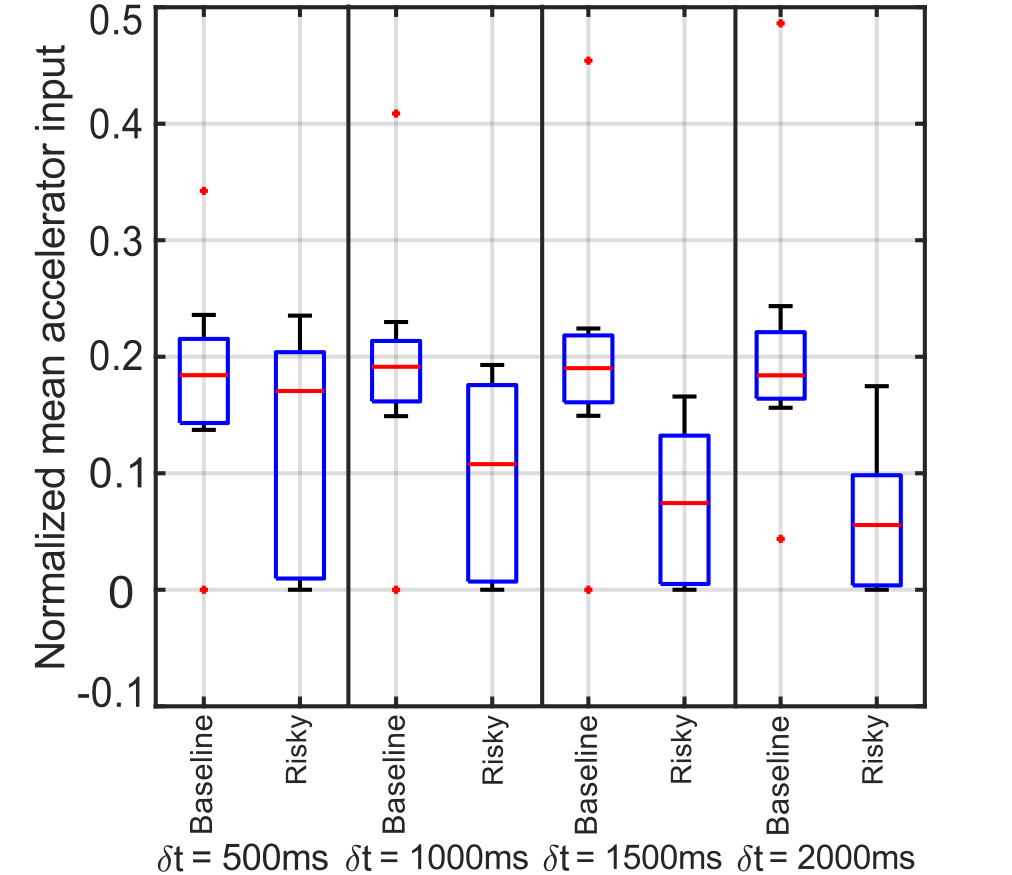} }}%
  \caption{Accelerator inputs - Driving timeframe relationship.}
  \label{fig:acceleratorInputs_SAP19}%
\end{figure}

\subsubsection{Pupil Diameter}
Since pupil dilation is one of the indicators of cognitive load increase, we analyzed normalized pupil diameters of the drivers in the same way as accelerator inputs between baseline and risky timeframes using paired T-test with alpha level of $0.05$. Since \acs{HMD}s and \acs{VR} offer controlled illumination, we expected that pupil dilation would happen due to the increase in cognitive load, and pupil diameters of the experimental group would increase earlier than the control group. Analyses showed that significant difference in pupil diameters between baseline and risky timeframes for the experimental group starts from $\delta{t} \approx 1.4s$ ($p=0.048$, Cohen's $d \approx 0.85$), whereas it starts from $\delta{t} \approx 2.4s$ ($p=0.0489$, Cohen's $d \approx 0.84$) for the control group. Figure~\ref{fig:pupilDiameter_SAP19} shows the results. Overall, there is a strong indication that cues for the critical pedestrians increased cognitive load of the experimental group earlier so that they behaved accordingly.

\begin{figure}[ht]
  \centering
   \subfigure[Pupil diameter - Baseline \& Risky driving relationship for the experimental group.]{{\includegraphics[height=4.5cm]{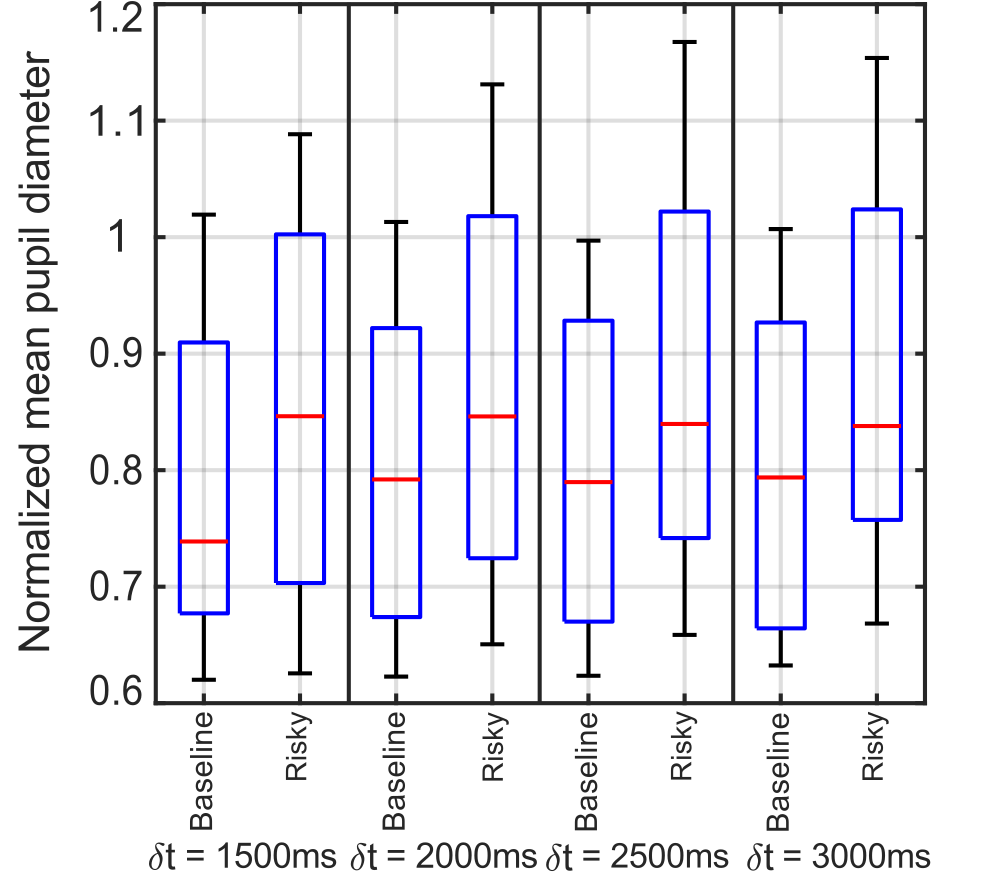}}}
   \qquad
   \qquad
   \subfigure[Pupil diameter - Baseline \& Risky driving relationship for the control group.]{{\includegraphics[height=4.5cm]{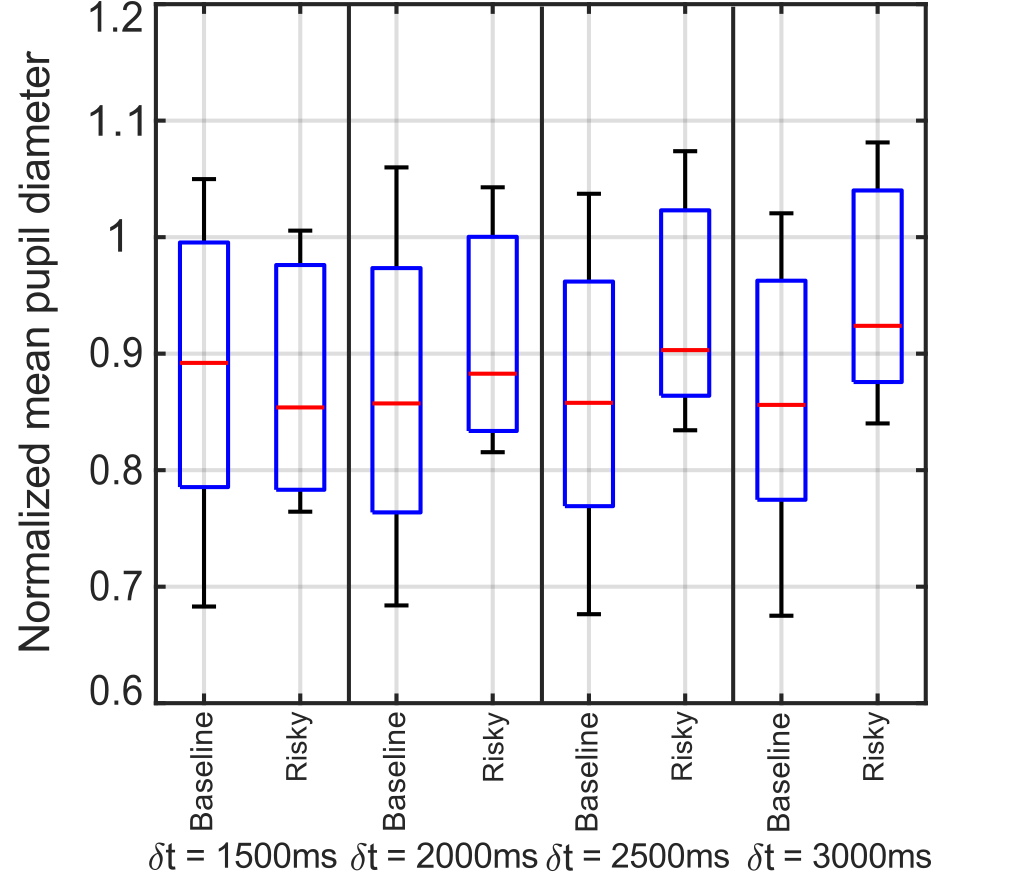} }}%
  \caption{Pupil diameter - Driving timeframe relationship.}
  \label{fig:pupilDiameter_SAP19}%
\end{figure}

\subsection{Conclusion}
We introduced a \acs{VR} driving simulation environment and a safety critical pedestrian crossing to study whether \acs{VR} setups and gaze-aware cues can increase driver attention in critical situations despite the prevalent disadvantages, such as narrow field-of-view, low resolution or weight of \acs{HMD}s, so that low-cost \acs{VR}-based training for safety critical situations can be proposed and further evaluated. To the best of our knowledge, this is the first work that assesses \acs{VR} setups using gaze-aware cues for safety critical situations in driving by analyzing eye tracking and performance metrics. We found significant differences in the distances to crossing pedestrians, accelerator inputs, and pupil diameters between baseline and risky timeframes. Results indicate that driver attention can be increased earlier with minimalistic gaze-aware cues properly in safety critical situations in \acs{VR}. Most of the previous work on driving simulation and training were done using physical driving simulators. However, \acs{VR} setups can decrease cost of implementation and time. Overall, we suggest that driver attention increasing training packages can be introduced in \acs{VR}. Since many modern cars have different warnings for safety critical situations, \acs{VR} could be used to assess these systems and train people to get acclimated with them as well. 

As future work, detailed eye-tracking analyses, a study to generate better attention grabbing cues, and a driver training study for critical situations to assess whether drivers improve their bad driving habits by \acs{VR}-based training can be done.

\newpage

\section[Person Independent, Privacy Preserving, and Real Time Assessment of Cognitive Load using Eye Tracking in a Virtual Reality Setup]{Person Independent, Privacy Preserving, and Real Time Assessment of Cognitive Load using Eye Tracking in a Virtual Reality Setup}
\label{appendix:A4}

\subsection{Abstract}
Eye tracking is handled as key enabling technology to \acs{VR} and \acs{AR} for multiple reasons, since it not only can help to massively reduce computational costs through gaze-based optimization of graphics and rendering, but also offers a unique opportunity to design gaze-based personalized interfaces and applications. Additionally, the analysis of eye tracking data allows to assess the cognitive load, intentions and actions of the user. In this work, we propose a person-independent, privacy-preserving and gaze-based cognitive load recognition scheme for drivers under critical situations based on previously collected driving data from a driving experiment in \acs{VR} including a safety critical situation. Based on carefully annotated ground-truth information, we used pupillary information and performance measures (inputs on accelerator, brake, and steering wheel) to train multiple classifiers with the aim of assessing the cognitive load of the driver. Our results show that incorporating eye tracking data into the \acs{VR} setup allows to predict the cognitive load of the user at a high accuracy above $80$\%. Beyond the specific setup, the proposed framework can be used in any adaptive and intelligent \acs{VR}/\acs{AR} application.

\subsection{Introduction}
Cognitive load is referred to as the amount of information processing activity that is imposed on working memory \cite{cavanaughInbook}. Cognitive load recognition is important and beneficial for many applications. It has been studied extensively in various domains, such as in education, psychology, or driving, since information on the cognitive load of an individual can be helpful to design user-adaptive interfaces. Various ways have therefore been proposed to assess the cognitive load of a subject, such as by means of N-back tasks (e.g., Appel et al. ~\cite{Appel:2018}), through the analysis of electroencephalography (EEG) signals (e.g., Zarjam et al. \cite{7074062}, Walter et al. \cite{10.3389/fnhum.2017.00286}), by means of eye movements studies or through assessment of facial expressions (e.g., Hussain et al. \cite{Hussain2014AutomaticCL}). Eye tracking offers a particularly non-invasive way of cognitive load assessment, especially through the measurement and analysis of the pupil diameter.

Meanwhile, eye tracking has also found its way into the driving domain, not only as a means to study driving behavior, but also as a powerful input modality for advanced driver assistance systems (e.g., K{\"u}bler et al. \cite{kubler2014stress}) or even as a means of driver observation on context of automated driving (e.g. Braunagel et al. ~\cite{braunagel2017online,8082802}). Modern cars are already capable of tasks such as lane following, traffic sign and light detection, automated parking, and collision warning. However, the full autonomous driving task is still too complex without human input and guidance. For this reason, current cars employ a variety of multi-modal warning systems for many different purposes to ensure driving safety and provide smooth driving experience. Augmented reality (AR) and head-up-display (HUD) technologies have been used as interfaces to such systems both in practice and driving simulation research. In the following, we will briefly review related work in this context.

Many driving simulation studies have been conducted in driving simulators or virtual reality (VR) environments in order to analyze driving behavior, safety, performance and training using \acs{HUD}s or virtual warnings. For example, \acs{HUD}s for blind spot detection and warning were discussed in a related work by Kim et al. \cite{Kim:2013:EHA:2516540.2516566}. Tran et al. \cite{Tran:2013:LDA:2516540.2516581} addressed the usage and benefits of \acs{HUD}s during left turns. Moreover, benefits and improvement of driving behavior for lane keeping using adaptive warnings were discussed by Dijksterhuis et al. \cite{doi:10.1177/0018720811430502}. The effect of improving bad driving habits using \acs{VR} in a user-study was discussed by Lang et al. \cite{8448290}. In the context of eye tracking and driving, there are several studies with various goals. For example, Konstantopoulos et al. \cite{KONSTANTOPOULOS2010827} studied eye movements during day, night, and rainy driving in a driving simulator. Braunagel et al. \cite{7313360} introduced a novel approach for driver activity recognition using head pose and eye tracking data. Furthermore, Braunagel et al. \cite{8082802} proposed a classification scheme to recognize driver take-over readiness using gaze, traffic, and a secondary task in conditional automated driving. Pomarjanschi et al. \cite{Pomarjanschi:2012:GGR:2070719.2070721} showed that gaze guidance reduced the number of pedestrian collisions in a driving simulator.

In the driving context, there are many studies that focus on cognitive load and driving. Engstr{\"o}m et al. \cite{doi:10.1177/0018720817690639} analyzed the effect of cognitive load on driving performance and found out that the effects of cognitive load on driving are task dependent. Yoshida et al. \cite{ClassifyCognitiveLoadWithML} proposed an approach to classify driver cognitive load to improve in-vehicle information service using real world driving data. Gabaude et al. \cite{CognitiveLoadMeasurementWhileDriving} conducted a study in a driving simulator to understand the relationship between mental effort and driving performance using cardiac activity, driving performance and subjective data measurements. Mizoguchi et al. \cite{10.1007/978-3-642-38812-5_12} proposed an approach to identify cognitive load of the driver using inductive logic programming with eye movement and driver input measurements in real driving situations. Fridman et al. \cite{Fridman:2018:CLE:3173574.3174226} proposed a scheme to estimate cognitive load in a 3-class problem in the wild for driving scenarios using convolutional neural networks. 

Driving simulation studies for safety critical situations using warnings and cognitive load recognition in driving exist in the literature. However, it is still an open question whether it is possible to recognize cognitive load of the driver in safety-critical situations and especially when the driver is confronted with visual gaze-aware warnings. In order to tackle this issue, we used the data collected using a \acs{VR} setup from our previous work \cite{bozkir_vr_attention_et} where drivers encountered a dangerously crossing pedestrian in an urban road. In order to keep the situation safety critical, Time-to-Collision (TTC) between driving vehicle and crossing pedestrian was kept $1.8sec < TTC < 5sec$. Rasouli et al. \cite{DBLP:journals/corr/RasouliKT17} discussed that in this range of \acs{TTC}, there is a high likelihood that pedestrian or joint attention between driver and pedestrian happens. However, if it does not happen, the outcome can be fatal. Our study was conducted with 16 participants. Half of them received gaze-aware pedestrian warning cues, whereas the other half did not receive any cue. 

In our proposed scheme, the cognitive load of the drivers are recognized using critical and non-critical time frames of driving for each participant. Since critical time frames are very short, we kept non-critical time frames also short in order to have a uniform distribution in the training data. We trained multiple classifiers and evaluated them leave-one-person-out fashion in order to obtain person-independent results. Furthermore, since the time frames that are used in training and testing are very short, they do not reflect the complete intention of driver during driving. Therefore, we obtained a privacy-preserving scheme. In addition to person-independence and privacy-preserving features, our system also works in real time, which brings the opportunity to implement the same system in real life.

In general, when the cognitive load of the driver is recognized in a safety critical situation, visual cues and support can be adapted accordingly in order to provide safer, smoother, and less stressful driving experience even in very risky situations. In this work, we focused on a proof-of-concept in the driving scenario due to its highly dynamic and uncertain nature. However, our results show that the same methodology can be applied to any adaptive and gaze-aware application, especially in \acs{VR}/\acs{AR}.

\subsection{Proposed Approach}

Since the proposed system depends on the driving data which were collected using a \acs{VR} setup, Section~\ref{sec:VRSetup_VRW19} describes first the \acs{VR} setup and the collected data from our previous work \cite{bozkir_vr_attention_et}. Then in Section \ref{sec:CognitiveStateRecognition_VRW19}, data processing, training, and testing procedures are discussed.

\subsubsection{VR Setup and Environment}
\label{sec:VRSetup_VRW19}
In our previous work~\cite{bozkir_vr_attention_et}, we conducted a user-study to evaluate safety during driving in VR.

The hardware setup was created using HTC-Vive, Logitech G27 Steering Wheel and Pedals, Phillips headphones and Pupil-Labs HTC-Vive Binocular Add-on. Figure~\ref{fig:HardwareSetup_VRW19} shows the dedicated setup.

\begin{figure}[ht]
 \centering
 \includegraphics[width=\linewidth]{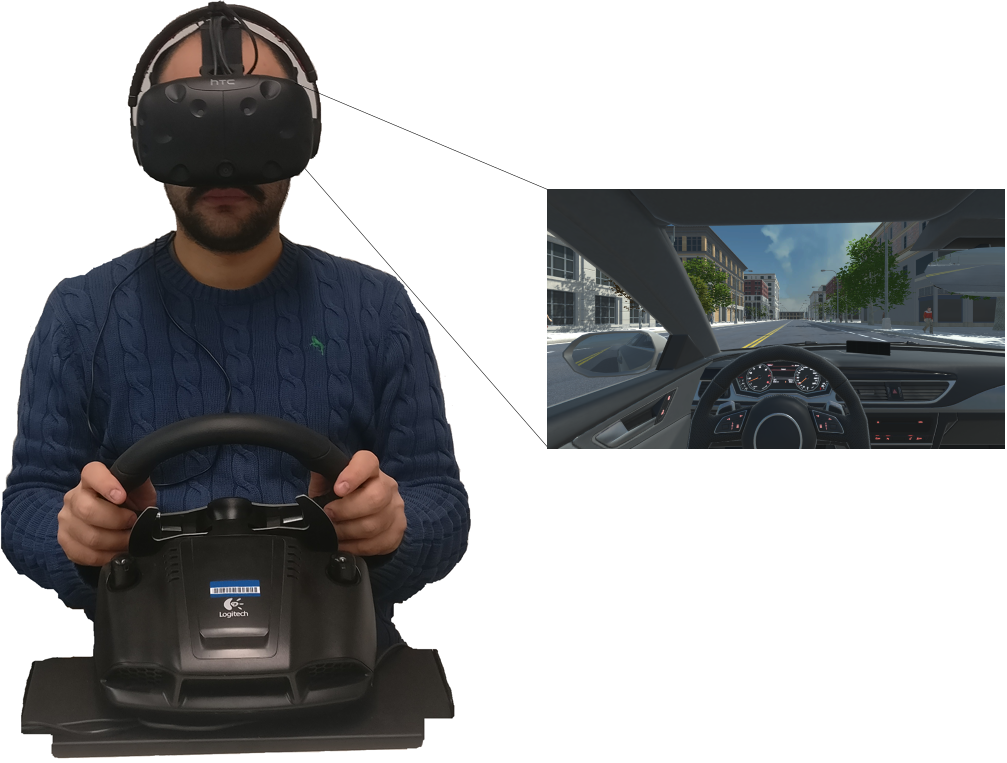}
 \caption{Experimental setup for VR.}
 \label{fig:HardwareSetup_VRW19}
\end{figure} 

The hardware setup was used in a virtual environment created using Unity3D. We used 3D models from Urban City Pack, City Park Exterior, and Traffic Sign Sets packages to design the virtual city. Since we had not only critically crossing pedestrians, but also other pedestrians, we used Modern People asset packages for pedestrian models. Vehicle models and helper scripts were obtained from Realistic Car HD02, Traffic Cars, and Realistic Car Controller asset packages. Lastly, Playmaker and Simple Waypoint System packages were used to make pedestrian and vehicle movements smoother. For the eye tracking measures, Pupil Service version 1.7 of open source hmd-eyes\footnote{https://github.com/pupil-labs/hmd-eyes} from Pupil-Labs was used. Examples scenes from \acs{VR} environment are shown in Figure~\ref{fig:ExampleScenes_VRW19}.

\begin{figure}[ht]
  \centering
   \subfigure[Cockpit of driving vehicle.]{{\includegraphics[width = 0.45\linewidth]{Image1_VRW.png}}}
   \qquad
   \subfigure[Intersection from driving vehicle.]{{\includegraphics[width = 0.45\linewidth]{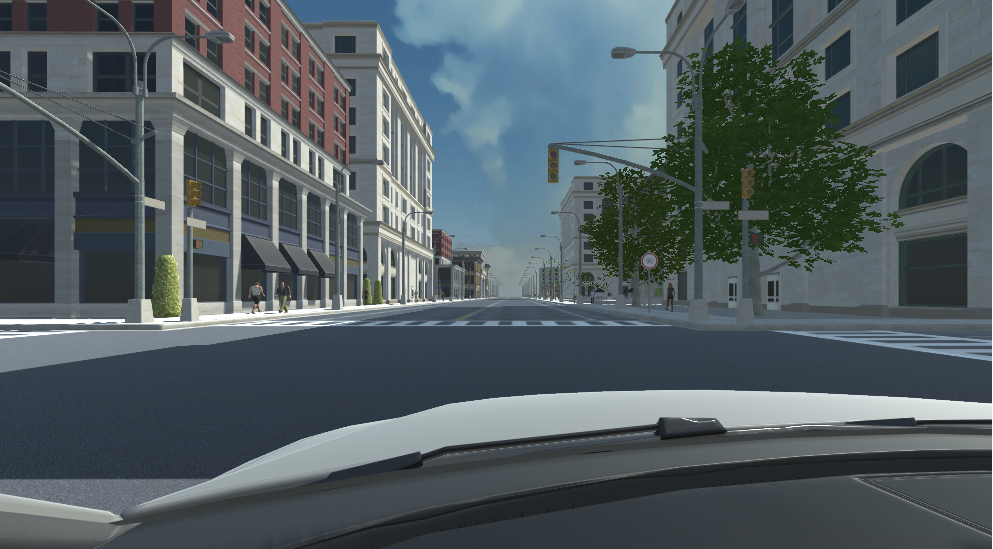} }}%
   \qquad
   \subfigure[Intersection.]{{\includegraphics[width = 0.45\linewidth]{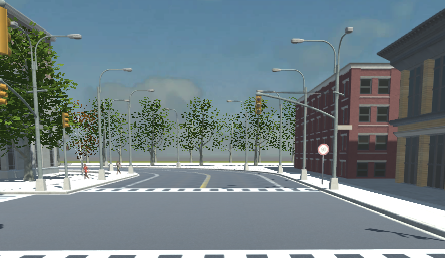} }}%
   \qquad
   \subfigure[Main road.]{{\includegraphics[width = 0.45\linewidth]{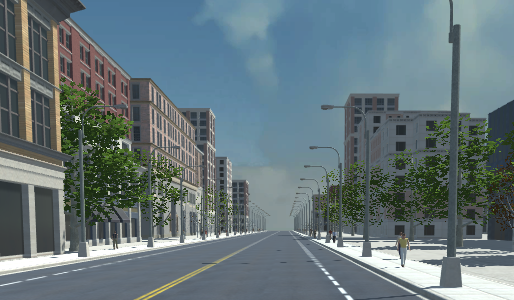} }}%
  \caption{Example scenes from VR environment.}
  \label{fig:ExampleScenes_VRW19}%
\end{figure}

The user-study consisted of acclimation and data collection phases. In the acclimation phase, no data was collected. In the data collection phase, participants encountered with a dangerously crossing pedestrian. Two critical pedestrians on the side walk of main road were generated. In the beginning of the experiment, one of them was marked as crossing pedestrian. This critical pedestrian started crossing the road when the distance between driving vehicle and pedestrian was ($d_{critical}\approx45m$). Every participant encountered with critically crossing pedestrian due to the start position of the vehicle in the data collection phase. They had the opportunity to speed up or slow down until the pedestrian crossing. Half of the participants started observing critical pedestrian warning cues around the pedestrian model with red color $\approx32$ meters in advance to pedestrian crossing. These parameters helped to keep \acs{TTC} as $1.8s<TTC<5s$, since the speed limit of main road was $90km/h$. Participants were supposed to realize the speed limit via speed signs on the road. Otherwise, the vehicle was equipped with maximum speed warning on a small in-car board. The pedestrian cues were made gaze-aware and were deactivated when gaze signal of the driver was closer than $5$ meters to pedestrian for $\approx0.85$ seconds. Gaze signal on 2D canvas was obtained from Pupil Service from Pupil-Labs and then mapped from 2D to 3D with the help of ray-casting and Unity colliders. The hyper-parameters were determined empirically. The measurements, which changed over the time, were recorded in real time and were available for offline analysis. Since the pupil diameter values are very important for recognizing cognitive load, we post-processed pupil diameter measurements to remove the noise and normalize the data. For smoothing and normalization, we applied Savitzky-Golay filter and divisive baseline correction using a baseline duration of $0.5$ seconds respectively.

Corresponding setup and experiments were run on a PC equipped with NVIDIA Titan X graphics card with 12GB memory, a 3.4GHz Intel i7-6700 processor and 16GB of RAM. 

\subsubsection{Cognitive Load Recognition}
\label{sec:CognitiveStateRecognition_VRW19}
The data we obtained from the experiment (mentioned in Section~\ref{sec:VRSetup_VRW19}) is not annotated with regard to the cognitive load levels. Therefore, we first annotated our data with two levels of cognitive load: Low and high. We set $t_{critical}$ for both with-and without-pedestrian cue scenarios. The purpose of $t_{critical}$ is to separate the time domain into low and high cognitive load levels. It is taken as $t_{warning}$ and $t_{movement}$ for with-and without-warning scenarios respectively. The reason of taking two different $t_{critical}$ values is that cognitive load of the drivers who receive critical pedestrian cues starts increasing from $t_{warning}$, whereas cognitive load of others who do not receive any cue increases after the start of pedestrian movement. 

In order to find the time frames to annotate exactly, we applied T-test using the pupil diameter data of each participant between $[t_{critical}-\delta{t}, t_{critical}]$ and $[t_{critical}, t_{critical}+\delta{t}]$. We used pupil diameter measurements due to the fact that pupil diameter is one of the main indicators of cognitive load. Once a significant difference in T-test was found with $p<0.05$, we assumed that we found a proper $\delta{t}$ value. In order to keep the distributions significantly different but rather close to each other, we did not accept distributions where $p<0.01$. Since cognitive load also depends on biological factors, which do not happen immediately, we shifted $t_{critical}$ by $+\theta{t}_{shift} = 0.8s$. In the end, we annotated each frame in the dedicated time frames with low or high cognitive load as it is shown in Table~\ref{tab:Annotations_VRW19}:

\begin{table}[ht]
  \begin{center}
    \caption{Cognitive load annotations for time-frames.}
    \label{tab:Annotations_VRW19}
    \begin{tabular}{c|c}
      \textbf{Time Frame} & \textbf{Cognitive Load} \\
      \hline
      $[t_{critical}+\theta{t}_{shift}-\delta{t}, t_{critical}+\theta{t}_{shift}]$ & Low \\
      $[t_{critical}+\theta{t}_{shift}, t_{critical}+\theta{t}_{shift}+\delta{t}]$ & High \\
    \end{tabular}
  \end{center}
\end{table}

In order to recognize cognitive load of the driver, we trained different classifiers including Support Vector Machines (SVM), decision trees, random forests, and k-Nearest Neighbors (k-NN) using each frame. For the feature set, we used pupil diameters, and driver inputs on accelerator and brake pedals and steering wheel. Min-max normalization was applied to input data. In order to make our approach person-independent, we evaluated the data of each driver against the trained model using rest of the drivers. For example, in order to evaluate the first participant, we trained classifiers with other 15 participants and then evaluated the first participant using the data and its labels. This approach assures that we obtain person-independent results in the end.

Offline analyses offer many insights from the collected data. However, real time working capability is as important as the accuracy of the system especially in \acs{VR}/\acs{AR} fields. With this motivation, we evaluated whether our proposed scheme is capable of working in real time.

\subsection{Results}
In the following, we report results of our automated cognitive load recognition and its real time working capabilities that was conducted using MATLAB on a PC which is equipped with NVIDIA GeForce GTX 1070 mobile graphics card with 8GB of RAM, a 2.2GHz Intel i7-8750 processor, and 32GB of RAM.

In our dataset, there are 1171 frames in total and from each frame, maximum four features were used in training and testing. In addition, there are $\approx 73$ frames (Mean) per participant (SD=$12.5$). We trained \acs{SVM}, decision trees, random forests, and \acs{k-NN} and tested according to the discussed setup in Section~\ref{sec:CognitiveStateRecognition_VRW19} and used different combinations of features along with pupil diameter. We observed that using steering wheel input of driver did not lead to more accurate recognition. Since participants did not need to change steering wheel angle too much during the encountered scenarios, it is acceptable. Taking into account that cognitive load does not change very dramatically in short amount of time and each participant was evaluated against the trained models using the rest of the participants, the cognitive load recognition results are reasonable. The highest accuracy of $80.7$\% was achieved by \acs{SVM}. Adding more training data and participants has a great potential to increase the accuracy of predictions. Accuracy, precision, recall, and F1-score results which were obtained using these classifiers and feature set of pupil diameter and driver inputs on accelerator and brake are shown in Table~\ref{tab:ResultsOfClassifiers_VRW19}. 

\begin{table}[ht]
  \caption{Results of cognitive load recognition.}
  \label{tab:ResultsOfClassifiers_VRW19}
	\centering%
  \begin{tabu}{%
	r%
	*{7}{c}%
	*{2}{r}%
	}
  \toprule
   \textbf{Method} & \rotatebox{90}{\textbf{Accuracy}} &   \rotatebox{90}{\textbf{Precision}} &   \rotatebox{90}{\textbf{Recall}} &   \rotatebox{90}{\textbf{F1-Score}}   \\
  \midrule
	Support Vector Machine & $0.8070$ & $0.7671$ & $0.8574$ & $0.8098$ \\
  	Decision Tree & $0.7344$ & $0.7332$ & $0.7005$ & $0.7165$  \\
  	Random Forest & $0.7436$ & $0.7372$ & $0.7230$ & $0.7299$ \\
  	\midrule
  	1-Nearest Neighbor & $0.6968$ & $0.6846$ & $0.6809$ & $0.6828$ \\
  	5-Nearest Neighbor & $0.7566$ & $0.7473$ & $0.7433$ & $0.7453$ \\
  	10-Nearest Neighbor & $0.7882$ & $0.7947$ & $0.7522$ & $0.7729$ \\
  \toprule
  \end{tabu}%
\end{table}

During the training of classifiers which are mentioned in Table~\ref{tab:ResultsOfClassifiers_VRW19}, we set some hyper-parameters. For \acs{SVM}, we used linear kernel function. For \acs{k-NN} approach, we evaluated 1-NN, 5-NN and 10-NN. The accuracy results increase by increasing the k value. For random forest classifier, we used five trees to train for classification purposes.

Since it is important to apply the proposed approach in real life scenarios, we evaluated whether cognitive load recognition can be done in real time. Table~\ref{tab:MeanTimeSpent_VRW19} shows the mean time spent for one prediction in each method.

\begin{table}[ht]
  \begin{center}
    \caption{Evaluation of mean prediction durations.}
    \label{tab:MeanTimeSpent_VRW19}
    \begin{tabular}{c|c}
      \textbf{Method} & \textbf{Mean Prediction Duration (ms)} \\
      \hline
      Support Vector Machine & $0.319$ \\
      Decision Tree & $0.305$ \\
      Random Forest & $5.42$ \\
      \hline
      1-Nearest Neighbor & $0.741$ \\
      5-Nearest Neighbor & $0.742$ \\
      10-Nearest Neighbor & $0.764$ \\
    \end{tabular}
  \end{center}
\end{table}

It is clear that all methods can be used for real time purposes. However, under this setup, it is reasonable to use \acs{SVM} due to its higher accuracy and low prediction duration. In addition, if the dataset size increases, the real time working capability of \acs{k-NN} is affected negatively. The same applies when the number of trees in random forests is increased.

\subsection{Conclusion and Discussion}
We proposed a scheme to recognize cognitive load of the drivers in safety critical situations using data collected during a driving study in \acs{VR}. The scheme is person-independent because it generalizes well cross-subject. With more training data, there is a high potential for this scheme to work in a generic way. If person-specific setup is requested, the same scheme can be applied by adjusting the training data. In this case, even a more accurate cognitive load recognition can be obtained.

Due to the fact that we concentrated on very short time frames, complete driving data of participants were not exposed and only small amount of frames was used in training and testing. Only, the pupil diameter measurements were baseline-corrected using the first $0.5$ seconds of driving. Therefore, it is a privacy-preserving scheme. Lastly, our scheme is capable of working real time. This outcome is very important and means that same scheme can be used in real driving studies and vehicles. It will enable more adaptive and intelligent feedbacks and inputs in driver warning systems; and eventually lead to safer and smoother driving experiences. We strongly suggest that similar schemes should be applied to real vehicles.

While this study is in driving domain, the outcome shows that our approach can be applied in similar adaptive user studies in \acs{VR} and \acs{AR} fields. The results indicate that there is a unique opportunity to design eye-tracking enabled interfaces and applications. Since we think that eye tracking has a great potential to transform \acs{VR} and \acs{AR} into another level, the outcome is valuable.

Despite the advantages and reasonable outcomes, there are some limitations as well. Firstly, since data were collected under \acs{VR} setup, there is a likelihood that drivers do not behave naturally in \acs{VR}. Virtual environment, weight of Head-Mounted-Display (HMD), or different dynamics of pedals or steering wheels can cause different behaviors than the real life. While we assume that participants became familiar with these in the acclimation phase, one should not ignore this possibility. Secondly, since the safety critical situations during driving happen in very short amount of time, it is difficult to collect big data in this context both using simulations or in real world. 

As future work, more data and features can be used. There is a high likelihood that the accuracy of cognitive load recognition increases with more data. The same scheme can be applied to real driving simulators along with safety critical scenarios. Therefore, the findings in \acs{VR} experiment can be compared with the future driving simulator experiments in terms of cognitive load recognition. Secondly, using raw eye videos along with other extracted features can be used to train deep models to estimate cognitive load. Furthermore, markov models or recurrent neural networks can be used to predict the cognitive load since they are suitable for time dependent data.

\chapter{Privacy Preserving Eye Tracking}
\label{appendix_B}

This chapter includes the following publications:
\vspace{1cm}
\begin{enumerate}
	\item\label{appendix_PLOSONEpaper_label} \textbf{Efe Bozkir*}, Onur Günlü*, Wolfgang Fuhl, Rafael F. Schaefer, and Enkelejda Kasneci. Differential privacy for eye tracking with temporal correlations. \emph{PLoS ONE}, 16(8):e0255979, 2021. doi: 10.1371/journal.pone.0255979.
		
	\item\label{appendix_ETRApaper_label} \textbf{Efe Bozkir*}, Ali Burak Ünal*, Mete Akgün, Enkelejda Kasneci, and Nico Pfeifer. Privacy preserving gaze estimation using synthetic images via a randomized encoding based framework. In~\emph{ACM Symposium on Eye Tracking Research and Applications (ETRA)}, New York, NY, USA, 2020. ACM. doi: 10.1145/3379156.3391364.
\end{enumerate}

\blfootnote{
\hspace{-14pt}{* indicates equal contribution.\\}
{Publications are included with minor templating modifications. Definitive versions are available via digital object identifiers at the relevant venues. Publication \ref{appendix_ETRApaper_label} is \textcopyright~2020 ACM, and included with relevant permission. Publication \ref{appendix_PLOSONEpaper_label} \textcopyright~2021 Bozkir et al. and is an open access article distributed under the terms of the \href{http://creativecommons.org/licenses/by/4.0/}{Creative Commons Attribution License}, which permits unrestricted use, distribution, and reproduction in any medium, provided the original author and source are credited.}
}

\newpage

\section{Differential Privacy for Eye Tracking with Temporal Correlations}
\label{appendix:B1}

\subsection{Abstract}
New generation head-mounted displays, such as \acs{VR} and \acs{AR} glasses, are coming into the market with already integrated eye tracking and are expected to enable novel ways of human-computer interaction in numerous applications. However, since eye movement properties contain biometric information, privacy concerns have to be handled properly. Privacy-preservation techniques such as differential privacy mechanisms have recently been applied to eye movement data obtained from such displays. Standard differential privacy mechanisms; however, are vulnerable due to temporal correlations between the eye movement observations. In this work, we propose a novel transform-coding based differential privacy mechanism to further adapt it to the statistics of eye movement feature data and compare various low-complexity methods. We extend the Fourier perturbation algorithm, which is a differential privacy mechanism, and correct a scaling mistake in its proof. Furthermore, we illustrate significant reductions in sample correlations in addition to query sensitivities, which provide the best utility-privacy trade-off in the eye tracking literature. Our results provide significantly high privacy without any essential loss in classification accuracies while hiding personal identifiers.

\subsection{Introduction}
Recent advances in the field of head-mounted displays (HMDs), computer graphics, and eye tracking enable easy access to pervasive eye trackers along with modern \acs{HMD}s. Soon, the usage of such devices might result in a significant increase in the amount of eye movement data collected from users across different application domains such as gaming, entertainment, or education. A large part of this data is indeed useful for personalized experience and user-adaptive interaction. Especially in virtual and augmented reality (\acs{VR}/\acs{AR}), it is possible to derive plenty of sensitive information about users from the eye movement data. In general, it has been shown that eye tracking signals can be employed for activity recognition even in challenging everyday tasks~\cite{Steil:2015:DEH:2750858.2807520,braunagel2017online,9d76c83aeb334940b29019b5624ca501}, to detect cognitive load~\cite{Appel:2018,10.1371/journal.pone.0203629}, mental fatigue~\cite{YAMADA201839}, and many other user states. Similarly, assessment of situational attention~\cite{bozkir_vr_attention_et}, expert-novice analysis in areas such as medicine~\cite{castner2018scanpath} and sports~\cite{10.1371/journal.pone.0186871}, detection of personality traits~\cite{Berkovsky:2019:DPT:3290605.3300451}, and prediction of human intent during robotic hand-eye coordination~\cite{Razin:2017:LPI:3298023.3298235} can also be carried out based on eye movement features. Additionally, eye movements are useful for early detection of anomias~\cite{10.3389/fnhum.2019.00354} and diseases~\cite{alzeimers_eye_tracking}. More importantly, eye movement data allow biometric authentication, which is considered to be a highly sensitive task~\cite{Onurbio1}. A task-independent authentication using eye movement features and Gaussian mixtures is, for example, discussed by Kinnunen et al.~\cite{Kinnunen:2010:TTP:1743666.1743712}. Additionally, biometric identification based on an eye movements and oculomotor plant model are introduced by Komogortsev and Holland~\cite{6712725} and by Komogortsev et al.~\cite{Komogortsev2010}. Eberz et al.~\cite{Eberz:2016:LLE:2957761.2904018} discuss that eye movement features can be used reliably also for authentication both in consumer level devices and various real world tasks, whereas Zhang et al.~\cite{Zhang:2018:CAU:3178157.3161410} show that continuous authentication using eye movements is possible in \acs{VR} headsets. While authentication via eye movements could be useful in biometric applications, the applications that do not require any authentication step possess privacy risks for the individuals if such information is not hidden in the data. In addition, if such information is linked to personal identifiers, the risk might be even higher.

As biometric content can be retrieved from eye movements, it is important to protect them against adversarial behaviors such as membership inference. According to Steil et al.~\cite[p.~3]{steil_diff_privacy}, people agree to share their eye tracking data if a governmental health agency is involved in owning data or if the purpose is research. Therefore, privacy-preserving techniques are needed especially on the data sharing side of eye tracking considering that the usage of \acs{VR}/\acs{AR} devices with integrated eye trackers increases. As removing only the personal identifiers from data is not enough for anonymization due to linkage attacks~\cite{4531148}, more sophisticated techniques for achieving user level privacy are necessary. Differential privacy~\cite{dwork2006,dwork_DP_only} is one effective solution, especially in the area of database applications. It protects user privacy by adding randomly generated noise for a given sensitivity and desired privacy parameter, $\epsilon$. The differentially private mechanisms provide aggregate statistics or query answers while protecting the information of whether an individual's data was contained in a dataset. However, high dimensionality of the data and temporal correlations can reduce utility and privacy, respectively. Since eye movement features are high dimensional, temporally correlated, and usually contain recordings with long durations, it is important to tackle utility and privacy problems simultaneously. For eye movement data collected from \acs{HMD}s or smart glasses, both local and global differential privacy can be applied. Applying differential privacy mechanisms to eye movement data would optimally anonymize the query outcomes that are carried out on such data while keeping data utility and usability high enough. As opposed to global differential privacy, local differential privacy adds user level noise to the data but assumes that the user sends data to a central data collector after adding local noise~\cite{Erlingsson:2014:RRA:2660267.2660348,NIPS2017_6948}. While both could be useful depending on the application use-case, for this work, we focus on global differential privacy, considering that in many \acs{VR}/\acs{AR} applications which collect eye movement data, there is a central trusted user-level data collector and publisher.

To apply differential privacy to the eye movement data, we evaluate the standard Laplace Perturbation Algorithm (LPA)~\cite{dwork2006} of differential privacy and Fourier Perturbation Algorithm (FPA)~\cite{Rastogi:2010:DPA:1807167.1807247}. The latter is suitable for time series data such as the eye movement feature signals. We propose two different methods that apply the \acs{FPA} to chunks of data using original eye movement feature signals or consecutive difference signals. While preserving differential privacy using parallel compositions, chunk-based methods decrease query sensitivity and computational complexity. The difference-based method significantly decreases the temporal correlations between the eye movement features in addition to the decorrelation provided by the \acs{FPA} that uses the discrete Fourier transform (DFT) as, e.g., in the works of Günlü and İşcan~\cite{Onurbio2} and Günlü et al.~\cite{OnurHWimplforlowcomp}. The difference-based method provides a higher level of privacy since consecutive sample differences are observed to be less correlated than original consecutive data. Furthermore, we evaluate our methods using differentially private eye movement features in document type, gender, scene privacy sensitivity classification, and person identification tasks on publicly available eye movement datasets by using similar configurations to previous works by Steil et al.~\cite{steil_diff_privacy,Steil:2019:PPH:3314111.3319913}. To generate differentially private eye movement data, we use the complete data instead of applying a subsampling step, used by Steil et al.~\cite{steil_diff_privacy} to reduce the sensitivity and to improve the classification accuracies for document type and privacy sensitivity. In addition, the previous work~\cite{steil_diff_privacy} applies the exponential mechanism for differential privacy on the eye movement feature data. The exponential mechanism is useful for situations where the best enumerated response needs to be chosen~\cite{dwork2014}. In eye movements, we are not interested in the ``best'' response but in the feature vector. Therefore, we apply the Laplace mechanism. In summary, we are the first to propose differential privacy solutions for aggregated eye movement feature signals by taking the temporal correlations into account, which can help provide user privacy especially for \acs{HMD} or smart glass usage in \acs{VR}/\acs{AR} setups.

Our main contributions are as follows. (1) We propose chunk-based and difference-based differential privacy methods for eye movement feature signals to reduce query sensitivities, computational complexity, and temporal correlations. Furthermore, (2) we evaluate our methods on two publicly available eye movement datasets, i.e., MPIIDPEye~\cite{steil_diff_privacy} and MPIIPrivacEye~\cite{Steil:2019:PPH:3314111.3319913}, by comparing them with standard techniques such as \acs{LPA} and \acs{FPA} using the multiplicative inverse of the absolute normalized mean square error (NMSE) as the utility metric. In addition, we evaluate document type and gender classification, and privacy sensitivity classification accuracies as classification metrics using differentially private eye movements in the MPIIDPEye and MPIIPrivacEye datasets, respectively. Classification accuracy is used in the literature as a practical utility metric that shows how useful the data and proposed methods are. Our utility metric also provides insights into the divergence trend of differentially private outcomes and is analytically trackable unlike the classification accuracy. For both datasets, we also evaluate person identification task using differentially private data. Our results show significantly better performance as compared to previous works while handling correlated data and decreasing query sensitivities by dividing the data into smaller chunks. In addition, our methods hide personal identifiers significantly better than existing methods.

\subsubsection{Previous Research}
There are few works that focus on privacy-preserving eye tracking. Liebling and Preibusch~\cite{Liebling2014} provide motivation as to why privacy considerations are needed for eye tracking data by focusing on gaze and pupillometry. Practical solutions are; therefore, introduced to protect user identity and sensitive stimuli based on a degraded iris authentication through optical defocus~\cite{John:2019:EDI:3314111.3319816} and an automated disabling mechanism for the eye tracker's ego perspective camera with the help of a mechanical shutter depending on the detection of privacy sensitive content~\cite{Steil:2019:PPH:3314111.3319913}. Furthermore, a function-specific privacy model for privacy-preserving gaze estimation task and privacy-preserving eye videos by replacing the iris textures are proposed by Bozkir and Ünal et al.~\cite{bozkir_ppge} and by Chaudhary and Pelz~\cite{10.1145/3379156.3391375}, respectively. In addition, solutions for privacy-preserving eye tracking data streaming~\cite{davidjohn2021privacypreserving} and real-time privacy control for eye tracking systems using area-of-interests~\cite{263891} are also introduced in the literature. These works lack studying effects of temporal correlations.

For the user identity protection on aggregated eye movement features, works that focus on differential privacy are more relevant for us. Recently, standard differential privacy mechanisms are applied to heatmaps~\cite{Liu2019} and eye movement data that are obtained from a \acs{VR} setup~\cite{steil_diff_privacy}. These works do not address the effects of temporal correlations in eye movements over time in the privacy context. In the privacy literature, there are privacy frameworks such as the Pufferfish~\cite{pufferfish_transaction_paper} or the Olympus~\cite{olympus_framework_privacy} for correlated and sensor data, respectively. These works, however, have different assumptions. For instance, the Pufferfish requires a domain expert to specify potential secrets and discriminative pairs, and Olympus models privacy and utility requirements as adversarial networks. As our focus is to protect user identity in the eye movements, we opt for differential privacy by discussing the effects of temporal correlations in eye movements over time and propose methods to reduce them. It has been shown that standard differential privacy mechanisms are vulnerable to temporal correlations as such mechanisms assume that data at different time points are independent from each other or adversaries lack the information about temporal correlations, leading an increased privacy loss of a differential privacy mechanism over time due to the temporal correlations~\cite{7930028,8333800}. The aggregated eye movement features over time might end up in an extreme case of such correlations due to various user behaviors. Therefore, in addition to discussing the effects of such correlations on differential privacy over time, we propose methods to reduce the correlations so that the privacy leakage due to the temporal correlations are minimal.

\subsection{Materials and Methods}
In this section, the theoretical background of differential privacy mechanisms, proposed methods, and evaluated datasets are discussed.

\subsubsection{Theoretical Background}
Differential privacy uses a metric to measure the privacy risk for an individual participating in a database. Considering a dataset with weights of $N$ people and a mean function, when an adversary queries the mean function for $N$ people, the average weight over $N$ people is obtained. After the first query, an additional query for $N-1$ people automatically leaks the weight of the remaining person. Using differential privacy, noise is added to the outcome of a function so that the outcome does not significantly change based on whether a randomly chosen individual participated in the dataset. The amount of noise added should be calibrated carefully since a high amount of noise might decrease the utility. We next define differential privacy.

\begin{definition}{\emph{$\epsilon$-Differential Privacy ($\epsilon$-DP)~\cite{dwork2006,dwork_DP_only}.}} 
\label{def:DiffPrivacy} 
A randomized mechanism $M$ is $\epsilon$-differentially private if for all databases $D$ and $D^\prime$ that differ at most in one element for all $S \subseteq Range(M)$ with
\begin{equation} 
\Pr[M(D) \in S] \leq e^{\epsilon} \Pr[M(D^\prime) \in S].
\end{equation}
\end{definition}

The variance of the added noise depends on the query sensitivity, which is defined as follows.

\begin{definition}{\emph{Query sensitivity~\cite{dwork2006}.}}
\label{def:QuerySensitivity}
For a random query $X^n$ and $w \in \{1,2\}$, the query sensitivity $\Updelta_{w}$ of $X^n$ is the smallest number for all databases $D$ and $D^\prime$ that differ at most in one element such that
\begin{equation} 
||\, X^n(D) - X^n(D^\prime) ||_{w} \leq \Updelta_{w}(X^n)
\end{equation}
where the $L_{w}$-distance is defined as 
\begin{equation} 
||\,X^n||_w = \sqrt[w]{\sum_{i=1}^{n}{\big(|X_{i}|\big)^{w}}}.
\end{equation}
\end{definition}

We list theorems that are used in the proposed methods. 
\begin{theorem} {Sequential composition theorem~\cite{McSherry:2009:PIQ:1559845.1559850}.}
\label{Theorem:SequentialCompTheorem} 
Consider $n$ mechanisms $M_{i}$ that randomization of each query is independent for $i=1,2,...,n$. If $M_{1}, M_{2}, ..., M_{n}$ are $\epsilon_{1}, \epsilon_{2}, ..., \epsilon_{n}$-differentially private, respectively, then their joint mechanism is $\displaystyle\left(\sum_{i=1}^{n}\epsilon_{i}\right)$-differentially private.
\end{theorem}

\begin{theorem} {Parallel composition theorem~\cite{McSherry:2009:PIQ:1559845.1559850}.}
\label{Theorem:ParallelCompTheorem} 
Consider $n$ mechanisms as $M_{i}$ for $i=1,2,...,n$ that are applied to disjoint subsets of an input domain. If $M_{1}, M_{2}, ..., M_{n}$ are $\epsilon_{1}, \epsilon_{2}, ..., \epsilon_{n}$-differentially private, respectively, then their joint mechanism is
$\left(\underset{i \in [1,n]}{\max}\epsilon_{i}\right)$-differentially private.
\end{theorem}

We define the Laplace Perturbation Algorithm (LPA)~\cite{dwork2006}. To guarantee differential privacy, the \acs{LPA} generates the noise according to a Laplace distribution. $Lap(\lambda)$ denotes a random variable drawn from a Laplace distribution with a probability density function (PDF): $\Pr[Lap(\lambda) = h] = \frac{1}{2\lambda} e^{-\mid h \mid/\lambda}$, where $Lap(\lambda)$ has zero mean and variance $2\lambda^{2}$. We denote the noisy and differentially private values as $\widetilde{X}_i = X_i(D) + Lap(\lambda)$ for $i=1,2,\ldots,n$. Since we have a series of eye movement observations, the final noisy eye movement observations are generated as $\widetilde{X}^n = X^n(D) + Lap^{n}(\lambda)$, where $Lap^{n}(\lambda)$ is a vector of $n$ independent $Lap(\lambda)$ random variables and $X^n(D)$ is the eye movement observations without noise. The \acs{LPA} is $\epsilon$-differentially private for $\lambda = \Updelta_{1}(X^n)/{\epsilon}$~\cite{dwork2006}.

We define the error function that we use to measure the differences between original $X^n$ and noisy $\widetilde{X}^n$ observations. For this purpose, we use the metric normalized mean square error (NMSE) defined as

\begin{equation}
    \text{NMSE} = \frac{1}{n}\sum_{i=1}^{n}{\frac{(X_{i} - \widetilde{X}_{i})^2}{\overline{X}\overline{\widetilde{X}}}}
\end{equation} where 
\begin{equation}
    \overline{X} = \frac{1}{n} \sum_{i=1}^{n}{X_{i}}\,,\qquad
    \overline{\widetilde{X}} = \frac{1}{n} \sum_{i=1}^{n}{\widetilde{X}_{i}}.
\end{equation}

We define the utility metric as
\begin{equation}
\label{eqn:Utility}
    \text{Utility} = \frac{1}{|\text{NMSE}|}.
\end{equation}

As differential privacy is achieved by adding random noise to the data, there is a utility-privacy trade-off. Too much noise leads to high privacy; however, it might also result in poor analyses on the further tasks on eye movements. Therefore, it is important to find a good trade-off.

\subsubsection{Methods}
Standard differential privacy mechanisms are vulnerable to temporal correlations, since the independent noise realizations that are added to temporally correlated data could be useful for adversaries. However, decorrelating the data without the domain knowledge before adding the noise might remove important eye movement patterns and provide poor results in analyses. Many eye movement features are extracted by using time windows, as in previous work~\cite{steil_diff_privacy,Steil:2019:PPH:3314111.3319913}, which makes the features highly correlated. Another challenge is that the duration of eye tracking recordings could change depending on the personal behaviors, skills, or personalities of the users. The longer duration causes an increased query sensitivity, which means that higher amounts of noise should be added to achieve differential privacy. In addition, when correlations between different data points exist, $\epsilon^{\prime}$ is defined as the actual privacy metric instead of $\epsilon$~\cite{8269219} that is obtained considering the fact that correlations can be used by an attacker to obtain more information about the differentially private data by filtering. In this work, we discuss and propose generic low-complexity methods to keep $\epsilon^{\prime}$ small for eye movement feature signals. To deal with correlated eye movement feature signals, we propose three different methods: \acs{FPA}, chunk-based FPA (CFPA) for original feature signals, and chunk-based FPA for difference based sequences (DCFPA). The sensitivity of each eye movement feature signal is calculated by using the $L_{w}$-distance such that
\begin{align}
\Updelta^{f}_{w}(X^n) &= \max_{p,\,q}{\left|\left|\,X^{n,(p,f)} - X^{n,(q,f)}\right|\right|}_{w} \nonumber\\  &=\max_{p,\,q}{\sqrt[w]{\sum_{t=1}^{n}{\Big(\Big|X^{(p,f)}_{t} - X^{(q,f)}_{t}\Big|\Big)^{w}}}}\label{eq:ObservationVectSensitivities}
\end{align}
where $X^{n,(p,f)}$ and $X^{n,(q,f)}$ denote observation vectors for a feature $f$ from two participants $p$ and $q$, $n$ denotes the maximum length of the observation vectors, and $w \in \{1,2\}$.

\paragraph{Fourier Perturbation Algorithm (FPA) \\}
In the \acs{FPA}~\cite{Rastogi:2010:DPA:1807167.1807247}, the signal is represented with a small number of transform coefficients such that the query sensitivity of the representative signal decreases. A smaller query sensitivity decreases the noise power required to make the noisy signal differentially private. In the \acs{FPA}, the signal is transformed into the frequency domain by applying Discrete Fourier Transform (DFT), which is commonly applied as a non-unitary transform. The frequency domain representation of a signal consists of less correlated transform coefficients as compared to the time domain signal due to the high decorrelation efficiency of the \acs{DFT}. Therefore, the correlation between the eye movement feature signals is reduced by applying the \acs{DFT}. After the \acs{DFT}, the noise sampled from the \acs{LPA} is added to the first $k$ elements of $DFT(X^n)$ that correspond to $k$ lowest frequency components, denoted as $F^k=DFT^{k}(X^n)$. Once the noise is added, the remaining part (of size $n-k$) of the noisy signal $\widetilde{F}^k$ is zero padded and denoted as $PAD^{n}(\widetilde{F}^k)$. Lastly, using the Inverse DFT (IDFT), the padded signal is transformed back into the time domain. We can show that $\epsilon$-differential privacy is satisfied by the \acs{FPA} for $\lambda = \frac{\sqrt{n}\sqrt{k}\Updelta_{2}(X^n)}{\epsilon}$ unlike the value claimed in previous work~\cite{Rastogi:2010:DPA:1807167.1807247}, as observed independently by Kellaris and Papadopoulos~\cite{FPACorrector}. The procedure is summarized in Figure~\ref{algo:FPA}, and the proof is provided below. Since not all coefficients are used, in addition to the perturbation error caused by the added noise, a reconstruction error caused by the lossy compression is introduced. It is important to determine the number of used coefficients $k$ to minimize the total error. We discuss how we choose $k$ values for \acs{FPA}-based methods below.

\begin{figure}[!h]
   \includegraphics[width=1\linewidth]{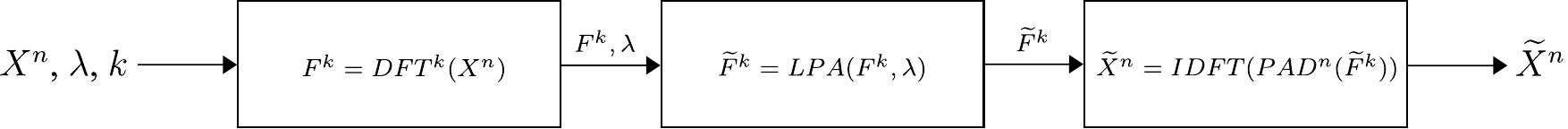}
	\caption{Flow of the Fourier Perturbation Algorithm (FPA).}  
  \label{algo:FPA}%
\end{figure}

\textbf{Proof of FPA being differentially private.}
We next prove that the \acs{FPA} is $\epsilon$ differentially private for $\displaystyle \lambda = (\sqrt{n}\sqrt{k}\Updelta_{2}(X^n))/\epsilon$. First, we prove the inequalities $(a)$ and $(b)$ in the following.\\

\begin{equation}
\Updelta_{1}(\widehat{F}^n) \stackrel{(a)}{\leq} \sqrt{k} \cdot \Updelta_{2}(\widehat{F}^n) \stackrel{(b)}{\leq} \sqrt{n} \cdot \sqrt{k} \cdot \Updelta_{2}(X^n) 
\label{main_eq_for_FPA_proof}
\end{equation} where
$\widehat{F}^n(I) = [\widehat{F}^k(I), 0,0,\ldots,0]$ such that $n-k$ zeros are padded. Consider~(\ref{main_eq_for_FPA_proof})$(a)$, which follows since we have

\begin{eqnarray}
\Updelta_{1}(\widehat{F}^n) &=& \max_{I,\,I^\prime}{\left|\left|\,\widehat{F}^{n}(I) - \widehat{F}^{n}(I^\prime)\right|\right|}_{1} = \max_{I,\,I^\prime}\sum_{j=1}^{n}{\left|\,\widehat{F}_{j}(I) - \widehat{F}_{j}(I^\prime)\right|} \nonumber\\
&=& \max_{I,\,I^\prime}\sum_{j=1}^{k}{\left|\,\widehat{F}_{j}(I) - \widehat{F}_{j}(I^\prime)\right|\cdot1}
\end{eqnarray} so that by applying Cauchy-Schwarz inequality, we obtain

\begin{eqnarray}
    \max_{I,\,I^\prime}\sum_{j=1}^{k}{\left|\,\widehat{F}_{j}(I) - \widehat{F}_{j}(I^\prime)\right|\cdot1} &\leq& \max_{I,\,I^\prime}\Big(\sum_{j=1}^{k}{\left|\,\widehat{F}_{j}(I) - \widehat{F}_{j}(I^\prime)\right|^{2}\Big)^{1/2}\cdot \Big(\sum_{j=1}^{k}{1^2}\Big)^{1/2}} \nonumber\\
    &\leq& \max_{I,\,I^\prime}{\left|\left|\,\widehat{F}^{n}(I) - \widehat{F}^{n}(I^\prime)\right|\right|}_{2} \cdot \sqrt{k} \nonumber \\
    &\leq& \sqrt{k} \cdot \Updelta_{2}(\widehat{F}^n).
\end{eqnarray}

Consider next~(\protect\ref{main_eq_for_FPA_proof})$(b)$, which follows since we obtain

\begin{equation}
    \Updelta_{2}(\widehat{F}^n) = \max_{I,\,I^\prime}{\left|\left|\,\widehat{F}^{n}(I) - \widehat{F}^{n}(I^\prime)\right|\right|}_{2} = \max_{I,\,I^\prime}\Big(\sum_{j=1}^{n}{\left|\,\widehat{F}_{j}(I) - \widehat{F}_{j}(I^\prime)\right|^{2}\Big)^{1/2}}
\end{equation}
and since $F^{n}$ has more non-zero elements than $\widehat{F}^n$, we have

\begin{equation}
    \Updelta_{2}(\widehat{F}^n) \leq \max_{I,\,I^\prime}\Big(\sum_{j=1}^{n}{\left|\,F_{j}(I) - F_{j}(I^\prime)\right|^{2}\Big)^{1/2}}.
    \label{eqn:b_first}
\end{equation}

Recall that $F^{n}(I) = DFT(X^{n}(I))$, $F^{n}(I^{\prime}) = DFT(X^{n}(I^{\prime}))$, and \acs{DFT} is linear, so we have

\begin{equation}
    DFT(X^{n}(I)-X^{n}(I^{\prime})) = F^{n}(I) - F^{n}(I^{\prime}).
\end{equation}

By applying Parseval's theorem to the \acs{DFT}, we obtain

\begin{equation}
    \Big(\frac{1}{n} \cdot \sum_{j=1}^{n}{\left|\,F_{j}(I) - F_{j}(I^\prime)\right|^{2}\Big)^{1/2}} = \Big(\sum_{j=1}^{n}{\left|\,X_{j}(I) - X_{j}(I^\prime)\right|^{2}\Big)^{1/2}}.
    \label{eqn:b_second}
\end{equation}

Combining~(\ref{eqn:b_first}) and~(\ref{eqn:b_second}), we prove~(\ref{main_eq_for_FPA_proof})$(b)$ since we have

\begin{eqnarray}
     \Updelta_{2}(\widehat{F}^n) &\leq& \max_{I,\,I^{\prime}}{\sqrt{\sum_{j=1}^{n}{\left|X_{j}(I) - X_{j}(I^{\prime})\right|^{2}}} \cdot \sqrt{n}} \nonumber\\
     &\leq& \max_{I,\,I^\prime}{\left|\left|\,X^{n}(I) - X^{n}(I^\prime)\right|\right|}_{2} \cdot \sqrt{n} \nonumber\\
     &\leq& \Updelta_{2}(X^{n}) \cdot \sqrt{n}.
\end{eqnarray}

Finally, since the \acs{LPA} that is applied to $\widehat{F}^{k}$ is $\epsilon$-\acs{DP} for $\lambda = \frac{\Updelta_{1}(\widehat{F}^n)}{\epsilon}$~\cite{dwork2006},~(\ref{main_eq_for_FPA_proof}) proves that the \acs{FPA} is $\epsilon$-\acs{DP} for $\displaystyle \lambda = \frac{\sqrt{n}\sqrt{k}\Updelta_{2}(X^n)}{\epsilon}$.

\paragraph{Chunk-based FPA (CFPA) \\}
One drawback of directly applying the \acs{FPA} to the eye movement feature signals is large query sensitivities for each feature $f$ due to long signal sizes. To solve this, Steil et al.~\cite{steil_diff_privacy} propose to subsample the signal using non-overlapping windows, which means removing many data points. While subsampling decreases the query sensitivities, it also decreases the amount of data. Instead, we propose to split each signal into smaller chunks and apply the \acs{FPA} to each chunk so that complete data can be used. We choose the chunk sizes of $32$, $64$, and $128$ since there are divide-and-conquer type tree-based implementation algorithms for fast \acs{DFT} calculations when the transform size is a power of $2$~\cite{Onurbio3}. When the signals are split into chunks, chunk level query sensitivities are calculated and used rather than the sensitivity of the whole sequence. Differential privacy for the complete signal is preserved by Theorem~\ref{Theorem:ParallelCompTheorem}~\cite{McSherry:2009:PIQ:1559845.1559850} since the used chunks are non-overlapping. As the chunk size decreases, the chunk level sensitivity decreases as well as the computational complexity. However, the parameter $\epsilon^{\prime}$ that accounts for the sample correlations might increase with smaller chunk sizes because temporal correlations between neighboring samples are larger in an eye movement dataset. On the other hand, if the chunk sizes are kept large, then the required amount of noise to achieve differential privacy increases due to the increased query sensitivity. Therefore, a good trade-off between computational complexity, and correlations is needed to determine the optimal chunk size.

\paragraph{Difference- and chunk-based FPA (DCFPA) \\}
To tackle temporal correlations, we convert the eye movement feature signals into difference signals where differences between consecutive eye movement features are calculated as
\begin{equation}
    \widehat{X}_t^{(f)} = 
    \Big\{ X_t^{(f)} -  X_{t-1}^{(f)} \Big\} \Big|_{t=2}^{n} \quad \text{,}\quad \widehat{X}_1^{(f)} =  X_1^{(f)}.
\end{equation}

Using the difference signals denoted by $\widehat{X}^{n,(f)}$, we aim to further decrease the correlations before applying a differential privacy method. We conjecture that the ratio $\epsilon^{\prime}/{\epsilon}$ decreases in the difference-based method as compared to the \acs{FPA} method. To support this conjecture, we show that the correlations in the difference signals decrease significantly as compared to the original signals. This results in lower $\epsilon^{\prime}$ and better privacy for the same $\epsilon$. The difference-based method is applied together with the \acs{CFPA}. Therefore, the differences are calculated inside chunks. The first element of each chunk is preserved. Then, the \acs{FPA} mechanism is applied to the difference signals by using query sensitivities calculated based on differences and chunks. For each chunk, noisy difference observations are aggregated to obtain the final noisy signals. This mechanism is differentially private by Theorem~\ref{Theorem:SequentialCompTheorem}~\cite{McSherry:2009:PIQ:1559845.1559850}, and described in Algorithm~\ref{algo:diff-chunk-FPA}. 
\begin{algorithm}
  \KwInput{$X^n$, $\lambda$, $k$}
  \KwOutput{$\widetilde{X}^n$}
  1) $\widehat{X}_t = \Big\{ X_t-X_{t-1}\Big\}\Big|_{t=2}^{n}\quad \text{,}\quad \widehat{X}_1=X_1$. \\
  2) $\widetilde{\widehat{X}}^{n} = FPA(\widehat{X}^{n}, \lambda, k)$.  \\
  3) $\widetilde{X}_t = \Big\{\widetilde{\widehat{X}}_t + \widetilde{\widehat{X}}_{t-1} \Big\} \Big|_{t=2}^{n}\quad \text{,}\quad \widetilde{X}_1 = \widetilde{\widehat{X}}_1$.
\caption{DCFPA.}
\label{algo:diff-chunk-FPA}
\end{algorithm}

Since Theorem~\ref{Theorem:SequentialCompTheorem} can be applied to the \acs{DCFPA} when the consecutive differences are assumed to be independent, which is a valid assumption for eye movement feature signals as we illustrate below, there is also a trade-off between the chunk sizes and utility for the \acs{DCFPA}. If a large chunk size is chosen, then the total $\epsilon$ value could be very large, which reduces privacy. Therefore, we choose chunk sizes of $32$, $64$, and $128$ for the \acs{DCFPA} as well for evaluation We illustrate the \acs{CFPA} and \acs{DCFPA} in Figure~\ref{fig:CFPA-DCFPA}, for instance with three chunks.

\begin{figure}[!h]
   \includegraphics[width=1\linewidth]{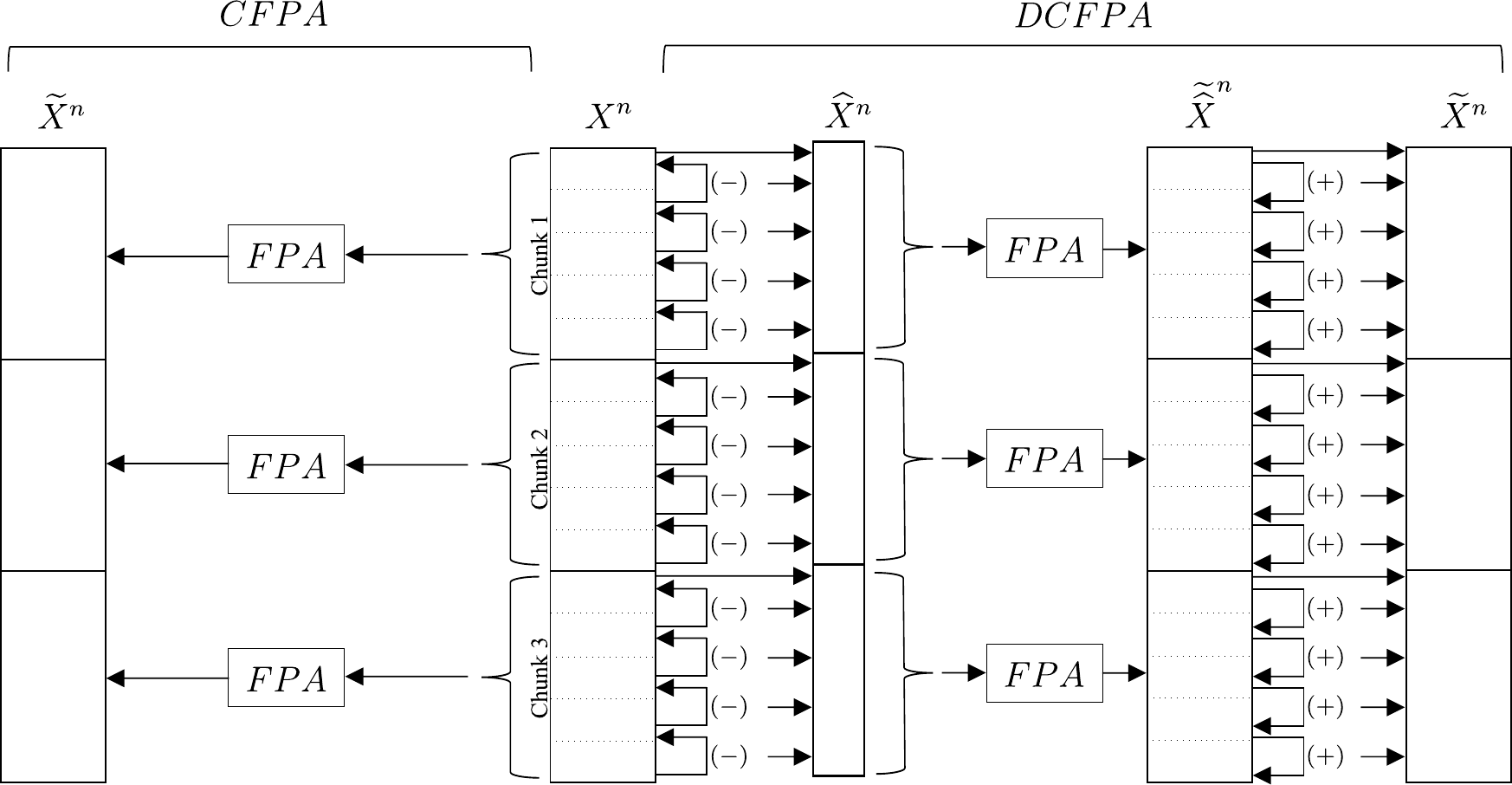}
   \caption{Flow of the CFPA and DCFPA.}
  \label{fig:CFPA-DCFPA}%
\end{figure}

\paragraph{Choice of the Number of Transform Coefficients \\}
The proposed methods require a selection of a value for $k$. A small $k$ value increases the reconstruction error, while a large $k$ value results in an increase in the perturbation error. Therefore, it is important to find an optimal $k$ value that minimizes the sum of the two errors. In this work, we compare a large set of possible $k$ values to choose the best values. 

We apply the aforementioned differential privacy mechanisms by using $100$ noisy evaluations to find optimal $k$ values applied to features or chunks. Optimal $k$ values have the minimum absolute \acs{NMSE} for each chunk, eye movement feature, and document or recording type. In a distributed setting, each party should know the $k$ values in advance. However, in a centralized setting, it is crucial to choose the $k$ values in a differentially private manner. To evaluate the differential privacy in the eye tracking area while taking the temporal correlations into account, we focus on optimal $k$ values for this work. One shortcoming of this approach is that the optimal $k$ value compromises some information about the data, which leaks privacy~\cite{Rastogi:2010:DPA:1807167.1807247}. Our observation is that the information leaked by optimizing the parameter $k$ is negligible as compared to the privacy reduction due to temporally correlated data. Thus, we illustrate the results with optimal $k$ values.

\subsubsection{Datasets}
We evaluate our methods on two different publicly available eye movement datasets namely, MPIIDPEye and MPIIPrivacEye that are dedicated to privacy-preserving eye tracking. Both datasets consist of aggregated and timely eye movement feature signals related to eye fixations, saccades, blinks, and pupil diameters which are commonly used in \acs{VR}/\acs{AR} applications as they represent individual user behaviors. As all minimum values of wordbook features ranging from $1$ to $4$ are zeros in both datasets, we exclude them from the utility and privacy calculations. In addition, we remark that both datasets are available for non-commercial scientific purposes.

\textbf{MPIIDPEye~\cite{steil_diff_privacy}:} A publicly available eye movement dataset consisting of $60$ recordings that is collected from \acs{VR} devices for a reading task of three document types (comics, newspaper, and textbook) from $20$ ($10$ female, $10$ male) participants. Each recording consists of $52$ eye movement feature sequences computed with a sliding window size of $30$ seconds and a step size of $0.5$ seconds.

\textbf{MPIIPrivacEye~\cite{Steil:2019:PPH:3314111.3319913}:} A publicly available eye movement dataset consisting of $51$ recordings from $17$ participants for $3$ consecutive sessions with a head-mounted eye tracker and a field camera, which is similar to an \acs{AR} setup. Each recording consists of $52$ eye movement feature sequences computed with a sliding window size of $30$ seconds and a step size of $1$ second, and each observation is annotated with binary privacy sensitivity levels of the scene that is being viewed. The dataset also consists of scene features extracted with convolutional neural networks. We do not evaluate the last part of the recording $1$ of the participant $10$, as the eye movement features are not available for this region. To detect the privacy level of the scene that is being viewed, we remark that the scene is very important~\cite{Orekondy_2017_ICCV}; however, an individual's eye movements can improve the detection rate when they are fused with the information from the scene.

\subsection{Results}
This section discusses data correlations in addition to evaluations using utility and classification metrics. The utility and classification results are averaged over $100$ noisy evaluations with the optimal $k$ values in MATLAB. We evaluate and compare the utility of differentially private eye movement feature signals by using absolute \acs{NMSE}, as this metric provides analytically trackable results. However, it does not provide implications regarding the practical usability of the private eye movement signals. Therefore, we also report classification accuracies of document type, scene privacy sensitivity, gender prediction, and person identification tasks in order to show the usability of the private data and proposed methods. An optimal trade-off between utility tasks (e.g., low absolute \acs{NMSE}, high classification accuracy in document type prediction) and privacy (e.g., low $\epsilon$, low classification accuracy in person identification or gender prediction tasks) is favorable.

\subsubsection{Correlation Analysis}
Using the correlation coefficient as the metric, we first illustrate high temporal correlation between eye movement feature data. Since there are $52$ eye movement features in both datasets, it is not feasible to illustrate all correlation results. Thus, in the following we illustrate the correlations for the features \textit{ratio large saccade} and \textit{blink rate} in the MPIIDPEye and MPIIPrivacEye datasets, respectively. The correlation coefficients of \textit{ratio large saccade} and \textit{blink rate} for three document and recording types over a time difference $\Updelta t$ w.r.t. the signal samples at, e.g., the fifth time instance for original eye movement feature signals and difference signals for all participants for both datasets are depicted in Figures~\ref{fig:corr_coeffs_raw},~\ref{fig:corr_coeffs_diffs} and Figures~\ref{fig:corr_coeffs_mpiiprivaceye_raw}, ~\ref{fig:corr_coeffs_mpiiprivaceye_diffs}, respectively. As correlations between the difference signals are significantly smaller than correlations between the original eye movement feature signals, the \acs{DCFPA} is less vulnerable to privacy reduction due to temporal correlations, thus ensuring that the value of ${\epsilon}^\prime$ is close to the differential privacy design parameter $\epsilon$.

\begin{figure}[!h]
   \includegraphics[width=1\linewidth]{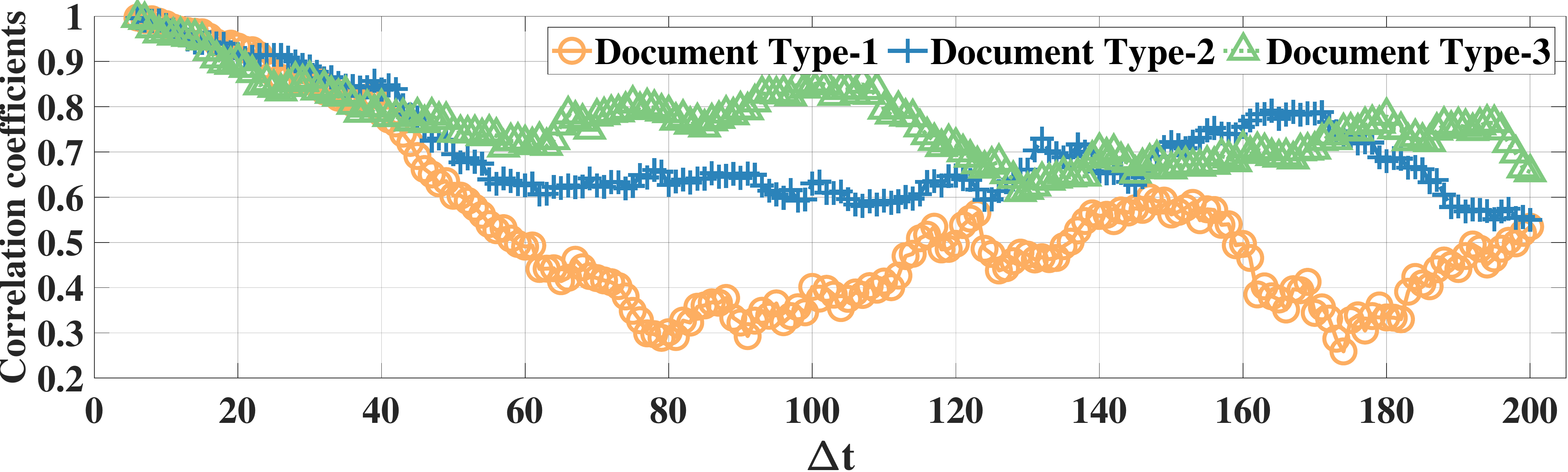}
   \caption{Correlation coefficients of the raw signals of feature \textit{ratio large saccade} in the MPIIDPEye dataset. The values are calculated over a time difference of $\Updelta t$ (Each time step corresponds to $0.5$s) w.r.t. the samples at the fifth time instance.}
  \label{fig:corr_coeffs_raw}%
\end{figure}

\begin{figure}[!h]
  \includegraphics[width=1\linewidth]{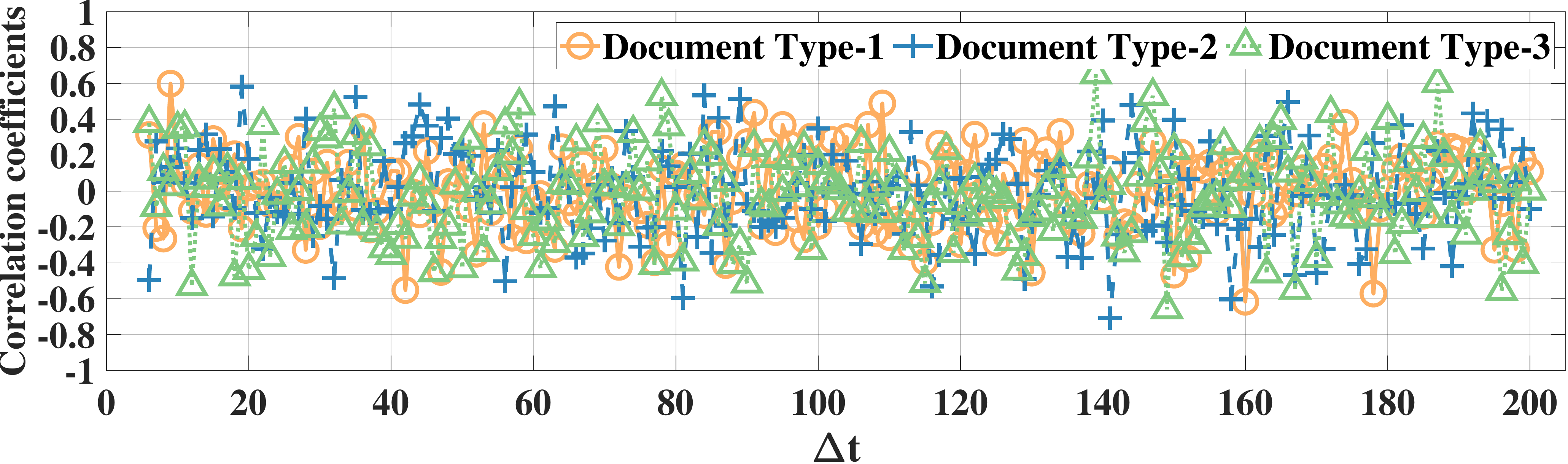}%
  \caption{ Correlation coefficients of the difference signals of feature \textit{ratio large saccade} in the MPIIDPEye dataset. The values are calculated over a time difference of $\Updelta t$ (Each time step corresponds to $0.5$s) w.r.t. the samples at the fifth time instance.}
  \label{fig:corr_coeffs_diffs}%
\end{figure}

\begin{figure}[!h]
  \includegraphics[width=1\linewidth]{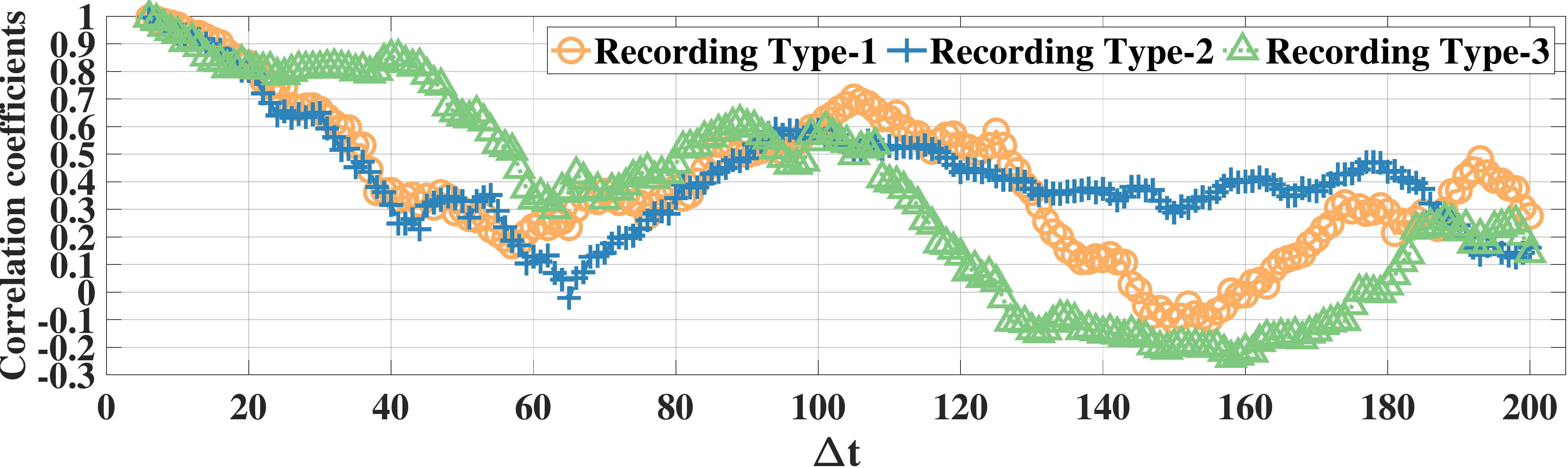}
    \caption{Correlation coefficients of the raw signals of feature \textit{blink rate} in the MPIIPrivacEye dataset. The values are calculated over a time difference of $\Updelta t$ (Each time step corresponds to $1$s) w.r.t. the samples at the fifth time instance.}
  \label{fig:corr_coeffs_mpiiprivaceye_raw}%
\end{figure}

\begin{figure}[!h]
   \includegraphics[width=1\linewidth]{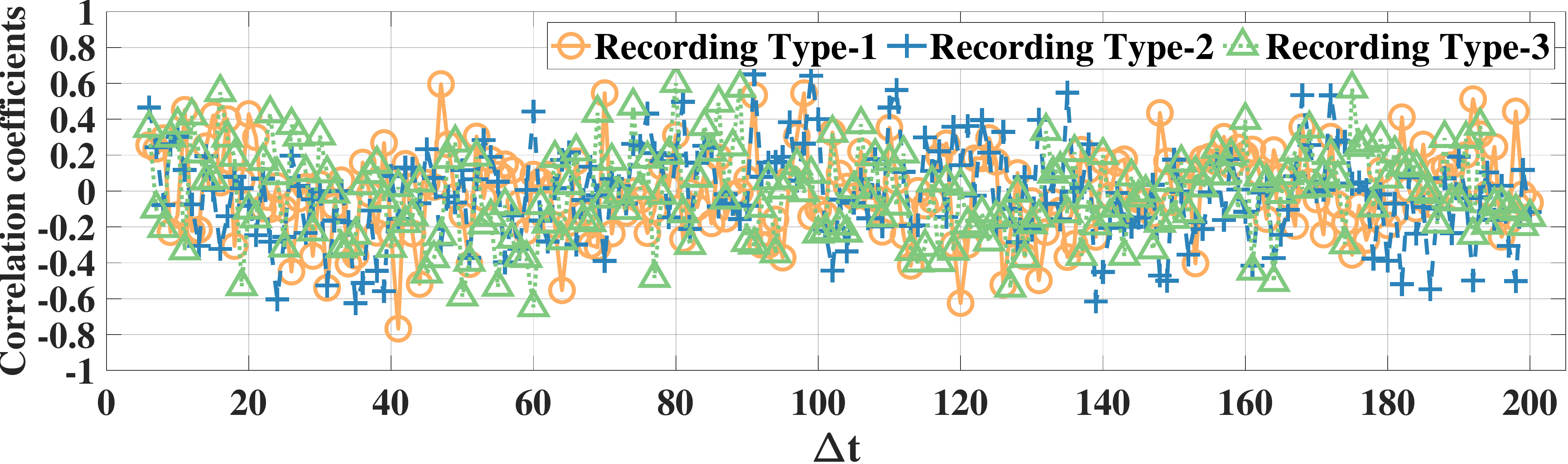}%
     \caption{Correlation coefficients of the difference signals of feature \textit{blink rate} in the MPIIPrivacEye dataset. The values are calculated over a time difference of $\Updelta t$ (Each time step corresponds to $1$s) w.r.t. the samples at the fifth time instance.}
  \label{fig:corr_coeffs_mpiiprivaceye_diffs}%
\end{figure}

\subsubsection{Utility Results}
We evaluate the utility defined in Eq~(\ref{eqn:Utility}) by applying our methods separately to different document and recording types; therefore, we report the utility results separately. As we apply the proposed methods separately to each eye movement feature, we first calculate the mean utility of each feature and then calculate the average utility over all features. The utility results for various $\epsilon$ values for aforementioned methods on the MPIIDPEye and MPIIPrivacEye datasets are given in Figures~\ref{fig:utility1},~\ref{fig:utility2},~\ref{fig:utility3} and Figures~\ref{fig:utility_mpiiprivaceye1},~\ref{fig:utility_mpiiprivaceye2},~\ref{fig:utility_mpiiprivaceye3}, respectively. 

\begin{figure}[!h]
   \includegraphics[width=\linewidth,keepaspectratio]{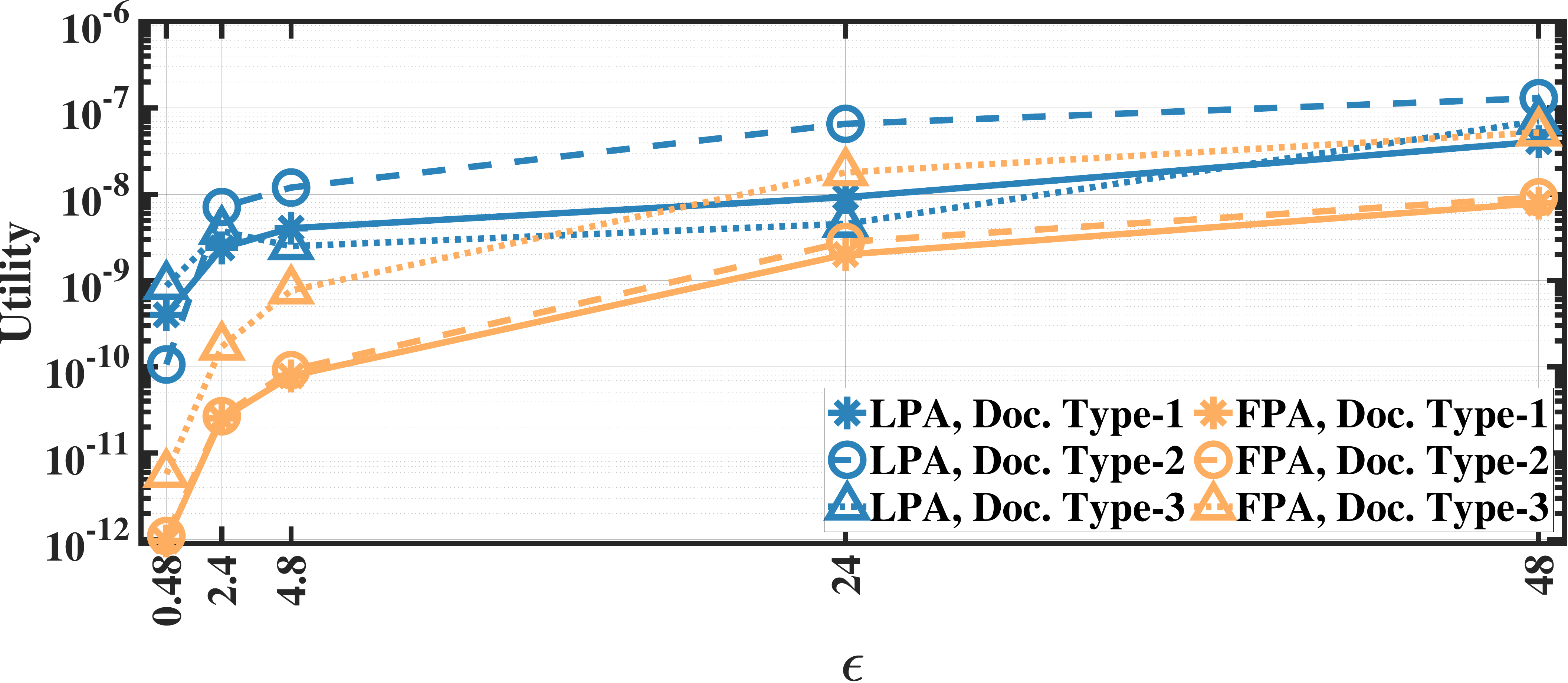}
   \caption{Utility of the LPA and FPA for MPIIDPEye.}
  \label{fig:utility1}%
\end{figure}

\begin{figure}[!h]
   \includegraphics[width=\linewidth,keepaspectratio]{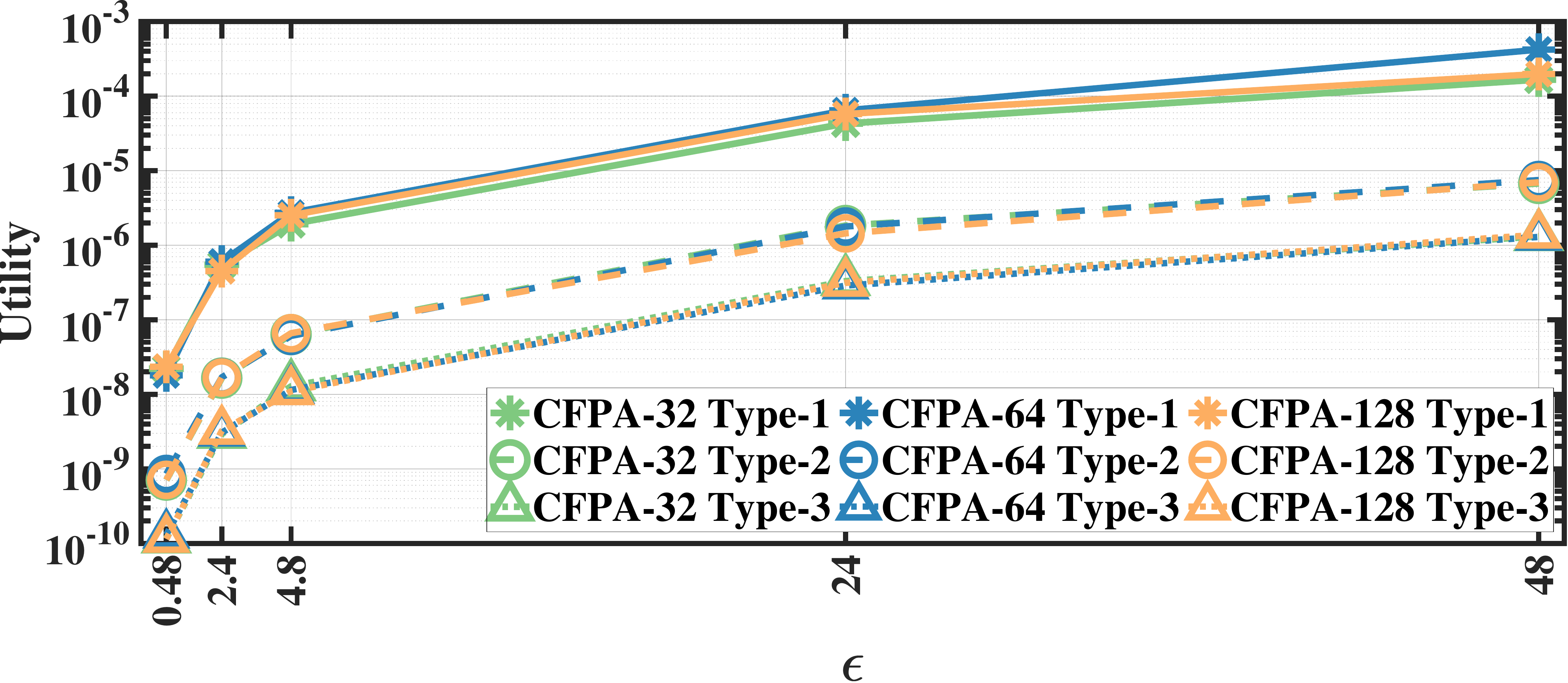}
   \caption{Utility of the CFPA for MPIIDPEye.}
  \label{fig:utility2}%
\end{figure}

\begin{figure}[!h]
   \includegraphics[width=\linewidth,keepaspectratio]{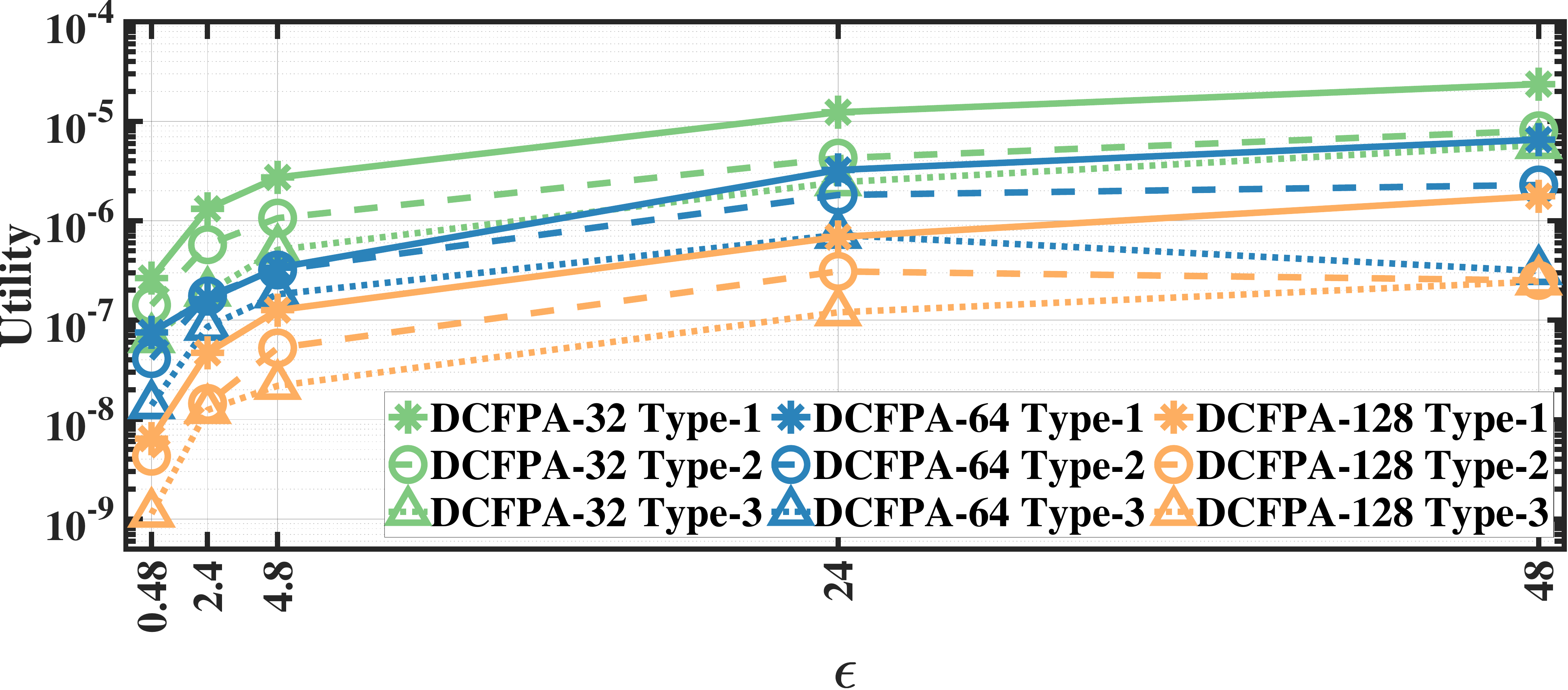}
   \caption{Utility of the DCFPA for MPIIDPEye.}
  \label{fig:utility3}%
\end{figure}

\begin{figure}[!h]
   \includegraphics[width=\linewidth,keepaspectratio]{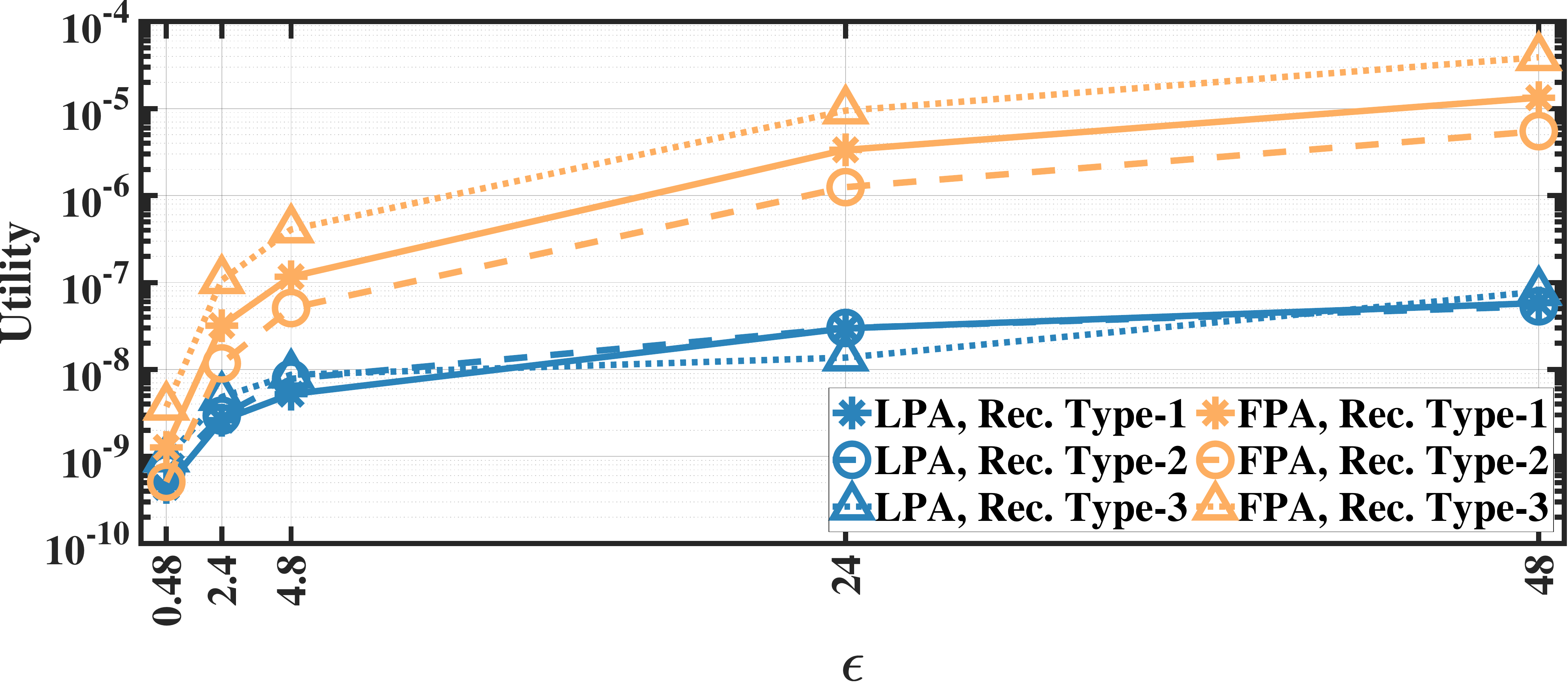}
   \caption{Utility of the LPA and FPA for MPIIPrivacEye.}
  \label{fig:utility_mpiiprivaceye1}%
\end{figure}

\begin{figure}[!h]
  \includegraphics[width=\linewidth,keepaspectratio]{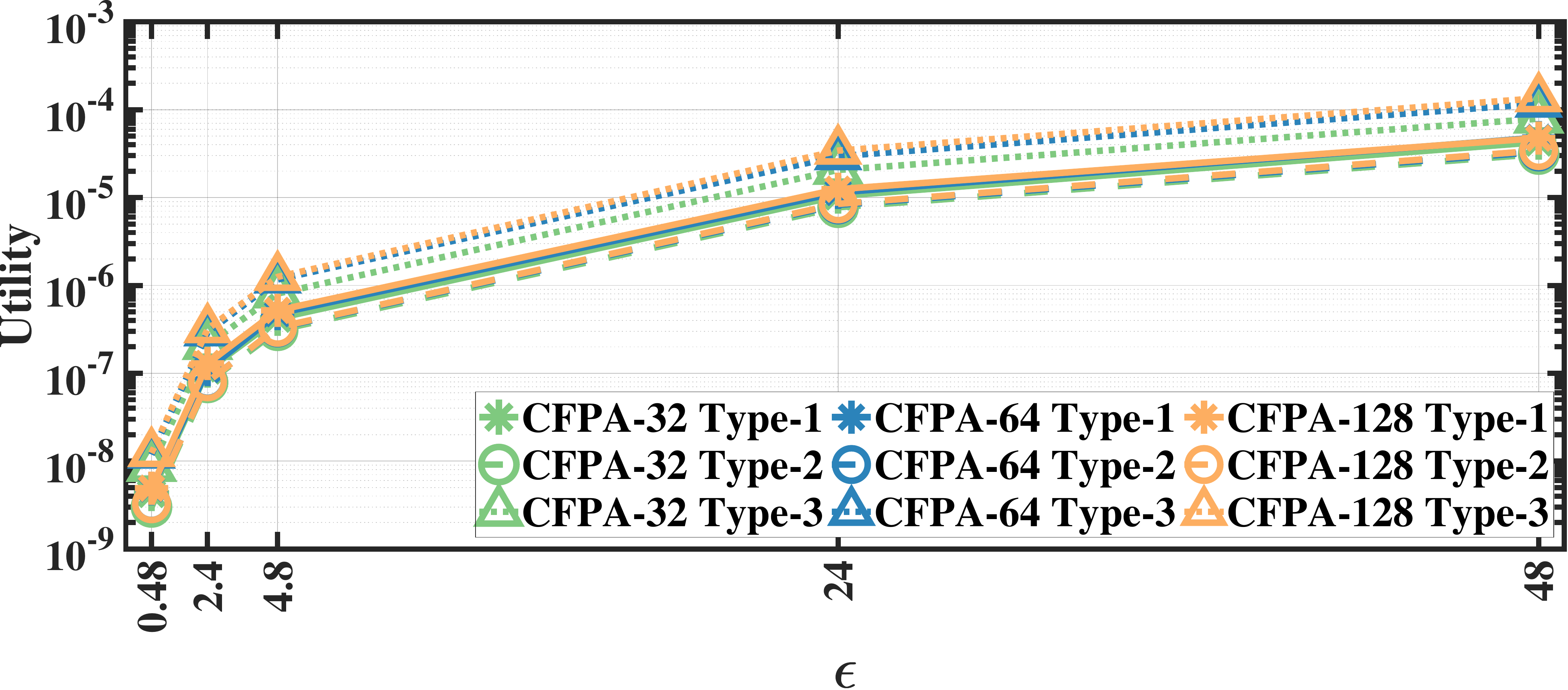}
    \caption{Utility of the CFPA for MPIIPrivacEye.}
    \label{fig:utility_mpiiprivaceye2}%
\end{figure}

\begin{figure}[!h]
   \includegraphics[width=\linewidth,keepaspectratio]{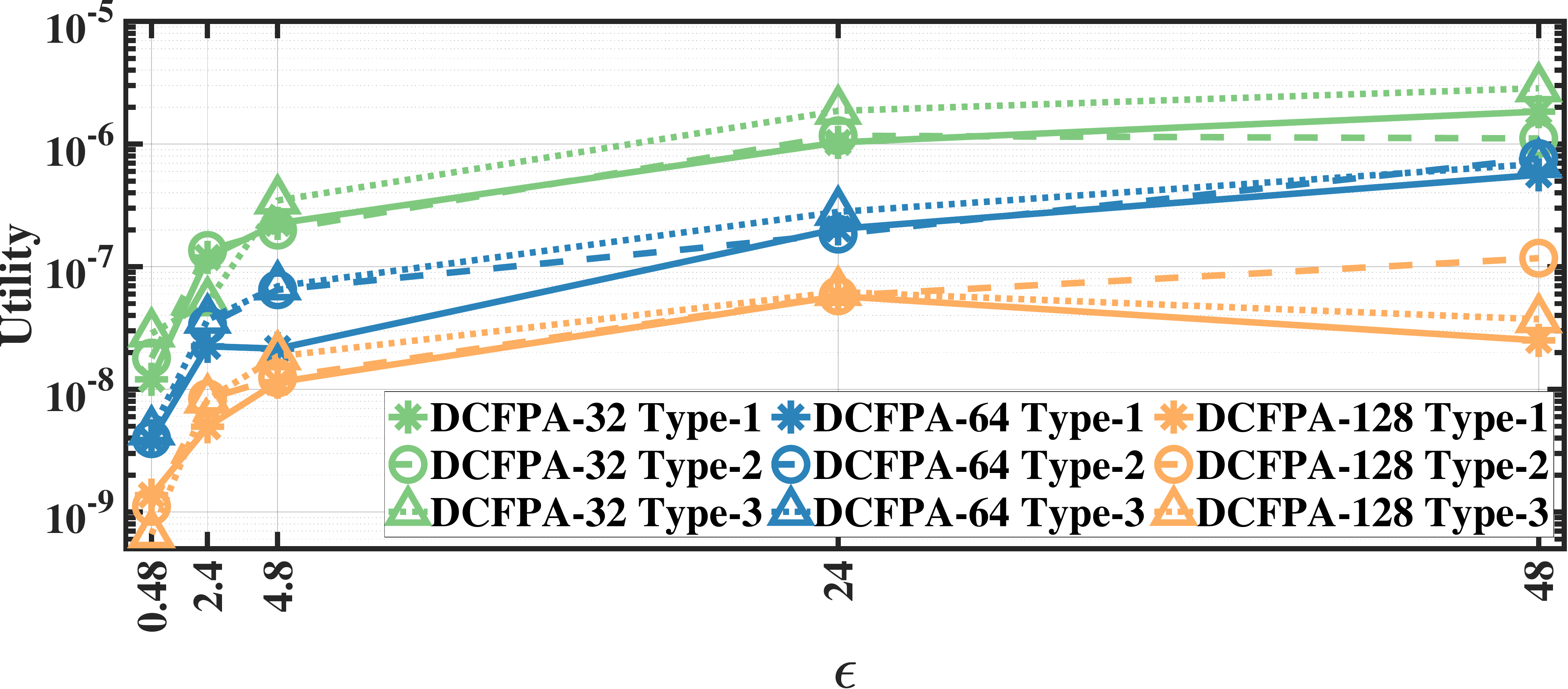}
    \caption{Utility of the DCFPA for MPIIPrivacEye.}
  \label{fig:utility_mpiiprivaceye3}%
\end{figure}

While a high absolute \acs{NMSE}, i.e., low utility, does not necessarily mean that a mechanism is completely useless, higher utility means that the mechanism would perform more effectively than low utility in various tasks. The trend in the utility results of both evaluated datasets are similar. As the query sensitivities are lower in \acs{CFPA}, utilities of \acs{CFPA} are always higher than the utilities of the \acs{FPA} as theoretically expected. The \acs{DCFPA} particularly with small chunks outperforms other methods in the most private settings, namely in the lowest $\epsilon$ regions. When different chunk sizes are compared within the \acs{CFPA} and \acs{DCFPA}, different chunk sizes perform similarly for the \acs{CFPA} method. For the \acs{DCFPA}, there is a significant trend for better utilities when the chunk sizes are decreased. However, as temporal correlations in the smaller chunk sizes higher and since a higher chunk size reduces the temporal correlations better, it is ideal to use a higher chunk size if the utilities are comparable. In general, while the \acs{LPA}, namely the standard Laplace mechanism used for differential privacy, is vulnerable to temporal correlations~\cite{8333800}, our methods also outperform it in terms of utilities. In addition to high utilities, the calculation complexities are decreased with the \acs{CFPA} and \acs{DCFPA} which is another advantage of chunk-based methods.

\subsubsection{Classification Accuracy Results}
We evaluate document type and gender classification results for the MPIIDPEye and privacy sensitivity classification results for the MPIIPrivacEye by using differentially private data generated by the methods which handle temporal correlations in the privacy context. In addition, for both datasets, we evaluate person identification tasks. While a \acs{NMSE}-based utility metric provides analytically trackable way for comparison, evaluating private data using classification accuracies give insights about the usability of the noisy data in practice. Instead of only using Support Vector Machines (SVM) as in previous works~\cite{steil_diff_privacy,Steil:2019:PPH:3314111.3319913}, we evaluate a set of classifiers including \acs{SVM}s, decision trees (DTs), random forests (RFs), and k-Nearest Neighbors (k-NNs). We employ a similar setup as in previous work~\cite{steil_diff_privacy} with radial basis function (RBF) kernel, bias parameter of $C=1$, and automatic kernel scale for the \acs{SVM}s. For \acs{RF}s and \acs{k-NN}s, we use $10$ trees and $k=11$ with a random tie breaker among tied groups, respectively. We normalize the training data to zero mean and unit variance, and apply the same parameters to the test data. Although we do not apply subsampling while generating the differentially private data, which is applied in previous work~\cite{steil_diff_privacy}, we use subsampled data for training and testing for document type, gender, and privacy sensitivity classification tasks with window sizes of $10$ and $20$ for MPIIDPEye and MPIIPrivacEye, respectively, to have a fair comparison and similar amount of data. Apart from the person identification task, all the classifiers are trained and tested in a leave-one-person-out cross-validation setup, which is considered as a more challenging but generic setup. For the person identification task, since it is not possible to carry out the experiments in a leave-one-person-out cross-validation setup, we opt for a similar configuration as in previous work~\cite{steil_diff_privacy} by using the first halves of the signals as training data and the remaining parts as test data. Such setup can be considered as one of the hypothetical best-case scenarios for an adversary as this simulates some set of prior knowledge for an adversary on participants' visual behaviors. For the person identification task, in order to use similar amount of data with other classification tasks from each signal, we use window sizes of $5$ and $10$ for MPIIDPEye and MPIIPrivacEye, respectively. For the MPIIDPEye, we evaluate results both with majority voting by summarizing classifications from different time instances for each participant and recording and without majority voting. Privacy sensitivity classification tasks for MPIIPrivacEye are carried only without majority voting since privacy sensitivity of the scene can change at each time step and applying majority voting to such task in our setup is not reasonable.

While classification results cannot be treated directly as the utility, they provide insights into the usability of the differentially private data in practice. We first evaluate document type classification task in the majority voting setting in Table~\ref{tbl_DocumentType} for MPIIDPEye dataset as it is possible to compare our results with the previous work~\cite{steil_diff_privacy}. As previous results quickly drop to the $0.33$ guessing probability in high privacy regions, we significantly outperform them particularly with \acs{DCFPA} and \acs{FPA} with the accuracies over $0.60$ and $0.85$, respectively. In the less private regions towards $\epsilon = 48$, this trend still exists with the \acs{CFPA} and \acs{FPA} with accuracy results over $0.7$ and $0.85$. Chunk-based methods perform slightly worse than the \acs{FPA} in the document type classifications even though the utility of the \acs{FPA} is lower. We observe that the reading patterns are hidden easier with chunk-based methods; therefore, document type classification task becomes more challenging. This is especially validated with \acs{DCFPA} methods using different chunk sizes, as \acs{DCFPA}-128 outperforms smaller chunk-sized \acs{DCFPA}s even though the sensitivities are higher. Therefore, we conclude that the differential privacy method should be selected for eye movements depending on the further task which will be applied. The document type classification results without majority voting are provided in the table in S1 Table.

Next, we analyze the gender classification results for MPIIDPEye. All methods are able to hide the gender information in the high privacy regions as it is already challenging to identify it with clean data as accuracies are $\approx0.7$ in previous work~\cite{steil_diff_privacy}. While we obtain similar results compared to previous work for the gender classification task, the \acs{CFPA} method is able to predict gender information correctly in the less private regions, namely $\epsilon = 48$, as it also has the highest utility values in these regions. The \acs{FPA} applied to the complete signal and the \acs{DCFPA} are not able to classify genders accurately. We observe that higher amount of noise that is needed by the \acs{FPA} and noising the fine-grained ``difference'' information between eye movement observations with \acs{DCFPA} are the reasons for hiding the gender information successfully in all privacy regions. Overall, the \acs{CFPA} provides an optimal equilibrium between gender and document type classification success in the less private regions if gender information is not considered as a feature that should be protected from adversaries. Otherwise, all proposed methods are able to hide gender information from the data in the higher privacy regions as expected. Gender classification results are depicted in Table~\ref{tbl_Gender}. Especially in some methods with \acs{k-NN}s and \acs{SVM}s, gender classification accuracies are close to zero because of the majority voting and it is validated by the results without majority voting in the table in S2 Table.

Lastly for the MPIIDPEye, we evaluate person identification task using differentially private data. The resulting classification accuracies with majority voting are depicted in Table~\ref{tbl_Person}. By using the \acs{FPA}, it is possible to identify the participants very accurately, which means that even though the document type classification accuracies of the \acs{FPA} are higher than the others, a strong adversary can also identify personal ids when this method is used. The same trend also exists in the without majority voting setting, which is reported in the table in S3 Table. The \acs{CFPA} and \acs{DCFPA} perform well against person identification attempts in the high privacy regions. However, when the \acs{CFPA} is used, it is possible to identify personal ids in the less private regions. Overall, for the MPIIDPEye dataset, the \acs{DCFPA} performs better than the others due to its resistance against person identification and gender classification and relatively high classification accuracies for the document type predictions. We conclude that this is due to the robust decorrelation effect of the \acs{DCFPA}.

For the MPIIPrivacEye, we report privacy sensitivity classification accuracies using differentially private eye movement features in the Table~\ref{tblprivacySensitivityDetection}. The \acs{FPA} performs worse than our methods. The \acs{DCFPA}, particularly with the chunk size of $32$, outperforms all other methods slightly in the higher privacy regions as it is also the case for the utility results. In the lower privacy regions, the \acs{CFPA} performs the best with $\approx0.60$ accuracy. Since performance does not drop significantly in the higher chunk sizes, it is reasonable to use higher chunk-sized methods as they decrease the temporal correlations better. While having an accuracy of approximately $0.6$ in a binary classification problem does not form the best performance, according to the previous work~\cite{Steil:2019:PPH:3314111.3319913}, privacy sensitivity classification using only eye movements with clean data in a person-independent setup only performs marginally higher than $0.60$. Therefore, we show that even though we use differentially private data in the most private settings, we obtain similar results to the classification results using clean data. This means that differentially private eye movements can be used along with scene features for detecting privacy sensitive scenes in \acs{AR} setups.

The results of the person identification task in the MPIIPrivacEye dataset are similar to the results of the MPIIDPEye dataset and the results with majority voting are depicted in Table~\ref{tblPersonIdsMPIIPrivacEye}. Personal identifiers are predicted very accurately when the \acs{FPA} is used. The \acs{CFPA} and \acs{DCFPA} are resistant to person identification attacks in all privacy regions performing around the random guess probability in almost all cases. Similar to the classification results of the MPIIDPEye dataset, the \acs{DCFPA} method performs the best when utility-privacy trade-off is taken into consideration. The person identification results without majority voting are presented in the table in S4 Table.

\subsection{Discussion}
We compared our differential privacy methods with the standard Laplace mechanism as well as the \acs{FPA} method, which is proposed for temporally correlated data, by using the MPIIDPEye and MPIIPrivacEye datasets. The utility results based on the \acs{NMSE} metric show that due to the reduced sensitivities as a result of the chunking operations, the \acs{CFPA} and \acs{DCFPA} perform better than the \acs{FPA} and standard Laplace mechanism. While larger chunk sizes applied with the \acs{CFPA} and \acs{DCFPA} decrease the effects of temporal correlations on the differential privacy mechanisms, they also increase the sensitivities, leading to higher amount of noise addition to the data and worse utility performance. Utility evaluations represent how much differentially private signals diverge from the original signals. Having eye movement feature signals less diverged from the original values by providing the differential privacy would lead better performance in various tasks. While the \acs{FPA}, \acs{CFPA}, and \acs{DCFPA} are appropriate for temporally correlated data, the \acs{DCFPA} uses the consecutive differences of eye movement feature signals, which are significantly less correlated than the original feature signals, as illustrated in Figures~\ref{fig:corr_coeffs_diffs} and~\ref{fig:corr_coeffs_mpiiprivaceye_diffs}. Thus, the \acs{DCFPA} is less vulnerable to temporal correlations in the differential privacy context.

In addition to utility results, we evaluated document type, gender, and person identification tasks for the MPIIDPEye dataset and privacy sensitivity classification of the observed scene and person identification task for the MPIIPrivacEye dataset and compared our results with the previous works especially in the eye tracking literature. The \acs{FPA} outperforms the \acs{CFPA} and \acs{DCFPA} in document type classification task since the chunks perturb ``Z''-type reading patterns. However, this might be a task-specific outcome as the \acs{CFPA} and \acs{DCFPA} perform better in terms of utility. In addition, when the \acs{FPA} is used, personal identifiers can be estimated with high accuracies in both datasets. On the contrary, especially the \acs{DCFPA} provides decreased probabilities for person identification tasks in the MPIIDPEye, and probabilities close to the random guess probability for the MPIIPrivacEye, which are optimal from a privacy-preservation perspective. We remark that this outcome is also related to decreased temporal correlations. Gender information is successfully hidden in all methods and scene privacy can be predicted to some extent using differentially private eye movement signals. In addition, privacy sensitivity detection results on the MPIIPrivacEye are consistent with the utility results based on the \acs{NMSE} metric.

Due to the significant reduction of temporal correlations and high utility and relatively accurate classification results in different tasks, the \acs{DCFPA} is the best performing differential privacy method for eye movement feature signals. In addition, it is not possible to recognize the person accurately from eye movement data when the \acs{DCFPA} is used. From correlation reduction point of view, in both methods namely, \acs{CFPA} and \acs{DCFPA}, when the performances are similar, it is reasonable to use higher chunk sizes as such chunks are less vulnerable to temporal correlations as illustrated in Figures~\ref{fig:corr_coeffs_raw} and~\ref{fig:corr_coeffs_mpiiprivaceye_raw}. Overall, our methods outperform the state-of-the-art for differential privacy for aggregated eye movement feature signals.

\subsection{Conclusion}
We proposed different methods to achieve differential privacy for eye movement feature signals by correcting, extending, and adapting the \acs{FPA} method. Since eye movement features are correlated over time and are high dimensional, standard differential privacy methods provide low utility and are vulnerable to inference attacks. Thus, we proposed privacy solutions for temporally correlated eye movement data. Our methods can be easily applied to other biometric human-computer interaction data as well since they are independent of the used data and outperform the state-of-the-art methods in terms of both \acs{NMSE} and classification accuracy and reduce the correlations significantly. In future work, we will analyze the actual privacy metric $\epsilon^{\prime}$ which takes the data correlations into account and choose $k$ values in a private manner for the centralized differential privacy setting.

\begin{landscape}
\begin{table}[!ht]
\centering
\caption{Document type classification accuracies in the MPIIDPEye dataset using differentially private eye movement features with majority voting.}
\begin{tabular}{cccccc}
    & \multicolumn{5}{c}{Document type classification accuracies (k-NN$|$SVM$|$DT$|$RF)} \\
    \hline
   Method & $\epsilon = 0.48$ & $\epsilon = 2.4$ & $\epsilon = 4.8$ & $\epsilon = 24$ & $\epsilon = 48$ \\
    \hline
    FPA & $0.50|0.63|0.82|\mathbf{0.87}$ & $0.51|0.63|0.81|\mathbf{0.87}$ & $0.5|0.61|0.81|\mathbf{0.87}$ & $0.52|0.63|0.82|\mathbf{0.87}$ & $0.52|0.64|0.83|\mathbf{0.88}$\\
    CFPA-32 & $0.39|0.37|0.45|0.44$ & $0.40|0.38|0.45|0.44$ & $0.40|0.44|0.46|0.44$ & $0.58|0.58|0.55|0.60$ & $\mathbf{0.71}|0.69|0.66|0.66$ \\
    CFPA-64 & $0.41|0.36|0.45|0.45$ & $0.40|0.37|0.44|0.45$ & $0.40|0.41|0.44|0.45$ & $0.57|0.59|0.55|0.59$ & $0.70|0.70|0.66|0.66$ \\
    CFPA-128 & $0.36|0.33|0.45|0.45$ & $0.36|0.33|0.44|0.44$ & $0.37|0.35|0.44|0.45$ & $0.52|0.56|0.52|0.57$ & $0.69|0.68|0.64|0.66$ \\
    DCFPA-32 & $0.51|0.37|0.46|0.44$ & $0.51|0.36|0.47|0.42$ & $0.47|0.35|0.47|0.43$ & $0.49|0.37|0.46|0.44$ & $0.48|0.36|0.47|0.45$ \\
    DCFPA-64 & $0.61|0.45|0.43|0.41$ & $0.55|0.35|0.43|0.41$ & $0.56|0.41|0.43|0.41$ & $\mathbf{0.60}|0.43|0.45|0.42$ & $0.59|0.40|0.44|0.43$ \\
    DCFPA-128 & $\mathbf{0.64}|0.45|0.46|0.48$ & $\mathbf{0.62}|0.42|0.45|0.46$ & $\mathbf{0.69}|0.50|0.44|0.46$ & $0.57|0.45|0.45|0.46$ & $0.60|0.42|0.45|0.46$ \\
    \hline
\end{tabular}
\label{tbl_DocumentType}
\end{table}
\vspace{3em}
\begin{table}[!ht]
\centering
\caption{Gender classification accuracies in the MPIIDPEye dataset using differentially private eye movement features with majority voting.}
\begin{tabular}{cccccc}
    & \multicolumn{5}{c}{Gender classification accuracies (k-NN$|$SVM$|$DT$|$RF)} \\
    \hline
   Method & $\epsilon = 0.48$ & $\epsilon = 2.4$ & $\epsilon = 4.8$ & $\epsilon = 24$ & $\epsilon = 48$ \\
    \hline
    FPA & $0.44|0.30|0.43|0.38$ & $0.45|0.30|0.41|0.37$ & $0.44|0.28|0.41|0.39$ & $0.43|0.27|0.43|0.38$ & $0.44|0.31|0.42|0.39$\\
    CFPA-32 & $0.04|0.01|0.26|0.24$ & $0.05|0.01|0.27|0.25$ & $0.05|0.02|0.28|0.27$ & $0.36|0.30|0.50|0.45$ & $0.62|0.50|0.67|0.53$ \\
    CFPA-64 & $0.08|0.05|0.27|0.26$ & $0.08|0.04|0.28|0.27$ & $0.10|0.06|0.31|0.27$ & $0.38|0.34|0.52|0.47$ & $0.62|0.51|0.68|0.54$ \\
    CFPA-128 & $0.18|0.15|0.32|0.30$ & $0.16|0.12|0.31|0.30$ & $0.18|0.10|0.32|0.31$ & $0.36|0.30|0.50|0.46$ & $0.60|0.47|0.68|0.54$ \\
    DCFPA-32 & $0.03|\approx0|0.22|0.31$ & $0.04|\approx0|0.23|0.32$ & $0.04|\approx0|0.22|0.32$ & $0.04|\approx0|0.23|0.31$ & $0.04|\approx0|0.23|0.32$ \\
    DCFPA-64 & $0.04|\approx0|0.30|0.33$ & $0.04|\approx0|0.30|0.34$ & $0.04|\approx0|0.30|0.32$ & $0.04|\approx0|0.29|0.34$ & $0.03|\approx0|0.30|0.34$ \\
    DCFPA-128 & $0.09|0.01|0.34|0.35$ & $0.08|\approx0|0.32|0.34$ & $0.08|0.01|0.32|0.35$ & $0.07|\approx0|0.33|0.34$ & $0.07|0.01|0.34|0.34$ \\
    \hline
\end{tabular}
\label{tbl_Gender}
\end{table}
\end{landscape}

\begin{landscape}
\begin{table}[!ht]
\centering
\caption{Person identification accuracies in the MPIIDPEye dataset using differentially private eye movement features with majority voting.}
\begin{tabular}{cccccc}
    & \multicolumn{5}{c}{Person identification accuracies (k-NN$|$SVM$|$DT$|$RF)} \\
    \hline
   Method & $\epsilon = 0.48$ & $\epsilon = 2.4$ & $\epsilon = 4.8$ & $\epsilon = 24$ & $\epsilon = 48$ \\
    \hline
    FPA & $1$ & $1$ & $1$ & $1$ & $1$\\
    CFPA-32 & $0.15|0.08|0.44|0.37$ & $0.13|0.08|0.46|0.39$ & $0.11|0.08|0.48|0.41$ & $0.40|0.11|0.64|0.70$ & $0.72|0.16|0.87|0.92$ \\
    CFPA-64 & $0.14|0.08|0.42|0.34$ & $0.13|0.08|0.44|0.37$ & $0.12|0.08|0.45|0.38$ & $0.39|0.11|0.63|0.71$ & $0.70|0.17|0.85|0.92$ \\
    CFPA-128 & $0.16|0.05|0.39|0.36$ & $0.15|0.05|0.41|0.36$ & $0.17|0.05|0.43|0.39$ & $0.45|0.07|0.55|0.63$ & $0.70|0.16|0.74|0.88$ \\
    DCFPA-32 & $0.06|0.10|0.39|0.37$ & $0.06|0.10|0.39|0.36$ & $0.08|0.10|0.39|0.36$ & $0.10|0.10|0.39|0.37$ & $0.10|0.10|0.40|0.38$ \\
    DCFPA-64 & $0.10|0.10|0.33|0.35$ & $0.10|0.10|0.32|0.34$ & $0.10|0.10|0.32|0.33$ & $0.13|0.10|0.31|0.34$ & $0.13|0.10|0.32|0.33$ \\
    DCFPA-128 & $0.09|0.05|0.24|0.28$ & $0.09|0.05|0.25|0.27$ & $0.10|0.05|0.23|0.27$ & $0.10|0.06|0.24|0.26$ & $0.10|0.05|0.22|0.25$ \\
    \hline
\end{tabular}
\label{tbl_Person}
\end{table}
\vspace{3em}
\begin{table}[!ht]
\centering
\caption{Privacy sensitivity classification accuracies in the MPIIPrivacEye dataset using differentially private eye movement features.}
\begin{tabular}{cccccc}
    & \multicolumn{5}{c}{Privacy sensitivity classification accuracies (k-NN$|$SVM$|$DT$|$RF)} \\
    \hline
   Method & $\epsilon = 0.48$ & $\epsilon = 2.4$ & $\epsilon = 4.8$ & $\epsilon = 24$ & $\epsilon = 48$ \\
    \hline
    FPA & $0.49|0.58|0.51|0.55$ & $0.49|0.58|0.51|0.55$ & $0.49|0.58|0.51|0.55$ & $0.50|0.58|0.51|0.55$ & $0.50|0.59|0.51|0.55$\\
    CFPA-32 & $0.55|0.59|0.52|0.56$ & $0.55|0.58|0.52|0.56$ & $0.55|0.58|0.52|0.56$ & $0.56|0.58|0.53|0.57$ & $0.58|\mathbf{0.60}|0.54|0.58$ \\
    CFPA-64 & $0.55|0.58|0.52|0.56$ & $0.55|0.58|0.52|0.56$ & $0.55|0.58|0.52|0.56$ & $0.56|0.58|0.53|0.57$ & $0.58|0.59|0.54|0.58$ \\
    CFPA-128 & $0.55|0.57|0.52|0.56$ & $0.55|0.57|0.52|0.56$ & $0.55|0.57|0.52|0.56$ & $0.56|0.58|0.53|0.57$ & $0.58|0.59|0.54|0.59$ \\
    DCFPA-32 & $0.54|\mathbf{0.59}|0.52|0.56$ & $0.55|\mathbf{0.59}|0.52|0.56$ & $0.55|\mathbf{0.59}|0.52|0.56$ & $0.54|\mathbf{0.59}|0.52|0.56$ & $0.55|0.59|0.52|0.56$ \\
    DCFPA-64 & $0.54|0.58|0.52|0.56$ & $0.54|0.58|0.52|0.56$ & $0.54|0.58|0.52|0.56$ & $0.54|0.58|0.52|0.56$ & $0.54|0.58|0.52|0.56$ \\
    DCFPA-128 & $0.54|0.57|0.52|0.56$ & $0.54|0.57|0.52|0.56$ & $0.54|0.57|0.52|0.56$ & $0.54|0.57|0.52|0.56$ & $0.54|0.57|0.52|0.56$ \\
    \hline
\end{tabular}
\label{tblprivacySensitivityDetection}
\end{table}
\end{landscape}

\begin{landscape}
\begin{table}[!ht]
\centering
\caption{Person identification classification accuracies in the MPIIPrivacEye dataset using differentially private eye movement features with majority voting.}
\begin{tabular}{cccccc}
    & \multicolumn{5}{c}{Person identification classification accuracies (k-NN$|$SVM$|$DT$|$RF)} \\
    \hline
   Method & $\epsilon = 0.48$ & $\epsilon = 2.4$ & $\epsilon = 4.8$ & $\epsilon = 24$ & $\epsilon = 48$ \\
    \hline
    FPA & $1$ & $1$ & $1$ & $1$ & $1$\\
    CFPA-32 & $0.05|0.06|0.07|0.07$ & $0.05|0.06|0.07|0.07$ & $0.06|0.06|0.08|0.07$ & $0.07|0.06|0.09|0.11$ & $0.11|0.06|0.14|0.16$ \\
    CFPA-64 & $0.06|0.06|0.06|0.07$ & $0.06|0.06|0.06|0.06$ & $0.06|0.06|0.07|0.07$ & $0.07|0.06|0.09|0.09$ & $0.11|0.06|0.16|0.16$ \\
    CFPA-128 & $0.06|0.06|0.07|0.07$ & $0.06|0.06|0.07|0.07$ & $0.06|0.06|0.07|0.08$ & $0.07|0.06|0.09|0.11$ & $0.11|0.06|0.15|0.15$ \\
    DCFPA-32 & $0.06|0.05|0.08|0.07$ & $0.06|0.06|0.07|0.08$ & $0.07|0.05|0.08|0.08$ & $0.07|0.05|0.08|0.08$ & $0.07|0.06|0.08|0.08$ \\
    DCFPA-64 & $0.06|0.05|0.06|0.06$ & $0.06|0.05|0.06|0.06$ & $0.06|0.05|0.06|0.06$ & $0.05|0.06|0.06|0.06$ & $0.06|0.05|0.06|0.06$ \\
    DCFPA-128 & $0.05|0.05|0.06|0.06$ & $0.05|0.05|0.05|0.06$ & $0.06|0.05|0.06|0.06$ & $0.05|0.05|0.05|0.06$ & $0.06|0.05|0.05|0.06$ \\
    \hline
\end{tabular}
\label{tblPersonIdsMPIIPrivacEye}
\end{table}
\end{landscape}

\begin{landscape}
\subsection{Supporting Information}

\paragraph*{S1 Table.~} 
\label{lbl:S1_Table}
{\bf ~Document type classification results without majority voting for the MPIIDPEye dataset.}
\begin{table}[!ht]
\centering
\caption{Document type classification accuracies in the MPIIDPEye dataset using differentially private eye movement features without majority voting.}
\begin{tabular}{cccccc}
    & \multicolumn{5}{c}{Document type classification accuracies (k-NN$|$SVM$|$DT$|$RF)} \\
    \hline
   Method & $\epsilon = 0.48$ & $\epsilon = 2.4$ & $\epsilon = 4.8$ & $\epsilon = 24$ & $\epsilon = 48$ \\
    \hline
    FPA & $0.46|0.52|0.68|0.73$ & $0.46|0.52|0.67|0.73$ & $0.45|0.51|0.67|0.73$ & $0.46|0.52|0.68|0.73$ & $0.47|0.52|0.68|0.74$\\
    CFPA-32 & $0.34|0.35|0.36|0.38$ & $0.34|0.35|0.36|0.38$ & $0.34|0.36|0.36|0.38$ & $0.39|0.44|0.38|0.42$ & $0.47|0.53|0.44|0.49$ \\
    CFPA-64 & $0.34|0.35|0.36|0.38$ & $0.34|0.35|0.36|0.38$ & $0.34|0.36|0.36|0.38$ & $0.39|0.44|0.38|0.42$ & $0.47|0.53|0.44|0.49$ \\
    CFPA-128 & $0.34|0.34|0.36|0.39$ & $0.34|0.34|0.36|0.39$ & $0.34|0.34|0.36|0.39$ & $0.38|0.42|0.37|0.42$ & $0.46|0.51|0.43|0.48$ \\
    DCFPA-32 & $0.36|0.35|0.36|0.37$ & $0.36|0.34|0.35|0.37$ & $0.35|0.34|0.36|0.37$ & $0.36|0.35|0.35|0.37$ & $0.35|0.34|0.36|0.38$ \\
    DCFPA-64 & $0.38|0.37|0.35|0.37$ & $0.37|0.35|0.35|0.37$ & $0.37|0.36|0.35|0.37$ & $0.37|0.36|0.35|0.37$ & $0.37|0.36|0.35|0.37$ \\
    DCFPA-128 & $0.40|0.38|0.36|0.38$ & $0.39|0.37|0.35|0.37$ & $0.41|0.39|0.35|0.38$ & $0.38|0.37|0.35|0.38$ & $0.39|0.37|0.35|0.38$ \\
    \hline
\end{tabular}
\label{tbl_DocumentTypeWOMaj}
\end{table}

\vspace{-2em}
\paragraph*{S2 Table.~}
\label{lbl:S2_Table}
{\bf ~Gender classification results without majority voting for the MPIIDPEye dataset.}
\begin{table}[!ht]
\centering
\caption{Gender classification accuracies in the MPIIDPEye dataset using differentially private eye movement features without majority voting.}
\begin{tabular}{cccccc}
    & \multicolumn{5}{c}{Gender classification accuracies (k-NN$|$SVM$|$DT$|$RF)} \\
    \hline
   Method & $\epsilon = 0.48$ & $\epsilon = 2.4$ & $\epsilon = 4.8$ & $\epsilon = 24$ & $\epsilon = 48$ \\
    \hline
    FPA & $0.48|0.42|0.48|0.45$ & $0.48|0.42|0.47|0.44$ & $0.47|0.41|0.47|0.45$ & $0.47|0.41|0.48|0.44$ & $0.48|0.43|0.48|0.45$\\
    CFPA-32 & $0.43|0.31|0.44|0.40$ & $0.43|0.31|0.45|0.41$ & $0.43|0.32|0.46|0.41$ & $0.46|0.42|0.49|0.47$ & $0.51|0.47|0.53|0.53$ \\
    CFPA-64 & $0.44|0.35|0.45|0.40$ & $0.44|0.35|0.45|0.41$ & $0.44|0.35|0.46|0.42$ & $0.46|0.43|0.49|0.47$ & $0.51|0.48|0.54|0.53$ \\
    CFPA-128 & $0.45|0.39|0.46|0.42$ & $0.45|0.38|0.46|0.42$ & $0.45|0.38|0.46|0.42$ & $0.46|0.43|0.49|0.47$ & $0.51|0.47|0.53|0.53$ \\
    DCFPA-32 & $0.44|0.27|0.45|0.42$ & $0.44|0.27|0.45|0.42$ & $0.44|0.27|0.45|0.42$ & $0.44|0.27|0.45|0.42$ & $0.44|0.27|0.46|0.42$ \\
    DCFPA-64 & $0.44|0.30|0.46|0.43$ & $0.43|0.29|0.46|0.43$ & $0.44|0.30|0.46|0.43$ & $0.43|0.30|0.46|0.43$ & $0.43|0.30|0.46|0.43$ \\
    DCFPA-128 & $0.44|0.32|0.46|0.43$ & $0.44|0.32|0.46|0.43$ & $0.44|0.32|0.47|0.43$ & $0.44|0.31|0.46|0.43$ & $0.44|0.32|0.47|0.43$ \\
    \hline
\end{tabular}

\label{tbl_GenderWOMaj}
\end{table}
\end{landscape}

\newpage

\begin{landscape}
\paragraph*{S3 Table.~}
\label{lbl:S3_Table}
{\bf ~Person identification results without majority voting for the MPIIDPEye dataset.}
\begin{table}[!ht]
\centering
\caption{Person identification accuracies in the MPIIDPEye dataset using differentially private eye movement features without majority voting.}
\begin{tabular}{cccccc}
    & \multicolumn{5}{c}{Person identification accuracies (k-NN$|$SVM$|$DT$|$RF)} \\
    \hline
   Method & $\epsilon = 0.48$ & $\epsilon = 2.4$ & $\epsilon = 4.8$ & $\epsilon = 24$ & $\epsilon = 48$ \\
    \hline
    FPA & $1|1|0.98|1$ & $1|1|0.98|1$ & $1|1|0.98|1$ & $1|1|0.97|1$ & $1|1|0.95|1$\\
    CFPA-32 & $0.09|0.11|0.16|0.16$ & $0.08|0.11|0.16|0.17$ & $0.09|0.11|0.17|0.17$ & $0.12|0.15|0.18|0.21$ & $0.18|0.21|0.23|0.27$ \\
    CFPA-64 & $0.09|0.11|0.17|0.17$ & $0.09|0.11|0.17|0.17$ & $0.09|0.11|0.17|0.17$ & $0.12|0.15|0.19|0.21$ & $0.17|0.21|0.23|0.27$ \\
    CFPA-128 & $0.11|0.13|0.17|0.18$ & $0.11|0.13|0.17|0.18$ & $0.11|0.13|0.17|0.18$ & $0.13|0.16|0.18|0.20$ & $0.18|0.20|0.21|0.25$ \\
    DCFPA-32 & $0.09|0.10|0.15|0.16$ & $0.09|0.11|0.14|0.16$ & $0.09|0.11|0.14|0.16$ & $0.09|0.11|0.14|0.16$ & $0.09|0.11|0.15|0.16$ \\
    DCFPA-64 & $0.09|0.10|0.13|0.15$ & $0.09|0.10|0.13|0.15$ & $0.09|0.10|0.13|0.15$ & $0.09|0.10|0.13|0.15$ & $0.09|0.10|0.13|0.15$ \\
    DCFPA-128 & $0.08|0.09|0.12|0.13$ & $0.08|0.09|0.11|0.13$ & $0.08|0.09|0.11|0.13$ & $0.08|0.09|0.12|0.13$ & $0.08|0.09|0.11|0.13$ \\
    \hline
\end{tabular}
\label{tbl_PersonWOMaj}
\end{table}

\vspace{-2em}
\paragraph*{S4 Table.~} 
\label{lbl:S4_Table}
{\bf ~Person identification results without majority voting for MPIIPrivacEye dataset.}
\begin{table}[!ht]
\centering
\caption{
Person identification classification accuracies in the MPIIPrivacEye dataset using differentially private eye movement features without majority voting.}
\begin{tabular}{cccccc}
    & \multicolumn{5}{c}{Person identification classification accuracies (k-NN$|$SVM$|$DT$|$RF)} \\
    \hline
   Method & $\epsilon = 0.48$ & $\epsilon = 2.4$ & $\epsilon = 4.8$ & $\epsilon = 24$ & $\epsilon = 48$ \\
    \hline
    FPA & $1|1|0.99|1$ & $1|1|0.99|1$ & $1|1|0.99|1$ & $1|1|0.97|1$ & $1|1|0.94|1$\\
    CFPA-32 & $0.06|0.07|0.06|0.06$ & $0.06|0.07|0.06|0.06$ & $0.06|0.07|0.06|0.06$ & $0.07|0.09|0.07|0.07$ & $0.08|0.10|0.08|0.08$ \\
    CFPA-64 & $0.06|0.07|0.06|0.06$ & $0.06|0.07|0.06|0.06$ & $0.06|0.07|0.06|0.06$ & $0.07|0.09|0.07|0.07$ & $0.08|0.10|0.07|0.08$ \\
    CFPA-128 & $0.06|0.08|0.06|0.07$ & $0.06|0.08|0.06|0.07$ & $0.06|0.08|0.06|0.07$ & $0.07|0.09|0.07|0.08$ & $0.08|0.11|0.07|0.08$ \\
    DCFPA-32 & $0.06|0.07|0.06|0.06$ & $0.06|0.07|0.06|0.06$ & $0.06|0.07|0.06|0.06$ & $0.06|0.07|0.06|0.06$ & $0.06|0.07|0.06|0.06$ \\
    DCFPA-64 & $0.06|0.06|0.06|0.06$ & $0.06|0.06|0.06|0.06$ & $0.06|0.06|0.06|0.06$ & $0.06|0.06|0.06|0.06$ & $0.06|0.07|0.06|0.06$ \\
    DCFPA-128 & $0.06|0.06|0.06|0.06$ & $0.06|0.06|0.06|0.06$ & $0.06|0.06|0.06|0.06$ & $0.06|0.06|0.06|0.06$ & $0.06|0.06|0.06|0.06$ \\
    \hline
\end{tabular}
\label{tblPersonIdsMPIIPrivacEyeWOMaj}
\end{table}
\end{landscape}
\newpage

\subsection*{Acknowledgments}
O. G\"unl\"u thanks Ravi Tandon for his useful suggestions. E. Bozkir thanks Martin Pawelczyk and Mete Akg\"un for useful discussions.

\subsection*{Author Contributions}
\textbf{Conceptualization:} Efe Bozkir, Onur Günlü, Wolfgang Fuhl, Rafael F. Schaefer, Enkelejda Kasneci.\\
\textbf{Data curation:} Efe Bozkir, Onur Günlü.\\
\textbf{Formal analysis:} Efe Bozkir, Onur Günlü.\\
\textbf{Investigation:} Efe Bozkir, Onur Günlü.\\
\textbf{Methodology:} Efe Bozkir, Onur Günlü.\\
\textbf{Software:} Efe Bozkir, Onur Günlü.\\
\textbf{Supervision:} Efe Bozkir, Onur Günlü, Wolfgang Fuhl, Rafael F. Schaefer, Enkelejda Kasneci.\\
\textbf{Validation:} Efe Bozkir, Onur Günlü.\\
\textbf{Visualization:} Efe Bozkir, Onur Günlü.\\
\textbf{Writing -- original draft:} Efe Bozkir, Onur Günlü.\\
\textbf{Writing -- review \& editing:} Efe Bozkir, Onur Günlü, Wolfgang Fuhl, Rafael F. Schaefer, Enkelejda Kasneci.

\subsection*{Peer Review History}
PLOS recognizes the benefits of transparency in the peer review process; therefore, we enable the publication of all of the content of peer review and author responses alongside final, published articles. The editorial history of this article is available here: \url{https://doi.org/10.1371/journal.pone.0255979}.

\subsection*{Data Availability Statement}
Relevant data files are provided via following url: \url{https://atreus.informatik.uni-tuebingen.de/bozkir/dp_eye_tracking}.

\subsection*{Funding}
O. G\"unl\"u and R. F. Schaefer are supported by the German Federal Ministry of Education and Research (BMBF) within the national initiative for ``Post Shannon Communication (NewCom)'' under the Grant 16KIS1004. We acknowledge support by Open Access Publishing Fund of University of T\"ubingen. The funders had no role in study design, data collection and analysis, decision to publish, or preparation of the manuscript.

\subsection*{Competing Interests}
The authors have declared that no competing interests exist.

\newpage

\section[Privacy Preserving Gaze Estimation using Synthetic Images via a Randomized Encoding Based Framework]{Privacy Preserving Gaze Estimation using Synthetic Images via a Randomized Encoding Based Framework}
\label{appendix:B2}

\subsection{Abstract}
Eye tracking is handled as one of the key technologies for applications that assess and evaluate human attention, behavior, and biometrics, especially using gaze, pupillary, and blink behaviors. One of the challenges with regard to the social acceptance of eye tracking technology is however the preserving of sensitive and personal information. To tackle this challenge, we employ a privacy-preserving framework based on randomized encoding to train a Support Vector Regression model using synthetic eye images privately to estimate the human gaze. During the computation, none of the parties learn about the data or the result that any other party has. Furthermore, the party that trains the model cannot reconstruct pupil, blinks or visual scanpath. The experimental results show that our privacy-preserving framework is capable of working in real-time, with the same accuracy as compared to non-private version and could be extended to other eye tracking related problems.

\subsection{Introduction}
Recent advances in the fields of Head-Mounted-Display (HMD) technology, computer graphics, augmented reality (AR), and eye tracking enabled numerous novel applications. One of the most natural and non-intrusive ways of interaction with \acs{HMD}s or smart glasses is achieved by gaze-aware interfaces using eye tracking. However, it is possible to derive a lot of sensitive and personal information from eye tracking data such as intentions, behaviors, or fatigue since eyes are not fully controlled in a conscious way.

It has been shown that cognitive load \cite{Chen2014UsingTP,Appel:2018}, visual attention \cite{bozkir_vr_attention_et}, stress \cite{kubler2014stress}, task identification \cite{Borji2014}, skill level assessment and expertise \cite{Liu2009,eivazi2017optimal,castner2018scanpath}, human activities \cite{Steil:2015:DEH:2750858.2807520,braunagel2017online}, biometric information and authentication \cite{Kinnunen:2010:TTP:1743666.1743712,Komogortsev2010,6712725,Zhang:2018:CAU:3178157.3161410,Abdrabou:2019:JGW:3314111.3319837}, or personality traits \cite{Berkovsky:2019:DPT:3290605.3300451} can be obtained using eye tracking data. Since highly sensitive information can be derived from eye tracking data, it is not surprising that \acs{HMD}s or smart glasses have not been adopted by large communities yet. According to a recent survey \cite{steil_diff_privacy}, people agree to share their eye tracking data only when it is co-owned by a governmental health-agency or is used for research purposes. This indicates that people are hesitant about sharing their eye tracking data in commercial applications. Therefore, there is a likelihood that larger communities could adopt \acs{HMD}s or smart glasses if privacy-preserving techniques are applied in the eye tracking applications. The reasons why privacy preserving schemes are needed for eye tracking are discussed in \cite{Liebling2014} extensively. However, until now, there are not many studies in privacy-preserving eye tracking. Recently, a method to detect privacy sensitive everyday situations \cite{Steil:2019:PPH:3314111.3319913}, an approach to degrade iris authentication while keeping the gaze tracking utility in an acceptable accuracy \cite{John:2019:EDI:3314111.3319816}, and differential privacy based techniques to protect personal information on heatmaps and eye movements \cite{Liu2019,steil_diff_privacy} are introduced. While differential privacy can be applied to eye tracking data for various tasks, it introduces additional noise on the data which causes decrease in the utility \cite{Liu2019,steil_diff_privacy}, and it might lead to less accurate results in computer vision tasks, such as gaze estimation or activity recognition.

In light of the above, function-specific privacy models are required. In this work, we focus on the gaze estimation problem as a proof-of-concept by using synthetic data including eye landmarks and ground truth gaze vectors. However, the same privacy-preserving approach can be extended to any feature-based, eye tracking problem such as intention, fatigue, or activity detection, in \acs{HMD} or unconstrained setups due to the demonstrated real-time working capabilities. In our study, the gaze estimation task is solved by using Support Vector Regression (SVR) models in a privacy-preserving manner by computing the dot product of eye landmark vectors to obtain the kernel matrix of the \acs{SVR} for a scenario, where two parties have the eye landmark data, each of which we call {\it input-party}, and one {\it function-party} that trains a prediction model on the data of the input-parties. This scenario is relevant when the input-parties use eye tracking data to improve the accuracy of their models and do not share the data due to the privacy concerns. To this end, we utilize a framework employing randomized encoding \cite{unal2019framework}. In the computation, neither the eye images nor the extracted features are revealed to the function-party directly. Furthermore, the input-parties do not infer the raw eye tracking data or result of the computation. Eye images that are used for training and testing are rendered using UnityEyes \cite{wood2016_etra} synthetically and 36 landmark-based features \cite{Park2018} are used. To the best of our knowledge, this is the first work that applies a privacy-preserving scheme based on function-specific privacy models on an eye tracking problem.

\subsection{Threat Model}
We assume that the input-parties are semi-honest (honest but curious) that are not allowed to deviate from the protocol description while they try to infer some valuable information about other parties' private inputs using their views of the protocol execution. We also assume that the function-party is malicious and the input-parties and the function-party do not collude. 

\subsection{Methodology}
In this section, we discuss the data generation, randomized encoding, and privacy-preserving gaze estimation framework.

\subsubsection{Data Generation}
To train and evaluate the gaze estimator, we generate eye images and gaze vectors. As our work is a proof-of-concept and requires high amount of data, synthetic images from UnityEyes \cite{wood2016_etra}, which is based on the Unity3D, are used. \textit{Camera parameters} and \textit{Eye parameters} are chosen as $(0,0,0,0)$ (fixed camera) and $(0,0,30,30)$ (eyeball pose range parameters in degrees), respectively. $20,000$ images are rendered in \textit{Fantastic} quality setting and $512 \times 384$ screen resolution. Then, processing and normalization pipeline from \cite{Park2018} is employed. In the end, we obtain $128 \times 96$ sized eye images, 18 eye landmarks including eight iris edge, eight eyelid, one iris center, and one iris-center-eyeball-center vector normalized according to Euclidean distance between eye corners, and gaze vectors using pitch and yaw angles. Final feature vectors consist of $36$ elements. Figure~\ref{fig:eyeImages_ETRA20} shows an example illustration.

\begin{figure}[ht]
  \centering
   \subfigure[Landmarks.]{{\includegraphics[height = 3.75cm]{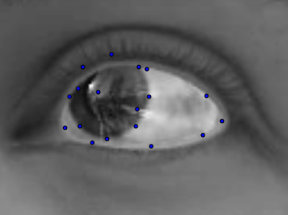}}}
   \qquad
   \qquad
   \qquad
   \subfigure[Gaze.]{{\includegraphics[height = 3.75cm]{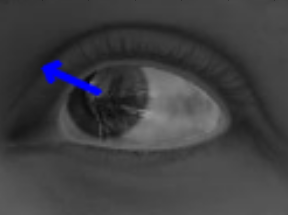} }}%
  \caption{Eye landmarks and gaze on a synthetic image.}
  \label{fig:eyeImages_ETRA20}
\end{figure}

\subsubsection{Randomized Encoding}
The utilized framework employs randomized encoding (RE) \cite{applebaum2006cryptography,applebaum2006computationally} to compute the  dot product of the landmark vectors. The dot product is needed to compute kernel matrix of the \acs{SVR} which is later used for training the gaze estimator and validation of the framework.

In the randomized encoding, the computation of a function $f(x)$ is performed by a randomized function $\hat{f}(x;r)$ where $x$ is the input value, which corresponds to eye landmarks in our setup, and $r$ is the random value. The idea is to encode the original function by using random value(s) such that the combination of the components of the encoding reveals only the output of the original function. In the framework, the computation of the dot product is accomplished by utilizing the decomposable and affine randomized encoding (DARE) of addition and multiplication~\cite{applebaum2017garbled}. The encoding of multiplication is as follows.

\begin{definition}[Perfect \acs{RE} for Multiplication \cite{applebaum2017garbled}]
\label{def:remul}
A multiplication function is defined as $f_m(x_1,x_2)=x_1 \cdot x_2$ over a ring $\mathsf{R}$. One can perfectly encode the $f_m$ by employing the \acs{DARE} $\hat{f_m}(x_1,x_2;r_1,r_2,r_3)$:
\begin{equation*}
\begin{aligned}
    \hat{f_m}(x_1,x_2;r_1,r_2,r_3) = ( & x_1+r_1, x_2+r_2, r_2x_1+r_3, r_1x_2+r_1r_2-r_3),
\end{aligned}
\end{equation*}
where $r_1,r_2$ and $r_3$ are uniformly chosen random values. The recovery of $f_m(x_1,x_2)$ can be accomplished by computing $c_1 \cdot c_2 - c_3 - c_4$ where $c_1=x_1+r_1$, $c_2=x_2+r_2$, $c_3=r_2x_1+r_3$ and $c_4=r_1x_2+r_1r_2-r_3$. The simulation of $\hat{f_m}$ can be done perfectly by the simulator $\mathsf{Sim}(y;a_1,a_2,a_3) := (a_1,a_2,a_3,a_1a_2-y-a_3)$ where $a_1, a_2$ $a_3$ are random values.
\end{definition}

\subsubsection{Framework}
To perform the private gaze estimation task in our scenario, we inspire from the framework as in \cite{unal2019framework} due to its efficiency compared to other approaches in the literature. The framework is proposed to compute the addition or multiplication of the input values of two input-parties in the function-party by utilizing randomized encoding. We utilize the multiplication operation over the eye landmark vectors to compute the dot product of these vectors to obtain kernel matrix of the \acs{SVR} in a privacy-preserving way.

\begin{figure}[ht]
  \centering
   \includegraphics[width = \linewidth]{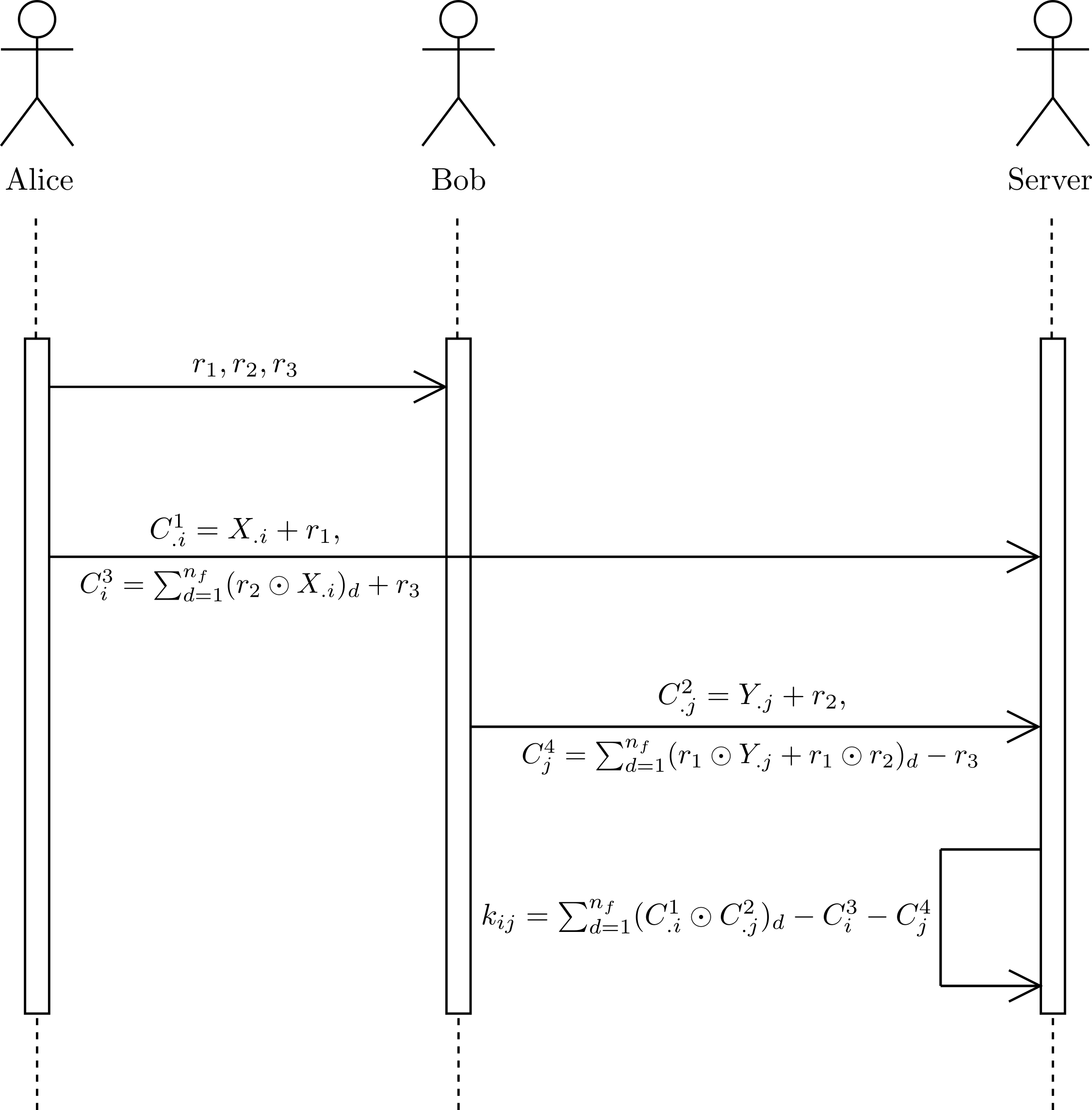}%
  \caption{Overall protocol execution.}
  \label{fig:FlowOfArchitecture_ETRA20}
\end{figure}

We have two input-parties as Alice and Bob, having the eye landmark data as $X \in \mathbb{R}^{n_f \times n_a}$ and $Y \in \mathbb{R}^{n_f \times n_b}$ where $n_a$ and $n_b$ represent the number of samples in Alice and Bob, respectively, and $n_f$ is the number of features. In addition to the input-parties, there exists a server that trains a model on the data of the input-parties. $A_{.j}$ for any matrix $A$ represents the $j$-th column of the corresponding matrix and ''$\odot$`` represents the element-wise multiplication of the vectors. As a first step, Alice creates a uniformly chosen random value $r_3 \in \mathbb{R}$ and two vectors $r_1, r_2 \in \mathbb{R}^{n_f}$ with uniformly chosen random values, which are used to encode the element-wise multiplication of the vectors and shares them with Bob. Afterwards, Bob computes $C^2_{.j}=Y_{.j}+r_2$ and $C^4_j= \sum_{d=1}^{n_f} (r_1 \odot Y_{.j} + r_1 \odot r_2)_d - r_3$,  $\forall j \in \{1,\cdots,n_b\}$ where $C^2 \in \mathbb{R}^{n_f \times n_b}$ and $C^4 \in \mathbb{R}^{n_b}$. Meanwhile, Alice computes $C^1_{.i}=X_{.i}+r_1$ and $C^3_{i}=\sum_{d=1}^{n_f}(r_2 \odot X_{.i})_d + r_3$, $\forall i \in \{1,\cdots,n_a\}$ where $C^1 \in \mathbb{R}^{n_f \times n_a}$ and $C^3 \in \mathbb{R}^{n_a}$. Input-parties send their share of the encoding to the server with the gram matrix of their samples, which is the dot product among their samples. Then, the server computes the dot product between samples of Alice and Bob to complete the missing part of the gram matrix of all samples. To achieve this, the server computes $k_{ij} = \sum_{d=1}^{n_f}(C^1_{.i} \odot C^2_{.j})_d - C^3_{i} - C^4_{j}$, $\forall i \in \{1,\cdots,n_a\}$ and $\forall j \in \{1,\cdots,n_b\}$ where $k_{ij}$ is the $i$-th row $j$-th column entry of the gram matrix between the samples of the input-parties. Once the server has all components of the gram matrix, it constructs the complete gram matrix $K$ by simply concatenating the parts of it. In our solution, Alice and Bob send to the server $(C^1, C^3)$ and $(C^2, C^4)$ tuples, respectively. These components reveal nothing but only the gram matrix of the samples after decoding. Furthermore, the input-parties shuffle their raw data before the computation to avoid the possibility of private information leakage such as the behavior of the person due to the nature of the visual sequence information. The overall flow is summarized in Figure~\ref{fig:FlowOfArchitecture_ETRA20}.

After having the complete gram matrix for all samples that Alice and Bob have, the server uses it as a kernel matrix as if it was computed by the linear kernel function on pooled data. Additionally, it is also possible to compute a kernel matrix as if it was computed by the polynomial or radial basis kernel function (RBF) by utilizing the resulting gram matrix. As an example, the calculation of \acs{RBF} from the gram matrix is as follows.

\begin{equation*} \label{rbfdotproduct}
K(x,y) = \exp\Bigg(-\dfrac{\norm{x \cdot x - 2 x \cdot y + y \cdot y}^2}{2 \sigma^2}\Bigg),
\end{equation*}

where ``$\cdot$'' represents the dot product of vectors, which is possible to obtain from the gram matrix, and $\sigma$ is the parameter utilized to adjust the similarity level. Once the desired kernel matrix is computed, it is possible to train an \acs{SVR} model by employing the computed kernel matrix to estimate the gaze. In the process of the computation of the dot product, the amount of data transferred among parties is $(n_fn_a + n_fn_b + n_a + n_b + 2n_f) \times d$ bytes where $d$ is the size of one data unit.

\subsection{Security Analysis}
A semi-honest adversary who corrupts any of the input-parties cannot learn anything about the private inputs of the other input-party. During the protocol execution, two vectors of random values and a single random value are sent from Alice to Bob. The views of the input-parties consist only of vectors with random values. Using these random values, it is not possible for one party to infer something about the other party's private inputs~\cite{unal2019framework}.

\begin{theorem}
A malicious adversary $\mathcal{A}$ corrupting the function-party learns nothing more than the result of gram matrix. It is computationally infeasible for $\mathcal{A}$ to infer any information about the input-parties’ data $X$ and $Y$ as long as Perfect \acs{RE} multiplication is semantically secure (Definition~\ref{def:remul}). 
\end{theorem}

\begin{proof}
We first show the correctness of our solution. We assume $n_f = 2$ and encode the function $f_d(x,y) = x_1y_1+x_2y_2$ over some finite ring $\mathsf{R}$ by the following \acs{DARE}:

\begin{equation*}
\begin{aligned}
    \hat{f_d}(x,y;r) = ( & x_1+r_{11}, y_1+r_{12}, x_2+r_{21}, y_2+r_{22}, \\
    & r_{12}x_1+ r_{22}x_2+r_3, \\ 
    & r_{11}y_1+r_{11}r_{12} + r_{21}y_2+r_{21}r_{22}-r_3)
\end{aligned}
\end{equation*}

Given an encoding $(c_1,c_2,c_3,c_4,c_5,c_6)$, $f_d(x,y)$ is recovered by computing $c_1c_2+c_2c_4+c_5+c_6$.

By the concatenation lemma in \cite{applebaum2017garbled}, we can divide $c_5$ and $c_6$ into $n_f$ shares by using $n_f$ random values instead of a single $r_3$ value.

\begin{equation*}
\begin{aligned}
    \hat{f_d}(x,y;r) = ( & x_1+r_{11}, y_1+r_{12}, r_{12}x_1 + r_{13}, r_{11}y_1+r_{11}r_{12} - r_{13}, \\
    & x_2+r_{21}, y_2+r_{22}, r_{22}x_2+r_{23}, r_{21}y_2+r_{21}r_{22}-r_{23})
\end{aligned}
\end{equation*}

Given an encoding $(c_1,c_2,c_3,c_4,c_5,c_6,c_7,c_8)$,

\begin{equation*}
\begin{aligned}
    \hat{f_m}(x_1,y_1;r) = (c_1,c_2,c_3,c_4)\\
    \hat{f_m}(x_2,y_2;r) = (c_5,c_6,c_7,c_8)
\end{aligned}
\end{equation*}

By the concatenation lemma in \cite{applebaum2017garbled}, $\hat{f_d}(x,y;r) = $ $(\hat{f_m}(x_1,y_1;r), \hat{f_m}(x_2,y_2;r))$ perfectly encodes the function $f_d(x,y)$ if Perfect \acs{RE} multiplication is semantically secure. 

After showing the correctness, we analyze the security with the simulation paradigm. In the simulation paradigm, there is a simulator who generates the view of a party in the execution. A party’s input and output must be given to the simulator to generate its view. Thus, security is formalized by saying that a party’s view can be simulatable given its input and output and the parties learn nothing more than what they can derive from their input and prescribed output.

The function-party $\mathcal{F}$ does not have any input and output. A simulator $\mathcal{S}$ can generate the views of incoming messages received by $\mathcal{F}$. $\mathcal{S}$ creates four vectors $C^{1'}$,$C^{2'}$,$C^{3'}$,$C^{4'}$ with uniformly distributed random values using a pseudorandom number generator $G'$. Finally, $\mathcal{S}$ outputs $\{C^{1'}$,$C^{2'}$,$C^{3'}$,$C^{4'}\}$. 

In the execution of the protocol $\pi$, $\mathcal{A}$ receives four messages which are masked with uniformly random values generated using a pseudorandom number generator $G$. The view of $\mathcal{A}$ includes $\{C^1$,$C^2$,$C^3$,$C^4\}$. The distribution over $G$ is statistically close to the distribution over $G'$. This implies that

\[\{\mathcal{S}(C^{1'},C^{2'},C^{3'},C^{4'})\} \overset{c}{\equiv} \{view^\pi_{\mathcal{A}}(C^1,C^2,C^3,C^4)\}\]

\end{proof}

\subsection{Results}
To demonstrate the performance, we conduct experiments on a PC equipped with Intel Core i7-7500U with 2.70 GHz processor and 16 GB memory RAM. We employ varying sizes of eye landmark data, that are $5,000$, $10,000$ and $20,000$ samples of which one-fifth is the test data and we split the data between the input-parties equally. The framework allows us to optimize the parameters of the model in the server without further communicating with the input-parties. Thanks to this, we utilize $5$-fold cross-validation to optimize the parameters, which are the similarity adjustment parameter $\gamma \in \{2^{-3},2^{-2},\cdots,2^{4}\}$ of the Gaussian \acs{RBF} kernel, the misclassification penalty parameter $C \in \{2^{-3},2^{-2},\cdots,2^{3}\}$, and the tolerance parameter $\epsilon \in \{0.005, 0.01, 0.05, 0.1, 0.5, 1\}$ of \acs{SVR}. After parameter optimization, we repeat the experiment on varying sizes of eye landmark data with the optimal parameter set $10$ times to assess the execution time. To evaluate the gaze estimation results, we employ mean angular error in the same way as in \cite{Park2018}. Table~\ref{tab:mae_ETRA20} demonstrates the relationship between the dataset size and the resulting mean angular error. Since no additional noise is introduced during the computation of the kernel matrix, the results from our privacy-preserving framework are the same with the non-private ones. The mean angular errors are lower compared to the state-of-the-art gaze estimation techniques since we use synthetic data and fixed camera position during image rendering.

\begin{table}[ht]
    \caption{The mean angular errors for varying dataset sizes.}
    \label{tab:mae_ETRA20}
    \centering
    \begin{tabular}{ c c }
     \toprule
     \makecell{\# of samples} & \makecell{Mean angular error} \\
     \midrule
     5k & 0.21 \\
     10k & 0.18 \\
     20k & 0.17 \\
     \bottomrule
    \end{tabular}
\end{table}

\begin{figure}[ht]
    \centering
    \subfigure[Execution time of Alice.]{{\includegraphics[width=0.4\linewidth]{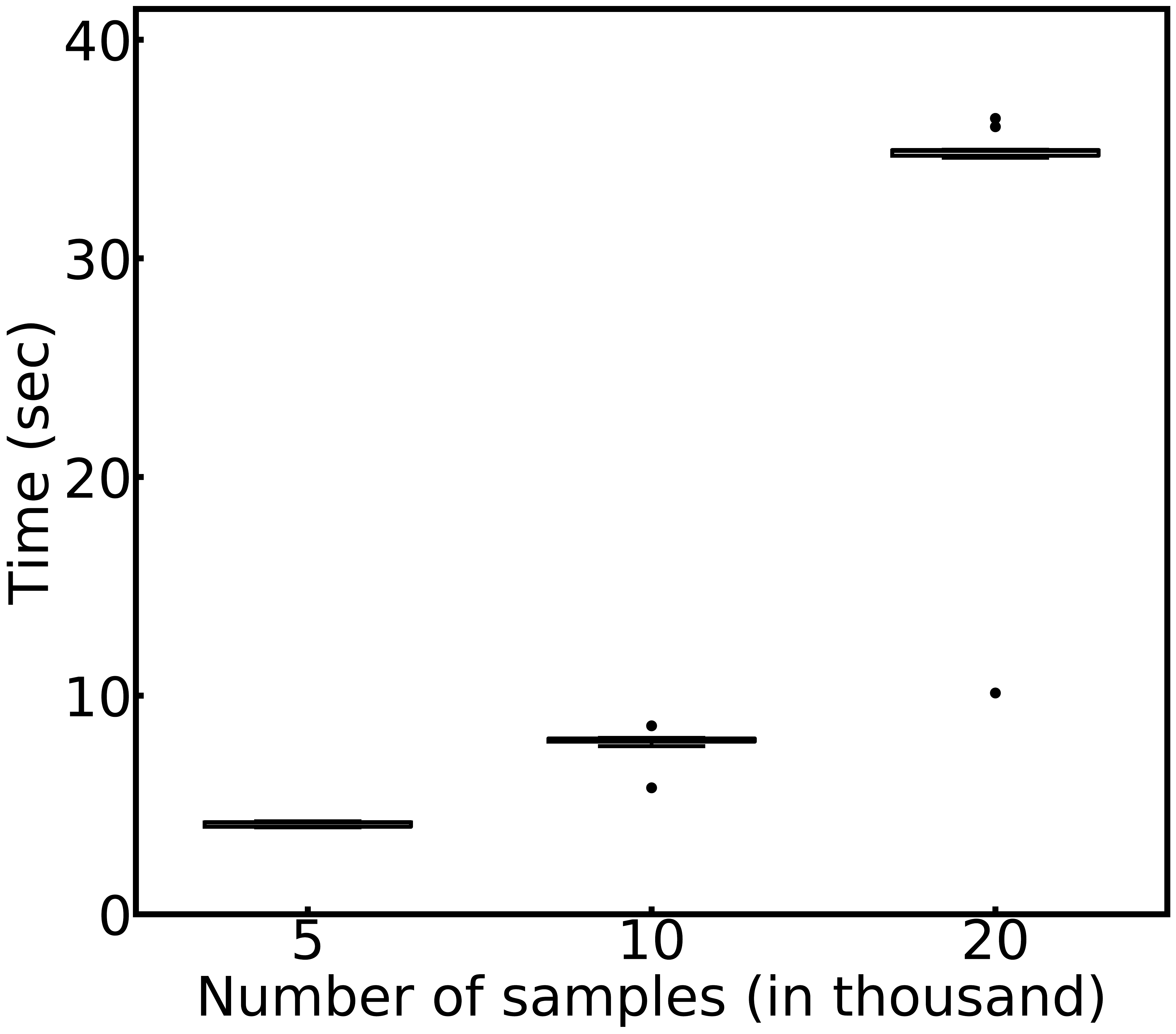}}}
    \hspace{20pt}
    \subfigure[Execution time of Bob.]{{\includegraphics[width=0.4\linewidth]{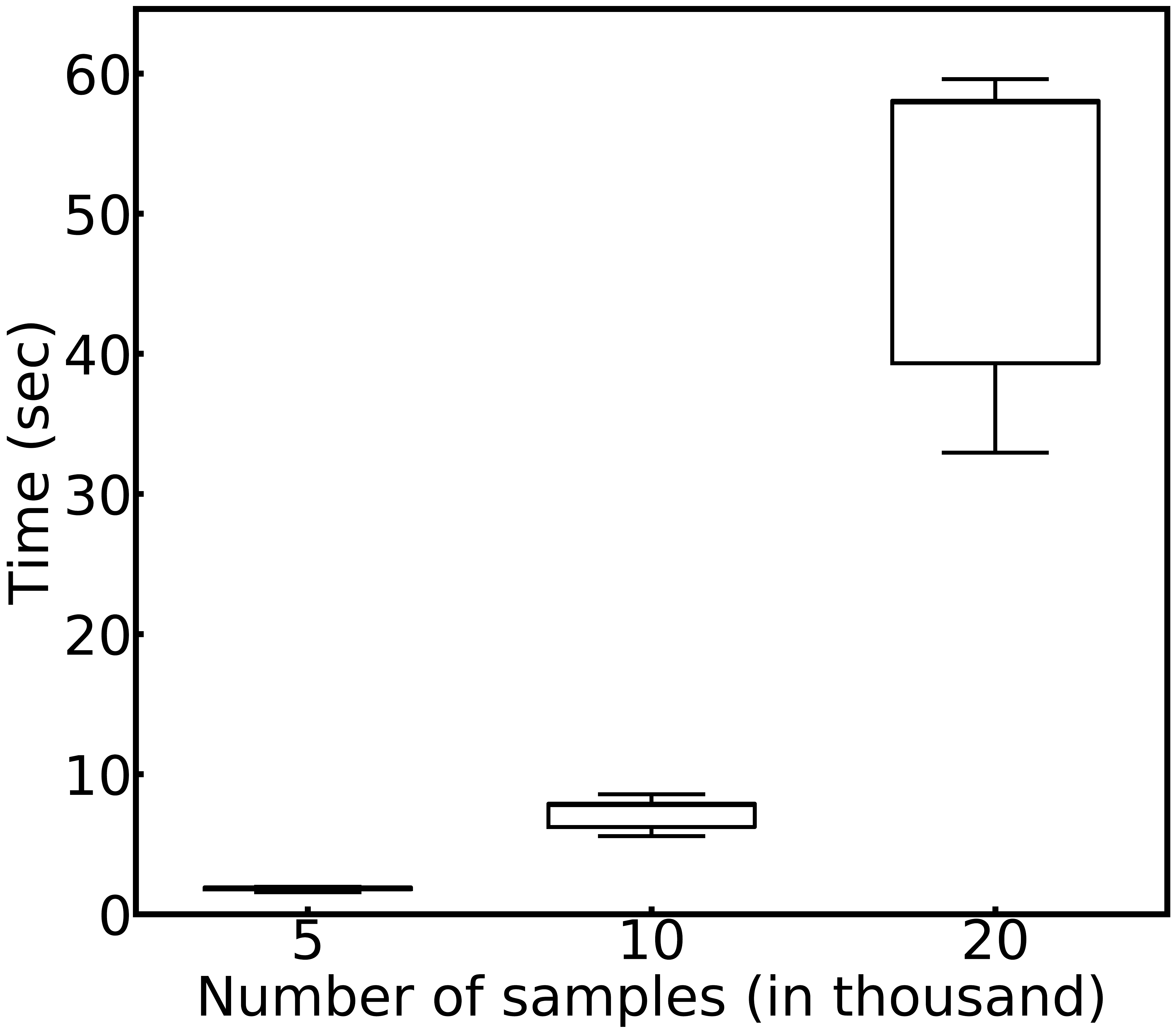}}}
    \hspace{20pt}
    \subfigure[Execution time of Server.]{{\includegraphics[width=0.4\linewidth]{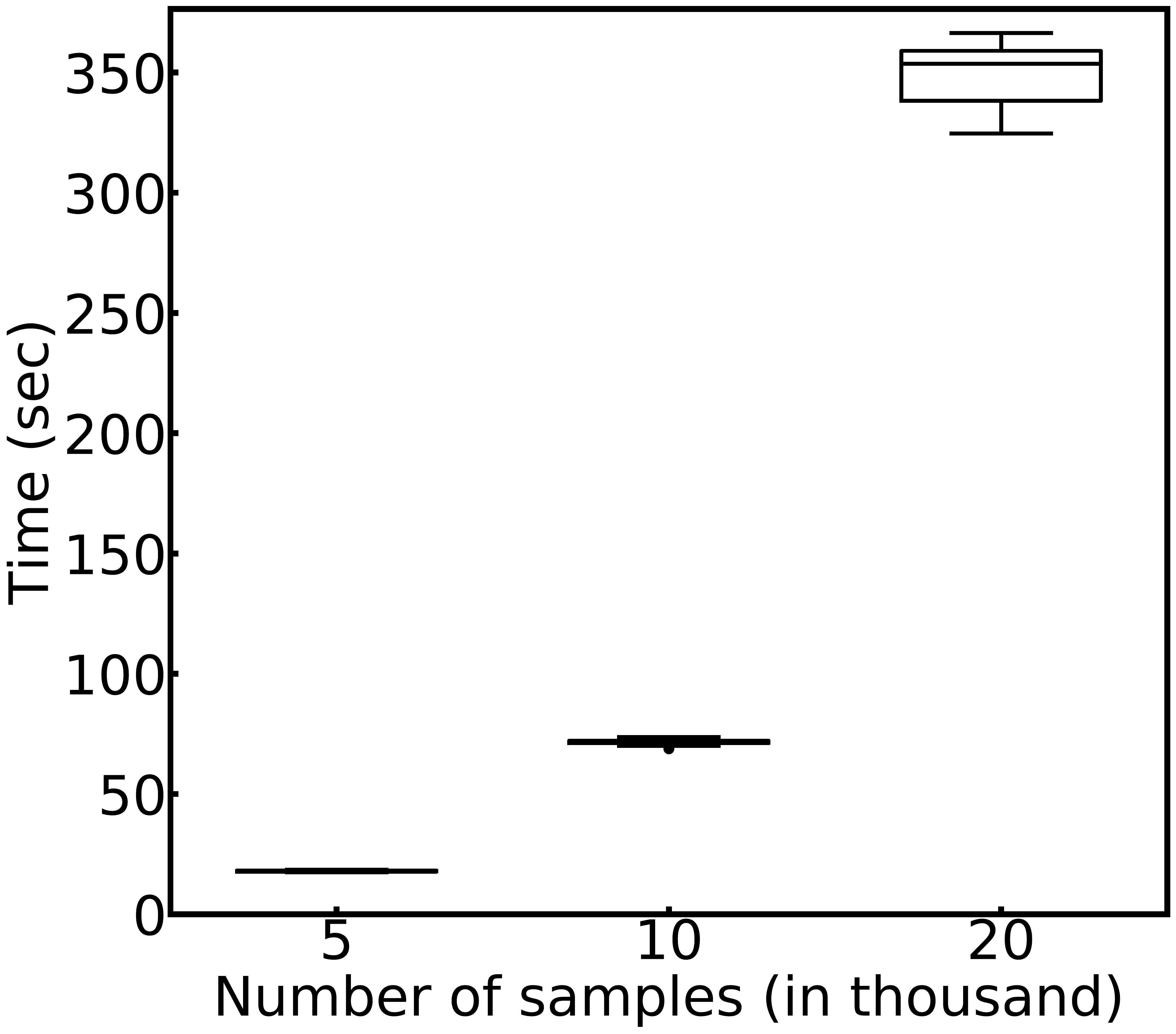}}}
    \hspace{20pt}
    \subfigure[Prediction time of Server.]{{\includegraphics[width=0.4\linewidth]{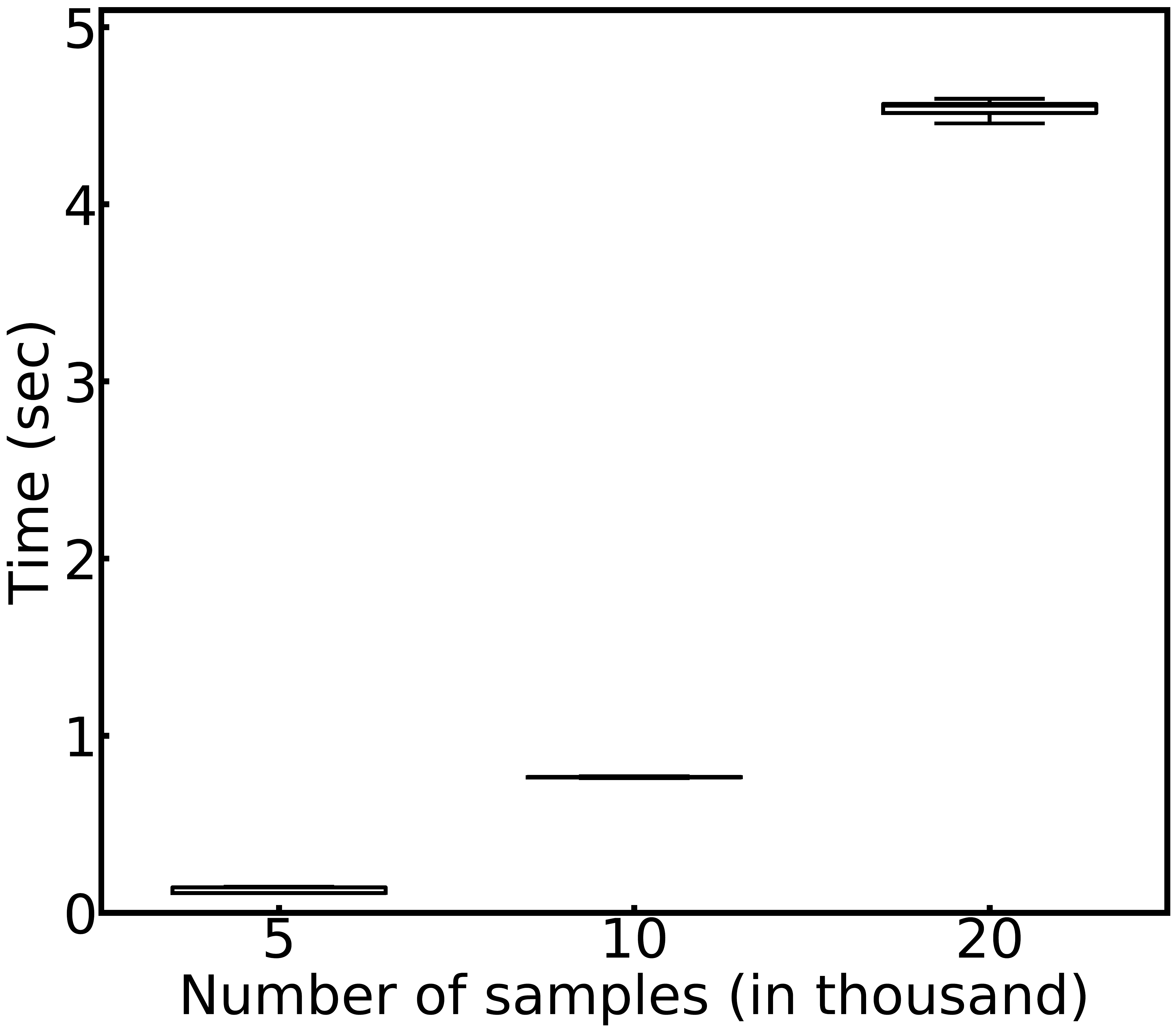}}}
    \caption{The execution time of (a) Alice, (b) Bob and (c) the server are given. We also demonstrate (d) the time required for the prediction of the test samples, which are $20\%$ of the total number of samples in each case.}
    \label{fig:exe_time_ETRA20}
\end{figure}

The amount of time to train and test the models increases as the sample sizes increase since computation requirements get larger. The increment in the dataset size increases the communication cost among parties. The execution times of all parties for $10$ runs with the optimal parameters are shown in Figure~\ref{fig:exe_time_ETRA20}. We also demonstrate the amount of time to predict the test samples, which corresponds to one-fifth of the total number of samples to emphasize the real-time working capabilities. In the experiment with $20,000$ samples, for instance, we spend $\approx4.5$ seconds to predict $4,000$ test samples, which corresponds to $1.125$ ms per sample. When the current sampling frequencies of eye trackers are taken into consideration, it is possible to deploy and use the framework to estimate gaze if an optimized communication between the parties is established.

\subsection{Conclusion}
In this work, we utilized a framework based on randomized encoding to estimate human gaze in a privacy-preserving way and in real-time. Our solution can provide improved gaze estimation if input-parties want to use each other's data for different reasons such as to account for genetic structural differences in the eye region. None of the input-parties has the access to the eye landmark data of the others or the result of the computation in the function party, while the function-party cannot infer anything about the data of the input-parties. Temporal information of the visual scanpath, pupillary, or blinks cannot be reconstructed due to the shuffling of the data, and lack of sensory information and direct access to the eye landmarks. Our solution works in real-time, hence it could be deployed along with \acs{HMD}s for different use-cases and extended to similar eye tracking related problems if similar amount of features is used. To the best of our knowledge, this is the first work based on function-specific privacy models in the eye tracking domain. The number of parties is a limitation of our solution. Thus, as future work we will extend our work to a larger number of parties.

\chapter{Accessibility of Eye Tracking in VR in Daily Setups}
\label{appendix_C}

This chapter includes the following publication:
\vspace{1cm}
\begin{enumerate}
	\item\label{appendix_AIVRpaper_label} \textbf{Efe Bozkir}, Shahram Eivazi, Mete Akgün, and Enkelejda Kasneci. Eye tracking data collection protocol for VR for remotely located subjects using blockchain and smart contracts. In~\emph{2020 IEEE International Conference on Artificial Intelligence and Virtual Reality (AIVR) Work-in-progress papers}, New York, NY, USA, 2020. IEEE. doi: 10.1109/AIVR50618.2020.00083.
\end{enumerate}

\blfootnote{
\hspace{-14pt}{\scriptsize Publication is included with minor templating modifications. Definitive version is available via digital object identifier at the relevant venue. The publication is \textcopyright~2020 IEEE. Reprinted, with permission, from \ref{appendix_AIVRpaper_label}. In reference to IEEE copyrighted material which is used with permission in this thesis, the IEEE does not endorse any of University of Tübingen’s products or services. Internal or personal use of this material is permitted. If interested in reprinting/republishing IEEE copyrighted material for advertising or promotional purposes or for creating new collective works for resale or redistribution, please go to \url{http://www.ieee.org/publications_standards/publications/rights/rights_link.html} to learn how to obtain a License from RightsLink. If applicable, University Microfilms and/or ProQuest Library, or the Archives of Canada may supply single copies of the dissertation.}
}

\newpage

\section[Eye Tracking Data Collection Protocol for VR for Remotely Located Subjects using Blockchain and Smart Contracts]{Eye Tracking Data Collection Protocol for VR for Remotely Located Subjects using Blockchain and Smart Contracts}
\label{appendix:C1}

\subsection{Abstract}
Eye tracking data collection in the virtual reality context is typically carried out in laboratory settings, which usually limits the number of participants or consumes at least several months of research time. In addition, under laboratory settings, subjects may not behave naturally due to being recorded in an uncomfortable environment. In this work, we propose a proof-of-concept eye tracking data collection protocol and its implementation to collect eye tracking data from remotely located subjects, particularly for virtual reality using Ethereum blockchain and smart contracts. With the proposed protocol, data collectors can collect high quality eye tracking data from a large number of human subjects with heterogeneous socio-demographic characteristics. The quality and the amount of data can be helpful for various tasks in data-driven human-computer interaction and artificial intelligence.

\subsection{Introduction}
Over past decades, head-mounted display (HMD) technologies have taken advantage of innovations from imaging and eye tracking research to improve image quality and utility of user interfaces. To date, several consumer level \acs{HMD}s have integrated eye trackers, providing opportunity for researchers to collect eye movement data for user behavior analysis and data-driven interaction.

In the virtual reality (VR) context, it has been shown that eye tracking is helpful for assessing human attention~\cite{bozkir_vr_attention_et}, detecting human stress~\cite{stress_vr}, assessing cognitive load~\cite{bozkir2019person}, predicting human future gaze locations~\cite{8998375}, supporting evaluation and diagnosis of diseases~\cite{7829437}, motion sickness detection~\cite{8642906}, foveated rendering~\cite{Arabadzhiyska2017, 9005240}, continuous authentication~\cite{Zhang:2018:CAU:3178157.3161410}, gaze-based interaction~\cite{VRPursuits_interaction}, training~\cite{8448290}, and redirected walking~\cite{redirected_walking_steinicke}. Many of these tasks are data-driven and require a large quantity of eye tracking data which are usually collected in laboratory settings. Subjects are frequently compensated with some amount of money or gifts for their participation. Two drawbacks of these settings are the lack of heterogeneity in socio-demographic characteristics of data collected subjects and potential for unnatural behaviors of subjects due to the constraints of the laboratory settings. While \acs{VR} is a unique and controlled environment and requires dedicated hardware such as \acs{HMD}s, as personal usage of such devices increases, we foresee that it should be possible to collect data from remotely located subjects, i.e., at their homes. Especially in situations such as COVID-19, this possibility could help experimental works continue in a remote setting. Currently, for crowd-sourcing or similar purposes, platforms such as Amazon Mechanical Turk\footnote{\url{https://www.mturk.com/}} are used. While it is not possible to collect \acs{VR} data with such platforms, for other types of data collection significant compensations are paid to manage the remote subjects' work. In addition, these third-party platforms store and manage data. In fact, as eye tracking and movement data represent unique information about the subjects, the data manipulation possibility of the third parties should be prevented. Third parties should only act as a bridge between the data collector and the subjects in case there is no direct communication between the parties.

To overcome the disadvantages of the laboratory setting and enable remotely located subject participation in eye tracking experiments in the \acs{VR} context, we propose a blockchain-based protocol on the Ethereum blockchain using smart contracts, where we use the blockchain for validation of data integrity and smart contract for compensation management. For this study, we focus mainly on collecting eye tracking data in \acs{VR} environments as many modern \acs{HMD}s come with integrated eye trackers. This means that subjects do not need any additional effort to integrate any sensor into their setup. It is relatively easier to control environmental configurations in \acs{HMD}s when compared to other setups such as illumination and light-sources which may affect subject behaviors or eye movement patterns. However, the proposed protocol can also be used in similar setups as long as identical experiment configurations are guaranteed.

While the first prominent usage of the blockchains is Bitcoin~\cite{bitcoin_whitepaper} and most of the applications are in the financial domain, blockchains also draw attention of the human-computer interaction (HCI), eye tracking, and \acs{VR} communities. Opportunities and challenges for the \acs{HCI} and interaction design and the role of \acs{HCI} community were discussed in ~\cite{foth_blockchain_for_interaction_design} and ~\cite{making_sense_blockchain_hci}, respectively. An augmented reality (AR)-based cryptocurrency wallet was developed in~\cite{crypto_ar_wallet} to familiarize users with blockchain wallet services. In addition, GazeCoin is a cryptocurrency for \acs{VR}/\acs{AR} which is exchanged between content makers, advertisers, and the users~\cite{gaze_coin_white_paper}. Apart from the financial use-cases, due to their immutability blockchains are used as notary. Additionally, Ethereum platform brings the smart contract~\cite{Szabo_1997} concept to the blockchains~\cite{eth_white_paper}. One of the straightforward usages of smart contracts is escrow services. For the remote purchase of goods, buyer and seller parties use the smart contracts without trusting one another and a trusted centralized party during the escrow. The smart contracts that are deployed on the blockchains distribute the money to the parties once buyer and seller parties fulfill their obligations in the remote purchase. In our protocol, we treat recorded eye tracking data as digital good so that compensation distribution is done by the smart contracts. To assure that the recorded data are not altered by the subjects, the hash of the recorded data using white-box cryptography~\cite{whitebox_crypto_alex_biryukov} is stored in the blockchain, which enables the blockchain as a notary for data integrity. Our major contributions are as follows. 

\begin{itemize}
    \item A blockchain-based eye tracking data collection protocol for remotely located subjects that can be used for eye tracking experiments in \acs{VR}, which presents the opportunity to collect data from a various number of subjects.
    \item Delegation of mutual trust issues for compensation management and integrity of the recorded eye tracking data to smart contracts and blockchains, respectively.
    \item Elimination of the centralized third parties for compensation management, data collection and manipulation, which is optimal from a privacy perspective.
\end{itemize}

\subsection{Preliminary Definitions}
As our protocol consists of interdisciplinary work from different domains such as virtual reality, blockchains, and cryptography, we provide some definitions that are used throughout the paper.

\textbf{Blockchain~\cite{bitcoin_whitepaper}:} An immutable ledger that consists of a chain of blocks that keeps records of transactions, maintained by several machines in a peer-to-peer network. Each block consists of a timestamp, transaction data, and the cryptographic hash of the previous block. As each block consists of the cryptographic hash of the one prior, immutability is automatically preserved unless one party has the majority of the computational power.

\textbf{Ethereum~\cite{eth_white_paper}:} Public, open-source, blockchain-based, and smart contract supporting distributed platform. 

\textbf{Ether (ETH)~\cite{eth_white_paper}:} The cryptocurrency of the Ethereum platform.

\textbf{Smart Contract~\cite{eth_white_paper}:} A self-executing, irreversible, and transparent contract between buyer and seller, implemented in the code.

\textbf{White-box cryptography~\cite{wyseur_whitebox}:} ``Software protection technology which allows for the application of cryptographic operations without revealing any critical information such as secret keys.''

\subsection{Protocol}
In this section, we discuss our protocol and its flow, assumptions, and details of the implementation.

\subsubsection{Flow}
\label{subsection:flow}
Our proposed protocol consists of two parties as data collector and subjects. The data collector is responsible for providing the \acs{VR} application for eye tracking data collection and subjects are tasked with carrying out the experiment and providing the recorded eye tracking data. At the end of a valid experiment, subjects are compensated for their participation. Let us assume that each subject is compensated with $X$ unit of \acs{ETH} for the valid data recorded from an experiment session. A relevant amount can be set for compensation depending on the experiment.

Figure~\ref{fig:protocol_AIVR20} shows the overall flow and short descriptions of each step of the protocol. As the \textbf{step 1}, subjects fetch the application from the data collector and carry out the experiment. While the content of the stimuli changes depending on the use-case, the \acs{VR} application validates the eye tracking data quality at the end of each experimental session by using tracking rates or confidence intervals that are provided by the eye tracker. If the recorded eye tracking data are too noisy, subjects are not supposed to send the data to the data collector, where they are informed by the \acs{VR} application. This obligation forces the subjects to follow the instructions of the \acs{VR} experiment, such as eye tracker calibration, carefully while the data collector obtains better quality data in the end. After the validation success, the \acs{VR} application calculates the hash output of the recorded eye tracking data and saves it. Saving the hash output is required for assessing the data integrity; however, adversarial subjects can easily find out the hashing algorithm using the executable of the \acs{VR} application on their own devices. Therefore, we opt for a white-box~\cite{wyseur_whitebox,whitebox_crypto_alex_biryukov} paradigm for calculating the data hash. In the white-box paradigm, the adversary is supposed to have visibility of the inputs, outputs, and other intermediate steps. White-box cryptography achieves protection of confidential information such as secret keys while keeping the application semantically the same. Even if adversaries infer the hash function, due to the lack of secret key, it is not possible to generate a hash output for altered data. Consequently, subjects are obliged to behave honestly, where honest behavior means not altering the recorded data. In the end of the first step, once the recorded data is validated and hash value is saved, the subjects are informed by the \acs{VR} application that the recorded eye tracking data is reportable.

\begin{figure*}[!ht]
  \centering
   \includegraphics[width=\linewidth]{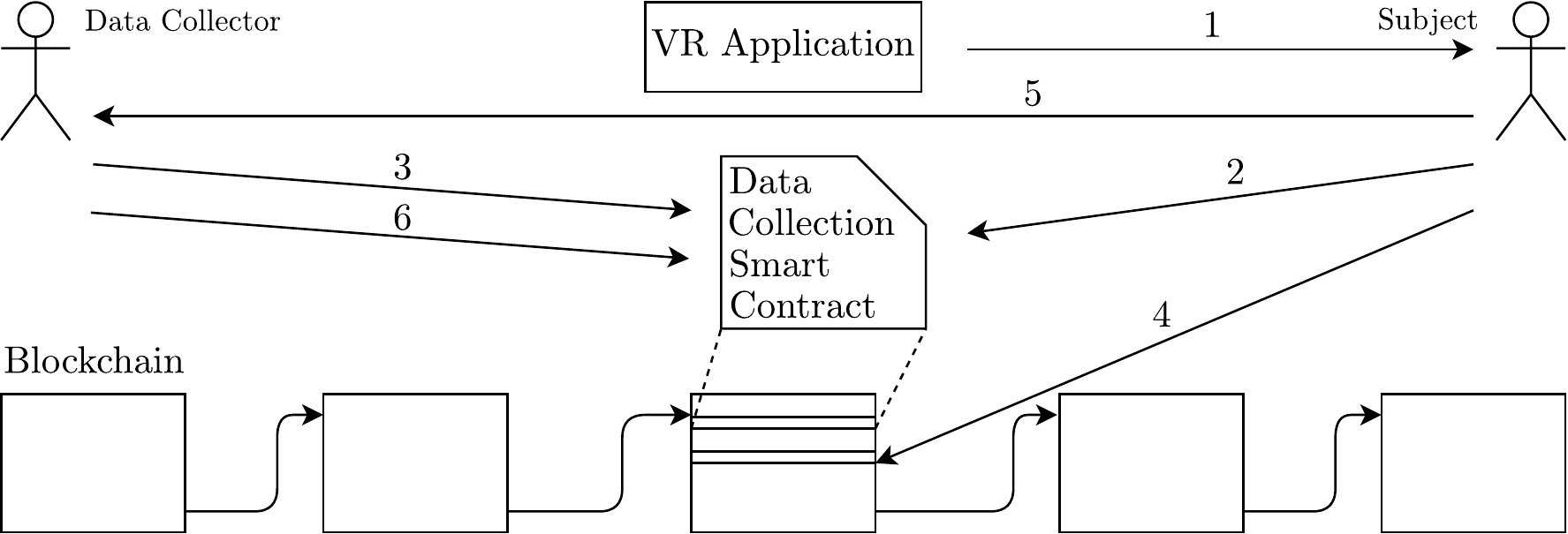}
  \caption{Blockchain-based protocol and its steps. ($1$) Subject fetches the application and carries out the experiment. ($2$) Subject initiates the smart contract. ($3$) Data collector confirms the contract creation and stakes. ($4$) Subject stores the recorded data hash in blockchain. ($5$) Subject transfers the recorded data to the data collector. ($6$) Data collector confirms the data collection.}
  \label{fig:protocol_AIVR20}%
\end{figure*}

As the \textbf{step 2}, the subjects initiate the smart contract and stake double the amount of compensation, which is $2X$ \acs{ETH} for our case. Staking double the amount of compensation that they will obtain from the smart contract forces subjects to act honestly; otherwise, they lose the amount that they stake. As the \textbf{step 3}, the data collector confirms the data collection and stakes the same amount as the subject, which is $2X$ \acs{ETH} to the smart contract. While the compensation is $X$ \acs{ETH} per experiment, the data collector is supposed to stake double the amount of compensation so that it also becomes an obligation to behave honestly. Otherwise, the doubled amount of compensation will be lost without obtaining the recorded data. As the \textbf{step 4}, the subjects store the hash output that is reported by the \acs{VR} application in the blockchain and, as the \textbf{step 5}, they send the recorded data along with the transaction hash of the transaction for storing the data hash in the blockchain to the data collector. If subjects try to alter the data, the hash in the blockchain and the altered data will not match and it will be discovered by the data collector. As the \textbf{step 6}, the data collector obtains the recorded eye tracking data and transaction hash of the data hash and checks whether or not the obtained data and the hash provided by the subjects overlap using the hash function that is implemented in the \acs{VR} application and secret keys. If the reported data and hash value stored in the blockchain overlap, the data collector confirms the smart contract and that the obtained data are valid. Then, the smart contract automatically distributes $3X$ and $X$ \acs{ETH} to the subject and the data collector, respectively. In the end, each subject earns $X$ unit of \acs{ETH} for participation in the experiment, where the data collector obtains the recorded eye tracking data. Due to the immutable nature of blockchains and smart contracts, none of the parties can alter the values in the blockchain and behave as an adversary.

In the protocol, as both parties stake more than the amount they are supposed to spend or earn, they have to act honestly in order to achieve successful data collection and compensation distribution, otherwise data collection is not finalized and parties lose the amount they stake. In particular, the subjects have to stake double the amount of compensation that they will receive whereas the data collectors have to stake double the amount of compensation that they will give. Since the smart contracts are immutable and stored in the blockchain, a third-party application is not needed for compensation distribution or data manipulation, which is useful from a privacy preservation point of view.

\subsubsection{Assumptions}
We have three main assumptions in our protocol. Firstly, validation of the quality of the recorded eye tracking data is automatically completed by the \acs{VR} application at the end of each experiment by using metrics such as tracking ratio or confidence levels reported by the eye tracker. Due to poor calibration for eye tracking, removal of the head-mounted display (HMD) in the middle of experiment, or similar reasons, recorded eye tracking data may have an extensive amount of noise level. Instead of cleaning data offline extensively after the experiments, our protocol assumes that data validity is checked at the end of each experiment by the \acs{VR} application and the application informs the subjects whether the quality of the data is valid and reportable.

Secondly, the recorded eye tracking data is hashed using white-box cryptography and stored at the end of the experiment by the \acs{VR} application to be stored in the blockchain for validation of the data integrity. In traditional eye tracking experiments, subjects participate in the experiments on the devices that are provided by the data collectors. However, in the remotely located subject participation, subjects run the applications on their own devices. Therefore, they have direct access to the provided application and if any adversarial subject analyzes the binary implementation of an application that does not use white-box paradigm, they can easily infer the used hash function and generate hash output for fake data. On the contrary, when using white-box cryptography, the secret keys are not leaked even if adversaries analyze the binary implementation. Even if an adversary infers the hash function, a hash output for fake data cannot be generated without secret keys. Therefore, white-box paradigm is used by the \acs{VR} application. If subjects alter the recorded data or send fake data to the data collector, the generated hash value will not match the recorded data, which leads subjects to lose their staked compensation in the smart contract.

Lastly, as our protocol does not use any centralized third party, a secure direct communication is needed for exchanging the application and the recorded data between the data collector and subjects. In case it is not available, a bridging third party only for communication purposes can be implemented.

\subsubsection{Implementation}
We select the Ethereum platform for our proof-of-concept due to its public blockchain, relatively higher number of nodes, and status as one of the most mature platforms in the blockchain domain. However, any blockchain-based platform that supports smart contracts can be opted in.

We implement the blockchain related part of the protocol, particularly the steps $2$, $3$, $4$, and $6$ discussed in Section~\ref{subsection:flow}, using Solidity\footnote{\url{https://docs.soliditylang.org/}} and a simple purchase smart contract~\cite{ng_smart_contract} on the Ropsten Testnet of the Ethereum platform. In the beginning of the data collection, both the data collector and the subject hold $1$ \acs{ETH} in their wallets. We select the compensation amount as $0.025$ \acs{ETH}. For the calculation of the hash output of the recorded eye tracking data, we use synthetic data; however, any eye tracker integrated to modern \acs{HMD}s can be used in a real-world implementation. The hash value of the data is calculated using Keyed-Hashing for Message Authentication (HMAC)~\cite{hmac_97} and Secure Hash Algorithm$3$-$512$ (SHA$3$-$512$)~\cite{sha3_standard} as it is possible to have white-box implementation of the \acs{HMAC}. The calculated hash value is stored in the input data field of a self transaction from the subject. After the protocol execution, the data collector and the subject hold $\approx 0.975$ and $\approx 1.025$ \acs{ETH} when the transaction fees are subtracted, respectively. The smart contract, overall procedure, the data collector, and the subject parties are available on the Ropsten Testnet via following link: \url{https://ropsten.etherscan.io/address/0x0e937a4a4618dd8d5a12ec4a9f8fd61d6bfd13e4}.

In the above link, there are three transactions in chronological order that correspond to steps $2$, $3$, and $6$ of our protocol. The subject (address starting with $0x89$) and the data collector (address starting with $0x44$) of our implementation are available in the source of the first and the second transactions of the smart contract, respectively. There are three transactions in the subject address. The first and second transactions are for depositing the test \acs{ETH} and initiating the smart contract, respectively. The third transaction in the subject address is a self transaction and corresponds to step $4$ of our protocol. In the ``Input Data'' field of the self-transaction, the calculated data hash is available.

\subsection{Conclusion and Discussion}
We proposed a blockchain-based protocol for collecting eye tracking data in \acs{VR} from remotely located subjects. As eye tracking experiments are usually conducted in laboratory settings with a limited number of subjects from similar backgrounds in terms of socio-demographic characteristics, it is a challenge to draw generic data-driven conclusions. Due to the laboratory settings, subjects may not behave naturally. While our protocol overcomes the drawbacks of the traditional eye tracking data collection setups without needing a centralized third party for data collection and compensation management, it also creates an opportunity to carry out the data collection anonymously, which is optimal for the privacy of subjects. We focused on the eye tracking data collection in \acs{VR} setups as validation of the eye tracking data and generation of the controlled environments with \acs{VR} can be done easily. In addition, current availability of eye tracker integrated \acs{HMD}s in the consumer market supports our protocol for \acs{VR} and eye tracking data; however, the proposed protocol may be useful for other types of eye trackers, sensors, or environments as long as identical configurations between subjects can be generated. In contrast to traditional eye tracking experiments, subject consent, additional questionnaire, or similar information should be collected digitally using our protocol. Our protocol may also require an application-level effort to have one-to-one mapping between subjects and experiments.

As future work, we plan to have an end-to-end implementation of our protocol along with a real \acs{VR} application and \acs{HMD}-integrated eye tracker. In addition, while transactions are applied anonymously on the public blockchains, it is possible to track them. Recent work on eye tracking, \acs{HCI}, and \acs{VR}~\cite{bozkir_ppge, steil_diff_privacy, fuhl2020reinforcement, sumer2020automated, bozkir2020differential} emphasize the importance of privacy preservation. Combining privacy-preserving methods with our protocol remains as part of future work.

\subsection*{Acknowledgments}
E.B. thanks Batuhan Sar{\i}o\u{g}lu for useful discussions on blockchains.

\backmatter

\cleardoublepage
\phantomsection
\bibliographystyle{unsrtnat}
\bibliography{thesis}

\end{document}